\PassOptionsToPackage{table}{xcolor}
\documentclass[11pt,onecolumn,a4paper,accepted=2026-09-15]{quantumarticle}
\pdfoutput=1
\usepackage[T1]{fontenc}
\usepackage{amsmath,amssymb}
\usepackage{tabularx}
\usepackage{makecell}

\usepackage{verbatim}
\usepackage{subcaption}

\usepackage{algorithm}
\usepackage{algpseudocode}
\usepackage{amsmath}

\usepackage[braket, qm]{qcircuit}
\usepackage{enumerate}

\usepackage{hyperref}
\hypersetup{
   colorlinks,
   citecolor=blue,
   filecolor=blue,
   linkcolor=blue,
   urlcolor=blue
}

\usepackage{ulem}
\usepackage{mathtools}
\mathtoolsset{showonlyrefs}
\usepackage{amsmath}
\usepackage{graphicx}
\usepackage{enumerate}
\usepackage{amssymb}

\usepackage{tikz-cd}
\usetikzlibrary{external}
\usepackage{xcolor}
\usepackage{colortbl} 
\usepackage{array}
\usepackage{booktabs,caption}

\usepackage{graphicx}
\usepackage{mdframed}
\usepackage{cite}

\usepackage{diagbox}

\usepackage{amsthm,amsmath}
\newcommand{\appref}[1]{\hyperref[#1]{{Appendix~\ref*{#1}}}}
\newcommand{\be}{\begin{eqnarray} \begin{aligned}}
\newcommand{\ee}{\end{aligned} \end{eqnarray} }
\newcommand{\benn}{\begin{eqnarray*} \begin{aligned}}
\newcommand{\eenn}{\end{aligned} \end{eqnarray*}}
\newcommand*{\textfrac}[2]{{{#1}/{#2}}}
\newcommand*{\bbN}{\mathbb{N}}
\newcommand*{\bbR}{\mathbb{R}}
\newcommand*{\bbC}{\mathbb{C}}
 
\newcommand*{\cB}{\mathcal{B}}

\newcommand*{\cE}{\mathcal{E}}
\newcommand*{\cF}{\mathcal{F}}
\newcommand*{\cG}{\mathcal{G}}
\newcommand*{\cH}{\mathcal{H}}

\newcommand*{\cK}{\mathcal{K}}
\newcommand*{\cL}{\mathcal{L}}

\newcommand*{\cN}{\mathcal{N}}

\newcommand*{\cU}{\mathcal{U}}

\newcommand*{\cV}{\mathcal{V}}
\newcommand*{\cW}{\mathcal{W}}

\newcommand*{\cY}{\mathcal{Y}}

\newcommand{\subspacegateerror}[2]{\mathsf{\underline{err}}_{#1}\bigl(#2\bigr)}

\newcommand*{\supp}{\mathrm{supp}}
\newcommand*{\myspan}{\mathrm{span}}

\newcommand*{\fr}[2]{\frac{#1}{#2}}

\newcommand{\bc}{\begin{center}}
\newcommand{\ec}{\end{center}}

\newcommand{\id}{\mathbb{I}}

\newtheorem{theorem}{Theorem}[section]
\newtheorem{lemma}[theorem]{Lemma}

\newtheorem{definition}[theorem]{Definition}

\newtheorem{corollary}[theorem]{Corollary}
\newtheorem{result}{Result}
\newtheorem{nogoresult}[result]{(No-go) Result}

\usepackage{amsfonts}

\def\id{\mathbb{I}}
\def\01{\{0,1\}}

\newcommand{\proj}[1]{|#1\rangle\langle#1|}
\newcommand{\braket}[2]{\langle #1, #2\rangle}

\newcommand*{\gkp}{\mathsf{GKP}}
\newcommand{\gkpcode}[3]{\cal{GKP}_{#1}^{#2}[#3]}

\DeclareMathOperator{\rank}{rank}

\newcommand*{\cZ}{\mathsf{C}Z}

\newcommand*{\Pgate}{\mathsf{P}}
\newcommand*{\Fgate}{\mathsf{F}}

\newcommand{\vertiii}[1]{{\left\vert\kern-0.25ex\left\vert\kern-0.25ex\left\vert #1 
    \right\vert\kern-0.25ex\right\vert\kern-0.25ex\right\vert}}

\newcommand{\posbos}[1]{\ket{\mathbf{#1}}}
\newcommand*{\encmap}{\mathsf{Enc}} 
\newcommand*{\decmap}{\mathsf{Dec}}
\newcommand*{\encmapgkp}{\mathsf{gkpEnc}} 
\newcommand*{\decmapgkp}{\mathsf{gkpDec}}
\newcommand{\encodergkp}[3]{\mathsf{gkpEnc}_{#1}^{#2}[#3]} 
\newcommand{\decodergkp}[3]{\mathsf{gkpDec}_{#1}^{#2}[#3]}

\newcommand*{\cHin}{\cH_{in}}
\newcommand*{\cHout}{\cH_{out}}

\newcommand*{\cLin}{{\cL_{in}}}
\newcommand*{\cLout}{{\cL_{out}}}
\newcommand*{\tr}{\mathsf{tr}}
\newcommand{\encoded}[1]{{\mkern1mu \bf{#1} \mkern1mu}}
\newcommand{\encodedC}[1]{\boldsymbol{\mathcal{#1}}}
\renewcommand{\cal}[1]{\mathcal{#1}}
\newcommand{\bb}[1]{\mathbb{#1}}
\renewcommand{\id}{\mathsf{id}}

\newcommand*{\interiorvertices}{V_{\mathsf{interior}}}
\newcommand*{\inputvertices}{V_{\mathsf{input}}}
\newcommand*{\outputvertices}{V_{\mathsf{output}}}
\newcommand{\inedges}[1]{E_{\mathsf{in}}(#1)}
\newcommand{\outedges}[1]{E_{\mathsf{out}}(#1)}
\newcommand{\indeg}[1]{\mathsf{deg}_{in}(#1)}
\newcommand{\outdeg}[1]{\mathsf{deg}_{out}(#1)}
\newcommand*{\tGKP}{\mathsf{gkp}}

\newcommand*{\gateerror}{\mathsf{err}}

\newcommand*{\CZ}{\mathsf{C}Z}
\newcommand*{\cn}{c}

\begin{document}
\title{Composable logical gate error
in approximate quantum error correction:
reexamining gate implementations in Gottesman-Kitaev-Preskill codes
}
\author{Lukas Brenner}
\author{Beatriz Dias}
\author{Robert Koenig}
\affiliation{Department of Mathematics, School of Computation, Information and Technology, Technical University of Munich, 85748 Garching, Germany}
\affiliation{Munich Center for Quantum Science and Technology, 80799 Munich, Germany}

\maketitle

\begin{abstract}
Quantifying the accuracy of logical gates is  paramount in approximate error correction, where perfect implementations are often unachievable with the available set of physical operations. To this end, we introduce a single scalar quantity we call  the (composable) logical gate error. It captures both the deviation of the logical action from the desired target gate as well as leakage out of the code space. It is subadditive under successive application of gates, providing a simple means for analyzing circuits. We show how to bound the composable logical gate error in terms of  matrix elements of physical unitaries between (approximate) logical computational basis states. In the continuous-variable context, this sidesteps the need for computing energy-bounded norms.

As an example, we study the composable logical gate error for linear optics  implementations of Paulis and Cliffords  in approximate Gottesman-Kitaev-Preskill (GKP) codes. We find that the logical gate error for implementations of Paulis 
depends linearly on the squeezing parameter. This implies that their accuracy improves monotonically with the amount of squeezing. For some Cliffords, however,  linear optics implementations which are exact for ideal GKP codes fail in the approximate case: they have a constant logical gate error even in the limit of infinite squeezing. This is consistent with previous results about the limitations of certain gate implementations for approximate GKP codes, see~\cite{RojkovNogo2024,Cubicphasenogo2021}. It shows that findings applicable to ideal GKP codes do not always translate to the realm of physically realizable approximate GKP codes.
\end{abstract}

\tableofcontents

\section{Introduction}
A central ingredient of current approaches to achieving fault-tolerant quantum computation is the notion of a quantum error-correcting code: Here logical information (e.g., $k$~logical qubits) is  embedded into a subspace of a larger-dimensional physical system (e.g., $n>k$~physical qubits) in such a way that certain physical errors can be detected and/or  corrected. To achieve computation, one additionally needs to identify physical mechanisms operating on encoded information. For example, a logical unitary, i.e., a logical gate acting on the code space, is typically implemented by applying a suitable unitary operation to the physical degrees of freedom. 
Importantly, such implementations should not involve  decoding in order to maintain fault-tolerance. More generally, it is desirable to have implementations with limited error-propagation properties: For example, the physical unitary should ideally be realized by a shallow-depth quantum circuit.

For many stabilizer codes, we have a good understanding of which logical unitaries have implementations with desirable fault-tolerance properties. For example, color codes~\cite{PhysRevLett.97.180501} permit to implement a generating set of Cliffords by transversal, i.e., depth-$1$ circuits. Importantly, an operation realized in this way preserves the code space, and perfectly acts on it by the desired logical Clifford. In contrast, identifying suitable  physical operations to realize logical gates is generally  more involved when using an approximate error-correcting code. Such codes typically arise in settings where  (e.g., for physical reasons), there is a distinguished basis consisting of  states which are only approximately pairwise orthogonal. Canonical examples are physically realizable continuous-variable (CV) codes.

 Consider for example the cat-code encoding a qubit~\cite{PhysRevA.94.042332, PhysRevLett.119.030502,PhysRevLett.111.120501,allopticalcatcode}: For concreteness, suppose that we encode the computational basis states~$\ket{0}$ and~$\ket{1}$  of a logical qubit into the two
 states
 \begin{align}
\left|\phi_\alpha(0)\right\rangle=\frac{|\alpha\rangle+|-\alpha\rangle}{\sqrt{2\left(1+e^{-2|\alpha|^2}\right)}}\qquad\textrm{ and }\qquad \left|\phi_\alpha(1)\right\rangle=\frac{|\alpha\rangle-|- \alpha\rangle}{\sqrt{2\left(1-e^{-2|\alpha|^2}\right)}}\ ,
\end{align}
where~$|\beta\rangle\in L^2(\mathbb{R})$ denotes the coherent state with parameter~$\beta\in\mathbb{C}$ and where we assume that~$\alpha>0$. These two orthogonal states define a~$2$-dimensional code space~$\cL_\alpha=\mathsf{span}\{|\phi_\alpha(z)\rangle\}_{z\in \{0,1\}}\subset L^2(\mathbb{R})$. 
It is therefore clear what we mean, for example, by a logical Pauli gate~$Z$: We would like to implement a map~$\encoded{Z}$ on~$\cL_\alpha$ such that
\begin{align}
\encoded{Z}\ket{\phi_\alpha(z)} &=(-1)^z\ket{\phi_\alpha(z)} \qquad\textrm{ for }z\in \{0,1\}\ .
\end{align}
It is straightforward to check that~$\encoded{Z}$ is 
realized exactly by the phase space rotation~
\begin{align}
W=\exp \left(i \pi a^\dagger a\right)\ ,\label{eq:idealphysical}
\end{align}
see Ref.~\cite{PhysRevA.68.042319}, i.e.,
we have 
\begin{align}
W|_{\cL_\alpha}=\encoded{Z}\ .
\end{align}
In other words, the physical (unitary) operation~$\encoded{Z}$ realizes what can be thought of as a Pauli-$Z$ operator on the code space.

Now suppose that instead of applying~$W$ we use the unitary
\begin{align}
W_\delta=\exp \left(i (\pi+\delta) a^\dagger a\right)\ ,\label{eq:vdelta}
\end{align}
which is identical to~$W$ except for a bit of ``overrotation'' given by a small angle~$\delta>0$. 
Then we clearly have  
\begin{align}
\lim_{\delta\rightarrow 0} \langle \phi_\alpha(z), W_\delta\phi_\alpha(z)\rangle &=(-1)^z\qquad\textrm{ for }z\in \{0,1\}\ 
\end{align}
suggesting that~$W_\delta$ also realizes~$\encoded{Z}$ approximately, with a vanishing error in the limit~$\delta\rightarrow 0$. For any non-zero (small)~$\delta>0$, however, the unitary~$W_\delta$ fails to implement~$\encoded{Z}$ in two significant ways. 
First, we have~$W_\delta\cL_\alpha\not\subseteq\cL_\alpha$, i.e., the unitary~$W_\delta$ does not preserve the code space~$\cL_\alpha$.  This property is sometimes referred to as leakage. Second, even ignoring this fact (i.e., projecting onto~$\cL_\alpha$ and normalizing), the action implemented by~$W_\delta$ does not coincide with that of~$\encoded{Z}$. Our goal is to quantify the deviation of the imperfect implementation~$W_\delta$ from the desired logical action given by the gate~$\encoded{Z}$.

We note that in this simple example, the deviation from the desired functionality (given by~$\encoded{Z}$) arises from the fact that an ill-chosen, i.e., erroneous physical unitary~$W_\delta$ (see Eq.~\eqref{eq:vdelta}) is used instead of the physical unitary~$W$ (see Eq.~\eqref{eq:idealphysical}). We note that in other settings (such as when considering approximate Gottesman-Kitaev-Preskill (GKP) codes, see Section~\ref{sec:idealapproxgkpcode}), the discrepancy between the 
desired logical unitary and the physical operation has an origin more fundamental than mere control imperfections (such as not being able to set~$\delta=0$ in Eq.~\eqref{eq:vdelta}). Instead, in these situations, there is no physically allowed operation (among those considered/available, such as the set of Gaussian unitaries) which perfectly implements the desired action. In other words, the consideration of gate errors is fundamentally unavoidable in such situations even when there is no physical source of noise.

\subsection{Our contribution}

\paragraph{Composable logical gate error.} 
In  practical fault-tolerance settings, the beautiful framework of ideal (e.g., stabilizer) codes often does not accurately reflect the given physical situation. Approximations are unavoidable, and understanding their validity is paramount. 

To quantify the accuracy of an implementation, let~$\cL\subset\cH$ be a code subspace of a physical Hilbert space~$\cH$, and let~$\encoded{U}:\cL\rightarrow\cL$ be a logical unitary gate. Let~$\encodedC{U}:\cB(\cL)\rightarrow\cB(\cL)$, $\encodedC{U}(\rho):=\encoded{U}\rho\encoded{U}^\dagger$  be the associated completely positive trace-preserving (CPTP) map. Consider an (approximate) implementation 
of~$\encoded{\cU}$ by a CPTP map~$\cW_\encoded{U}:\cB(\cH)\rightarrow\cB(\cH)$. We introduce 
a quantity called the composable logical gate error to quantify how well~$\cW_\encoded{U}$ approximates the desired logical gate~$\encoded{U}$ when acting on code states. Let~$\pi_\cL:\cH\rightarrow\cL$ denote the orthogonal projection onto~$\cL$ and let~$\Pi_\cL(\rho)=\pi_\cL\rho \pi_\cL^\dagger$. Then the composable logical gate error of the implementation~$\cW_\encoded{U}$ of~$\encoded{U}$ on~$\cL$ is defined as 

 \begin{align}
    \gateerror_{\cL}(\cW_\encoded{U},\encoded{U}) = \| ( \cW_\encoded{U} - \encodedC{U} )\circ\Pi_\cL \|_\diamond \ ,
\end{align}
where~$\|\cdot \|_\diamond$ denotes the diamond norm.  

We note that the quantity~$\gateerror_{\cL}(\cW_\encoded{U},\encoded{U})$ captures both the leakage out of the code space as well as the incurred logical error. Composability here means that this quantity is subadditive under  concatenations (see item~\eqref{it:composability} below). This also implies that it is subadditive under tensor products (see item~\eqref{it:propertylastcomposableerror} below). In particular, if an upper bound  on the gate error is known for the implementation of each (logical) gate in a circuit, this  immediately  gives an accuracy result for the entire (composed) circuit. That is, we obtain an upper bound on the gate error of an implementation where each (logical) gate in the  circuit is replaced by a corresponding implementation thereof.

 In more detail, the composable logical gate error has the following properties:
\begin{enumerate}[(a)]
\item\label{it:propertyonecomposableerror}
Operational interpretation:  Assume that the implementation~$\cW_\encoded{U}$ is applied to some state~$\rho$ supported on~$\cL$. Then the composable logical gate error~$\gateerror_{\cL}(\cW_\encoded{U},\encoded{U})$ provides an upper bound on the distance between the resulting state~$\cW_{\encoded{U}}(\rho)$  and the state~$\encoded{U}\rho \encoded{U}^\dagger$ obtained by applying the ideal (perfect) logical unitary~$\encoded{U}$.
\item
Definiteness: We have~$\gateerror_{\cL}(\cW_\encoded{U},\encoded{U})=0$ if and only if~$\cW_\encoded{U}$ perfectly implements the logical unitary~$\encoded{U}$ on the code space.
\item \label{it:composability}
Composability:  Let~$\cW_{\encoded{U}_j}$ be an approximate implementation of a logical unitary~$\encoded{U}_j$ on the code space~$\cL\subseteq \cH$ for~$j=1,2$. Then
\begin{align}
\gateerror_{\cL}(\cW_{\encoded{U}_2} \circ \cW_{\encoded{U}_1},\encoded{U}_2\encoded{U}_1)&\leq \gateerror_\cL(\cW_{\encoded{U}_1},\encoded{U}_1)+\gateerror_\cL(\cW_{\encoded{U}_2},\encoded{U}_2)\ .
\end{align}
This inequality provides an estimate on how errors accumulate under the sequential application of approximately implemented logical unitaries.
\item\label{it:propertylastcomposableerror}

Subadditivity under tensor products: Let~$\cH_1,\cH_2$ be Hilbert spaces with respective code subspaces~$\cL_1,\cL_2$ and let~$\cW_{\encoded{U}_j}$
be an approximate implementation of a logical unitary~$\encoded{U}_j$ on~$\cL_j$ for~$j=1,2$. 
Then 
\begin{align}
    \gateerror_{\cL_1\otimes \cL_2}(\cW_{\encoded{U}_1} \otimes \mathsf{id}_{\cB(\cH_2)}, \encoded{U}_1\otimes I_{\cL_2})&=\gateerror_{\cL_1}(\cW_{\encoded{U}_1},\encoded{U}_1)\ .
\end{align} 
In particular, we have 
\begin{align}
    \gateerror_{\cL_1 \otimes \cL_2}(\cW_{\encoded{U}_1} \otimes \cW_{\encoded{U}_2},\encoded{U}_1 \otimes \encoded{U}_2)&\leq \gateerror_{\cL_1}(\cW_{\encoded{U}_1},\encoded{U}_1) + \gateerror_{\cL_2}(\cW_{\encoded{U}_2},\encoded{U}_2)\, .
\end{align}
\end{enumerate}

Properties~\eqref{it:propertyonecomposableerror}--\eqref{it:propertylastcomposableerror} immediately 
provide a way to easily analyze implementations of circuits, assuming that for each logical gate~$\encoded{U}_j$ in the circuit, an approximate realization~$\cW_{\encoded{U}_j}$ on the appropriate code subspace~$\cL_j$, together with an upper bound on~$\gateerror_{\cL_j}(\cW_{\encoded{U}_j},\encoded{U}_j)$ is known. We explain this in detail in Section~\ref{sec: imperfect implementation of unitaries}. There we include an important generalization to the ``gate deformation'' case where 
the implementation~$\cW_U$ approximately maps an ``input'' code space~$\cL_{in}\subseteq\cH_{in}$ to an ``output'' code space~$\cL_{out}\subseteq\cH_{out}$, where~$\cL_{in}\cong\cL_{out}$ but the subspaces are non-identical (and may, in particular, be subspaces of different Hilbert spaces~$\cH_{in},\cH_{out}$). 

\paragraph{Computable bounds on the composable logical gate error.}
For a~$d$-dimensional code space~$\cL\subset \cH$,
a logical unitary~$\encoded{U}:\cL\rightarrow\cL$ is fully specified by fixing an orthonormal basis~$\{\ket{\encoded{j}}\}_{j=0}^{d-1}$ of~$\cL$ and considering the matrix elements~$U_{j,k}:=\langle \encoded{j},\encoded{U}\encoded{k}\rangle$ of~$\encoded{U}$ with respect to this basis. We often consider the matrix~$U\in\mathsf{Mat}_{d\times d}(\mathbb{C})$ as a linear map~$U:\mathbb{C}^d\rightarrow\mathbb{C}^d$. Then we may equivalently think of the logical unitary~$\encoded{U}=\encmap_{\cL}U\decmap_{\cL}$ acting on the code space~$\cL$
as the composition of three maps, where the encoding map~$\encmap_{\cL}(\ket{j})=\ket{\encoded{j}}$ isometrically maps~$\mathbb{C}^d$ into~$\cL$ and~$\decmap_{\cL}:=\encmap_{\cL}^{-1}$ is the inverse of the encoding map.  Here~$\{\ket{j}\}_{j=0}^{d-1}$ is the computational (orthonormal) basis of~$\mathbb{C}^d$. The associated (logical) CPTP map~$\encodedC{U}:\cB(\cL)\rightarrow\cB(\cL)$ can similarly be written   in terms of~$U$. Correspondingly, we often simply use~$U:\mathbb{C}^d\rightarrow\mathbb{C}^d$ (respectively the corresponding matrix~$U$) in place of~$\encoded{U}$ (respectively~$\encodedC{U}$) in the following, see e.g., Eq.~\eqref{eq:alternativembadv} for an example of this convention. We emphasize, however, that the corresponding maps~$\encoded{U}$ and~$\encodedC{U}$ depend on the choice of basis of~$\cL$.

Consider  a desired target logical unitary~$U:\mathbb{C}^d\rightarrow\mathbb{C}^d$ and a unitary implementation~$\cW_U(\rho)=W_U\rho W_U^\dagger$ with~$W_U:\cH\rightarrow\cH$. We give   general upper bounds on the composable logical gate error 
\begin{align}
 \gateerror_{\cL}(W_U,U) := \gateerror_\cL(W_\encoded{U}, \encoded{U}) := \gateerror_\cL(\cal{W}_\encoded{U}, \encoded{U}) \ \label{eq:alternativembadv}
\end{align}
in terms of the matrix elements~$M_{j,k}=\langle \encoded{j},W_U\encoded{k}\rangle$ of~$W_U$. 
In more detail, we introduce the operator~$B :=  \encoded{U}^\dagger \pi_{\cL}W_U\pi_{\cL}: \cL \rightarrow\cL$  which roughly amounts to the composition of~$W_U$ with the adjoint (inverse) of~$U$, restricted (by appropriate projections) to the code space~$\cL$ (see Definition~\ref{def:diagonalunitary}). 
In case $W_U$ exactly implements the unitary~$U$, the operator $B$ evaluates to the identity on the code space $\cL$. Importantly, this has the following properties:
\begin{enumerate}[(i)]
\item
The operator~$B$ is fully determined by~$U$ and the matrix elements~$\{M_{j,k}\}_{j,k}$, and can thus easily be  computed analytically in typical settings of interest. For example, in CV quantum error correction, our bounds sidestep the need for computing energy-bounded norms~\cite{winter2017energyconstraineddiamondnormapplications,2018_Shirokov} as used in Refs.~\cite{koenig2023limitationslocalupdaterecovery, matsuura2024continuousvariablefaulttolerantquantumcomputation}.
\item
The logical gate error~$\gateerror_\cL(W_U, U)$ can be expressed as a function of~$B$: We show that 
\begin{align}  \gateerror_{\cL} (W_U,U) =   2\sqrt{1-\cn(B)^2}\label{eq:cnbboundgeneral}
\end{align}
where~$\cn(B)$ is the so-called Crawford number (or inner numerical radius) of~$B$, see Corollary~\ref{cor:gateerrorunitaryimplement} and Lemma~\ref{lem:lowerboundgateerrorbzerozero}.
\end{enumerate}
We further show that the identity~\eqref{eq:cnbboundgeneral} implies more easily applicable bounds. For example, if~$B$ is sparse, then we can give an upper bound on the composable gate error which does not explicitly depend on the code space dimension~$d$. More generally,  we can give an upper bound on the composable gate error 
which depends polynomially on~$d$ and the matrix elements~$M_{j,k}$ (see Corollary~\ref{cor:shortmatrixBstatement}).

These bounds underlie our Result~\ref{thm:result2} for approximate GKP codes (see below). They should independently be useful in approximate quantum error correction because they reduce the problem of estimating the accuracy of an implementation to computing  matrix elements with respect to code states. 

\paragraph{Linear optics implementations of logical   gates in approximate GKP codes.} We emphasize that the framework introduced above -- our notion of composable logical gate error and our bounds  in terms of matrix elements -- applies to any approximate quantum error-correcting code. To demonstrate its utility, we apply it to approximate GKP codes. 

Consider the (ideal) GKP code~\cite{gkp}, a continuous-variable (CV)  stabilizer code with the  convenient and highly desirable feature that logical  Pauli and Clifford operations can be realized exactly by  linear optics (i.e., Gaussian) operations. Its code states can unfortunately not be created in practice: Instead, finitely-squeezed (approximate) GKP-states need to be used. This motivates the following question: How accurate  are linear optics realizations of logical gates in approximate  GKP-codes?

In more detail, approximate GKP codes encoding a~$d$-dimensional qudit are typically introduced as a two-parameter family~$\{\gkpcode{\kappa,\Delta}{ }{d}\}_{(\kappa,\Delta)}$, 
where~$\kappa,\Delta>0$ determine the amount of squeezing of the code states in the code~$\gkpcode{\kappa,\Delta}{}{d}\subset L^2(\mathbb{R})$. That is, $\kappa^{-2}$ is proportional to the variance of the envelope in the definition of a finitely squeezed GKP-state, whereas~$\Delta^2$ is that of individual peaks (local maxima). 
 We refer to Section~\ref{sec:truncatedapproximategkpcode}
 for detailed definitions. It is often said that the ideal GKP code~$\gkpcode{}{}{d}$, whose code states are formal superpositions of position-eigenstates (distributions), is obtained in the limit~$(\kappa,\Delta)\rightarrow (0,0)$.   Here we challenge this idea from the point of view of logical gate errors.

 Following the idea that the code~$\gkpcode{\kappa,\Delta}{}{d}$ approximates the ideal GKP code~$\gkpcode{}{}{d}$ in the limit~$(\kappa,\Delta)\rightarrow (0,0)$, 
 it is natural to ask how the parameters~$(\kappa,\Delta)$ relate to the accuracy of gate implementations. The folklore (informal) answer to this question is  along the lines of ``the more squeezing there is, the better". But how much squeezing is actually necessary? And what is the best choice for the (relative) dependence between~$\kappa$ and~$\Delta$? 

\begin{figure}[h]
    \centering
    \includegraphics[width=\linewidth]{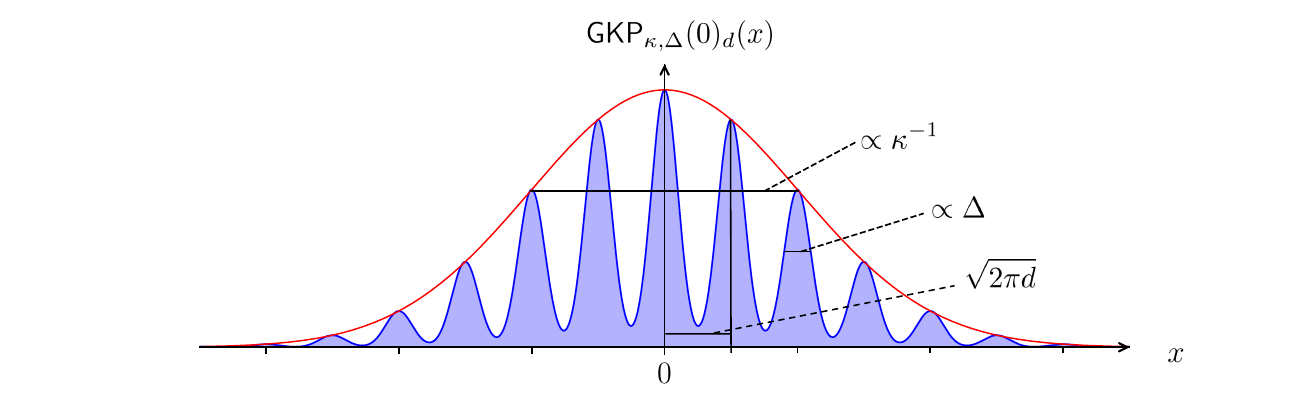}
    \caption{The computational logical basis state~$\ket{\overline{0}}$ in the code~$\gkpcode{\kappa,\Delta}{}{d}$
   is denoted $\ket{\gkp_{\kappa,\Delta}(0)_d} \in \gkpcode{\kappa,\Delta}{}{d}$. Its wavefunction is given by a superposition of Gaussians  with variance proportional to~$\Delta^2$, with an envelope (shown in red) whose   variance is proportional to~$\kappa^{-2}$. The individual peaks (local maxima) are located at integer multiples of~$\sqrt{2\pi d}$. The remaining (logical) basis states are constructed by translating this state in position-space. For a formal definition of the code $\gkpcode{\kappa,\Delta}{}{d}$, we refer to Section~\ref{sec: approximate GKP codes}.}
    \label{fig:untruncatedGKP}
\end{figure}

A particularly natural and popular choice is that of symmetric squeezing, where~$\Delta=\kappa/(2\pi d)$, i.e., we consider pairs of the form
\begin{align}
(\kappa,\Delta)&=(\kappa,\kappa/(2\pi d))\ .\label{eq:symmetricsqueezingcode}
\end{align} 
We call the GKP codes $\gkpcode{\kappa,\kappa/(2\pi d)}{}{d}$ associated with this choice of parameters symmetrically squeezed (approximate) GKP codes with squeezing parameter~$\kappa$. For technical reasons, we will focus on a slightly modified variant of these codes which we denote by~$\gkpcode{\kappa}{\star}{d}$. The latter is obtained from the code~$\gkpcode{\kappa,\kappa/(2\pi d)}{}{d}$ using a certain truncation procedure rendering the defining basis states orthogonal. It serves as a convenient proxy for the symmetrically squeezed code~ $\gkpcode{\kappa,\kappa/(2\pi d)}{}{d}$ which retains its essential features. This concerns, in particular,  the asymptotic behavior of the logical gate error for physical implementations in the limit of large squeezing.
Considering our modified family of symmetrically squeezed approximate GKP codes  facilitates our analysis.
We refer to Section~\ref{sec: approximate GKP codes} for detailed definitions, and to Corollary~\ref{lem:gateerruntruncatedcode}, where we argue that the two codes behave in an analogous manner with regards to the composable logical gate error. 

The choice of~$(\kappa,\Delta)$ in Eq.~\eqref{eq:symmetricsqueezingcode} is motivated by the action of the linear optics implementation 
\begin{align}
    W_{\Fgate}=e^{i\pi (Q^2+P^2)/4} \label{eq:wfgatequadratic}
\end{align}
of the (logical) Fourier transform~$\Fgate$.
In the setup of an ideal GKP code~$\gkpcode{}{}{d}$, the unitary~$W_{\Fgate}$ provides an exact implementation of the (logical) Fourier transform~$\Fgate$ (on~$\mathbb{C}^d$) on the code space (see Ref.~\cite{gkp}).
In contrast, when 
    the unitary~\eqref{eq:wfgatequadratic}
    acts on the approximate code~$\gkpcode{\kappa,\Delta}{}{d}$, it essentially interchanges~$\kappa$ and~$\Delta$. More precisely, it (approximately) maps the code~$\gkpcode{\kappa,\Delta}{}{d}$ to a new code~$\gkpcode{\kappa',\Delta'}{}{d}$, i.e., it transforms the parameters~$(\kappa,\Delta)$ as
         \begin{align}
        (\kappa,\Delta)\mapsto (\kappa',\Delta'):=(2\pi d \Delta,\kappa/(2\pi))\label{eq:transformationbehavior}
        \end{align}  
        while implementing the (logical) Fourier transform. 
       This motivates the choice~$(\kappa,\Delta)$ 
       associated with symmetric squeezing, see Eq.~\eqref{eq:symmetricsqueezingcode}, which is a fixed point of the map~\eqref{eq:transformationbehavior}. For this choice, the unitary~$W_{\Fgate}$ approximately maps the code space to itself. In fact, we can show that~$W_{\Fgate}$ indeed implements the Fourier transform on the symmetrically squeezed approximate GKP code~$\gkpcode{\kappa}{\star}{d}$ in the limit~$\kappa\rightarrow 0$ of infinite squeezing, i.e., we have
       \begin{align} \label{eq:limerrorF}
       \lim_{\kappa\rightarrow 0}\gateerror_{\gkpcode{\kappa}{\star}{d}} (W_{\Fgate},\Fgate)&=0\ ,
       \end{align}
       see Theorem~\ref{thm:result2theorem} for a more  detailed quantitative statement.
       
The logical (qudit) Pauli operators~$X$ and~$Z$ on~$\mathbb{C}^d$ have the linear optics implementations
\begin{align}
W_X&=e^{-i\sqrt{2\pi/d}P}\qquad\textrm{ and }\qquad W_Z=e^{i\sqrt{2\pi/d}Q}\label{eq:gaussianunitaryimplementationlogical}
\end{align}       
in the ideal GKP code~$\gkpcode{}{}{d}$. Our first main result shows that these unitaries are accurate implementations of the logical Pauli operators for the symmetrically squeezed GKP code as well.
\begin{mdframed}
\vspace{-2mm}
\begin{result}[Implementation of Pauli operators by linear optics]
\label{thm:result2}
Let~$d\geq 2$ be an integer. Consider the Gaussian unitary~$W_U$ associated with a logical (qudit) Pauli operator~$U\in \{X,Z\}$ defined by Eq.~\eqref{eq:gaussianunitaryimplementationlogical}.
Then~$W_U$ implements the logical Pauli operator~$U$ in the symmetrically squeezed GKP code~$\gkpcode{\kappa}{\star}{d}$ with logical gate error 
\begin{align}
\gateerror_{\gkpcode{\kappa}{\star}{d}}(W_U,U) &\leq 8\kappa\label{eq:gateerrorlogicalPauliGKP}
\end{align}
for every~$U\in \{X,Z\}$.
\end{result}
\vspace{1mm}
\end{mdframed} 
Result~\ref{thm:result2} gives a detailed quantitative answer to how much and what kind of squeezing is required to apply logical Pauli operators in an approximate GKP codes. The linear dependence of the bound~\eqref{eq:gateerrorlogicalPauliGKP} on~$\kappa$ makes it relevant and applicable in practical applications.
(We also establish a bound of the form
$\gateerror_{\gkpcode{\kappa}{\star}{d}}(W_\Fgate,\Fgate)=O(\mathsf{poly}(d)\mathsf{poly}(\kappa))$
on the gate error of the implementation~$W_{\Fgate}$ of the logical Fourier transform~$\Fgate$. While weaker, this implies that the gate error of this implementation also vanishes in the limit~$\kappa\rightarrow 0$ of infinite squeezing, see Eq.~\eqref{eq:limerrorF}.)

Given the accuracy of linear optics implementations of Paulis in approximate GKP codes (Result~\ref{thm:result2}), it is natural to ask if the linear optics implementations of Cliffords introduced for ideal GKP codes in Ref.~\cite{gkp} also extend to physically realistic GKP codes, that is, codes which are spanned by normalizable states with finite energy. Surprisingly, we find that this is not the case: We show that the standard (ideal GKP code) linear optics
implementation 
\begin{align}
    W_\Pgate&=e^{i(Q^2 + c_d \sqrt{2\pi/d}Q)/2}\qquad\textrm{where  }\qquad c_d = d \mod 2\ \label{eq:wpgatevd}
\end{align}
of the (logical) phase gate~$\Pgate$ which acts as
\begin{align}
    \Pgate \ket{x} = e^{i \pi x(x+c_d)/d} \ket{x} \qquad\text{ for }\qquad x \in \mathbb{Z}_d 
\end{align}
on logical basis states~$\ket{x}$
is unsuitable for use in an approximate GKP code even in the limit of infinite squeezing.
\begin{mdframed}
\vspace{-2mm}
\begin{nogoresult}[Linear optics implementations of certain Cliffords fail for  symmetric squeezing] 
    \label{thm:result1}
    Let~$d\geq 2$ be an integer. Let~$\kappa>0$. Consider the (logical) Clifford phase gate~$\Pgate$ in the symmetrically squeezed GKP code~$\gkpcode{\kappa}{\star}{d}$. Let~$W_{\Pgate}$ be the Gaussian unitary defined in Eq.~\eqref{eq:wpgatevd}. Then 
\begin{align}
    \gateerror_{\gkpcode{\kappa}{\star}{d}}(W_\Pgate, \Pgate) \geq \textfrac{3}{100} 
\end{align}
for any~$\kappa < 1/250$. In particular, the logical gate error of the implementation~$W_{\Pgate}$ of~$\Pgate$ 
is lower bounded by a constant even in the limit~$\kappa\rightarrow 0$ of infinite squeezing. 
\end{nogoresult}
\vspace{1mm}
\end{mdframed}
We note that a result analogous to Result~\ref{thm:result1} holds for the standard linear optics implementation of a CNOT-gate in the GKP code (with~$d=2$).
Similar observations were previously made for the Gaussian implementation of the two-qubit $\cZ$-gate~\cite{RojkovNogo2024} and the implementation of the single-qubit $T$-gate using the (non-Gaussian) cubic phase gate~\cite{Cubicphasenogo2021}.

Result~\ref{thm:result1} goes against the folklore idea that physical gate implementations proposed in the setting of ideal GKP codes are also suitable for approximate GKP codes. It begs the question of how to implement Cliffords (and more general gates) in approximate GKP codes. In separate work~\cite{cliffordshybrid2025}, we propose a solution to this problem which circumvents the no-go Result~\ref{thm:result1} in a qubit-oscillator model. It makes use  of  qubit Clifford operations and qubit-controlled single-mode phase space displacements of constant strength (i.e., shifts by a constant in position/momentum-space). Such hybrid qubit-oscillator operations are readily available in present-day experimental setups~\cite{Eickbusch_2022, CampagneEikbushetal20,DispersiveRegimeSCCircuit} (for a recent review see Ref.~\cite{liu2024hybridoscillatorqubitquantumprocessors}). 

Importantly, Results~\ref{thm:result2} and~\ref{thm:result1} extend to  alternative code spaces whose codewords are close to those of the truncated symmetrically squeezed GKP code in trace distance. In particular, our main conclusions also apply to standard, i.e., untruncated symmetrically-squeezed GKP codes. In more detail, we show continuity bounds on the composable logical gate error (see Appendix~\ref{sec:continuity}). These imply that the composable logical gate error is comparable for pairs of codes whose codewords are close in trace distance. This property is satisfied for  untruncated and truncated symmetrically-squeezed GKP codes (for suitable parameter choices). In particular, in the limit of large squeezing, 
the composable gate error asymptotically vanishes  for one of these families if and only if that is the case for the other family.

We emphasize that the regime of large squeezing (small~$\kappa$) is the physically relevant one for quantum error correction with GKP codes, as this is where the stabilizer constraints are approximately satisfied (yielding protection against small displacements). In this regime, our positive results (Result~1) guarantee accurate gate implementations, while our no-go results (Result~2) reveal fundamental limitations of specific physical implementations. In contrast, GKP codes with broad peaks (large~$\kappa$) are ill-suited for fault-tolerant quantum computation because they lack the key property of translation-invariance (stabilizer constraints) enabling recovery from displacement errors.

\subsection{Outline}

This work is structured as follows.
In Section~\ref{sec: imperfect implementation of unitaries}, we define the composable logical gate error~$\gateerror_\cL(\cW_U,U)$ and introduce notation used throughout this work. We show that the logical gate error is subadditive under composition, a property we use to estimate how errors accumulate in a quantum circuit composed of approximately implemented unitary gates. 

Section~\ref{sec: approximate implementation of logical unitaries} is devoted to deriving
easily computable bounds on the composable logical gate error~$\gateerror_\cL(\cW_U,U)$ for a given unitary implementation~$\cW_U(\rho) = W_U \rho W_U^\dagger$ of~$U$. These bounds only involve matrix elements of the implementation~$W_U$ with respect to families of states which are (close to) an orthonormal basis of the code space~$\cL$. 

In Section~\ref{sec:idealapproxgkpcode} we define ideal and approximate GKP-codes encoding a qudit into an oscillator.
In Section~\ref{sec: logical gates in ideal GKP codes} we introduce the underlying physical setup (linear optics) and  describe the elementary physical operations we allow for. Moreover, we examine (Gaussian) implementations of logical Pauli- and Clifford gates which are known to be exact in the limit of ideal (non-approximate) GKP codes.

In Section~\ref{sec: bound gkp cliffords} we analyze logical gate implementations for  approximate GKP codes.
For each implementation of a logical Clifford, we obtain an upper bound on the composable logical error. We find that for logical Pauli operators the gate error has a linear upper bound in the squeezing parameter. For the Fourier transform, the gate error is bounded by a polynomial in the code space dimension and the squeezing parameter of the approximate GKP code. Taken together, this gives widely applicable accuracy guarantees for quantum computations implemented using these gates in approximate GKP codes.

In Section~\ref{sec: no go} we investigate how certain linear optics implementations of logical Cliffords in ideal GKP codes -- namely that of the  logical phase gate~$\Pgate$  -- 
fail to accurately realize these gates in the case of symmetrically squeezed approximate GKP codes.

In the appendices, we discuss results used in proofs throughout this work. 
In Appendix~\ref{sec:mathfacts} we state a few general mathematical facts on the diamond norm, inner products in Hilbert spaces, and discrete Gaussian distributions. In Appendix~\ref{sec: logical gate error general circuits} we extend the framework of the composable logical gate error to general circuits which are represented as a directed acyclic graph.
In Appendix~\ref{sec:numericalradius} we prove a general lower bound on the inner numerical radius of an operator in terms of its matrix elements.
In Appendix~\ref{app:gateerror_bound_matrixelem} we use these bounds and provide the technical proof of Theorem~\ref{cor:shortmatrixBstatement}. In Appendix~\ref{sec:approximateGKPstates} we discuss different notions of GKP-states and their properties. We prove bounds showing that the different GKP-states considered are close to each other. In Appendix~\ref{sec: matrix elements bounds} we compute the matrix elements of approximate implementations of logical gates in approximate GKP-codes. These are used in Section~\ref{sec: bound gkp cliffords} to prove bounds on the composable logical gate error of each gate. In Appendix~\ref{sec:nogoresultasymmetric} we give an additional result extending the no-go result in Section~\ref{sec: no go}: We show that the typical linear optics implementation of the phase gate~$\Pgate$ is not accurate when considering a non-symmetrically squeezed approximate GKP code either. 
In Appendix~\ref{sec:nogoresultCZ} we show that the standard linear optics implementation of the two-qubit $\cZ$-gate also exhibits a non-vanishing logical gate error in the limit of large  squeezing.
Appendix~\ref{sec:continuity} gives a continuity bound on the logical gate error: It relates the logical gate error of  implementations in different codes which are close to each other.
We apply these results in Appendix~\ref{sec:nogoresultphase_noencoding} to prove an extended version of the no-go result for the Gaussian implementation of the $\Pgate$-gate: We show that  no choice of encoding inside the symmetrically squeezed GKP code can achieve a vanishing logical gate error.

\section{Imperfect implementations of unitaries} \label{sec: imperfect implementation of unitaries}

In this section, we formally define quantities of interest. In Section~\ref{sec:gateerrordef} we introduce the general setting we consider, define the composable logical gate error, and establish that it is subadditive under composition of channels. In Appendix~\ref{sec: logical gate error general circuits}, we further extend the framework to general circuits.

\subsection{Definition of the composable logical gate error \label{sec:gateerrordef}} 
We are interested in approximating a logical CPTP map by a suitable implementation when acting on a code space. 
Concretely, let~$\cH_{in}, \cH_{out}$ be (``physical'') Hilbert spaces, and let~$\cL_{in}\subseteq\cH_{in}$ and~$\cL_{out}\subseteq \cH_{out}$ be two isomorphic (``logical'') subspaces. We assume~$\dim \cL_{in} = \dim \cL_{out}= d <\infty$. 
Let~$\encoded{U} :\cL_{in}  \rightarrow \cL_{out}$ be a unitary ``logical'' map and 
\begin{alignat}{2}
    \encodedC{U} : \cal{B}(\cal{L}_{in}) &\ \rightarrow\ && \cal{B}(\cal{L}_{out}) \\
    \rho &\ \mapsto\ && \encoded{U}\rho \encoded{U}^\dagger \ 
\end{alignat}
the associated CPTP map. Our goal is to approximate~$\encodedC{U}$ by a suitable CPTP map ~$\cW : \cB(\cH_{in})  \rightarrow \cB(\cH_{out})~$. We are interested in the error incurred when approximating~$\encodedC{U}$ by the CPTP map~$\cW$.

To this end, we need the following definitions: For a subspace~$\cL\subset\cH$ of a Hilbert space~$\cH$, we denote 
by~$\pi_\cL:\cH\rightarrow\cL$ the orthogonal projection onto~$\cL$, and by 
\begin{alignat}{2}
    \Pi_\cL:\cB(\cH)&\  \rightarrow\  &&\cB(\cL)\\
    \rho & \  \mapsto\  &&\Pi_{\cL}(\rho)=\pi_\cL\rho\pi_\cL
\end{alignat}
the corresponding projection onto~$\cB(\cL)\subseteq\cB(\cH)$.  We define what we call the (composable) logical gate error of an implementation~$\cW: \cB(\cH_{in})  \rightarrow \cB(\cH_{out})$ of a (logical)
map~$\encodedC{U}:\cal{B}(\cal{L}_{in})\rightarrow\cal{B}(\cal{L}_{out})$ on  input and output code spaces~$\cL_{in}$  and~$\cL_{out}$, respectively, as 
\begin{align}
    \label{eq:gateerrordef}
\gateerror_{\cL_{in}, \cL_{out}}(\cW,\encoded{U})=\left\|(\cW-\encodedC{U})\circ \Pi_{\cLin}\right\|_\diamond\, .
\end{align}
The composable logical gate error~$\gateerror_{\cL_{in}, \cL_{out}}(\cW,\encoded{U})$ quantifies the worst-case error of the implementation~$\cW$ when acting on states supported on the input code space~$\cL_{in}$. 

We note that we often consider the case where~$\cH_{in} = \cH_{out} = \cH$, $\cL_{in} = \cL_{out}$ and an implementation of the form~$\cW(\rho)=W\rho W^\dagger$ where~$W$ is a unitary on~$\cH$. In this case, we write 
\begin{align}
    \gateerror_{\cL}(W,\encoded{U})=\gateerror_{\cL}(\cW,\encoded{U})
\end{align} by a slight abuse of notation.

The following lemma justifies the term composable logical gate error. We show that the quantity is subadditive under composition of channels. 

\begin{lemma}[Subadditivity of the logical gate error] \label{lem: additivity gate error} Let~$\cH_{in}^{(1)}, \cH_{out}^{(1)}, \cH_{out}^{(2)}$ be Hilbert spaces. Set~$\cH_{in}^{(2)} := \cH_{out}^{(1)}$.
Let 
\begin{align}
\begin{matrix}
\cal{L}^{(1)}_{in} &\subseteq &\cal{H}^{(1)}_{in}\\
 \cal{L}^{(1)}_{out} &\subseteq &\cal{H}^{(1)}_{out}
 \end{matrix}
 \qquad\textrm{ and }\qquad
 \begin{matrix}
 \cal{L}^{(2)}_{in}&\subseteq &\cal{H}^{(2)}_{in}\\
 \cal{L}^{(2)}_{out}&\subseteq&\cal{H}^{(2)}_{out}
 \end{matrix}
 \end{align} be (code) subspaces with~$\cal{L}^{(2)}_{in} := \cal{L}^{(1)}_{out}$, and~$\cal{L}_{in}^{(1)} \simeq \cal{L}_{out}^{(1)}=\cal{L}_{in}^{(2)}\simeq \cal{L}_{out}^{(2)}$.
 Let 
 \begin{align}
 \begin{matrix}
 \encoded{U}_1:&\cL_{in}^{(1)}&\rightarrow \cL_{out}^{(1)}\\
 \encoded{U}_2:&\cL_{in}^{(2)}&\rightarrow \cL_{out}^{(2)}
 \end{matrix} 
 \end{align}
 be unitaries.
 Let
 \begin{align}
 \begin{matrix}
 \cal{W}_1 :& \cal{B}(\cH^{(1)}_{in})& \rightarrow &\cal{B}(\cH^{(1)}_{out})\\
 \cal{W}_2 : &\cal{B}(\cH^{(2)}_{in})& \rightarrow &\cal{B}(\cH^{(2)}_{out})
 \end{matrix}
 \end{align}
 be two CPTP maps.
     Then we have
    \begin{align}
        \label{eq:def_logicalerrror}
    \gateerror_{\cL^{(1)}_{in},\cL^{(2)}_{out}}(\cW_2\circ\cW_1,\encoded{U}_2  \encoded{U}_1)&\leq 
    \gateerror_{\cL^{(1)}_{in}, \cL_{out}^{(1)}}(\cW_1,\encoded{U}_1)+\gateerror_{\cL^{(2)}_{in}, \cL_{out}^{(2)}}(\cW_2,\encoded{U}_2) \ .
    \end{align}
    \end{lemma}

\begin{proof}

Define the CPTP maps   associated  to the logical unitaries 
 \begin{align}
 \begin{matrix}
 \encodedC{U}_1:&\cB(\cL_{in}^{(1)})&\rightarrow \cB(\cL_{out}^{(1)})\\
 \encodedC{U}_2:&\cB(\cL_{in}^{(2)})&\rightarrow \cB(\cL_{out}^{(2)})
 \end{matrix} 
 \end{align}
  acting as~$\encodedC{U}_j(\rho):=\encoded{U}_j\rho \encoded{U}_j^\dagger$ for~$j\in \{1,2\}$.

By the definition of the logical gate error and the triangle inequality for the diamond norm we have 
\begin{align}
    \gateerror_{\cL^{(1)}_{in},\cL^{(2)}_{out}}(\cW_2\circ\cW_1,\encoded{U}_2 \encoded{U}_1)&= \left\| \left(\cW_2 \circ \cW_1 - \encodedC{U}_2 \circ \encodedC{U}_1\right) \circ \Pi_{\cL_{in}^{(1)}} \right\|_\diamond \\
    &\le \left\| \left(\cW_2 \circ \cW_1 - \cW_2 \circ \encodedC{U}_1\right) \circ \Pi_{\cL_{in}^{(1)}} \right\|_\diamond + \left\| \left(\cW_2 \circ \encodedC{U}_1 - \encodedC{U}_2 \circ \encodedC{U}_1\right) \circ \Pi_{\cL_{in}^{(1)}} \right\|_\diamond\\
    &= \left\| \cW_2 \circ \left(\cW_1 -\encodedC{U}_1\right) \circ \Pi_{\cL_{in}^{(1)}} \right\|_\diamond + \left\| \left(\cW_2  - \encodedC{U}_2 \right) \circ \encodedC{U}_1 \circ \Pi_{\cL_{in}^{(1)}} \right\|_\diamond\, . \label{eq: triangle logical gate err}
\end{align}
We proceed to bound each of the terms in Eq.~\eqref{eq: triangle logical gate err}.
First, by the monotonicity of the diamond norm under CPTP maps, it follows that
\begin{align}
    \left\| \cW_2 \circ \left(\cW_1 -\encodedC{U}_1\right) \circ \Pi_{\cL_{in}^{(1)}} \right\|_\diamond \le \left\|  \left(\cW_1 -\encodedC{U}_1\right) \circ \Pi_{\cL_{in}^{(1)}} \right\|_\diamond  \, . \label{eq: lge first}
\end{align}
Second, notice that~$\encodedC{U}_1 = \Pi_{\cL_{out}^{(1)}} \circ \encodedC{U}_1 = \Pi_{\cL_{in}^{(2)}} \circ \encodedC{U}_1$. Therefore we have
\begin{align}
    \left\| \left(\cW_2  - \encodedC{U}_2 \right) \circ \encodedC{U}_1 \circ \Pi_{\cL_{in}^{(1)}} \right\|_\diamond 
    &= \left\| \left(\cW_2  - \encodedC{U}_2 \right) \circ \Pi_{\cL_{in}^{(2)}} \circ \encodedC{U}_1 \circ \Pi_{\cL_{in}^{(1)}} \right\|_\diamond \\
    &= \left\| \left(\cW_2  - \encodedC{U}_2 \right) \circ \Pi_{\cL_{in}^{(2)}} \circ \encodedC{U}_1  \right\|_\diamond \\
    &= \left\| \left(\cW_2  - \encodedC{U}_2 \right) \circ \Pi_{\cL_{in}^{(2)}}   \right\|_\diamond \ , \label{eq: lge second}
\end{align}
where the last identity is due to the invariance of the diamond norm under compositions with unitary channels.  
The claim follows by combining~Eqs.~\eqref{eq: lge first} and~\eqref{eq: lge second} with Eq.~\eqref{eq: triangle logical gate err}.
\end{proof}

The following result shows that the logical gate error is invariant under tensoring additional systems.

\begin{lemma} \label{lem:invariancetensoredidentity}
    Let~$\cH_1,\cH_2$ be Hilbert spaces. Let~$\cL_j \subset \cH_j$ be a subspace for~$j=1,2$ and let~$\cW:\cB(\cH_1) \rightarrow \cB(\cH_1)$ be a 
    CPTP map. Let~$\encoded{U}: \cL_1 \rightarrow \cL_1$ be unitary. 
    Then 
    \begin{align}
        \gateerror_{\cL_1\otimes \cL_2}(\cW \otimes \mathsf{id}_{\cB(\cH_2)}, \encoded{U} \otimes I_{\cL_2})&=\gateerror_{\cL_1}(\cW,\encoded{U})\ .
    \end{align} 
\end{lemma}
  \begin{proof}
    By definition of the logical gate error we have 
    \begin{align}
        &\gateerror_{\cL_1 \otimes \cL_2} (\cW \otimes \mathsf{id}_{\cB(\cH_2)}, \encoded{U} \otimes I_{\cB(\cL_2)}) \\
        &\qquad= \left\|\left(\cW \otimes \mathsf{id}_{\cB(\cH_2)} - \encodedC{U} \otimes \mathsf{id}_{\cB(\cL_2)}\right) \circ \Pi_{\cL_1 \otimes \cL_2} \right\|_\diamond\\
        &\qquad= \left\|\left(\cW \otimes \mathsf{id}_{\cB(\cH_2)} - \encodedC{U} \otimes \mathsf{id}_{\cB(\cL_2)}\right) \circ \left(\Pi_{\cL_1} \otimes \Pi_{\cL_2}\right) \right\|_\diamond\\
        &\qquad= \left\|(\cW \circ \Pi_{\cL_1})\otimes \Pi_{\cL_2} - (\encodedC{U} \circ \Pi_{\cL_1}) \otimes \Pi_{\cL_2} \right\|_\diamond\\
        &\qquad= \left\|\left((\cW  - \encodedC{U}) \circ \Pi_{\cL_1}\right) \otimes \Pi_{\cL_2} \right\|_\diamond
    \end{align} where we used that~$\mathsf{id}_{\cB(\cH_2)} \circ \Pi_{\cL_2} = \Pi_{\cL_2}$. The multiplicativity of the diamond norm under tensor products (see~\cite[Theorem 3.49]{watrousbook}) implies that
    \begin{align}
        \left\|\left((\cW  - \encodedC{U}) \circ \Pi_{\cL_1}\right) \otimes \Pi_{\cL_2} \right\|_\diamond &= \left\|\left(\cW  - \encodedC{U}\right) \circ \Pi_{\cL_1}  \right\|_\diamond \cdot \left\| \Pi_{\cL_2} \right\|_\diamond\\
        &= \left\|\left(\cW  - \encodedC{U}\right) \circ \Pi_{\cL_1}  \right\|_\diamond\\
        &=  \gateerror_{\cL_1} (\cW, \encoded{U}) \, ,
    \end{align} where we used that~$\|\Pi_{\cL_2}\|_\diamond= 1$.
  \end{proof}

\section{Upper bounds on the logical gate error \label{sec: approximate implementation of logical unitaries}}
Let~$\cHin,\cHout$ be two Hilbert spaces.
 Let~$\cLin\subseteq \cHin$
and~$\cLout\subseteq \cHout$ be two (code) subspaces. We assume that they are isomorphic and finite-dimensional, i.e.,  
\begin{align}
\cLin \cong \cLout \cong \mathbb{C}^d\label{eq:clinloutprec}
\end{align}
for some~$d\geq 2$.  Now consider a logical unitary~$\encoded{U}:\cal{L}_{in}\rightarrow \cal{L}_{out}$. We assume that it is specified by a unitary~$U:\mathbb{C}^d\rightarrow\mathbb{C}^d$ (equivalently given by a matrix~$U=(U_{j,k})_{j,k}\in\mathsf{Mat}_{d\times d}(\mathbb{C})$ consisting of matrix elements 
$U_{j,k}=\langle j,Uk\rangle$ for~$j,k\in \{0,\ldots,d-1\}$
with respect to the computational basis of~$\mathbb{C}^d$), and two isometric encoding maps
\begin{align}
\begin{matrix}
\encmap_{\cL_{in}}:&\mathbb{C}^d&\rightarrow &\cL_{in}\\
\encmap_{\cL_{out}}:&\mathbb{C}^d&\rightarrow &\cL_{out}\ .
\end{matrix}\label{eq:enccmapsfreedomb}
\end{align}
Denoting the corresponding inverse (decoding) maps by
\begin{align}
\begin{matrix}
\decmap_{\cL_{in}}:=\encmap_{\cL_{in}}^{-1} :&\cL_{in}& \rightarrow&\mathbb{C}^d\\
\decmap_{\cL_{out}}:=\encmap_{\cL_{out}}^{-1} :&\cL_{out}& \rightarrow&\mathbb{C}^d\ ,
\end{matrix}
\end{align}
the map~$\encoded{U}:\cal{L}_{in}\rightarrow \cal{L}_{out}$ is given as a product 
\begin{align}
    \encoded{U}=\encmap_{\cL_{out}}  U \decmap_{\cL_{in}}  :\cL_{in}  \rightarrow \cL_{out}\ . \label{eq:encondedforfixedchoice}
\end{align}
See Table~\ref{tab:encdecgeneral} for the circuit schemes representing these maps. 
The map~$\encoded{U}$ gives  rise to the CPTP map
\begin{alignat}{2}
    \encodedC{U} : \cal{B}(\cal{L}_{in}) &\ \rightarrow\ && \cal{B}(\cal{L}_{out}) \\
    \rho &\ \mapsto\ && \encoded{U}\rho \encoded{U}^\dagger \ .
\end{alignat}
In the following, we derive upper bounds on the logical gate error~$\gateerror_{\cL_{in}, \cL_{out}}(\cW, \encoded{U})$ quantifying how well a CPTP map~$\cW:\cB(\cHin)\rightarrow\cB(\cHout)$ approximates the ``logical'' CPTP map~$\encodedC{U}:\cB(\cLin)\rightarrow\cB(\cLout)$.
By a slight abuse of notation, we will write
\begin{align}
\gateerror_{\cL_{in}, \cL_{out}}(\cW,U):=\gateerror_{\cL_{in}, \cL_{out}}(\cW, \encoded{U})\ \label{eq:gateerrrclinclout}
\end{align}
for this expression. We emphasize, however, that~$\encodedC{U}$ and thus the quantity~\eqref{eq:gateerrrclinclout} depends
on the choice of encoding maps~\eqref{eq:enccmapsfreedomb}. For a more general discussion and an encoding-independent notion of a logical gate error we refer to Appendix~\ref{sec:subspacegateerror}.

\begin{center}
    \begin{table}[H]
        \centering
        \renewcommand{\arraystretch}{1.3} 
        \setlength{\tabcolsep}{10pt} 
        \rowcolors{1}{white}{gray!15} 
        \centering
        \begin{tabular}{c|c|c}
        \textbf{map} & \textbf{circuit} & \textbf{decomposition}\\ \hline
        \raisebox{-0cm}{$\mathsf{Enc}_{\cL}$} &\raisebox{-0.25cm}{\hspace{6mm}\includegraphics{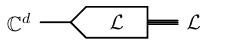}} & \diagbox{\tiny\hspace{3.5cm}}{\tiny\hspace{3.5cm}} \\
        \raisebox{-0cm}{$\mathsf{Dec}_{\cL}$} &  \raisebox{-0.25cm}{\hspace{6mm}\includegraphics{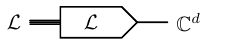}} & \diagbox{\tiny\hspace{3.5cm}}{\tiny\hspace{3.5cm}} \\
        \raisebox{0.1cm}{$I_{\mathbb{C}^d}$} & \raisebox{-0.05cm}{\includegraphics{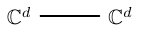}}  &  \raisebox{-0.2cm}{\includegraphics{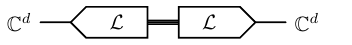}}\\
         \raisebox{0.2cm}{$\encoded{U}$} & \raisebox{-0.15cm}{\includegraphics{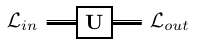}} & \raisebox{-0.15cm}{\includegraphics{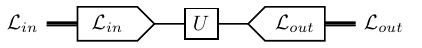}}\\
        \raisebox{0.2cm}{$\encoded{U}^\dagger$} &\raisebox{-0.15cm}{\includegraphics{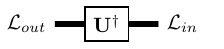}} & \raisebox{-0.15cm}{\includegraphics{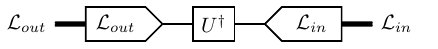}}
           \end{tabular}
           \caption{Diagrammatic representation of the encoding map~$\mathsf{Enc}_{\cL}$ and the decoding map~$\mathsf{Dec}_{\cL}$ maps and specification of a logical unitary~$\encoded{U}$ (and its adjoint~$\encoded{U}^\dagger$) with the corresponding input and output spaces.}
           \label{tab:encdecgeneral}
    \end{table}
    \end{center}

In more detail, the encoding maps~\eqref{eq:enccmapsfreedomb} can be fixed by choosing orthonormal bases of the corresponding spaces and defining associated maps. That is, let~$\{\ket{j}\}_{j=0}^{d-1}$ be the computational (orthonormal) basis of~$\bb{C}^d$ and~$\{\ket{\encoded{j}}_{\cLin}\}_{j=0}^{d-1}$  an orthonormal basis of~$\cLin$.
Then we can embed~$\mathbb{C}^d$ into~$\cLin\subseteq\cHin$ by mapping basis states as 
\begin{align}
\begin{matrix}
\encmap_{\cLin}: & \mathbb{C}^d & \rightarrow & \cHin\\
             & \ket{j} & \mapsto &\ket{\encoded{j}}_{\cLin}
             \end{matrix}\qquad\textrm{ for }\qquad j\in \{0,\ldots,d-1\}\ ,
\end{align}
and linearly extending to all of~$\mathbb{C}^d$, obtaining an isometry~$\encmap_{\cLin}$. For the subspace~$\cLout \subseteq \cHout$, we define the encoding map~$\encmap_{\cLout}:\mathbb{C}^d\rightarrow\cLout$ in an analogous manner using  an orthonormal basis 
$\{\ket{\encoded{j}}_{\cLout}\}_{j=0}^{d-1}$ of~$\cLout$.

We note that according to this definition, the unitary~$\encoded{U}$ acts on basis states as
\begin{align}
\encoded{U}\ket{\encoded{j}}_{\cLin}=\sum_{k=0}^{d-1} U_{j,k}\ket{\encoded{k}}_{\cLout}\qquad\textrm{ for all }j\in \{0,\ldots,d-1\}\ .
\end{align}     

In Section~\ref{sec:diamondnormKrausrankone}, we derive an exact expression for the logical gate error in terms of the so-called inner numerical radius for unitary implementations. Section~\ref{sec:numericalradius} provides lower bounds on the inner numerical radius based on matrix elements of implementations. Section~\ref{sec:gateerror_bound_matrixelem} bounds the logical gate error directly in terms of matrix elements.

\subsection{Gate error of unitary implementations \label{sec:diamondnormKrausrankone}}

In this section, we consider the case where the implementation is given by a unitary~$W:\cHin\rightarrow\cHout$. That is, we consider the CPTP map~$\cW:\cB(\cHin)\rightarrow\cB(\cHout)$ defined by $\cW(\rho)=W\rho W^\dagger$.
The following definition will be useful. 
\begin{definition}\label{def:diagonalunitary}
We define the operator 
\begin{align}
    B=B^U_{\cLin,\cLout}(W,U) =  
    \encoded{U}^\dagger \pi_{\cLout}W\pi_{\cLin}:\cL_{in}\rightarrow\cL_{in}\, .
\end{align}
\end{definition}
\noindent 
The operator~$B$ is characterized by the matrix elements with respect to the basis~$\{\ket{\encoded{k}}_{\cL_{in}}\}_{k=0}^{d-1}$ of~$\cLin$ given by 
\begin{align}
 B_{j,k}=\langle\encoded{j}|_{\cLin}
\encoded{U}^\dagger \pi_{\cLout}W\pi_{\cLin} \ket{\encoded{k}}_{\cLin}
\end{align}
for~$j,k\in \{0,\ldots,d-1\}$. It is straightforward to check that 
\begin{align}
B_{j,k}&:= \sum_{m=0}^{d-1} \overline{U_{m,j}}  \langle \encoded m|_{\cLout} W |\encoded{k}\rangle_{\cL_{in}} \qquad\textrm{for}\qquad j,k\in \{0,\ldots,d-1\}\ .\label{eq:bjkdefinition}
\end{align}

We can express this by a  circuit identity as shown in Fig.~Fig.~\ref{fig:circuitB}, where we indicate the input and output spaces~$\cLin$ and~$\cLout$, respectively.
\begin{figure}
\centering
\resizebox{\linewidth}{!}{%
\begin{tabular}{ccc}
\raisebox{-0.25cm}{\includegraphics{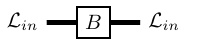}} & = & \hspace{-3.9cm} \raisebox{-0.25cm}{\includegraphics{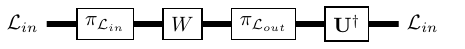}}\\ [1.2em]
& = & \raisebox{-0.25cm}{\includegraphics{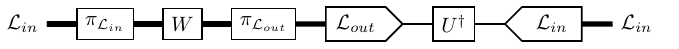}}\, .
\end{tabular}
}
\caption{Diagrammatic illustration of the definition of~$B$.}
\label{fig:circuitB}
\end{figure}

Also observe that 
\begin{align}
\|B\|\leq 1\qquad\textrm{ if } W\textrm{ is unitary}\ \label{eq:unitaryBW}
\end{align}
because the operator norm satisfies~$\|X Y\|\leq \|X\|\cdot\|Y\|$, $\|V\|=1$ for an isometry~$V$, and~$\|\pi\|=1$ for a projection~$\pi$.

\begin{lemma}[Logical gate error of unitary implementations] \label{lem:gateerrorinnernumericalradius}
    Let~$W:\cHin \rightarrow \cHout$ be unitary. Consider the CPTP map~$\cW(\rho)=W\rho W^\dagger$. Let~$B:\cLin\rightarrow\cLin$ 
    be the map introduced in Definition~\ref{def:diagonalunitary}. Set
\begin{align}
\cn(A)&:=\min_{\psi\in\cLin: \|\psi\|=1} |\langle \psi,A\psi\rangle|\qquad\textrm{for any linear map } A: \cLin \rightarrow\cLin\, . \label{eq:crawford}
\end{align} 
    Then
    \begin{align}
        \gateerror_{\cLin, \cLout}(W,\encoded{U}) = 2\sqrt{1 - \cn(B)^2}\, .
    \end{align}
\end{lemma}
The quantity~$\cn(A)$ is sometimes referred to as the inner numerical radius or the Crawford number of~$A$, see~\cite{Kressner2018Crawford}.
\begin{proof}
    By definition of the logical gate error in combination with Lemma~\ref{lem:finitedimdiamond} we can restrict the supremum in the definition of the operator norm to all operators in $\cB(\cLin \otimes \cL')$ where $\cL' \cong \cLin \cong \mathbb{C}^d$. This gives
    \begin{align}
        \gateerror_{\cL}(W,\encoded{U}) &= \left\| (\cW - \encodedC{U})\circ \Pi_{\cL}\right\|_\diamond \\
        &= \sup_{\substack{X \in \cB(\cLin\otimes \cL')\\ \|X\|_1 =1}} \left\| \left((\cW - \encodedC{U})\otimes \mathsf{id}_{\cB(\cL')}\right)(X) \right\|_1\\
        &= \sup_{\substack{\Phi \in \cLin\otimes \cL'\\ \|\Phi\|=1}} \left\| \left((\cW - \encodedC{U}) \otimes \mathsf{id}_{\cB(\cL')}\right) (\proj{\Phi}) \right\|_1\\
        &=  \sup_{\substack{\Phi \in \cLin\otimes \cL'\\ \|\Phi\|=1}} \left\| (W \otimes I_{\cL'}) \proj{\Phi} (W \otimes I_{\cL'})^\dagger - (\encoded{U} \otimes I_{\cL'}) \proj{\Phi} (\encoded{U} \otimes I_{\cL'})^\dagger \right\|_1 \\
        &= \sup_{\substack{\Phi \in \cLin\otimes \cL'\\ \|\Phi\|=1}} 2\sqrt{1 -\left|\langle (\encoded{U} \otimes I_{\cL'}) \, \Phi, (W \otimes I_{\cL'}) \,  \Phi\rangle\right|^2} \, . \label{eq:gateerrorintermsofinnerprod}
    \end{align}
    In the third identity we used that $(\cW - \encodedC{U})$ is Hermicity-preserving which implies that it is enough to restrict the supremum to pure states, see Lemma~\ref{lem: diamond hermicity preserving}.
    In the penultimate step we used that both $(W \otimes I_{\cL'}) \ket{\Phi}$ and $(\encoded{U} \otimes I_{\cL'}) \ket{\Phi}$ are normalized due to the unitarity of $W:\cHin \rightarrow \cHout$ and $\encoded{U}: \cLin \rightarrow \cLout \subset \cHout$.
    
    We have  
    \begin{align}
    \inf_{\substack{\Phi \in \cLin\otimes \cL'\\ \|\Phi\|=1}} \left|\langle (\encoded{U} \otimes I_{\cL'}) \, \Phi, (W \otimes I_{\cL'})\,  \Phi\rangle\right| &= \inf_{\substack{\Phi \in \cLin\otimes \cL'\\ \|\Phi\|=1}} \left|\langle (\pi_{\cLout}\encoded{U} \otimes I_{\cL'}) \, \Phi, ((W \, \pi_{\cLin}) \otimes I_{\cL'})\,  \Phi\rangle\right|\\
    &=\inf_{\substack{\Phi \in \cLin\otimes \cL'\\ \|\Phi\|=1}} \left|\langle  \, \Phi, (\encoded{U}^\dagger \pi_{\cLout} W \pi_{\cLin} \otimes I_{\cL'}) \,  \Phi\rangle\right| \\
    &= \inf_{\substack{\rho \in \cB(\cLin)\\
     \rho \ge 0, \tr (\rho) =1}} \left| \tr \left( \encoded{U}^\dagger \pi_{\cLout} W \pi_{\cLin} \rho \right)\right|\\
     &=  \inf_{\substack{\rho \in \cB(\cLin)\\
     \rho \ge 0, \tr (\rho) =1}} \left| \tr \left( B \rho \right)\right|\, . \label{eq:intermediatestepgateerror}
    \end{align}
    The first equality follows from the fact that $U \cLin = \cLout$ and that $(\pi_{\cLin} \otimes I_{\cL'}) \ket{\Phi}=\ket{\Phi}$ for all $\ket{\Phi} \in \cLin \otimes \cL'$.
    The third equality follows by taking the partial trace over the subsystem affiliated with the subspace~$\cL'$. In the last step we used the definition of $B$.\\
    \indent Let $A:\cLin \rightarrow \cLin$ be a linear operator. The numerical range of $A$ is defined as the subset of~$\mathbb{C}$ given by
\begin{align}
\cN(A)=\big\{\langle \psi,A \psi\rangle\ |\ \psi\in\cLin, \|\psi\|=1\big\}\ .
\end{align}
According to the Toeplitz-Hausdorff theorem (see e.g., \cite[Theorem 3.54]{watrousbook})
the set~$\cN(A)$ is (compact as $\cLin$ is finite-dimensional and) convex.
Since every density operator $\rho \in \cB(\cLin)$ is a convex combination of pure states, it follows that 
\begin{align} \label{eq:numericalrangeconvexcomb}
\left\{ \tr(A \rho) \mid \rho \in \cB(\cLin), \rho\ge 0, \tr(\rho)=1\right\} &= \cN(A)\ .
\end{align}
Combining Eq.~\eqref{eq:numericalrangeconvexcomb} with Eq.~\eqref{eq:intermediatestepgateerror} gives 
\begin{align}
    \inf_{\substack{\Phi \in \cLin\otimes \cL'\\ \|\Phi\|=1}} \left|\langle (\encoded{U} \otimes I_{\cL'}) \, \Phi, (W \otimes I_{\cL'})\,  \Phi\rangle\right| &= \min_{\substack{\psi\in\cLin\\ \|\psi\|=1}} |\langle \psi,B\psi\rangle| = \cn(B)\, , \label{eq:gateerrorintermsofinnerprodnonunitary}
\end{align}
where we note that the infimum can be replaced by a minimum because the unit ball in $\cLin$ is compact and the map $\psi \mapsto |\langle \psi,B\psi\rangle|$ is continuous. The claim follows by combining with Eq.~\eqref{eq:gateerrorintermsofinnerprodnonunitary}.
\end{proof}

We can easily extend Lemma~\ref{lem:gateerrorinnernumericalradius} to give an upper bound on the logical gate error in the case where $\cW$ is a general Kraus-rank one map, that is, $\cW(\rho) = W \rho W^\dagger$ with $W:\cHin \rightarrow\cHout$ an arbitrary bounded linear operator.
\begin{lemma}
[Gate error of general Kraus-rank one implementations] \label{lem:gateerrornonunitary}
    Let~$W:\cHin \rightarrow \cHout$ be linear. Consider the CP map~$\cW(\rho)=W\rho W^\dagger$. Let~$B:\cLin\rightarrow\cLin$ 
    be the map introduced in Definition~\ref{def:diagonalunitary}.
    Then
    \begin{align}
        \gateerror_{\cLin, \cLout}(W,\encoded{U}) &\leq \sqrt{(\|W\|^2 + 1)^2 - 4 c(B)^2}\, . \label{eq:gateerrornonunitary}
    \end{align}
\end{lemma}
\begin{proof}
    We note that two rank-one operators of the form $\proj{v}$ and $\proj{w}$ with $v,w \in \cH$ not necessarily normalized vectors in a Hilbert space $\cH$ satisfy 
    \begin{align}
        \left\| \proj{v} - \proj{w} \right\|_1 &= \sqrt{\left(\|v\|^2 + \|w\|^2\right)^2 - 4 |\langle v,w \rangle|^2}\ . \label{eq:trace norm difference rank one}
    \end{align}
    Using this identity, we can follow the same lines of the proof of Lemma~\ref{lem:gateerrorinnernumericalradius} which gives 
    \begin{align}
        \gateerror_{\cLin, \cLout}(W,\encoded{U}) &= \sup_{\substack{\Phi \in \cLin\otimes \cL'\\ \|\Phi\|=1}} \sqrt{\left(\|(W\otimes I_\cL')\Phi\|^2 + \|(\encoded{U} \otimes I_{\cL'})\Phi\|^2 \right)^2 -4\left|\langle (W \otimes I_{\cL'}) \, \Phi, (\encoded{U} \otimes I_{\cL'})\,  \Phi\rangle\right|^2|} \\
        &\le \sup_{\substack{\Phi \in \cLin\otimes \cL'\\ \|\Phi\|=1}}\sqrt{(\|W\|^2 + 1)^2 - 4  \left|\langle (W \otimes I_{\cL'}) \, \Phi, (\encoded{U} \otimes I_{\cL'})\,  \Phi\rangle\right|^2}\, . 
    \end{align} 
    The claim follows by the same arguments used to derive Eq.~\eqref{eq:gateerrorintermsofinnerprodnonunitary}.
\end{proof}

Given an implementation $W$ of a logical unitary $\encoded{U}$, we will use the following result to derive an expression for the logical gate error of the implementation $W^\dagger$ of the adjoint~$\encoded{U}^\dagger$.
\begin{lemma}[Invariance of the inner numerical radius under implementing the adjoint] \label{lem:cnadjoint}
    Let~$W:\cHin \rightarrow \cHout$ be linear. Consider the CP map~$\cW(\rho)=W\rho W^\dagger$. Let~$B = B_{\cL_{in}, \cL_{out}}^U(W,U):\cLin\rightarrow\cLin$ and~$\widetilde{B} = B_{\cL_{out}, \cL_{in}}^{U^\dagger}(W^\dagger, U^\dagger):\cLout\rightarrow\cLout$
be the maps introduced in Definition~\ref{def:diagonalunitary}. Then
\begin{align}
    \cn(B) = \cn(\widetilde{B})\, .
\end{align}
\end{lemma}
\begin{proof}
Define the linear map~$C: \mathbb{C}^d \rightarrow \mathbb{C}^d$ as 
\begin{align}
    C = \decmap_{\cLout} \pi_{\cLout} W \pi_{\cLin} \encmap_{\cLin}\, .
\end{align}
It is easy to check that 
\begin{align}
    \decmap_{\cLin} B \encmap_{\cLin} &= U^\dagger C \qquad\textrm{and} \qquad \encmap_{\cLout}  \widetilde{B}  \decmap_{\cLout} =  U C^\dagger  \, .
\end{align}
This implies that
\begin{align}
    \cn(B) = \cn(U^\dagger C) \qquad \textrm{and} \qquad \cn(\widetilde{B}) = \cn(U C^\dagger)\,. \label{eq:cnBC}
\end{align}
Let~$A: \mathbb{C}^d \rightarrow \mathbb{C}^d$ be a linear map. Recall that  the numerical range of~$A$ is defined as 
\begin{align}
    \cN(A)=\left\{\langle \psi,A\psi\rangle\ |\ \psi\in\mathbb{C}^d, \|\psi\|=1\right\}\, .
    \end{align}
    Therefore we can write the inner numerical radius of~$A$ as 
    \begin{align}
        \cn(A) = \min_{z \in \cN(A)} |z| \label{eq:defnumericalrange}\, .
    \end{align}
    It is easy to verify that the numerical range satisfies 
    \begin{align}
        \cN(A^\dagger) &= \overline{\cN(A)} \label{eq:n_adjoint} \qquad \textrm{and}\\
        \cN(V A V^\dagger) &= \cN(A) \qquad \text{for every unitary } V \label{eq:n_unitary}\, ,
        \end{align}
        where~$\overline{\cN(A)}= \{z \mid \overline{z} \in \cN(A)\}$.
    It follows that
    \begin{align}
    \begin{alignedat}{2}
        \cN(\widetilde{B}) &= \cN(U C^\dagger) &\qquad &\textrm{by Eq.~\eqref{eq:cnBC}}\\ 
        &= \cN(U^\dagger U C^\dagger U) &\qquad& \textrm{by Eq.~\eqref{eq:n_unitary}}\\
        &= \cN(C^\dagger U) \\
        &= \overline{\cN(U^\dagger C)} &\qquad& \textrm{by Eq.~\eqref{eq:n_adjoint}}\\
        &= \overline{\cN(B)} &\qquad& \textrm{by Eq.~\eqref{eq:cnBC}} \, .
    \end{alignedat}
    \end{align}
    The claim then follows from Eq.~\eqref{eq:defnumericalrange}.
\end{proof}

The following is an immediate consequence of Lemma~\ref{lem:gateerrorinnernumericalradius}: the logical gate error of a unitary implementation~$\cW(\rho)=W\rho W^\dagger$ for a logical unitary $U$ is equal to the logical gate error of the adjoint map $\cW^\dagger$ for the adjoint logical unitary $U^\dagger$.
\begin{corollary}[Gate error of unitary implementations and their adjoints]\label{cor:gateerrorunitaryimplement}
Let~$W: \cH_{in} \rightarrow \cH_{out}$ be unitary and~$\cW(\rho) = W \rho W^\dagger$. Then 
\begin{align}
\gateerror_{\cLout,\cLin} (W^\dagger,U^\dagger) &=  \gateerror_{\cLin,\cLout} (W,U)\,.
\end{align}
\end{corollary}
\begin{proof}
Let $B= B_{\cLin, \cLout}^U(W,U)$ and $\widetilde{B} = B_{\cL_{out}, \cL_{in}}^{U^\dagger}(W^\dagger, U^\dagger):\cLout\rightarrow\cLout$ be the maps introduced in Definition~\ref{def:diagonalunitary}.
Since~$W^\dagger$ is unitary,  Lemma~\ref{lem:gateerrorinnernumericalradius}  implies that
\begin{align}
\gateerror_{\cLout,\cLin} (W^\dagger,U^\dagger)&= 2 (1-\cn(\widetilde{B}))^{1/2}\\
 &= 2 (1-\cn(B))^{1/2}\\
 &= \gateerror_{\cLin,\cLout} (W,U)\, , 
\end{align}
where the second equality follows from Lemma~\ref{lem:cnadjoint} and the last equality is a consequence of Lemma~\ref{lem:gateerrorinnernumericalradius}.
\end{proof}

\subsection{Bounds on the gate error in terms of matrix elements \label{sec:gateerror_bound_matrixelem}}

In this section we state bounds on the logical gate error which only depend on matrix elements of the implementation and are thus easily applicable, see Theorem~\ref{cor:shortmatrixBstatement}.
 These bounds exploit the special form of the operator~$B$ whose Crawford number is relevant for the logical gate error (see Definition~\ref{def:diagonalunitary}).
In particular, the bounds are sensitive to sparsity of the operator $B$.  In the following, we call a matrix~$B\in\mathsf{Mat}_{d\times d}(\mathbb{C})$~$s$-sparse if the number of non-zero entries in each row and column is at most~$s$.

The following is our main technical tool to analyze logical gate errors for approximate GKP codes.

\begin{theorem}[Logical gate error of unitary implementations in terms of matrix elements]\label{cor:shortmatrixBstatement}
Let~$W:\cHin\rightarrow\cHout$ be unitary and~$\cW(\rho)=W\rho W^\dagger$. Let~$U:\mathbb{C}^d\rightarrow\mathbb{C}^d$ be unitary, and let~$B=B^{U}_{\cLin,\cLout}(W,U)$ be the operator introduced in Definition~\ref{def:diagonalunitary}.
Then we have the following.
\begin{enumerate}[(i)]
\item\label{it:ssparsityassumptiononea}
Suppose that~$B$ is~$s$-sparse and has positive entries~$\{B_{j,j}\}_{j=0}^{d-1}$ on the diagonal. 
Then
\begin{align}
\max\left\{\gateerror_{\cLin,\cLout} (\cW,U), \gateerror_{\cLout,\cLin} (\cW^\dagger,U^\dagger) \right\} &\leq 3 \left((1-\min_j B_{j,j})+(s-1)\max_{j,\ell:j\neq \ell }|B_{j,\ell}|\right)^{1/2}\ .
\end{align}
\item \label{it:claimsecondbassumptionb} Furthermore, we have 
\begin{align} 
    \max\left\{\gateerror_{\cLin,\cLout} (\cW,U), \gateerror_{\cLout,\cLin} (\cW^\dagger,U^\dagger) \right\} &\leq 8 d^{3/8} \left(\max_{j,k}\left| U_{j,k} - \langle \encoded{j}|_{\cLout} W |\encoded{k} \rangle_{\cLin} \right|\right)^{1/4}  \ .
\end{align}
\end{enumerate}
\end{theorem}

We refer to Section~\ref{app:gateerror_bound_matrixelem} in the Appendix for a rigorous proof of this statement. It heavily relies on bounds on the Crawford number of a matrix in terms of its matrix elements. We develop the corresponding technical results in Appendix~\ref{sec:numericalradius}.

We note that these techniques can also be used to establish a simple lower bound on the logical gate error in terms of matrix elements as follows.
\begin{lemma}[Lower bound on the logical gate error]\label{lem:lowerboundgateerrorbzerozero}
Let~$\cH$ be a Hilbert space and~$\cL \subset \cH$ a subspace. Let~$W:\cH\rightarrow\cH$ be unitary and~$\cW(\rho)=W\rho W^\dagger$. Let~$U:\mathbb{C}^d\rightarrow\mathbb{C}^d$ be unitary and let~$B=B^{U}_{\cL,\cL}(W,U)$ be the operator introduced in Definition~\ref{def:diagonalunitary} where we fix an isometric encoding map~$\encmap_\cL: \mathbb{C}^d \rightarrow \cL$. 
Then 
\begin{align} \label{eq:claimB00}
    \gateerror_\cL(\cW,U)&\geq 2\sqrt{1-|B_{0,0}|^2}\ .
    \end{align}
\end{lemma}
\begin{proof}
By Lemma~\ref{lem:gateerrorinnernumericalradius} we have

    \begin{align}
        \gateerror_\cL(W,U) = 2   \sqrt{1 - c(B)^2}
    \end{align}

Inequality~\eqref{eq:claimB00} then follows from~$\cn(B) \leq |B_{0,0}|$.
    \end{proof}

We use Lemma~\ref{lem:lowerboundgateerrorbzerozero} to establish our no-go result (see Section~\ref{sec: no go}) on linear optics implementations of Clifford gates in approximate GKP codes.

\section{Ideal and approximate GKP codes\label{sec:idealapproxgkpcode}}

In this section, we discuss GKP codes encoding a qudit into an oscillator. Concretely, in Section~\ref{sec:idealgkpcode} we define ideal GKP codes as superpositions of infinitely squeezed states.     
In Section~\ref{sec: approximate GKP codes} we consider physically meaningful approximations to ideal GKP codes.

To enhance readability, we denote a single-mode squeezing unitary by 
\begin{align}
    M_\alpha := e^{-i(\log\alpha) (QP+PQ)/2}\qquad \textrm{for} \qquad \alpha >0 \label{def:Malpha}\, ,
\end{align}
Here~$Q$ and~$P$ are the canonical position- and momentum-operators on~$L^2(\bbR)$, respectively.
This notation is used throughout the remainder of the text. The map~$M_\alpha$ acts on elements~$\Psi\in L^2(\mathbb{R})$ by 
\begin{align}
    \label{eq:sq_statesL2}
\left(M_\alpha \Psi\right)(x)=\frac{1}{\sqrt{\alpha}} \Psi(x / \alpha)\qquad\textrm{ for }x\in \mathbb{R}\ ,
\end{align}
and satisfies 
\begin{align}
\begin{aligned}
    \label{eq:sq_PQ}
    \left(M_\alpha\right)^\dagger P M_\alpha &= P/\alpha
    \qquad\text{ and }\qquad
    \left(M_\alpha\right)^\dagger Q M_\alpha = \alpha Q
    \end{aligned}\qquad\textrm{ for }\qquad \alpha>0 \ .
    \end{align}

\subsection{Definition of the ideal GKP code\label{sec:idealgkpcode}}
The ideal GKP code defined in Ref.~\cite{gkp} encodes a qudit into formal linear combinations of position-eigenstates of an oscillator. In more detail, let~$(Q,P)$ denote the canonical position- and momentum operators on~$L^2(\mathbb{R})$. 
For~$x\in\mathbb{R}$ we denote by~$\ket{{\bf x}}$ the position-eigenstate to eigenvalue~$x$, i.e., the distribution satisfying~$Q\ket{{\bf x}}=x\ket{{\bf x}}$.  The ideal GKP-code is then defined as follows. Let~$d\geq 2$ be an integer. Then we define 
\begin{align} \label{eq: def ideal gkp}
\ket{\gkp(j)_d}&=\sum_{s\in\mathbb{Z}}\posbos{\sqrt{\frac{2\pi}{d}}j + s\cdot \sqrt{2\pi d}}
\qquad\textrm{ for }\qquad j\in\mathbb{Z}_d := \{0,\ldots, d-1\} \ .\end{align}
Formally, we can think of~$\ket{\gkp(j)_d}$ as being obtained from the integer-spaced
GKP-state (distribution)
\begin{align}
\ket{\gkp}&=\sum_{s\in\mathbb{Z}} \ket{\encoded{s}}\ \label{eq:formalgkpidealinteger}
\end{align}
by changing the spacing to~$\sqrt{2\pi d}$ by means of a squeezing unitary~$M_{\sqrt{2\pi d}}$, and subsequent translation by~$j\sqrt{2\pi /d}$, that is,
\begin{align}
\ket{\gkp(j)_d}&=e^{-ij\sqrt{2\pi/d}P} M_{\sqrt{2\pi d}}\ket{\gkp}\qquad\textrm{ for }\qquad j\in\mathbb{Z}_d\ .\label{eq:generalizmhms}
\end{align}
The states (distributions)~$\{\ket{\gkp(j)}_d\}_{j\in\mathbb{Z}_d}$ form a basis of the (ideal) GKP-code ~\[\gkpcode{}{}{d}:= \mathrm{span}\{\ket{\gkp(j)_d}\}_{j\in\mathbb{Z}_d} \, . \] 
The code has stabilizer generators
\begin{align}
S^{(1)}_d=e^{i\sqrt{2\pi d}Q}\qquad\textrm{ and }\qquad S^{(2)}_d=e^{-i\sqrt{2\pi d}P}\ 
\end{align}
and encodes a~$d$-dimensional qudit. We use the encoding map
\begin{align}
\begin{matrix}
\encmapgkp[d]: & \mathbb{C}^d & \rightarrow & \gkpcode{}{}{d}\\
 & \ket{j} & \mapsto &\ket{\gkp(j)_d}
 \end{matrix}\qquad\textrm{ for }\qquad j\in\mathbb{Z}_d\ ,
\end{align}
and the corresponding decoding map defined by
\begin{align}
\begin{matrix}
\decmapgkp[d]=(\encmapgkp[d])^{-1}:&\gkpcode{}{}{d}&\rightarrow& \mathbb{C}^d\\
&\ket{\gkp(j)_d}& \rightarrow & \ket{j} 
\end{matrix}\qquad\textrm{ for }\qquad j\in \mathbb{Z}_d\ .
\end{align}

\subsection{Definition of approximate GKP codes} \label{sec: approximate GKP codes} 
In this section, we introduce approximate GKP codes. We first review the standard definition of
finitely squeezed GKP-states (see Section~\ref{sec:nonorthogonalapproximateGKP}). 
In Section~\ref{sec:truncatedapproximategkpcode} 
we introduce truncated approximate GKP-states which are pairwise orthogonal and more suitable for computational purposes.
Finally, in Section~\ref{sec:symmetricallysqueezedapproximateGKP} 
we discuss symmetrically squeezed (approximate) GKP codes. 

\subsubsection{Non-truncated (non-orthogonal) approximate GKP-states\label{sec:nonorthogonalapproximateGKP}}
Following the notation of Ref.~\cite{brenner2024complexity}, for~$\kappa,\Delta>0$, let~$\eta_\kappa,\Psi_\Delta\in L^2(\mathbb{R})$ be the centered Gaussians defined by
\begin{align} \label{eq:envelopepeakfunctions}
    \eta_\kappa(x)&=\frac{\sqrt{\kappa}}{\pi^{1 / 4}} e^{-\kappa^2 x^2 / 2} \ , \\
    \Psi_{\Delta}(x)&=\frac{1}{\pi^{1/4}\sqrt{\Delta}}e^{-x^2 /\left(2 \Delta^2\right)} \ , \label{eq:singlepeakfunctionvad}
\end{align}
for~$x\in\bb{R}$. We note that~$\eta_\kappa$ and~$\Psi_\Delta$ are normalized such that~$\|\eta_\kappa\|=\|\Psi_\Delta\|=1$. 

Replacing the position-eigenstate (distribution)~$\ket{\encoded{s}}$ in Eq.~\eqref{eq:formalgkpidealinteger}
by a width-$\Delta$ Gaussian centered at~$s$, i.e.,
the state~$\ket{\chi_\Delta(s)}$ defined by 
\begin{align}
    (\chi_\Delta(s))(x) &= \Psi_\Delta(x-s) \qquad\textrm{for }\qquad x\in \mathbb{R}\ ,
\end{align} 
 and introducing a Gaussian envelope~$\eta_{\kappa}$ (with width~$1/\kappa$), we obtain   the element
\begin{align}
\ket{\gkp_{\kappa,\Delta}}&=C_{\kappa,\Delta}\sum_{s\in \mathbb{Z}} \eta_\kappa(s)\ket{\chi_\Delta(s)}
\label{eq:nontruncatedapproximateGKP}
\end{align}
of~$L^2(\mathbb{R})$. Here the constant~$C_{\kappa,\Delta}>0$ is chosen such that this is normalized, i.e., a state. 
Using~$\ket{\gkp_{\kappa,\Delta}}$ instead of~$\ket{\gkp}$ in Eq.~\eqref{eq:generalizmhms},  we obtain the family of states
\begin{align}
\ket{\gkp_{\kappa,\Delta}(j)_d}&=e^{-ij\sqrt{2\pi/d}P} M_{\sqrt{2\pi d}}\ket{\gkp_{\kappa,\Delta}}\qquad\textrm{ for }\qquad j\in\mathbb{Z}_d\ .
\end{align}
The (untruncated) approximate GKP code then is the space~$\gkpcode{\kappa,\Delta}{}{d}=\myspan\{\gkp_{\kappa,\Delta}(j)_d\}_{j\in\mathbb{Z}_d}$.

We note that this a commonly used definition of an  approximate GKP code~\cite{Terhal_2016,DuivenvoordenSingleMode17,Weigand_2018,Hastrup_2021,brenner2024factoring}. It was shown in Ref.~\cite{Matsuura_2020} that the underlying basis states are close to Fock-damped GKP states~\cite{MenicucciCVCluster} for suitable parameters.
 However, it has the drawback that the states~$\{\gkp_{\kappa,\Delta}(j)_d\}_{j\in\mathbb{Z}_d}$ are not pairwise orthogonal. We address this issue in the next section by introducing truncated GKP-states. 

\subsubsection{Truncated approximate GKP-states and approximate GKP codes\label{sec:truncatedapproximategkpcode}}
In this section, we introduce an orthonormal family~$\{\gkp_{\kappa,\Delta}^\varepsilon(j)_d\}_{j\in\mathbb{Z}_d}$  of~$L^2$-functions which we consider in place of the (untruncated)  GKP-states introduced in Section~\ref{sec:nonorthogonalapproximateGKP}. We note that these states  span a slightly different subspace~$\gkpcode{\kappa,\Delta}{\varepsilon}{d}:=\myspan\{\gkp_{\kappa,\Delta}^\varepsilon(j)_d\}_{j\in\mathbb{Z}_d}$ 
(i.e., they do not form a basis of~$\gkpcode{\kappa,\Delta}{}{d}$). However, we argue that they are close to the previous states~$\{\gkp_{\kappa,\Delta}(j)_d\}_{j\in\mathbb{Z}_d}$  for parameters of interest. As a consequence, we will subsequently focus on the code~$\gkpcode{\kappa,\Delta}{\varepsilon}{d}$ only.  

Fix~$\kappa,\Delta > 0$.
For~$\varepsilon\in (0,1/2)$, we define the state
\begin{align} \label{eq:deftruncatedGassian}
\Psi_{\Delta}^{\varepsilon}=\frac{\Pi_{[-\varepsilon, \varepsilon]} \Psi_{\Delta}}{\left\|\Pi_{[-\varepsilon, \varepsilon]} \Psi_{\Delta}\right\|} \ ,
\end{align}
where~$\Pi_{[-\varepsilon,\varepsilon]}$ is the orthogonal projection onto the subspace of~$L^2(\mathbb{R})$ of functions with support contained in the interval~$[-\varepsilon,\varepsilon]$.
We define the translated truncated Gaussians~$\chi^\varepsilon_\Delta$ by 
\begin{align}
    (\chi^\varepsilon_\Delta(z))(x) &= \Psi_\Delta^\varepsilon(x-z) \qquad\textrm{for }\qquad x\in \mathbb{R}\ .
\end{align} 
The~$\varepsilon$-truncated approximate GKP-state~$\gkp_{\kappa,\Delta}^\varepsilon\in L^2(\mathbb{R})$ is defined as
\begin{align} \label{eq:GKP_eps}
\ket{\gkp^\varepsilon_{\kappa,\Delta}}&=C_{\kappa} \sum_{z \in \mathbb{Z}} \eta_\kappa(z)|\chi^\varepsilon_{\Delta}(z)\rangle\ ,
\end{align}
where~$C_\kappa>0$ is such that this function is normalized. 
That is, $\gkp_{\kappa,\Delta}^{\varepsilon}$ has support~$\mathsf{supp}(\gkp_{\kappa,\Delta}^{\varepsilon}) = \mathbb{Z}(\varepsilon)$ where
\begin{align}
\mathbb{Z}(\varepsilon)&=\{x\in\mathbb{R}\ |\ \mathsf{dist}(x,\mathbb{Z})\leq \varepsilon\}\qquad \textrm{where} \qquad \mathsf{dist}(x, \mathbb{Z}) = \min_{z\in \mathbb{Z}}|x -z|\, .
\end{align}
We note that~$1/\kappa$ determines the width of the envelope~$\eta_\kappa$ and~$\Delta$ the width of the individual peaks, see Fig.~\ref{fig:GKP}.

\begin{figure}[H]
    \centering
    \includegraphics[width = 0.8 \textwidth]{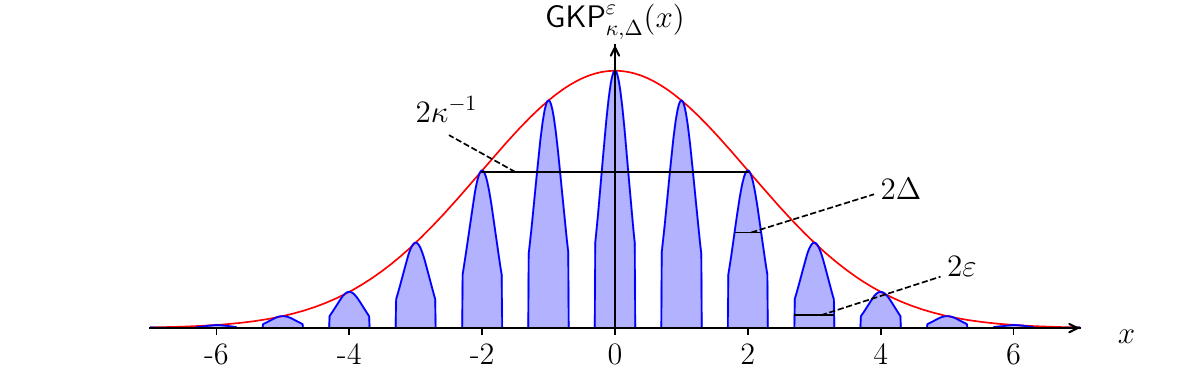}
    \caption{The approximate GKP-state~$\ket{\gkp_{\kappa, \Delta}^\varepsilon}$ in position-space. The red line represents the envelope~$\eta_\kappa(x)\propto e^{-\kappa^2x^2/2}$ of the state, a Gaussian  with variance~$\kappa^{-2}$.    
    According to our convention, this function has~$\varepsilon$-truncated Gaussian peaks of variance~$\Delta^2$ centered at all integers.   
    \label{fig:GKP}}
\end{figure}
Let~$d\in\mathbb{N}$ be an integer. Let
\begin{align}
\varepsilon\le 1/(2d)\ .\label{eq:varepsilonchoice}
\end{align}
We define the states
\begin{align}
    \ket{\gkp^\varepsilon_{\kappa,\Delta}(j)_d} = e^{-ij \sqrt{2\pi/d}P} M_{\sqrt{2\pi d}} \ket{\gkp^\varepsilon_{\kappa,\Delta}} \qquad \textrm{for} \qquad j \in \{0,\dots, d-1\}\, .
\end{align}
Then we can embed~$\mathbb{C}^d$ into an oscillator using a map defined as
\begin{align}
\begin{matrix}
\encodergkp{\kappa,\Delta}{\varepsilon}{d}: & \mathbb{C}^d & \rightarrow & L^2(\mathbb{R})\\
 & \ket{j} & \mapsto &\ket{\gkp^\varepsilon_{\kappa,\Delta}(j)_d}
 \end{matrix}\qquad\textrm{ for }\qquad j\in \{0,\ldots,d-1\}\ ,
\end{align}
linearly extended to  all of~$\mathbb{C}^d$. 
The assumption~\eqref{eq:varepsilonchoice} on~$\varepsilon$ ensures that 
the states~$\{\ket{\gkp^\varepsilon_{\kappa,\Delta}(j)_d}\}_{j=0}^{d-1}$ 
are pairwise orthogonal, implying that~$\encmapgkp_{\kappa,\Delta}^{\varepsilon}[d]$ is isometric. The states form an orthonormal basis of the code space
\begin{align} \label{eq:defGKPcodespace}
\gkpcode{\kappa,\Delta}{\varepsilon}{d}&:=\encodergkp{\kappa,\Delta}{\varepsilon}{d} (\mathbb{C}^d)\ .
\end{align}
We denote the corresponding decoding map by
\begin{align}
\begin{matrix}
\decodergkp{\kappa,\Delta}{\varepsilon}{d}=(\encodergkp{\kappa,\Delta}{\varepsilon}{d})^{-1}:&\gkpcode{\kappa,\Delta}{\varepsilon}{d}&\rightarrow& \mathbb{C}^d\\
&\ket{\gkp^\varepsilon_{\kappa,\Delta}(j)_d} & \rightarrow & \ket{j}
\end{matrix}\qquad\textrm{ for }\qquad j\in \{0,\ldots,d-1\}\ .
\end{align}

Let us briefly comment on suitable choices of the truncation parameter~$\varepsilon$ and the relationship between the code~$\gkpcode{\kappa,\Delta}{}{d}$
introduced in the previous section and the subspace~$\gkpcode{\kappa,\Delta}{\varepsilon}{d}\subset L^2(\mathbb{R})$ spanned by the truncated GKP-states. For suitable parameter choices of~$(\Delta,\varepsilon)$ (specifically, small values of~$\Delta$ and~$\Delta/\varepsilon$), the code state~$\ket{\gkp_{\kappa,\Delta}(j)_d}\in \gkpcode{\kappa,\Delta}{}{d}$ is well-approximated by the truncated GKP-state~$\ket{\gkp_{\kappa,\Delta}^\varepsilon(j)_d}\in\gkpcode{\kappa,\Delta}{\varepsilon}{d}$ for each~$j\in \mathbb{Z}_d$. A corresponding quantitative statement is given in 
 Lemma~\ref{lem:truncatedapproximateGKPstates} in the appendix for the integer-spaced states~$\ket{\gkp_{\kappa,\Delta}}$ and~$\ket{\gkp^\varepsilon_{\kappa,\Delta}}$, and implies the corresponding statement
 for~$\ket{\gkp_{\kappa,\Delta}(j)_d}$ and~$\ket{\gkp^\varepsilon_{\kappa,\Delta}(j)_d}$
  because unitaries preserve inner products. In particular, this closeness guarantee is tightest  for the choice
\begin{align}
\varepsilon_d&:=1/(2d)\ .\label{eq:optimaltruncationparametergeneraldefinition}
\end{align}
Correspondingly, we call~$\varepsilon_d$ the  optimal truncation parameter.

Here we are interested in the relationship between the amount of squeezing (expressed by~$(\kappa,\Delta)$) and the accuracy of an implementation (expressed by the logical gate error). In other words, we typically fix the code dimension~$d\geq 2$ and are interested in the (asymptotic) scaling of the gate error in the limit of large squeezing (corresponding to~$(\kappa,\Delta)\rightarrow (0,0)$). In this setup, we can use
the optimal truncation parameter~$\varepsilon=\varepsilon_d$ (Eq.~\eqref{eq:optimaltruncationparametergeneraldefinition}),
and  study the code~$\gkpcode{\kappa,\Delta}{\varepsilon_d}{d}$ instead of~$\gkpcode{\kappa,\Delta}{}{d}$. Correspondingly, we henceforth concentrate on codes of the form~$\gkpcode{\kappa,\Delta}{\varepsilon}{d}$ (and typically make the choice~$\varepsilon=\varepsilon_d)$. To emphasize this, let us write
\begin{align}
\gkpcode{\kappa,\Delta}{\star}{d}:=\gkpcode{\kappa,\Delta}{\varepsilon_d}{d}=\gkpcode{\kappa,\Delta}{1/(2d)}{d}\ .
\end{align}
Then~$\{\gkpcode{\kappa,\Delta}{\star}{d}\}_{(\kappa,\Delta)}$ is a two-parameter family of codes parameterized by the squeezing parameters~$\kappa>0$ and~$\Delta>0$. We will usually refer to this as the approximate GKP code in the following.

\subsubsection{Symmetrically squeezed approximate GKP codes \label{sec:symmetricallysqueezedapproximateGKP}}
So far, we have considered GKP codes parametrized by  potentially unrelated squeezing parameters~$(\kappa,\Delta)$. 
Symmetrically squeezed GKP codes are defined using a single squeezing parameter~$\kappa>0$ and a truncation parameter~$\varepsilon>0$. They are obtained by setting~$\Delta$ equal to
\begin{align}
\Delta = \kappa / (2 \pi d)\ ,\label{eq:symmetricsqueezinggeneraldef}
\end{align}
 i.e., we define
\begin{align}
\gkpcode{\kappa}{\varepsilon}{d} &:=
\gkpcode{\kappa,\kappa/(2\pi d)}{\varepsilon}{d}
\end{align}
and write associated encoding- and decoding-maps as
\begin{align}
\encmapgkp_{\kappa}^{\varepsilon}[d]&:=\encmapgkp_{\kappa,\kappa/(2\pi d)}^\varepsilon[d] \ ,\\
\decmapgkp_{\kappa}^{\varepsilon}[d]&:=\decmapgkp_{\kappa,\kappa/(2\pi d)}^\varepsilon[d]\ .
\end{align}
The special choice made in Eq.~\eqref{eq:symmetricsqueezinggeneraldef} is motivated by the fact that this provides a particularly simple linear optics implementation of the Fourier transform (see Theorem~\ref{thm:result2theorem}) below.

The main code of interest here is the symmetrically squeezed approximate GKP code with optimal truncation parameter~$\varepsilon_d$ (cf.~Eq.~\eqref{eq:optimaltruncationparametergeneraldefinition}),
which we will denote by
\begin{align}
\gkpcode{\kappa}{\star}{d}:=\gkpcode{\kappa}{\varepsilon_d}{d}=\gkpcode{\kappa,\kappa/(2\pi d)}{1/(2d)}{d}\ .
\end{align}
We denote the corresponding encoding- and decoding-map as 
$\encmapgkp_{\kappa}^{\star}$ and~$\decmapgkp_{\kappa}^{\star}$, respectively. We note that for a fixed code space dimension~$d\geq 2$, this defines a one-parameter family~$\{\gkpcode{\kappa}{\star}{d}\}_{\kappa>0}$ of symmetrically squeezed approximate GKP codes.

\section{Linear optics implementations of Cliffords in GKP codes} \label{sec: logical gates in ideal GKP codes}

In Section~\ref{sec:physicalsetupv}, we describe the underlying physical system and specify the set of bounded-strength Gaussian unitaries. In Section~\ref{sec:logicalops} we define the logical operations on qudits we consider. 
In Section~\ref{sec:logicalpaulicliffordideal}, we review the known (linear optics) implementations of logical Pauli- and Clifford operations in the ideal GKP code. 

\subsection{Bounded-strength Gaussian unitary operations\label{sec:physicalsetupv}}
Here we introduce Gaussian (linear optics) operations in multimode systems and associated unitary operations. 
In the following, we typically consider systems with~$n$ oscillators and associated Hilbert space~$\cH_n=L^2(\mathbb{R})^{\otimes n}$.
The set of canonical mode operators will be denoted~$R=(Q_1,P_1,\ldots,Q_n,P_n)$. 
(We write~$(Q,P)$ instead of~$(Q_1,P_1)$ if~$n=1$.) They satisfy the canonical commutation relations
\begin{align}
\left[R_k, R_{\ell}\right]=i J_{k, \ell} I_{\cH_n}\quad \text { where } \quad J=\left(\begin{array}{cc}
0 & 1 \\
-1 & 0
\end{array}\right)^{\oplus n}\  .
\end{align}
By the Campbell-Baker–Hausdorff formula, these imply the important relation
\begin{align}
W(\xi)^\dagger
 R_k W(\xi)=R_k+\xi_kI_{\cH_n} \quad \text { for all }\quad k=1, \ldots, 2 n \quad\textrm{ where }\quad W(\xi)=e^{-i\xi\cdot JR}
\label{eq:displacementrelationsetup}
\end{align}
for any vector~$\xi\in\mathbb{R}^{2n}$, where~$\xi\cdot \eta:=\sum_{j=1}^{2n} \xi_j \eta_j$ denotes the inner product of two vectors~$\xi,\eta\in \mathbb{R}^{2n}$, i.e., $\xi\cdot JR=\sum_{j,k=1}^n J_{j,k}\xi_jR_k$.  Because of Eq.~\eqref{eq:displacementrelationsetup}
a unitary of the form~$e^{-i\xi\cdot JR}$ is called a Weyl operator or (phase space) displacement operator. We say that it is of constant strength if~$\|\xi\|$ is bounded by a constant (independent of other parameters  such as the number of modes).

More generally, we consider unitaries of the form
\begin{align}
U(A,\xi)=e^{iR\cdot AR-i\xi\cdot JR}\qquad\textrm{ where }\qquad A=A^T\in\mathsf{Mat}_{2n\times 2n}(\mathbb{R})\ 
\end{align}
generated by Hamiltonians which are quadratic in the mode operators.  Similarly as for displacement operators, $U(A,\xi)$ is called a bounded-strength unitary of both~$\|A\|$ and~$\|\xi\|$ are bounded by a constant.

\subsection{Paulis and Cliffords for (logical) qudits\label{sec:logicalops}}

In this section, we define the (logical) operations we will consider subsequently.
Throughout, we consider qudit systems with a Hilbert space~$\mathbb{C}^d$ of dimension~$d\geq 2$. We use an orthonormal (``computational'') basis~$\{\ket{x}\}_{x\in\mathbb{Z}_d}$ indexed by residues modulo~$d$, which we often identify with the set~$\mathbb{Z}_d=\{0,\ldots,d-1\}$.

 The (generalized)  Pauli operators~$X,Z$ for a qudit are  defined by their action
 \begin{align}
 \begin{matrix}
 X\ket{x}&=&\ket{x\oplus 1}\\
 Z\ket{x}&=&\omega_d^x \ket{x}
 \end{matrix}\qquad\textrm{ for all }x\in\mathbb{Z}_d\ ,
 \end{align}
 where~$\oplus$ denotes addition modulo~$d$ and~$\omega_d=e^{2\pi i/d}$ is a~$d$-th root of unity. The Pauli operators~$X,Z$ satisfy
 \begin{align}
ZX&=\omega_d X Z
 \end{align}
 and  generate the single-qudit Pauli group.

The Clifford group is the normalizer of the Pauli group in the unitary group (and thus  contains the Pauli group). Ignoring global phases, the single-qudit Clifford group is generated by the gates~$\{\Pgate,\Fgate\}$, where~$\Pgate$ is the phase gate and~$\Fgate$ the Fourier transform. These act on computational basis states as 
 \begin{align}
 \begin{matrix}
 \Pgate \ket{x}&=&\omega_d^{x(x+c_d)/2}\ket{x}\\
 \Fgate \ket{x}&=&\frac{1}{\sqrt{d}}\sum_{y\in\mathbb{Z}_d}\omega_d^{xy}\ket{y}
 \end{matrix}\qquad\textrm{ for all } x\in\mathbb{Z}_d\ ,
 \end{align}
 where~$c_d=0$ if~$d$~is even and~$c_d=1$ otherwise. 
  
To generate the~$n$-qubit Clifford group we also need the two-qubit controlled phase gate~$\CZ$ which acts on~$\bbC^d \otimes \bbC^d$ as 
\begin{align}
    \CZ (\ket{x} \otimes \ket{y}) = \omega_d^{xy} \ket{x} \otimes \ket{y} \quad \text{ for all } x,y\in\mathbb{Z}_d \ .
\end{align}
The considered set of gates is illustrated in Fig.~\ref{tab:singlequditops}. We note that for~$d=2$, 
the phase gate~$\Pgate= Z(e^{i\pi/2})$ acts as~$\Pgate\ket{x}=i^x\ket{x}$ and is usually denoted as~$S$, and the Fourier transform~$\Fgate$ is  the familiar Hadamard gate~$H$.
\begin{table}[H]
\renewcommand{\arraystretch}{1.3} 
\setlength{\tabcolsep}{10pt} 
\rowcolors{1}{white}{gray!15} 
    \centering
    \begin{center}
    \begin{tabular}{c|c|c}
        \textbf{diagram} & \textbf{gate} &  \textbf{type of unitary} \\ \hline 
         \hspace{-1.5cm}\includegraphics{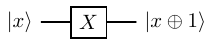}  &  \raisebox{0.25cm}{modular  shift~$X$}  & \raisebox{0.25cm}{Pauli}\\
        \hspace{-1.65cm}\includegraphics{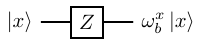}  &\raisebox{0.25cm}{ modular phase shift~$Z$} & \raisebox{0.25cm}{Pauli}\\
        \hspace{-0.6cm}\includegraphics{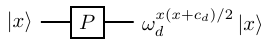}  & \raisebox{0.25cm}{phase gate~$\Pgate$}      & \raisebox{0.25cm}{Clifford} \\ 
        \raisebox{0.0cm}{\includegraphics{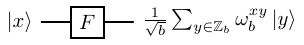}}  &  \raisebox{0.25cm}{Fourier transform~$\Fgate$} & \raisebox{0.25cm}{Clifford} \\
        \raisebox{0.0cm}{\includegraphics{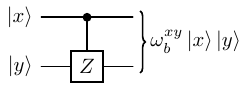}}  &    \raisebox{0.5cm}{controlled modular phase ~$\cZ$} & \raisebox{0.5cm}{Clifford}                
    \end{tabular}        
    \end{center}
\caption{Elementary single- and two-qudit operations defined using an orthonormal basis~$\{\ket{x}\}_{x\in\mathbb{Z}_d}$ of~$\mathbb{C}^d$. Here we distinguish Pauli from Clifford unitaries that are not Pauli.}\label{tab:singlequditops}
\end{table}

\subsection{Implementations of logical Cliffords\label{sec:logicalpaulicliffordideal}}
Consider the (ideal) GKP code~$\gkpcode{}{}{d}$ encoding a~$d$-dimensional qudit into a single oscillator.
As discussed in Ref.~\cite{gkp}, logical Pauli- and Clifford operators for this code have exact realizations by Gaussian unitaries, i.e., unitaries generated by linear or quadratic Hamiltonians in the mode operators. This is expressed by the following:

\begin{lemma}[Logical gates in ideal GKP codes] \label{lem:cliffordgkpphysicalI} Let~$d\in\mathbb{N}$ be arbitrary. 
Each logical unitary~$U \in \{X,Z,\Pgate,\Fgate, \CZ\}$ can be implemented exactly in the ideal GKP code~$\gkpcode{}{}{d}$ (respectively~$\gkpcode{}{}{d}^{\otimes 2}$) by a bounded-strength Gaussian unitary~$W_U$, see Table~\ref{fig:logicalopsgkpone}.
\begin{center}
\begin{table}[H]
    \centering
    \renewcommand{\arraystretch}{1.3} 
    \setlength{\tabcolsep}{10pt} 
    \rowcolors{1}{white}{gray!15} 
    \centering
    \resizebox{\textwidth}{!}{\begin{tabular}{c|c}
        \textbf{logical unitary~$U$} & \textbf{implementation~$W_U$}  \\ \hline
        \hspace{-1.5cm}\includegraphics[scale=0.9]{Xgate.pdf}  &           \hspace{-2.2cm} \includegraphics[scale=0.9]{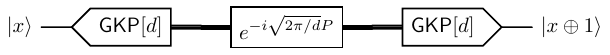}  \\
        \hspace{-1.7cm}\includegraphics[scale = 0.9]{Zgate.pdf}  & \hspace{-2.5cm} \includegraphics[scale=0.9]{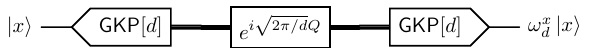}   \\
        \hspace{-0.7cm}\includegraphics[scale=0.9]{Pgate.pdf}  & \hspace{-0.2cm} \includegraphics[scale=0.9]{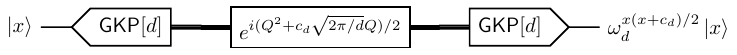}   \\
                            \hspace{-0.2cm}\raisebox{0.0cm}{\includegraphics[scale=0.9]{Fgate.pdf}}  &                       \hspace{-0.2cm}        \raisebox{0.0cm}{\includegraphics[height=0.8cm]{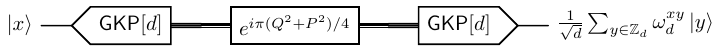}}\\
                            \hspace{-0.2cm}\hspace{-0.9cm} \raisebox{0.2cm}{\includegraphics[scale= 0.9]{CZ.pdf}} & \hspace{-2.7cm} \includegraphics[scale=0.9]{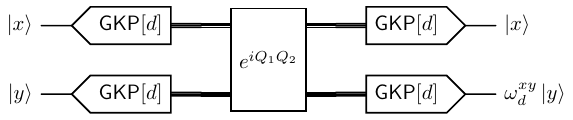}
       \end{tabular}}
       \caption{Exact implementations of logical Pauli- and Clifford group generators by Gaussian unitaries in the ideal GKP code~$\gkpcode{}{}{d}$.}
       \label{fig:logicalopsgkpone}
\end{table}
\end{center}
\end{lemma}

\begin{proof}
It is straightforward to check that the  logical Pauli operators can be implemented by
\begin{align}
Z= \decmapgkp[d]  e^{i  \sqrt{2 \pi/d}Q}  \encmapgkp[d] \qquad\textrm{ and }\qquad  X =\decmapgkp[d]   e^{-i \sqrt{2 \pi/d}P}  \encmapgkp[d]\ ,
\end{align}
whereas the Clifford group generators have the realizations
\begin{align}
\Pgate &= \decmapgkp[d]  e^{i(Q^2+c_d\sqrt{2\pi/d} Q)/2}  \encmapgkp[d] \ , \\
\Fgate&= \decmapgkp[d]  e^{i\pi (Q^2+P^2)/4}  \encmapgkp[d] \ , \\
\CZ&=\decmapgkp[d]  e^{iQ_1Q_2}  \encmapgkp[d]\ .
\end{align}
All Gaussian unitaries in these expressions are of bounded strength for any~$d\in\mathbb{N}$, implying the claim.
\end{proof}

\section{Gate error of logical Clifford implementations} \label{sec: bound gkp cliffords}

In this section we state our main result (Result~\ref{thm:result2} in the introduction) concerning the physical implementation of logical gates in approximate GKP codes: We quantify -- in terms of the logical gate error -- how well generators of the Clifford group can be implemented in a symmetrically squeezed  GKP code, i.e., the truncated code~$\gkpcode{\kappa}{\star}{d}$ with squeezing parameter~$\kappa>0$ and qudit dimension~$d\in\mathbb{N}$ (see Section~\ref{sec:symmetricallysqueezedapproximateGKP}). 

In more detail, we consider the Clifford group generators~$\{X,Z,\Fgate\}$. We show the following:
\begin{enumerate}[(i)]
\item\label{it:onepartproofm}
The standard linear optics implementations of the Pauli gates~$\{X,Z\}$ introduced for ideal GKP codes in Section~\ref{sec:logicalpaulicliffordideal} are accurate implementations for the code~$\gkpcode{\kappa}{\star}{d}$, with a logical gate error independent of~$d$ and depending  linearly on~$\kappa$.
\item\label{it:secondpartproofm}
 Similarly, we show that the linear optics implementation of the Fourier transform~$\Fgate$ is also suitable for the symmetrically squeezed GKP code~$\gkpcode{\kappa}{\star}{d}$. However, here we can only establish an upper bound on the logical gate error which depends polynomially on the qudit dimension~$d$ and the squeezing parameter~$\kappa$. 
 \end{enumerate}

We note that the asymptotic behavior of these bounds in the limit of large squeezing is not an artifact of the truncation: We show that the logical gate error also asymptotically vanishes for symmetrically squeezed untruncated approximate GKP codes, see Corollary~\ref{lem:gateerruntruncatedcode}.

The technically detailed version of
Result~\ref{thm:result2}  is the following.

\begin{theorem}[Detailed statement corresponding to Result~\ref{thm:result2}]\label{thm:result2theorem}
Let~$d\geq 2$ be an integer and~$\kappa \in (0,1/4)$.  Then the following holds for the symmetrically squeezed approximate GKP code~$\gkpcode{\kappa}{\star}{d}$. 
Consider the set~$\cG=\{X,Z,\Fgate\}$ of Cliffords on~$\mathbb{C}^{d}$. For~$U\in \cG$, let~$W_U$ be the linear
 optics implementation of~$U$ (see Table~\ref{fig:logicalopsgkpone}). 
These implementations have a logical gate error upper bounded as indicated in Table~\ref{tab:resulttwo}.
\begin{table}[H]
\renewcommand{\arraystretch}{1.3}
\centering
\small
\setlength{\tabcolsep}{8pt}
\rowcolors{1}{white}{gray!15}
\begin{tabular}{
    c|
    c|
    c
}
\textbf{logical unitary~$U$}  & \textbf{$\gateerror_{\gkpcode{\kappa}{\star}{d}}(W_U,U) \leq~$} & \textbf{established in} \\

$X, Z$  &~$3 \kappa$ & Lemma~\ref{lem: error XZF} \\

$\Fgate$ &~$21 d^{3/8}\kappa^{1/16}$ & Lemma~\ref{lem: gate error F} \\
\end{tabular}
\caption{List of implementations~$W_U$ for Cliffors~$U$ from Lemma~\ref{lem:cliffordgkpphysicalI} and associated logical gate errors (Theorem~\ref{thm:result2theorem}). 
\label{tab:resulttwo}}
\end{table}
\end{theorem}
The last column of Table~\ref{tab:resulttwo} indicates in which Lemma the corresponding upper bound on the logical gate error is established. 
The bounds are derived from  bounds on the matrix elements of these implementations (with respect to GKP code states).

In the remainder of this section, we establish Theorem~\ref{thm:result2theorem}. 
 
\begin{lemma}[Implementation of logical Pauli operators]\label{lem: error XZF} Let~$d \ge 2$ be an integer. Let~$\kappa\in (0,1)$.   Consider the symmetrically squeezed  GKP code~$\gkpcode{\kappa}{\star}{d}$. Then the following holds.
 \begin{enumerate}[(i)]
 \item\label{it:logicalXoperatorgaterrorbnd} The logical gate error of
 the linear optics implementation~$W_{X}=e^{-i\sqrt{2\pi/d}P}$ of the logical Pauli operator~$X$ 
satisfies~$\gateerror_{\gkpcode{\kappa}{\star}{d}}(W_{X}, X) \le 3 \kappa$.
 \item\label{it:logicalZoperatorgaterrorbnd}
The logical gate error of the linear optics implementation~$W_{Z}=e^{i\sqrt{2\pi/d}Q}$ of the logical Pauli operator~$Z$ satisfies~$\gateerror_{\gkpcode{\kappa}{\star}{d}}(W_{Z}, Z) \le  3\kappa$.
 \end{enumerate}
\end{lemma}
\begin{proof}
Let us write
\begin{align}
\ket{\encoded{j}}&=\ket{\gkp_{\kappa,\kappa/(2\pi d)}^{\varepsilon_d}(j)}\qquad\textrm{ for }\qquad  j\in \mathbb{Z}_d\ ,
\end{align} 
where we define~$\varepsilon_d = 1/(2d)$. Then~$\{\ket{\encoded{j}}\}_{j\in\mathbb{Z}_d}$ is an orthonormal basis of the space~$\cLin=\cLout=\gkpcode{\kappa}{\star}{d}$. 

The logical Pauli operator~$X$ has matrix elements~$X_{j,k}=\delta_{j,k\oplus 1}$ where~$\oplus$ denotes addition modulo~$d$. It follows that the matrix elements of the operator~$B$ in Definition~\ref{def:diagonalunitary} are given by
\begin{align}
B_{r,s}&=\sum_{j=0}^{d-1} \overline{X_{j,r}}\langle \encoded{j}|W_{X}|\encoded{s}\rangle\\
&=\langle \encoded{r\oplus 1}|W_{X} |\encoded{s}\rangle\\
&=M_{r\oplus 1,s}\qquad\textrm{ for }\qquad r,s\in \mathbb{Z}_d\ ,\label{eq:bmmatrixrel}
\end{align}
where~$M_{j,k}=\langle\encoded{j}|W_{X}|\encoded{k}\rangle$ 
are the matrix elements of the implementation~$W_X$. In Lemma~\ref{lem: matrix elements X gate}, we show that
these satisfy
\begin{align}
(1-\kappa^2)\cdot \delta_{j,k\oplus 1}\leq M_{j,k}\leq \delta_{j,k\oplus 1}\qquad\textrm{ for all }j,k\in\mathbb{Z}_d\ .\label{eq:mjkdleltakj}
\end{align}
It follows from Eqs.~\eqref{eq:bmmatrixrel} and~\eqref{eq:mjkdleltakj}
that
\begin{align}
B_{r,s}=0\qquad\textrm{ unless }\qquad r=s\ ,
\end{align}
i.e., the matrix~$B$ is~$1$-sparse. Furthermore, the diagonal entries~$B_{r,r}=M_{r\oplus 1,r}$ are real and satisfy
\begin{align}
B_{r,r}\geq 1-\kappa^2>0
\end{align}
according to Eq.~\eqref{eq:mjkdleltakj}. Corollary~\ref{cor:shortmatrixBstatement}~\eqref{it:ssparsityassumptiononea}
therefore implies that 
\begin{align}
\gateerror_{\gkpcode{\kappa}{\star}{d}}(W_X,X) &\leq 3 (1-\min_r B_{r,r})^{1/2}\leq 3\kappa\ .
\end{align}
This proves Claim~\eqref{it:logicalXoperatorgaterrorbnd}.
The proof of Claim~\eqref{it:logicalZoperatorgaterrorbnd} proceeds in an analogous fashion.  Recall that the symmetrically squeezed GKP code corresponds to the choice
\begin{align}
\Delta &=\frac{\kappa}{2\pi d}\ .\label{eq:symmetricchoicegkpde}
\end{align}
Using that~$Z_{j,k}=\omega_d^{j}\delta_{j,k}$ where~$\omega_d = e^{2\pi i / d}$, the matrix~$B$ in  Definition~\ref{def:diagonalunitary} is given by the entries 
\begin{align}
    B_{r,s}&=\sum_{j=0}^{d-1} \overline{Z_{j,r}}\langle \encoded{j}|W_Z |\encoded{s}\rangle\\
&= \overline{\omega_d^r}\langle \encoded{r}|W_Z |\encoded{s}\rangle\\
&= \overline{\omega_d^r} M_{r,s}\qquad\textrm{ for }\qquad r,s\in \mathbb{Z}_d\ ,
\end{align}
where~$M_{r,s}$ are the matrix elements of the implementation~$W_Z$. 
We show in Lemma~\ref{lem: matrix element Z gate} that these satisfy
\begin{align}
    (1- 10 d^2 \Delta^2 - 16(\Delta/\varepsilon)^4 )\cdot  \delta_{j,k} \le  \overline{\omega}_d^{j} \cdot M_{j,k}&\le \delta_{j,k}  
    \qquad \textrm{for all} \qquad j,k \in \mathbb{Z}_d\, ,
\end{align}
which implies that~$B$ is a diagonal (hence~$1$-sparse) matrix with positive diagonal entries satisfying
\begin{align}
    B_{r,r}&= \overline{\omega_d^r}M_{r,r}\\
    &\ge 1- 10 (d\Delta)^2 - 16 (\Delta/\varepsilon)^4\ .
\end{align}
It therefore follows from Corollary~\ref{cor:shortmatrixBstatement}~\eqref{it:ssparsityassumptiononea} that
\begin{align}
\gateerror_{\gkpcode{\kappa}{\star}{d}}(W_Z,Z) &\leq 3 
\left(10 (d\Delta)^2+16(\Delta/\varepsilon_d)^4\right)^{1/2}\\
&\leq 3\left(\sqrt{10}d\Delta+4(\Delta/\varepsilon_d)^2\right)\\
&\leq 3\left(\sqrt{10}\frac{\kappa}{2\pi}+4(\kappa/\pi)^2\right)\\
&\leq 3\kappa\ ,
\end{align}
where we used that~$a^2+b^2\leq (a+b)^2$ for~$a,b\geq 0$, inserted~$\Delta$ (see Eq.~\eqref{eq:symmetricchoicegkpde}) and~$\varepsilon_d=1/(2d)$ and used that~$\sqrt{10}/(2\pi) + 4/\pi^2 < 1$.  This is the claim.
\end{proof}

For the Fourier transform, we establish a slightly more general bound than that given in Theorem~\ref{thm:result2theorem} by also considering (some) non-symmetrically squeezed codes. (We use this generalization in Section~\ref{sec: no go}.)
We note that the Gaussian unitary~$e^{i\frac{\pi}{4}(Q^2+P^2)}$ 
(approximately) maps the input  code space~$\cL_{in}=\gkpcode{\kappa,\Delta}{\varepsilon}{d}$  to the output code space~$\cL_{out}=\gkpcode{\kappa',\Delta'}{\varepsilon}{d}$
(where the parameters~$(\kappa',\Delta')$ are linearly related to~$(\kappa,\Delta)$)
while implementing the (logical) Fourier transform, see Lemma~\ref{lem: matrix elements fourier}. This translates into the following  statement for the logical gate error, where we consider the optimal truncation parameter~$\varepsilon_d=1/(2d)$ only. 

\begin{lemma}[Implementation of the logical Fourier transform]\label{lem: gate error F}
Let~$d\geq 2$ be an integer, and let~$\kappa, \Delta >0$ be such that 
\begin{align}
\Delta &\leq \frac{\kappa}{2\pi d} \ .
\label{eq:dekadeltav}
\end{align}
Let~$(\kappa',\Delta')$ be defined by 
\begin{align}
\kappa'&=2\pi d\Delta \ , \\
\Delta'&=\kappa/(2\pi d)\ .
\end{align}
Then 
    \begin{align} \label{eq:FgateerrorkappaDelta}
         \gateerror_{\gkpcode{\kappa,\Delta}{\star}{d},\gkpcode{\kappa',\Delta'}{\star}{d}}\left(
        e^{i\pi(Q^2+P^2)/4},\Fgate
        \right)& \leq  21 d^{3/8}\kappa^{1/16}\ .
        \end{align}
In particular, we have 
    \begin{align}
         \gateerror_{\gkpcode{\kappa}{\star}{d}}\left(
        e^{i\pi(Q^2+P^2)/4},\Fgate
        \right)& \leq 21 d^{3/8}\kappa^{1/16}\  \label{eq:kappvarspdve}
        \end{align}
        for the symmetrically squeezed code~$\gkpcode{\kappa}{\star}{d}$,  for any~$\kappa>0$.         
\end{lemma}
\begin{proof}
We note that the claim about the symmetrically squeezed code~$\gkpcode{\kappa}{\star}{d}$
(where the optimal truncation parameter~$\varepsilon_d=1/(2d)$ is used)   follows immediately from
 the general claim because Condition~\eqref{eq:dekadeltav}  is satisfied for
the choice~$\Delta=\kappa/(2\pi d)$, for any~$\kappa>0$.

We note that the bound in Eq.~\eqref{eq:FgateerrorkappaDelta} is trivial if~$\kappa \ge 1/d^2$. Therefore, in the following, we assume that~$\kappa \in (0,1/d^2)$.
We are interested in the logical Fourier transform~$\Fgate$ defined by its matrix elements~$F_{j,k}=\omega_d^{jk}/\sqrt{d}$ and its (approximate) implementation by
\begin{align}
W_\Fgate&=e^{i\frac{\pi}{4}(Q^2 + P^2)}\ .
\end{align}
We use the 
orthonormal bases 
\begin{align}
    \ket{\encoded{k}_{in}} &= \ket{\gkp^{\varepsilon_d}_{\kappa ,\Delta}(k)_d} 
    \quad\textrm{ and }\quad
    \ket{\encoded{j}_{out}} = \ket{\gkp^{\varepsilon_d}_{\kappa' ,\Delta'}(j)_d} 
    \qquad \text{ for } \qquad j, k \in \bb{Z}_d \ \label{eq:notationshortvectors}
\end{align}
of~$\cLin=\gkpcode{\kappa,\Delta}{\star}{d}$ and~$\gkpcode{\kappa',\Delta'}{\star}{d}$, respectively, consisting of truncated approximate GKP-states with optimal truncation parameter~$\varepsilon_d=1/(2d))$ (using the notation introduced in Eq.~\eqref{eq:defGKPcodespace}).
We have
\begin{align}
    \left| \Fgate_{j,k} -  \langle\encoded{j}_{\cLout}|W_\Fgate|\encoded{k}_{\cLin}\rangle \right| 
    &=\left| \frac{1}{\sqrt{d}}\omega_d^{jk} -  \langle\encoded{j}_{\cLout}|W_{\Fgate}|\encoded{k}_{\cLin}\rangle \right| \leq  40 \kappa^{1/4}  
    \label{eq:upperboundkpafdv}
    \end{align}
according to Lemma~\ref{lem: matrix elements fourier}  (see Eq.~\eqref{eq:optimalsqueezingpmvd}). With Corollary~\ref{cor:shortmatrixBstatement}~\eqref{it:claimsecondbassumptionb} and Eq.~\eqref{eq:upperboundkpafdv} we conclude that 
\begin{align}
    \gateerror_{\gkpcode{\kappa,\Delta}{\star}{d},\gkpcode{\kappa',\Delta'}{\star}{d}}\left( W_{\Fgate},\Fgate \right)
    & \leq 8 d^{3/8}\left(40 \kappa^{1/4}\right)^{1/4}\leq  21 d^{3/8}\kappa^{1/16}\ .
\end{align}
\end{proof}

Theorem~\ref{thm:result2theorem} gives bounds on the logical gate error for implementations of Paulis and the Fourier transform in the symmetrically squeezed truncated approximate GKP code~$\gkpcode{\kappa}{\star}{d} = \gkpcode{\kappa,\kappa/(2\pi d)}{\varepsilon_d}{d}$. We now turn to the untruncated approximate GKP code~$\gkpcode{\kappa,\kappa/(2\pi d)}{}{d}$, and show that a good implementation~$W_U$ of a logical gate~$U$ in the truncated code also is a good implementation in the untruncated code. By good implementation we mean an implementation with vanishing logical gate error in the limit~$\kappa \rightarrow 0$ of infinite squeezing.
Specifically, the following result can be applied 
 to the implementations 
 of the logical gates~$X,Z$ and~$\Fgate$ given in Table~\ref{fig:logicalopsgkpone}. Together with Theorem~\ref{thm:result2theorem}, it implies that these gates have a logical gate error bounded by a polynomial function of~$d$ and~$\kappa$ in the symmetrically squeezed untruncated approximate GKP code. We note that this statement depends on using an appropriate encoding map for this code. 
\begin{corollary}
    \label{lem:gateerruntruncatedcode}
    Let~$d\geq 2$ be an integer and~$\kappa \in (0,1/4)$. Consider the symmetrically squeezed truncated and untruncated approximate GKP codes~$\cL_\kappa = \gkpcode{\kappa}{\star}{d}$ and~$\widetilde{\cL}_\kappa = \gkpcode{\kappa,\kappa/(2\pi d)}{}{d}$, respectively. Let~$W_U$ be an implementation of a logical gate~$U$ in~$\bbC^d$. 
    There exists an isometric encoding map~$\encmap_{\widetilde{\cL}_\kappa}: \mathbb{C}^d \rightarrow \widetilde{\cL}_\kappa$ such that 
    \begin{align} \label{eq:claimuntruncgkpimplementation1}
        \gateerror_{\widetilde{\cL}_\kappa}(W_U, U) \le \gateerror_{\cL_\kappa}(W_U, U) + 8 d^{3/4} \kappa^{1/4}\, .
    \end{align}
    In particular, if
    \begin{align}
        \lim_{\kappa \rightarrow 0} \gateerror_{\cL_\kappa}(W_U, U) = 0 \ ,
    \end{align}
    then 
    \begin{align}
        \label{eq:erruntruncatedGKP}
        \lim_{\kappa \rightarrow 0} \gateerror_{\widetilde{\cL}_\kappa}(W_U, U) = 0 \ . 
    \end{align}
In particular, the implementations~$W_X, W_Z$ and~$W_\Fgate$ given in Table~\ref{fig:logicalopsgkpone} for the logical gates~$X,Z$ and~$\Fgate$ satisfy Eq.~\eqref{eq:erruntruncatedGKP}. 
    \end{corollary}
    We prove Corollary~\ref{lem:gateerruntruncatedcode} by using a continuity statement about the logical gate error (see Lemma~\ref{lem:codecontinuity}) showing that the logical gate error depends continuously on the encoder used. 
    
\begin{proof}
The proof relies on the bounds established for the logical gate error in the truncated GKP code, combined with the fact that each untruncated code state~$|\gkp_{\kappa,\kappa/(2\pi d)} (j)_d \rangle \in \gkpcode{\kappa,\kappa/(2\pi d)}{}{d}$ is close to the corresponding truncated code state~$|\gkp_{\kappa,\kappa/(2\pi d)}^{\varepsilon_d} (j)_d  \rangle \in \gkpcode{\kappa}{\star}{d}$, for~$j\in\mathbb{Z}_d$.

    Due to Lemma~\ref{lem:truncatedapproximateGKPstates} we have 
    \begin{align}
        \langle\gkp_{\kappa,\kappa/(2\pi)}^{\varepsilon_d}, \gkp_{\kappa,\kappa/(2\pi)}\rangle &\ge 1 - 7\kappa/(2\pi d) - 2((\kappa/(2\pi d))/(1/(2d)))^4\\
        &\ge 1 - \kappa - 2(\kappa/\pi)^4 \\
        &\ge 1 - 2\kappa\, ,
    \end{align}
     where we used that by definition~$\varepsilon_d = 1/(2d)$ and the assumption~$\kappa \in (0,1/4)$.
    Let~$j \in \mathbb{Z}_d$. We have 
\begin{align}
        \left\| \gkp_{\kappa,\kappa/(2\pi)}^{\varepsilon_d}(j)_d - \gkp_{\kappa,\kappa/(2\pi)}(j)_d \right\| & = \left\| e^{-i \sqrt{2\pi/d}jP} M_{\sqrt{2\pi d}}\left(\gkp_{\kappa,\kappa/(2\pi)}^{\varepsilon_d} - \gkp_{\kappa,\kappa/(2\pi)} \right) \right\| \\
        &= \left\| \gkp_{\kappa,\kappa/(2\pi)}^{\varepsilon_d} - \gkp_{\kappa,\kappa/(2\pi)}  \right\|\\
        &= \sqrt{2\left(1 - \langle\gkp_{\kappa,\kappa/(2\pi)}^{\varepsilon_d}, \gkp_{\kappa,\kappa/(2\pi)}\rangle\right)}\\
        &\le 2\sqrt{\kappa}\, ,
    \end{align} where the second ideintity follows from the unitarity of~$e^{-i \sqrt{2\pi/d}jP} M_{\sqrt{2\pi d}}$.
According to Lemma~\ref{lem:symorthogonalization} there exists an orthonormal basis~$\{\xi_j\}_{j=0}^{d-1}$ of~$\widetilde{\cL}_\kappa$ such that 
    \begin{align}
        \left\| \gkp_{\kappa,\kappa/(2\pi d)}^{\varepsilon_d}(j)_d -  \xi_j \right\| \le 4 \sqrt{ d \kappa} \qquad \textrm{for all} \qquad j \in \mathbb{Z}_d\, .
    \end{align}
    Eq.~\eqref{eq:claimuntruncgkpimplementation1} then follows from the continuity statement for the logical gate error given in Lemma~\ref{lem:codecontinuity} applied with~$\phi_j= \gkp_{\kappa,\kappa/(2\pi d)}^{\varepsilon_d}(j)_d$ and~$\widetilde{\phi}_j = \xi_j$ for~$j \in \mathbb{Z}_d$ and~$\delta = 4\sqrt{d\kappa}$. 

    The claim about the implementations of~$W_X, W_Z$ and~$W_\Fgate$ follows from Eq.~\eqref{eq:erruntruncatedGKP} together with Theorem~\ref{thm:result2theorem}.
\end{proof}

\section{No-go result for linear optics implementations \label{sec: no go}} 

In this section, we prove Result~\ref{thm:result1}: We show that 
a standard linear optics implementation of a Clifford gate in the symmetrically squeezed GKP code has a constant logical gate error irrespective of the amount of squeezing. 
More precisely, using the example of the~$\Pgate$-gate, we show that the linear optics implementation introduced in the context of ideal GKP codes (see Section~\ref{sec:logicalpaulicliffordideal}) has a logical gate error lower bounded by a constant in the symmetrically squeezed (approximate) GKP code~$\gkpcode{\kappa}{\varepsilon}{d}$, even in the limit~$\kappa\rightarrow 0$ of infinite squeezing. 
We emphasize again that this behavior in the limit of large squeezing is not a peculiarity resulting from truncation:
even for untruncated GKP codes, the logical gate error of 
 the Gaussian $\mathsf{P}$-gate implementation
is lower bounded by a constant in the limit of large squeezing. For a  general statement which additionally extends this result to arbitrary encoding isometries, we refer to Theorem~\ref{thm:noencodingnogoP} in the Appendix.
  
  The following is a detailed version of Result~\ref{thm:result1} in the introduction.

\begin{theorem}[No-go result for linear optics implementation] \label{lem:nogoP}
Let~$d\ge 2$ be an integer and let~$c_d = d \mod 2$. Let~$\varepsilon\in (0,1/(2d)]$ and~$\kappa < 1/250$ be arbitrary. Then the following holds for the logical gate~$\Pgate$.
\begin{enumerate}[(i)]
\item \label{it:firstclaimnogoPone}We have
\begin{align}
     \gateerror_{\gkpcode{\kappa,\Delta}{\varepsilon}{d}}(e^{i(Q^2 +c_d\sqrt{2\pi/d}Q)/2}, \Pgate) &\geq \textfrac{1}{25}-32 (\Delta/\varepsilon)^4\  
\end{align}
for all~$\Delta$ satisfying
\begin{align}
1/2 \ge 2\pi d\Delta \geq \kappa .\label{eq:pideltakappa}
\end{align}

\item \label{it:firstclaimnogoPtwo} For the symmetrically squeezed GKP code~$\gkpcode{\kappa}{\star}{d}$ 
 we have 
\begin{align}
     \gateerror_{\gkpcode{\kappa}{\star}{d}}(e^{i(Q^2 +c_d\sqrt{2\pi/d}Q)/2}, \Pgate) &\geq \textfrac{3}{100} \ . 
\end{align}
\end{enumerate}
\end{theorem}
We note that a similar argument also applies to the linear optics implementation of the logical~$\CZ$ gate. The detailed statement and proof of this is given in Appendix~\ref{sec:nogoresultCZ}.

\begin{proof}
We note that
Claim~\eqref{it:firstclaimnogoPone} implies that
\begin{align}
     \gateerror_{\gkpcode{\kappa}{\varepsilon}{d}}(e^{i(Q^2 +c_d\sqrt{2\pi/d}Q)/2}, \Pgate) &\geq \textfrac{1}{25}-32 (\Delta/\varepsilon)^4\ 
\end{align}
for all~$\kappa<1/250$. This is because by definition, symmetrically squeezed GKP codes correspond to the choice
$\Delta = \kappa/(2\pi d)$, and this satisfies condition~\eqref{eq:pideltakappa} with equality. Claim~\eqref{it:firstclaimnogoPtwo} follows from this by substituting~$\varepsilon_d=1/(2d)$.

It remains to prove Claim~\eqref{it:firstclaimnogoPone}.   Define~$\cL=\cL_{in}=\cL_{out}=\gkpcode{\kappa,\Delta}{\varepsilon}{d}$
and
\begin{align}
\ket{\encoded{k}}=\ket{\gkp_{\kappa,\Delta}^\varepsilon(k)_d}\qquad\textrm{ for }\qquad k\in \mathbb{Z}_d\ .
\end{align}
We are interested in the logical~$\Pgate$-gate defined by its matrix elements~$\Pgate_{j,k}=\delta_{j,k} \cdot \omega_d^{(j^2 + c_d j)/2}$ and its (approximate) implementation by
\begin{align}
W_\Pgate&= e^{i(Q^2/2 +c_d\sqrt{2\pi/d}Q)/2} .
\end{align}
The matrix elements of the matrix~$B$ in Definition~\ref{def:diagonalunitary} are given by
\begin{align}
B_{j,k}&=\sum_{m=0}^{d-1} \overline{\Pgate_{m,j}}\langle \encoded{m}|W_\Pgate |\encoded{k}\rangle\\
&=  \omega_d^{-(j^2 + c_d j)/2} \langle \encoded{j}|W_\Pgate |\encoded{k}\rangle\qquad\textrm{ for }\qquad j,k\in \mathbb{Z}_d\ . \label{eq: Brs P}
\end{align}
In particular,
\begin{align}
B_{j,j}&= \omega_d^{-(j^2 + c_d j)/2}  \langle \encoded{j}|W_\Pgate |\encoded{j}\rangle\, .
\end{align}
Therefore by definition of the GKP code~$\gkpcode{\kappa,\Delta}{\varepsilon}{d}$ we have    \begin{align}
         B_{0,0} &= \langle \gkp_{\kappa,\Delta}^\varepsilon,\left(M_{\sqrt{2\pi d}}\right)^\dagger e^{i(Q^2 +c_d\sqrt{2\pi/d}Q)/2} M_{\sqrt{2\pi d}} \gkp_{\kappa,\Delta}^\varepsilon\rangle\\
          &= \langle \gkp_{\kappa,\Delta}^\varepsilon, e^{i((\sqrt{2\pi d}Q)^2 +c_d\sqrt{2\pi/d}\sqrt{2\pi d}Q)/2}  \gkp_{\kappa,\Delta}^\varepsilon\rangle\\
          &= \langle \gkp_{\kappa,\Delta}^\varepsilon, e^{\pi i(dQ^2  +c_dQ)}  \gkp_{\kappa,\Delta}^\varepsilon\rangle\, , \label{eq: B00Pgate}
    \end{align}
    where we used that~$(M_{\sqrt{2\pi d}})^\dagger Q M_{\sqrt{2\pi d}} = \sqrt{2\pi d} Q$.
        Using Eq.~\eqref{eq: B00Pgate}, the assumption~\eqref{eq:pideltakappa}  and Lemma~\ref{lem: bound B_00 P} (with~$b= 1/(2\pi)$)
we obtain 
    \begin{align}
    |B_{0,0}| &= \left|\langle \gkp_{\kappa,\Delta}^\varepsilon, e^{\pi i (d Q^2+c_d Q)} \gkp_{\kappa,\Delta}^\varepsilon\rangle\right|\\
    & \le 49/50+16(\Delta/\varepsilon)^4\  \label{eq: B00 P final}
    \end{align}
    for all~$\kappa< 1/250$.

We note that by Lemma~\ref{lem:lowerboundgateerrorbzerozero}, we have
\begin{align}
    \gateerror_\cL(W_\Pgate,\Pgate) &\geq 2\sqrt{1 - |B_{0,0}|^2}\\
    &\ge 2 \left(1 - |B_{0,0}|\right)\, , \label{eq: gaterror B00 P}
\end{align}
where we used that~$\sqrt{1-x^2} \ge 1 - x$ for~$x \in [0,1]$.

The claim follows combining Eqs.~\eqref{eq: B00 P final} and~\eqref{eq: gaterror B00 P}\, .
\end{proof}

While the no-go result of Theorem~\ref{lem:nogoP} shows that the standard linear optics implementation fails for symmetrically squeezed GKP codes (see Claim~\eqref{it:firstclaimnogoPtwo}) it does -- a priori -- not preclude the possibility that this implementation works for a general, i.e., asymmetrically squeezed GKP code~$\gkpcode{\kappa,\Delta}{\star}{d}$, i.e., for choices of~$(\kappa,\Delta)$ where~$\Delta\neq \kappa/(2\pi d)$.  
In Appendix~\ref{sec:nogoresultasymmetric}, we rule out the possibility of sidestepping the no-go result in this way
even with very small values of~$(\kappa,\Delta)$, i.e., high squeezing.

\section*{Acknowledgments}

LB, BD and RK gratefully acknowledge support by the European Research Council under
grant agreement no. 101001976 (project EQUIPTNT), as well as the Munich Quantum
Valley, which is supported by the Bavarian state government through the Hightech Agenda
Bayern Plus. 

\section*{Author contributions}
All three authors contributed extensively to the development and the writing of this manuscript.

\newpage
\appendix

\section{Mathematical facts used in the analysis \label{sec:mathfacts}}
In this section, we collect a number of definitions and mathematical statements which are used in our analysis. 
In Section~\ref{sec:background}, we discuss various properties of the diamond norm.
In Section~\ref{sec:boundsinnerproducts} we state bounds on inner products we use repeatedly. Finally, in Section~\ref{sec:periodicgaussians} we introduce periodic Gaussians and related results.
  
\subsection{Properties of the diamond norm \label{sec:background}}

Consider two finite-dimensional Hilbert spaces~$\cal{X}$ and~$\cal{Y}$. The trace norm of a bounded operator~$X\in\cal{B(X)}$ is defined as~$\|X\|_1 = \tr\sqrt{X ^\dagger X}$, and its spectral or operator norm is defined as~$\|X\| = \sup_{\psi\in\cal{X} \,:\, \|\psi\|\leq 1} \| X \psi \|$.

Consider a linear map~$\cal{E}: \cal{B}(\cal{X}) \rightarrow \cal{B}(\cal{Y})$. Its adjoint map~$\cal{E}^\dagger : \cal{B}(\cal{Y}) \rightarrow \cal{B}(\cal{X})$ is the unique map satisfying 
\begin{align}
    \label{eq:inner-product-def}
    \langle Y , \cal{E}(X) \rangle = \langle \cal{E}^\dagger(Y) , X \rangle \qquad \text{ for all } X\in\cal{B(X)},Y\in\cal{B(Y)} \ ,
\end{align}
where the (Hilbert-Schmidt) inner product is defined as~$\langle Y, X \rangle = \tr(Y^\dagger X)$ for all~$X\in\cal{B(X)},Y\in\cal{B(Y)}$. 
The trace norm of a linear map 
$\cal{E}:\cal{B}(\cal{X}) \rightarrow \cal{B}(\cal{Y})$ is
\begin{align}
\|\cal{E}\|_1 = \sup_{X\in\cal{B(X)} \,:\, \|X\|_1\leq 1} \|\cal{E}(X)\|_1 
\end{align} and its spectral or operator norm is
\begin{align}
\|\cal{E}\|= \sup_{X\in\cal{B(X)} \,:\, \|X\|\leq 1} \|\cal{E}(X)\|\, .
\end{align}

The diamond norm (also called completely bounded trace norm) of a linear map~$\cal{E} : \cal{B}(\cal{X}) \rightarrow \cal{B}(\cal{Y})$ is the quantity 
\begin{align}
    \label{eq:diamondnormdef}
    \| \cal{E} \|_\diamond = \sup_{n\in\bb{N}} \| \cal{E} \otimes \id_{\cal{B}(\bb{C}^n)} \|_1 \ .
\end{align}
For Hermiticity-preserving linear maps, we have the following.

\begin{lemma} \label{lem: diamond hermicity preserving}
    The diamond norm of a Hermiticity-preserving linear map~$\cal{E} : \cal{B}(\cal{X}) \rightarrow \cal{B}(\cal{Y})$ satisfies
    \begin{align}
        \| \cal{E} \|_\diamond = \sup_{\ket{\psi} \in \cal{X} \otimes \cal{X} \, : \, \| \psi \| = 1} \| (\cal{E}  \otimes \id_{\cal{B}(\cal{X})} )(\proj{\psi}) \|_1 \ .
    \end{align}
\end{lemma}

\begin{proof}
    See e.g., \cite[Theorem 3.51]{watrousbook}.
\end{proof}

A linear map~$\cE: \cB(\cal{X}) \rightarrow \cB(\cY)$ admits the following representation.
Consider an additional Hilbert space~$\cal{Z}$. Let~$A,B: \cal{X} \rightarrow \cal{Y} \otimes \cal{Z}$ be such that 
\begin{align}
    \label{eq:stinespringrep}
    \cal{E}(X) = \tr_\cal{Z} ( A X B^\dagger )
    \qquad \text{ for all } X \in \cal{B(X)} \ .
\end{align}
The pair~$(A,B)$ is called a Stinespring pair for the linear map~$\cal{E}$. 
We denote by
\begin{align}
    \ket{\Phi}_{\cal{XX}'} = \frac{1}{\sqrt{d_x}} \sum_{j=0}^{d_x-1} \ket{j}_{\cal{X}} \otimes \ket{j}_{\cal{X}'}
\end{align}
the maximally entangled state on~$\cal{X}\otimes \cal{X}'$ where~$\cal{X}'$ is isomorphic to~$\cal{X}$, $d_x = \dim\cal{X}$, and~$\{\ket{j}_\cal{X}\}_{j=0}^{d_x-1}$ and~$\{\ket{j}_{\cal{X}'}\}_{j=0}^{d_x-1}$ are orthonormal bases of~$\cal{X}$ and~$\cal{X}'$, respectively. 
Denote by 
\begin{align}
    \label{eq:Choirep}
    J(\cal{E}) = (\cal{E} \otimes I_{\cal{B}(\cal{X}')})( \proj{{\Phi}}_{\cal{XX}'})
\end{align}
the Choi-Jamio{\l}kowski representation of the map~$\cal{E}$.

A characterization of the expression~\eqref{eq:diamondnormdef} for the diamond norm which will be useful is (see~\cite[Definition 11.2 and Theorem 11.1]{10.5555/863284}):
\begin{align}
    \label{eq:diamondnorm_def2}
    \| \cal{E} \|_\diamond = \inf \{ \|A\| \cdot \|B\| \mid \tr_{\cal{Z}}(A X B^\dagger) = \cal{E}(X) \text{ for all } X \in \cal{B(X)}  \} \ ,
\end{align}
i.e., the infimum is taken over all Stinespring pairs~$(A,B)$ for~$\cal{E}$.
The expression for the diamond norm in Eq.~\eqref{eq:diamondnorm_def2} does not depend on the choice of the auxiliary space~$\cal{Z}$, provided that at least one Stinespring pair~$(A,B)$ exists. This is the case for~$\dim\cal{Z} \geq \rank J(\cal{E})$ (see Ref.~\cite{10.5555/863284}). Without loss of generality, we choose
\begin{align}
    \label{eq:dzeqrank}
    \cal{Z} = \bb{C}^{\rank J(\cal{E}) } \ .
\end{align}
Additionally, notice that the set of Stinespring pairs~$(A,B)$ for~$\cal{E}$ is compact, which implies that the infimum in Eq.~\eqref{eq:diamondnorm_def2} is achieved. 

The following lemma relates the diamond norm of the adjoint map~$\cal{E}^\dagger$ to that of~$\cal{E}$.

\begin{lemma}
\label{lem:diamondnormadjoint}
Let~$\cal{E} : \cal{B(X)} \rightarrow \cal{B(Y)}$ be linear and let~$J(\cal{E})$ be the Choi-Jamio{\l}kowski representation of~$\cal{E}$. We have
\begin{align}
    \label{eq:Erank}
    \| \cal{E}^\dagger \|_\diamond \leq \rank J(\cal{E})  \| \cal{E} \|_\diamond \ .
\end{align}
Moreover, if~$\cal{E}$ has Kraus-rank one, i.e., $\rank J(\cal{E}) = 1$, then
\begin{align}
    \label{eq:diamondadjointunitary}
    \| \cal{E}^\dagger \|_\diamond = \| \cal{E} \|_\diamond \ .
\end{align}
\end{lemma}

\begin{proof}
Denote by~$(A,B)$ a Stinespring pair for~$\cal{E}$ which achieves the infimum in Eq.~\eqref{eq:diamondnorm_def2}, i.e., 
\begin{align}
    \label{eq:E_AB_optimal}
    \| \cal{E}\|_\diamond = \|A\| \cdot \|B\| \ .
\end{align}
(Recall that the infimum is achieved because the set of Stinespring pairs is compact.)
    
Let~$\cal{Z}'$ be isomorphic to~$\cal{Z}$. Define~$\widetilde{A},\widetilde{B}:\cal{Y} \rightarrow \cal{X} \otimes \cal{Z}'$ by
\begin{align}
\label{eq:deftildeAB} 
    \begin{aligned}
    \widetilde{A} \ket{\psi} &= \sqrt{d_z} (A^\dagger \otimes I_{\cal{Z}'})(\ket{\psi} \otimes \ket{\Phi}_{\cal{ZZ}'}) \\ 
    \widetilde{B} \ket{\psi} &= \sqrt{d_z} (B^\dagger \otimes I_{\cal{Z}'})(\ket{\psi} \otimes \ket{\Phi}_{\cal{ZZ}'})
    \end{aligned} 
    \qquad \text{ for all } \ket{\psi} \in \cal{Y} \ ,
\end{align}
where~$\ket{\Phi}_{\cal{ZZ}'}$ is defined similarly to expression~\eqref{eq:Choirep} and~$d_z = \dim(\cal{Z})$. We show the following.

Claim: The pair~$(\widetilde{A},\widetilde{B})$ is a Stinespring pair for the adjoint map~$\cal{E}^\dagger: \cal{B(Y)} \rightarrow \cal{B(X)}$. 

\begin{proof}[Proof of Claim]
By Eq.~\eqref{eq:stinespringrep}, we have
\begin{align}
    \langle Y, \cal{E}(X) \rangle 
    &= \langle Y, \tr_{\cal{Z}} (A X B^\dagger) \rangle \\
    &= \tr (( Y \otimes I_{\cal{Z}})^\dagger A X B^\dagger )\\
    &= \tr (B^\dagger ( Y \otimes I_{\cal{Z}})^\dagger A X) \\
    &= \langle A^\dagger ( Y \otimes I_{\cal{Z}}) B, X \rangle 
\end{align}
for all~$X\in\cal{B(X)},Y\in\cal{B(Y)}$. By Eq.~\eqref{eq:inner-product-def}, this implies that
\begin{align}
    \cal{E}^\dagger(Y) = A^\dagger ( Y \otimes I_{\cal{Z}}) B \ .
\end{align}
We prove the claim by showing that
\begin{align}
    \cal{E}^\dagger(Y) &= \tr_{\cal{Z}'} (\widetilde{A} Y \widetilde{B}^\dagger)  = A^\dagger ( Y \otimes I_{\cal{Z}}) B 
    \qquad \text{ for all } \qquad Y \in \cal{B(Y)} \ ,
\end{align}
with~$\widetilde{A}$ and~$\widetilde{B}$ defined in Eq.~\eqref{eq:deftildeAB}. It suffices to observe that
\begin{align}
    \widetilde{A} Y \widetilde{B}^\dagger &= d_z (A^\dagger \otimes I_{\cal{Z}'})(Y \otimes \ket{\Phi} \! \bra{\Phi}_{\cal{ZZ}'})(B \otimes I_{\cal{Z}'}) \\
    &= 
    \sum_{j,k=0}^{d_z-1} (A^\dagger \otimes I_{\cal{Z}'})(Y \otimes \ket{j j} \! \bra{k k}_\cal{ZZ'})(B \otimes I_{\cal{Z}'}) 
\end{align}
and 
\begin{align}
    \tr_{\cal{Z}'} (\widetilde{A} Y \widetilde{B}^\dagger) 
    &=  A^\dagger \left(Y \otimes \sum_{j=0}^{d_z-1} \ket{j}_\cal{Z} \! \bra{j}_\cal{Z}\right) B \\
    &=  A^\dagger (Y \otimes I_{\cal{Z}})B 
\end{align}
for all~$Y \in \cal{B(Y)}$.
\end{proof}
    
Because~$(\widetilde{A}, \widetilde{B})$ is a Stinespring pair for~$\cE$, it follows from Eq.~\eqref{eq:diamondnorm_def2} that
\begin{align}
    \| \cal{E}^\dagger\|_\diamond 
    &\leq \|\widetilde{A}\|  \cdot \|\widetilde{B}\|  \ .
\end{align}
Observe that
\begin{align} 
    \|\widetilde{A}\|  
    &= \sup_{\psi \in \cal{Y} \,:\,\|\psi\|  \leq 1}
    \sqrt{d_z} \| (A^\dagger \otimes I_{\cal{Z}'})(\psi \otimes \Phi_{\cal{ZZ}'} ) \|  \\
    &\leq \sup_{\psi \in \cal{Y\otimes Z \otimes Z'} \,:\,\|\psi\| \leq 1}
    \sqrt{d_z} \| (A^\dagger \otimes I_{\cal{Z}'}) \psi \|  \\
    &= \sqrt{d_z} \|A^\dagger \otimes I_{\cal{Z}'}\| \\
    &= \sqrt{d_z} \|A^\dagger\|  
\end{align}
where~$d_z = \dim\cal{Z}$, and similarly~$\|\widetilde{B}\|  \leq \sqrt{d_z} \|B^\dagger\|~$.
Therefore 
\begin{alignat}{2}
    \| \cal{E}^\dagger\|_\diamond &\leq d_z \| A^\dagger \|  \cdot \| B^\dagger \|  \\
    &= d_z \| A \|  \cdot \| B \|  \\
    &= d_z \| \cal{E} \|_\diamond 
    &&\qquad \text{  by Eq.~\eqref{eq:E_AB_optimal} } \\
    &= \rank J(\cal{E}) \| \cal{E} \|_\diamond 
    &&\qquad \text{ by Eq.~\eqref{eq:dzeqrank}}\ .
\end{alignat}
This is the Claim~\eqref{eq:Erank}.
If~$\cal{E}$ has Kraus-rank one, i.e., $\rank{J(\cal{E})} = 1$, then~$\|\cal{E}^\dagger\|_\diamond \leq \|\cal{E}\|_\diamond$. Interchanging~$U$ and~$U^\dagger$ then gives Eq.~\eqref{eq:diamondadjointunitary}.
\end{proof}

In the following we argue that the definition of the diamond norm in terms of finite-dimensional stabilizing systems from Eq.~\eqref{eq:diamondnormdef} is applicable in the case where we compose a linear map with a projection onto a finite-dimensional space.
We denote the set of trace-class operators on a (possibly infinite-dimensional) Hilbert space~$\cH$ by~$\cB_1(\cH)$. 
\begin{lemma}\label{lem:finitedimdiamond}
    Let~$\cal{X}, \cal{Y}$ be separable Hilbert spaces and~$\cL \subset \cal{X}$ a finite-dimensional subspace. Let~$\cE: \cB_1(\cal{X}) \rightarrow \cB_1(\cal{Y})$ be linear.
    Let~$\pi_\cL: \cal{X} \rightarrow \cL$ be the orthogonal projection onto the space~$\cL$. Define~$\Pi_\cL: \cB_1(\cal{X}) \rightarrow \cB_1(\cL)$ by~$\Pi_\cL(A) = \pi_\cL A \pi_\cL$.
    Then 
    \begin{align} \label{eq:claimfinitedim}
        \sup_{\cal{X}'} \sup_{\substack{A \in \cB_1(\cal{X}\otimes \cal{X}')\\ \|A\|_1 \le 1}} \left\| \left(\left(\cE \circ \Pi_\cL\right) \otimes \cal{I}_{\cB(\cal{X}')}\right)(A)\right\|_1 = \sup_{n\in \mathbb{N}} \sup_{\substack{B \in \cB(\cal{L}\otimes \mathbb{C}^n)\\ \|B\|_1 \le 1}} \left\| \left(\cE  \otimes \cal{I}_{\cB(\mathbb{C}^n)}\right)(B)\right\|_1
    \end{align}
    where the supremum on the lhs.\ is taken over all separable Hilbert spaces~$\cal{X}'$ and the rhs.\ is understood using the natural embedding of~$\cB(\cL\otimes \mathbb{C}^n) = \cB_1(\cL\otimes \mathbb{C}^n)$ into~$\cB_1(\cal{X} \otimes \mathbb{C}^n)$.
\end{lemma}
\begin{proof}
    Fix an additional separable Hilbert space~$\cal{X}'$. Let~$A \in \cB_1(\cal{X}\otimes \cal{X}')$ with~$\|A\|_1=1$ be given. Then~$A$ is a compact operator. Using the singular value decomposition for compact operators,
    there are orthonormal bases~$\{e_j\}_{j}, \{f_j\}_{j} \subset \cal{X}\otimes \cal{X}'$ and a sequence of positive numbers~$\{s_j\}_j$ such that 
    \begin{align}
         A = \sum_{j} s_j |e_j\rangle \langle f_j|\, .
    \end{align}
    Note that~$\|A\|_1 = \|s\|_1 = \sum_{j} s_j$ as
    \begin{align}
       1 = \|A\|_1 = \tr |A| = \tr \sqrt{A^\dagger A} = \tr \left(\sum_{j} s_j \proj{f_j}\right) = \sum_{j} s_j\, .
    \end{align}
Thus~$A$ is contained in the closure of the convex hull of all rank one operators~$|\phi\rangle \langle \psi|$ where~$\phi,\psi \in \cal{X}\otimes \cal{X}'$ and~$\|\phi\|, \|\psi\| = 1$. In particular we can restrict the supremum on the lhs. of Eq.~\eqref{eq:claimfinitedim} to these operators.
    Note that~$(\Pi_\cL \otimes  \cal{I}_{\cB_1(\cal{X}')} )(|\phi\rangle \langle \psi|) = |\widetilde{\phi}\rangle \langle \widetilde{\psi}|$ for some~$\widetilde{\phi}, \widetilde{\psi} \in \cL\otimes \cal{X}'$ with~$\|\widetilde{\phi}\|, \|\widetilde{\psi}\|\le 1$. 
    Therefore it is enough to consider the supremum over rank-one operators~$A = |\widetilde{\phi}\rangle \langle \widetilde{\psi}|\in  \cB_1(\cL \otimes \cal{X}')$ with~$\|\widetilde{\phi}\|, \|\widetilde{\psi}\|\le 1$ on the lhs.\ of Eq.~\eqref{eq:claimfinitedim}.
     Using the Schmidt-decomposition we can further restrict the supremum in Eq.~\eqref{eq:claimfinitedim} to finite-dimensional additional systems~$\cal{X}'$, see~\cite[Lemma 3.45]{watrousbook} for a detailed proof. In fact, we can assume~$\cal{X}' \simeq \cal{\cL}$. 
    The claim follows.
\end{proof}

\subsection{Bounds on inner products \label{sec:boundsinnerproducts}}
Here we state a few bounds on inner products which we use repeatedly. These  are simple consequences  of the Cauchy-Schwarz inequality.

Let~$\cH$ be a Hilbert space. We use the bound
\begin{align}
\left|\langle \Psi,\Phi\rangle-1\right|\leq \|\Psi-\Phi\|&\leq \sqrt{2}\sqrt{\left|\langle \Psi,\Phi\rangle-1\right|}\ .\label{eq:upperboundpsiphiinner}
\end{align} on pure states~$\Psi, \Phi \in \cH$.
\begin{proof}
The upper bound on~$\|\Psi-\Phi\|$ in Eq.~\eqref{eq:upperboundpsiphiinner} follows from
\begin{align}
\|\Psi-\Phi\|^2&=2\left(1-\mathsf{Re}\langle \Psi,\Phi\rangle \right)\\
&\leq  2|\mathsf{Re}(\langle \Psi,\Phi\rangle -1)|\\
&\le 2|\langle \Psi,\Phi\rangle -1|\ .
\end{align}
For the lower bound on~$\|\Psi-\Phi\|$ in Eq.~\eqref{eq:upperboundpsiphiinner}, we use the Pythagorean theorem, which implies that
\begin{align}
\|\Pi(\Psi-\Phi)\|&\leq \|\Psi-\Phi\|
\end{align}
for any orthogonal projection~$\Pi$. Applying this with the rank-$1$ projection~$\Pi=\proj{\Psi}$ gives the claim.
\end{proof}

The following lemma gives an explicit bound expressing continuity of the map~$\Psi\mapsto \langle \Psi, U\Psi\rangle$. 
\begin{lemma}\label{lem:continuitymatrixelem}
Let~$\psi_1,\psi_2\in\cH$ be arbitrary states and let~$U$ be a unitary on~$\cH$.  Then 
\begin{align}
\left|\langle \psi_1,U\psi_1\rangle-\langle \psi_2,U\psi_2\rangle\right| & \leq 2\|\psi_1-\psi_2\|\ .\label{eq:firstboundcontinuitymatrixelem}
\end{align}
In particular, we have 
\begin{align}
\left|\langle \psi_1,U\psi_1\rangle\right| & \leq \left|\langle \psi_2,U\psi_2\rangle\right|+2\|\psi_1-\psi_2\|\ .\label{eq:secondboundcontinuitymatrixelem}
\end{align}
\end{lemma}
\begin{proof}
    Note that for any two states~$\psi_1,\psi_2 \in \cH$ and any unitary~$U$, we have 
    \begin{align}
    \langle \psi_1,U\psi_1\rangle &= \langle \psi_2+(\psi_1-\psi_2),U(\psi_2+(\psi_1-\psi_2))\rangle\\
    &=\langle \psi_2,U\psi_2\rangle+ \langle \psi_2,U(\psi_1-\psi_2)\rangle+\langle (\psi_1-\psi_2),U\psi_1\rangle\ .
    \end{align}
   The triangle inequality  thus implies that 
   \begin{align}
\left|\langle \psi_1,U\psi_1\rangle-\langle \psi_2,U\psi_2\rangle\right|&\leq \left|\langle \psi_2,U(\psi_1-\psi_2)\rangle\right|+\left|\langle (\psi_1-\psi_2),U\psi_1\rangle\right|\ .
\end{align}
The Eq.~\eqref{eq:firstboundcontinuitymatrixelem} now follows from the Cauchy-Schwarz inequality using the unitarity of~$U$.

Observe that Eq.~\eqref{eq:secondboundcontinuitymatrixelem} immediately follows from Eq.~\eqref{eq:firstboundcontinuitymatrixelem} and the triangle inequality. This concludes the proof.
\end{proof}

The following lemma expresses continuity of matrix elements more generally: The matrix element~$\langle \psi_1,U\psi_2\rangle$ of a unitary~$U$ with respect to two states~$(\psi_1,\psi_2)$ is close to the matrix element~$\langle \varphi_1,U\varphi_2\rangle$ with respect to a different pair~$(\varphi_1,\varphi_2)$ of states if~$\psi_j$ is close to~$\varphi_j$ for~$j=1,2$.
\begin{lemma}
\label{lem:triangle_ineq}
Let~$\psi_1,\psi_2,\varphi_1,\varphi_2 \in \cal{H}$ be states. Let~$U$ be a unitary on~$\cH$. 
Then
\begin{align}
    \label{eq:lem_triangle_ineq_claimfirst}
    | \braket{\psi_1}{U\psi_2} - \braket{\varphi_1}{U\varphi_2} | 
    \leq 
    \|\psi_1-\varphi_1\|+\|\psi_2-\varphi_2\|\ .
\end{align}
In particular, we have 
\begin{align}
    \label{eq:lem_triangle_ineq_claim1}
    | \braket{\psi_1}{U\psi_2} - \braket{\varphi_1}{U\varphi_2} | 
    \leq \sqrt{2} \left(\sqrt{ |\langle \psi_1,\varphi_1\rangle-1|} + \sqrt{|\langle \psi_2,\varphi_2\rangle-1|}\right)  \ .
\end{align}
and  
\begin{align}
    \label{eq:lem_triangle_ineq_claim2}
    |\braket{\psi_1}{U\psi_2}-a|
    \leq 
    |\braket{\varphi_1}{U\varphi_2} - a| + \sqrt{2}\left( \sqrt{ |\langle \psi_1,\varphi_1\rangle-1|} + \sqrt{|\langle \psi_2,\varphi_2\rangle-1|}\right) \ .
\end{align} for all~$a\in\bbC$.
\end{lemma}

\begin{proof}
We have
\begin{alignat}{2}
    &| \braket{\psi_1}{U\psi_2} - \braket{\varphi_1}{U\varphi_2} |   \\
    &\quad=  | \braket{\psi_1}{U\psi_2} - \braket{\psi_1}{U\varphi_2} + \braket{\psi_1}{U\varphi_2} - \braket{\varphi_1}{U\varphi_2} |  \\
    &\quad\leq | \langle\psi_1, U ( \psi_2 - \varphi_2)| + | \langle (\psi_1 - \varphi_1), U \varphi_2\rangle | 
    && \text{ by the triangle inequality} \\
    &\quad\leq \| U^\dagger \psi_1 \| \cdot \|\psi_2 - \varphi_2\| + \| \psi_1 - \varphi_1\| \cdot \| U \varphi_2 \| 
    && \text{ by the Cauchy-Schwarz inequality} \\
    &\quad\leq \|\psi_2 - \varphi_2\| + \| \psi_1 - \varphi_1 \| 
    && \text{ by unitarity of~$U$, $\|\psi_1\|=1$ and~$\|\varphi_2\|=1$\ .}
    \end{alignat}
    This establishes Eq.~\eqref{eq:lem_triangle_ineq_claimfirst}.
    
    Eq.~\eqref{eq:lem_triangle_ineq_claim1} follows immediately from Eq.~\eqref{eq:lem_triangle_ineq_claimfirst}
    using the inequality~\eqref{eq:upperboundpsiphiinner}. 
    
    Finally, Claim~\eqref{eq:lem_triangle_ineq_claim2} follows   from   Eq.~\eqref{eq:lem_triangle_ineq_claim1} and the triangle inequality.
   \end{proof}
The following lemma quantitatively expresses the following fact for a family~$\{U_t\}_{t=1}^T$ of unitaries and a state~$\Psi$: If each unitary~$U_t$ leaves the state~$\Psi$ (approximately) invariant, then the same is approximately the case for the composition~$U_T\cdots U_1$.
   
\begin{lemma}[Matrix elements of composed unitaries]\label{lem:composedunitaries}
Let~$\Psi\in\cH$ be a state and~$U_1,\ldots,U_T$ unitary operators on~$\cH$. 
Then
\begin{align}
\left|\langle \Psi,U_T\cdots U_1\Psi\rangle -1\right|&\leq \sqrt{2}\sum_{t=1}^T \sqrt{|\langle \Psi,U_t\Psi\rangle-1|}\qquad\textrm{ for all }\qquad t\in \{1,\ldots,T\}\ .
\end{align}
\end{lemma}
\begin{proof}
Set~$\Psi_0=\Psi$. 
Define~$\Psi_t=U_t\cdots U_1 \Psi$. 
For any~$t\in \{1,\ldots,T\}$, we have
\begin{align}
\left\|\Psi_t-\Psi\right\|
&=\left\|U_t\Psi_{t-1}-\Psi\right\|\\
&\leq \left\|U_t\Psi_{t-1}-U_t\Psi\right\|+\left\|U_t\Psi-\Psi\right\|\\
&= \left\|\Psi_{t-1}-\Psi\right\|+\left\|U_t\Psi-\Psi\right\|\ .
\end{align}
by the triangle inequality and invariance of the norm under unitaries.
By induction, it follows that
\begin{align}
\left\|\Psi_T-\Psi\right\|&\leq \sum_{t=1}^T \left\|U_t\Psi-\Psi\right\|\ .
\end{align}
In particular,
\begin{align}
\left|\langle \Psi,\Psi_T\rangle -1\right|\leq \left\|\Psi_T-\Psi\right\|\leq \sqrt{2}\sum_{t=1}^T \sqrt{|\langle U_t\Psi,\Psi\rangle-1|}\ .
\end{align}
by  Eq.~\eqref{eq:upperboundpsiphiinner}.  This is the claim. 
\end{proof}

\subsection{Discrete Gaussian distributions\label{sec:periodicgaussians}}

We need a few facts about discrete Gaussian distributions.
Their study was initiated by Banaszczyk~\cite{Banaszczyk1993} who proved so-called transference theorems for lattices, and subsequently obtained measure inequalities for convex bodies~\cite{Banaszczyk1995,Banaszcyk96}.
 They are defined as follows.
For~$s>0$, define~$\rho_s(x) = \exp(-\pi x^2/s^2)$ for~$x \in \mathbb{R}$. The discrete Gaussian random variable~$X_s$ then has distribution
\begin{align}
\operatorname{Pr}\left[X_s=z\right]=\frac{\rho_s(z)}{\rho_s(\mathbb{Z})}\qquad\textrm{for}\qquad z\in \mathbb{Z} \label{eq:discretegaussiandistributionv}
\end{align}
where~$\rho_s(\mathbb{Z})=\sum_{z\in\mathbb{Z}}\rho_s(z)$.

\begin{lemma}\label{lem:discretegaussiandistribution}
Let~$s>0$. The discrete Gaussian random variable~$X_s$ satisfies
\begin{align}
\operatorname{Pr}\left[\left|X_s\right| \geq r\right] \leq 2 e^{-\frac{3\pi }{4} (r/s)^2}\label{eq:tailboundbrenner2024}\ 
\end{align}
for any~$r>0$. 
\end{lemma}
\begin{proof}
For the proof of Eq.~\eqref{eq:tailboundbrenner2024}, see~\cite[Lemma 8.3]{brenner2024factoring}. 
\end{proof}

We will also need the following result on so-called periodic Gaussians. For~$s>0$ and~$x\in \mathbb{R}$  define
\begin{align}
    f_s(t)&=\sum_{z\in\mathbb{Z}}\rho_s(z+t)=\rho_s(\mathbb{Z}+t)\qquad\textrm{ for }t\in\mathbb{R}\ \label{eq:fsdefinitionperiodicgaussian}\, .
\end{align} 
\begin{lemma}\label{lem: periodic gaussians}
    Let~$s>0$ and~$t\in \mathbb{R}$. Then 
    \begin{align}
        e^{-\pi t^2/s^2}f_s(0)\leq  f_s(t) \le f_s(0) \qquad \textrm{for all} \qquad t\in \mathbb{R}\, .
    \end{align}
    In particular, $f_s(t)$ is maximized when~$t \in \mathbb{Z}$. 
\end{lemma}
\begin{proof}
    We refer to~\cite[Lemma~3]{DadushDucas2018} for a proof. 
\end{proof}

\section{General quantum circuits} \label{sec: logical gate error general circuits}

Here we briefly comment on how our considerations generalize to arbitrary quantum circuits. We use the graph representation of quantum circuits, see e.g.~\cite{PhysRevLett.125.150504}.
Let~$G=(V,E)$ be a connected, directed acyclic graph.
For a vertex~$v\in V$, let~$\inedges{v}\subseteq E$ denote the edges incoming to~$v$, and let~$\outedges{v}\subseteq E$ be the edges outgoing from~$v$.
We define the in-degree and the out-degree of~$v$ as~$\indeg{v}=|\inedges{v}|$ and~$\outdeg{v}=|\outedges{v}|$, respectively.
We call
\begin{align}
\inputvertices&:=\{v\in V\ |\ \indeg{v}=0,\outdeg{v}=1\}\\
\outputvertices&:=\{v\in V\ |\ \indeg{v}=1,\outdeg{v}=0\}\ .
\end{align}
the sets of input- and output-vertices. The remaining set of vertices
\begin{align}
\interiorvertices&:=V\backslash (\inputvertices\cup\outputvertices)
\end{align}
will be referred to as the set of interior vertices. We denote the size of this set by~$T:=|\interiorvertices|$.  The set of interior vertices with the property that all in-neighbors (predecessors) are input vertices is called the in-boundary of~$G$. 

We assume that the graph is equipped with the following additional data:
\begin{enumerate}[(i)]
\item \label{it:firstgraph}
Every edge~$e\in E$ carries a Hilbert space~$\cK_e$ of dimension~$d_e=\dim{\cK_e}\geq 2$. These  satisfy the condition
\begin{align}
\label{eq:inputoutputspace}
\bigotimes_{e \in \inedges{v}} \cK_e &\cong\bigotimes_{e\in \outedges{v}} \cK_e\qquad\textrm{ for every }\qquad v\in \interiorvertices\ .
\end{align}
\item
For every interior vertex~$v\in \interiorvertices$, we assume that there is a unitary
\begin{align}
U_v:\bigotimes_{e\in\inedges{v}} \cK_e\rightarrow \bigotimes_{e\in\outedges{v}}\cK_e\ .
\end{align}
\item \label{it:circuitgraphconstruction}
We assume that the set of interior vertices is indexed as
\begin{align}
\interiorvertices=\{v_t\}_{t=1}^T
\end{align}
subject to the following condition. 
Let~$G^{(0)}=G$. Then the following holds for all~$t=1,\ldots,T$:
\begin{enumerate}[(a)]
\item
The vertex~$v_t$ is a vertex belonging to the in-boundary of the graph~$G^{(t-1)}$.
\item\label{it:contractionordergraph}
The graph~$G^{(t)}$ is obtained from~$G^{(t-1)}$ by 
detaching each edge~$e\in \outedges{v_t}$ from~$v_t$ and reattaching it to a newly added (input) vertex~$v_e$,
and removing all input vertices previously attached to~$v_t$ as well as~$v_t$ itself, see an example in Fig.~\ref{fig:circuitgraph}.
We denote the set of output edges of a vertex~$v \in G^{(t)}$ by~$E_{\mathsf{out}}^{(t)}(v)$ and the set of input vertices of~$G^{(t)}$ by~$\interiorvertices^{(t)}$.
    \end{enumerate}
\end{enumerate}
A graph with this data defines  a quantum circuit realizing a unitary
\begin{align}
U^{(T)}\cdots U^{(1)}\ 
\end{align}
where~\eqref{it:contractionordergraph} defines in which sequence the gates are applied.  Here~$U^{(t)}$ acts non-trivially only on the subsystem~$\bigotimes_{e\in\inedges{v_t}}\cK_e$, mapping it to~$\bigotimes_{e\in\outedges{v_t}}\cK_e$ by application of~$U_{v_t}$. 

\begin{figure}[h]
    \begin{center}
    \includegraphics[width = \textwidth]{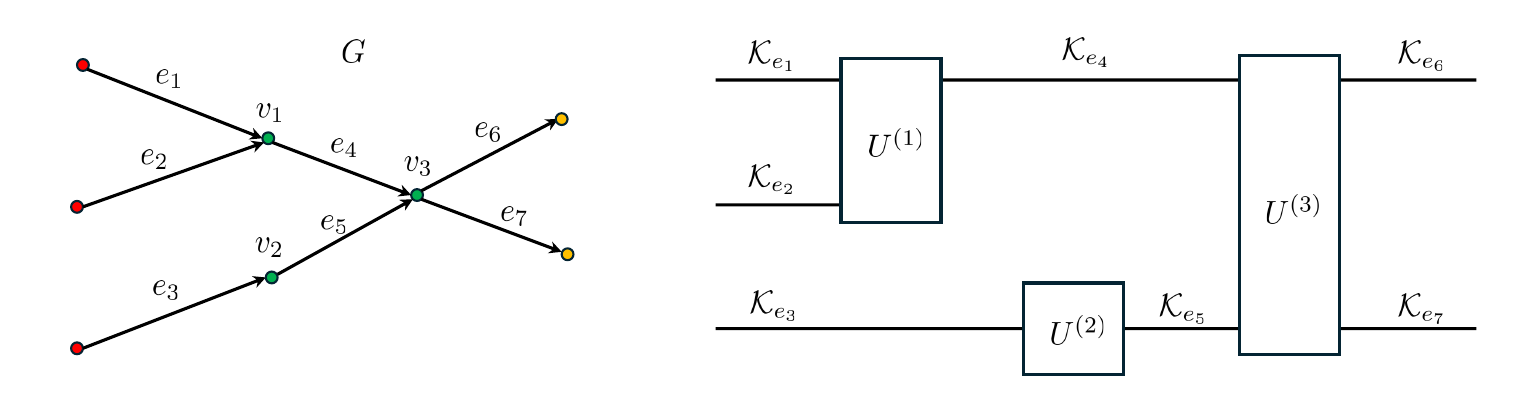}
    \end{center}
    \caption{Illustration of the construction of the circuit~$U^{(3)} U^{(2)} U^{(1)}$ from the graph~$G$ as described in~\eqref{it:circuitgraphconstruction}. Input vertices are marked in red, interior vertices in green and output vertices in orange. 
    To each interior vertex~$v_t$ for~$t \in \{1,2,3\}$ we associate a unitary~$U^{(t)}$. The (input and output) spaces~$\cK_{e_j}$ associated with each edge~$e_j$ for~$j \in \{1,\dots,7\}$ satisfy~Eq.~\eqref{eq:inputoutputspace}, e.g., $\cK_{e_1} \otimes \cK_{e_2} \simeq \cK_{e_4}$.}
    \label{fig:circuitgraph}
\end{figure}

We now consider an implementation where instead of physical qudits (associated with the spaces~$\cK_e$), we use logical qudits encoded in subspaces~$\{\cL_e\}_e$ of the physical Hilbert spaces~$\{\cH_e\}_e$. 
We further assume that for each internal vertex~$v$, we have a (possibly approximate) implementation of a logical (i.e., encoded version) of the unitary~$U_v$. 
That is, we assume that 
\begin{enumerate}[(i)]
\setcounter{enumi}{3}
\item
For each~$e\in E$, there is a Hilbert space~$\cH_e$ and a subspace~$\cL_e\subseteq\cH_e$ satisfying~$\cL_e\cong\cK_e$ (i.e., of the same dimension), with the isomorphism given by an isometric encoding map~$\encmap_{\cL_e}: \cK_e \rightarrow \cL_e$. We denote its inverse by~$\decmap_{\cL_e}:=(\encmap_{\cL_e})^{-1}:\cL_e\rightarrow\cK_e$.
\item \label{eq:fivegraph}
For each interior vertex~$v\in\interiorvertices$, we have a unitary
\begin{align}
W_v: \bigotimes_{e\in\inedges{v}} \cH_e\rightarrow \bigotimes_{e\in\outedges{v}}\cH_e\ .
\end{align}
\end{enumerate}

For~$v\in\interiorvertices$, define
\begin{align}
\cL_{in}(v)&=\bigotimes_{e\in\inedges{v}} \cL_e\\
\cH_{in}(v)&=\bigotimes_{e\in\inedges{v}} \cH_e
\end{align}
and define~$\cL_{out}(v)$ and~$\cH_{out}(v)$ similarly. 
Then~$\cL_{in}(v)\cong \cL_{out}(v)$. 
We define the map
\begin{align}
    \encoded{U}_v = \left(\bigotimes_{e\in\outedges{v}} \encmap_{\cL_e}\right) U_v \left(\bigotimes_{e\in\inedges{v}} \decmap_{\cL_e}\right): \cLin(v) \rightarrow \cLout(v)
\end{align}
and the logical gate error at the vertex~$v$ as
\begin{align} \label{eq:deferrgraph}
\gateerror_v&= \left\|(\cW_v-\encodedC{U}_v)\circ \Pi_{\cLin(v)}\right\|_\diamond\ ,
\end{align}
where~$\encodedC{U}_v(\rho)=\encoded{U}_v\rho \encoded{U}_v^\dagger$
and~$\cW_v(\rho)=W_v\rho W_v^\dagger$.

Then the composition of maps obtained from the graph~$G$ by using the CPTP map~$\cW_v$ at each interior vertex~$v\in\interiorvertices$ satisfies the following.
\begin{lemma}\label{lem:compositiongeneralcircuitmodular}
Let~$G= G_0 =(V,E)$ be a graph with indexed interior vertices~$\interiorvertices = \{v_t\}_{t=1}^{T}$ and induced graphs~$\{G^{(t)} \}_{t=1}^T$. Let~$\cH_e$, $\cL_e$, $\cK_e$ be Hilbert spaces for~$e\in E$ and let~$U_v$, $W_v$ be unitaries for~$v \in \interiorvertices$ as introduced in~\eqref{it:firstgraph}--\eqref{eq:fivegraph}.
Define the set of input edges of the graph~$G^{(t)}$ as~$E_{\mathsf{in}}^{(t)} = \bigcup_{v \in \inputvertices^{(t)}} E_{\mathsf{out}}^{(t)}(v)$ for~$t=0,\dots,T$ and define the unitaries
\begin{align}
    \begin{aligned}
        \encoded{U}^{(t)} &= \encoded{U}_{v_t}  \otimes \left(\bigotimes_{e \in E_{\mathsf{in}}^{(t-1)} \setminus \inedges{v_t}} I_{\cL_e}\right) \\
    W^{(t)} &= W_{v_t} \otimes  \left(\bigotimes_{e \in E_{\mathsf{in}}^{(t-1)} \setminus \inedges{v_t}} I_{\cH_e}\right)
    \end{aligned}
    \qquad\qquad\textrm{for} \qquad t\in \{1,\dots,T\}\, .
\end{align}
Define~$\encodedC{U}^{(t)} (\rho) = \encoded{U}^{(t)} \rho \left(\encoded{U}^{(t)}\right)^\dagger$ and~$\cW^{(t)}(\rho) = W^{(t)} \rho \left(W^{(t)}\right)^\dagger$. Then 
\begin{align}
    \left\|\left(\cW^{(T)}\circ \dots \circ \cW^{(1)} - \encodedC{U}^{(T)} \circ \dots \circ \encodedC{U}^{(1)}\right) \circ \bigotimes_{e\in E_{\mathsf{in}}^{(0)}} \Pi_{\cL_e} \right\|_\diamond \le  \sum_{t=1}^{T} \gateerror_{v_t}\, ,
\end{align}
where~$\gateerror_{v}$ is defined by Eq.~\eqref{eq:deferrgraph}.
\end{lemma}
\begin{proof}
    This is a straightforward application of Lemma~\ref{lem: additivity gate error}.
\end{proof}
Lemma~\ref{lem:compositiongeneralcircuitmodular} provide a means of analyzing the logical gate error associated with a circuit in terms of its building blocks,  allowing for a modular analysis. We note, however, that the resulting bound will generally not be tight. For concrete implementations, a more direct examination of the logical gate error may yield tighter bounds.

\section{Lower bounds on the inner numerical radius \label{sec:numericalradius}}

 In this section we establish a lower bound on the inner numerical radius  of an operator. 
It serves as the main technical ingredient for the upper bounds on the logical gate error stated in Theorem~\ref{cor:shortmatrixBstatement} from the main text.

Fix an orthonormal basis~$\{\ket{j}\}_{j=0}^{d-1}$ of~$\mathbb{C}^d$. In the following, we often identify a linear operator~$A:\mathbb{C}^d\rightarrow\mathbb{C}^d$ with its matrix~$A\in\mathsf{Mat}_{d\times d}(\mathbb{C})$  with respect to the computational basis, i.e., $A_{j,k}=\langle j,Ak\rangle$. 
The Crawford number~$\cn(A)$ of an operator~$A:\mathbb{C}^d\rightarrow\mathbb{C}^d$ (see Eq.~\eqref{eq:crawford}) can be bounded as follows.
\begin{lemma}\label{lem:linearoplowerboundmx}
Let~$A:\mathbb{C}^d\rightarrow\mathbb{C}^d$ be linear. 
Define 
\begin{align} \label{def:alphaA}
\alpha(A):=
\min_{j} \left(|\mathsf{Re}\,A_{j,j}|-\frac{1}{2}\sum_{\ell\neq j} |A_{j,\ell} + \overline{A_{\ell, j}}|\right)\ .
\end{align}
Suppose that 
\begin{align}
    \mathsf{Re}\ A_{j,j} \geq 0  \quad\text{ for every }\quad  j\in \{0,\ldots,d-1\} \ .
    \label{it:secondconditionapositive}
\end{align}
Then 
\begin{align}
\cn(A)&\geq \alpha(A)\ .
\end{align}
\end{lemma}
\begin{proof}
  Decompose~$A$ as 
\begin{align}
A=K+iL  \qquad\textrm{ where }\qquad
\begin{matrix}
    &K=\frac{1}{2}(A+A^\dagger)\\
    &L=\frac{1}{2i}(A-A^\dagger)
\end{matrix}
\ .
\end{align}
Then~$K$ and~$L$ are self-adjoint operators and thus we have~$K_{\ell,j} = \overline{K_{j,\ell}}$ and~$L_{\ell,j} = \overline{L_{j,\ell}}$ for all~$j,\ell \in \{0,\dots,d-1\}$. Therefore using the definition~\eqref{def:alphaA} we observe that
\begin{align}
    \alpha(K) = \min_{j} \left( |K_{j,j}| - \sum_{\ell\neq j} |K_{\ell,j}|\right) = \min_{j} \left( |K_{j,j}| - \sum_{\ell\neq j} |K_{j, \ell}|\right)\ .\label{eq:kalphaArelation}
\end{align}
Moreover, since~$\mathsf{Re}\, A_{j,j} = K_{j,j}$ and~$K_{j,\ell} + \overline{K_{\ell,j}}= A_{j,\ell} + \overline{A_{\ell,j}}$ for~$j \neq \ell$ we have 
\begin{align}
    \alpha(A) = \alpha(K)\, . \label{eq:alphaAK}
\end{align}
Since~$\cn(A)\geq 0$ by definition, the claim is trivial if~$\alpha(A)\le0$.
Thus we can assume that
\begin{align}
\alpha(A)> 0\ .\label{eq:alphagreaterzero}
\end{align}
Combined with Eqs.~\eqref{eq:kalphaArelation} and~\eqref{eq:alphaAK}, Condition~\eqref{eq:alphagreaterzero} states that (in the terminology of Ref.~\cite{varah})~$K$ is diagonally dominant by both its rows and columns.

The operator~$K$ is self-adjoint by definition and diagonally dominant by rows, and every diagonal entry~$K_{j,j}=\mathsf{Re}\,A_{j,j}$ is non-negative by the assumption~\eqref{it:secondconditionapositive}. We conclude using Gershgorin's circle theorem~\cite[Theorem 6.1.1]{horn2012matrix} that~$K$ is positive semidefinite. 

By definition of~$K$ and~$L$, we have 
\begin{align}
|\langle \Psi, A\Psi\rangle |&=
\left|\langle \Psi,K\Psi\rangle+i \langle \Psi,L\Psi\rangle\right|\\
&=\sqrt{|\langle \Psi,K\Psi\rangle|^2+|\langle \Psi,L\Psi\rangle|^2}\\
&\geq |\langle \Psi,K\Psi\rangle|\qquad\textrm{ for any } \Psi\in\mathbb{C}^d\ \label{eq:upperboundmvinner} \ ,
\end{align}
where we used the fact that expectation values of self-adjoint operators are real. Using the spectral decomposition of~$K$  and the fact that~$K$ is positive semidefinite, it is easy to see that 
\begin{align}
|\langle \Psi,K\Psi\rangle|^2&\geq \min \{\lambda^2\ |\ \lambda\in\mathsf{spec}(K)\}\qquad\textrm{ for any}\qquad \Psi\in\mathbb{C}^d\textrm{ with }\|\Psi\|=1\ ,
\label{eq:cbound_aux1}
\end{align}
where we denote by~$\mathsf{spec}(K)$ the set of distinct eigenvalues of~$K$. 
We note that
\begin{align}
\min \{\lambda^2\ |\ \lambda\in\mathsf{spec}(K)\}&=\sigma_d(K)^2
\label{eq:cbound_aux2}
\end{align}
 where~$\sigma_d(K)$ is the smallest singular value of~$K$.
 
Since~$K$ is diagonally dominant by both rows and columns, we can apply~\cite[Corollary 2]{varah} which implies (in the case of self-adjoint matrices)
\begin{align}
\sigma_d(K)&\geq \alpha(K)\ .
\end{align}
Combined with Eqs.~\eqref{eq:upperboundmvinner}, \eqref{eq:cbound_aux1}, \eqref{eq:cbound_aux2} and~\eqref{eq:alphaAK} we obtain
\begin{align}
|\langle \Psi,A\Psi\rangle|&\geq |\langle \Psi,K\Psi\rangle|\geq \sigma_d(K)\geq \alpha(K) = \alpha(A)
\end{align}
for every unit vector~$\Psi\in \mathbb{C}^d$. This implies the claim by definition of the Crawford number~$\cn(A)$.
\end{proof}

The following property of the inner numerical radius will be useful.
\begin{lemma}\label{lem:triangleinequalitynumericalradius}
Let~$S,T,A:\mathbb{C}^d\rightarrow\mathbb{C}^d$ be linear. Then
\begin{align}
\cn(SAT)&\geq \cn(A)-  \|A\|\cdot \|S-I\|\cdot \|T\|-\|T-I\| \cdot \|A\| \ \label{eq:firstclaimsat}
\end{align}
and 
\begin{align}
\cn(SAT)&\geq \cn(A)- \|A\|\cdot \|S\|\cdot \|T-I\|-\|S-I\| \cdot \|A\| \ . \label{eq:secondclaimsat}
\end{align}
\end{lemma}
\begin{proof}
Let~$\Psi\in\mathbb{C}^d$ with~$\|\Psi\|=1$ be arbitrary. Then 
\begin{align}
\langle \Psi, SAT\Psi\rangle &= \langle \Psi, (S-I)AT\Psi\rangle+\langle \Psi, AT\Psi\rangle\\
&=\langle \Psi, (S-I)AT\Psi\rangle+\langle \Psi,A(T-I)\Psi\rangle+\langle\Psi,A\Psi\rangle
\end{align}
and thus 
\begin{align}
    |\langle \Psi, SAT \Psi \rangle| 
    &\geq |\langle \Psi, A \Psi \rangle| - | \langle \Psi, (S-I) AT \Psi \rangle| - | \langle \Psi, A (T-I) \Psi \rangle | \\ 
    &\geq c(A) - | \langle \Psi, (S-I) AT \Psi \rangle| - | \langle \Psi, A (T-I) \Psi \rangle | 
\end{align}
by the triangle inequality and the Definition~\eqref{eq:crawford} of~$c(A)$.
The Claim~\eqref{eq:firstclaimsat} follows using that 
\begin{align}
|\langle \Psi, (S-I)AT\Psi\rangle|&\leq \|(S-I)AT\|\leq \|S-I\|\cdot\|A\|\cdot\|T\|\\
|\langle \Psi, A(T-I)\Psi\rangle|&\leq \|A(T-I)\|\leq \|A\|\cdot \|T-I\|\ 
\end{align}
for any unit vector~$\Psi$, and by taking the minimum over~$\Psi$. The Claim~\eqref{eq:secondclaimsat} is shown in an analogous manner.
\end{proof}
\begin{lemma} \label{lem:continuitycrawford}
    The Crawford number~$\cn: \cB(\mathbb{C}^d) \rightarrow [0,\infty)$ is (Lipschitz-)continuous.
\end{lemma}
\begin{proof}
    Let~$A, B \in \cB(\mathbb{C}^d)$. Let~$\Psi \in \mathbb{C}^d$ be normalized. Then 
    \begin{align}
        \left|\langle \Psi, A \Psi\rangle \right| &= \left|\langle \Psi, (A-B) \Psi\rangle + \langle \Psi, B \Psi\rangle\right|\\
        &\le \left| \langle \Psi, B \Psi\rangle\right| + \left|\langle \Psi, (A-B) \Psi\rangle\right|\\
        &\le \left| \langle \Psi, B \Psi\rangle\right| + \|A-B\|\, ,\label{eq:cnupper}
    \end{align}
    where we used the triangle inequality to obtain the second inequality and the Cauchy-Schwarz inequality to obtain the last inequality.
    Similarly, we can bound 
    \begin{align}
        \left|\langle \Psi, A \Psi\rangle \right| &\ge  \left||\langle \Psi, B \Psi\rangle - \langle \Psi, (A-B) \Psi\rangle|\right|\\
        &\ge |\langle \Psi, B \Psi\rangle| - |\langle \Psi, (A-B) \Psi\rangle| \\
        &\ge |\langle \Psi, B \Psi\rangle| - \|A-B\|\, .\label{eq:cnlower}
    \end{align}
    Combining Eqs.~\eqref{eq:cnupper} and~\eqref{eq:cnlower} we find 
    \begin{align}
        \cn(B) - \|A-B\| \le \cn(A) \le \cn(B) + \|A-B\|\, ,
    \end{align}
    or equivalently
    \begin{align}
        \left|\cn(A) - \cn(B) \right| \le \|A-B \|\, .
    \end{align}
    This shows the claim.
\end{proof}

We can eliminate the assumption~\eqref{it:secondconditionapositive}
in Lemma~\ref{lem:linearoplowerboundmx} as follows.
\begin{lemma}\label{lem:linearoplowerboundmxmodified}
Let~$B:\mathbb{C}^d\rightarrow\mathbb{C}^d$ be linear. Assume that~$B_{j,j}\neq 0$ for every~$j\in \{0,\ldots,d-1\}$. Define
\begin{align}
\varphi_j&=\arg B_{j,j}\qquad\textrm{ for }\qquad j\in \{0,\ldots,d-1\}\ .
\end{align}
Then 
\begin{align}
\cn(B)&\geq\min_{j} \left(|B_{j,j}|-\frac{1}{2}\sum_{\ell\neq j} |e^{-i\varphi_j}B_{j,\ell} +e^{i\varphi_\ell} \overline{B_{\ell, j}}|\right)-\|B\|\max_{j\in \{0,\ldots,d-1\}} \left|e^{i\varphi_j}-1\right| 
\end{align}
\end{lemma}
\begin{proof}
Let us  introduce the diagonal matrix
\begin{align}
D&=\mathsf{diag}\left(\{e^{i\varphi_j}\}_{j=0}^{d-1}\right)\ 
\end{align}
and the matrix
\begin{align}
A&=D^{-1}B\ .
\end{align}
Then we have 
\begin{align}
A_{j,j}=D^{-1}_{j,j}B_{j,j}=|B_{j,j}|\in\mathbb{R}\qquad\textrm{ for  all }\qquad j\in \{0,\ldots,d-1\}\ .\label{eq:realentriesa}
\end{align}
In particular, we have
\begin{align}
\cn(B)&=\cn(DA) \\
&\geq \cn(A)-\|A\|\cdot \|D-I\|\qquad\textrm{ by Lemma~\ref{lem:triangleinequalitynumericalradius}}\\
&\geq \alpha(A)-\|A\|\cdot \|D-I\|\qquad \textrm{ by Lemma~\ref{lem:linearoplowerboundmx} using Eq.~\eqref{eq:realentriesa}}\ .\label{eq:loweralphabadad}
\end{align}
The claim follows from Eq.~\eqref{eq:loweralphabadad} because
\begin{align}
\|A\|&=\|D^{-1}B\|=\|B\| \ , \\
\|D-I\|&=\max_{j\in \{0,\ldots,d-1\}} \left|e^{i\varphi_j}-1\right|\ ,
\end{align}
by definition of~$D$ and 
\begin{align}
\alpha(A)&=
\min_{j} \left(|\mathsf{Re}\,A_{j,j}|-\frac{1}{2}\sum_{\ell\neq j} |A_{j,\ell} + \overline{A_{\ell, j}}|\right)\\
&=\min_{j} \left(|B_{j,j}|-\frac{1}{2}\sum_{\ell\neq j} |e^{-i\varphi_j}B_{j,\ell} +e^{i\varphi_\ell} \overline{B_{\ell, j}}|\right) \ .
\end{align}
\end{proof}
An immediate consequence of Lemma~\ref{lem:linearoplowerboundmxmodified} is the following result for sparse matrices. We call a matrix~$B\in\mathsf{Mat}_{d\times d}(\mathbb{C})$~$s$-sparse if the number of non-zero entries in each row and column is at most~$s$.

\begin{corollary}[Inner numerical radius for sparse matrices]\label{cor:numericalradiussparse}
Let~$B \in \mathsf{Mat}_{d\times d}(\mathbb{C})$ be~$s$-sparse for~$s\in\{1,\ldots,d\}$ with~$B_{j,j} \neq 0$ for every~$j \in \{0,\dots,d-1\}$. Then 
\begin{align}
\cn(B)&\geq \left(\min_{j} |B_{j,j}|\right)-(s-1) \max_{j,\ell:j\neq \ell} |B_{j,\ell}|-\|B\|\max_{j\in \{0,\ldots,d-1\}}\left|
\frac{B_{j,j}}{|B_{j,j}|}
-1\right|\ .
\end{align}
\end{corollary}

We also need the following result for matrices with subnormalized rows and columns.

\begin{corollary} \label{cor: non sparse \cn(B)}
Let~$B\in \mathsf{Mat}_{d\times d}(\mathbb{C})$ be such that~$B_{j,j} \neq 0$ for every~$j \in \{0,\dots,d-1\}$. Assume that
\begin{align}
\sum_{\ell=0}^{d-1} |B_{j,\ell}|^2 &\leq 1 \qquad \textrm{and} \qquad\sum_{\ell=0}^{d-1} |B_{\ell,j}|^2 \leq 1 \label{eq:upperboundajellassumption}\qquad \textrm{for every} \qquad j\in \{0,\ldots,d-1\}\, .
\end{align}
Then the inner numerical radius of~$B$ is bounded as
\begin{align}
\cn(B)\geq 1-7\left(d\max_{j}|1- B_{j,j}|\right)^{1/2}\ .
\end{align}
\end{corollary}

\begin{proof}
    Define~$\varphi_j=\arg B_{j,j}$ for~$j\in \{0,\ldots,d-1\}$.
Let~$j\in \{0,\ldots,d-1\}$ be arbitrary. Then
\begin{align}
    \frac{1}{2}\sum_{\ell\neq j} |e^{-i\varphi_j}B_{j,\ell} +e^{i\varphi_\ell} \overline{B_{\ell, j}}| &\leq \frac{1}{2}\sum_{\ell\neq j}\left(|B_{j,\ell}| + |B_{\ell, j}| \right)\\
&\leq \frac{1}{2}\sqrt{d-1}\left(\left(\sum_{\ell\neq j}|B_{j,\ell}|^2\right)^{1/2} + \left(\sum_{\ell\neq j}|B_{\ell,j}|^2\right)^{1/2}\right)\\
&\leq \sqrt{d-1} \left(1-|B_{j,j}|^2\right)^{1/2}\\
&\leq  \left(2(d-1)|1- B_{j,j}|\right)^{1/2} \, , \label{eq:boundBoddgiag}
\end{align}
where we used the Cauchy-Schwarz inequality in the second step and the assumption~\eqref{eq:upperboundajellassumption} in the penultimate step. 
The last inequality follows from 
\begin{align}
    1- |B_{j,j}|^2 &= (1-|B_{j,j}|)(1+|B_{j,j}|) \le 2(1-|B_{j,j}|) = 2|1- |B_{j,j}|| \le  2|1 - B_{j,j}| \label{eq:complexnumberfirst} 
\end{align}  
 where we used that~$|B_{j,j}|\le 1$.
It follows from Eq.~\eqref{eq:boundBoddgiag} that
\begin{align}
\left| B_{j,j} \right| - \frac{1}{2} \sum_{\ell \neq j} \left| e^{-i \varphi_j} B_{j,\ell} + e^{i \varphi_\ell} \overline{B_{\ell,j}} \right|
&\geq 1 - |1 - B_{j,j}| - \left(2(d-1)|1 - B_{j,j}|\right)^{1/2} \\
&\geq 1 - 2 |1 - B_{j,j}| ^{1/2} - \left(2(d-1) |1 - B_{j,j}| \right)^{1/2} \notag \\
&\geq 1 - 3(d |1 - B_{j,j}|)^{1/2}, \label{eq:auxKjjmsumBjj}
\end{align}
for all~$j \in \{0, \ldots, d-1\}$, where we used~$|B_{j,j}| = |1 - (1 - B_{j,j})| \geq 1 - |1 - B_{j,j}|$ in the first step. The second inequality follows from~$|1 - B_{j,j}| \leq 2$ and 
\begin{align}
    \label{eq:cleqsqrt}
    x \leq 2\sqrt{x} \quad \text{ for } \quad x \in [0,2] \ .
\end{align}
In the last step we used that~$2 + \sqrt{2(d-1)} \leq 3\sqrt{d}$ for all~$d \geq 2$.

It follows from 
\begin{align}
    |1- B_{j,j}|B_{j,j}|^{-1}| &\le |1 - B_{j,j}| + |B_{j,j} -  B_{j,j}|B_{j,j}|^{-1}| = |1-B_{j,j}| + |1-|B_{j,j}|| \le 2|1-B_{j,j}| \label{eq:complexnumbersecond}
\end{align}  
and Eq.~\eqref{eq:cleqsqrt} that
\begin{align}
    \left|\frac{B_{j,j}}{|B_{j,j}|}- 1\right| &\le 2 |1-B_{j,j}| \le 4\left(|1-B_{j,j}| \right)^{1/2}\label{eq:boundBjj/Bjj}\, .
\end{align}
Moreover, the assumption~\eqref{eq:upperboundajellassumption} implies that
\begin{align}
    \|B\| \le \|B\|_2 \le \sqrt{d}  \label{eq: bound B norm}
\end{align}
by the Cauchy-Schwarz inequality.
The claim follows by combining Eqs.~\eqref{eq:auxKjjmsumBjj}, \eqref{eq:boundBjj/Bjj} and~\eqref{eq: bound B norm} with Lemma~\ref{lem:linearoplowerboundmxmodified}. 
\end{proof}

\section{Proof of Theorem~\ref{cor:shortmatrixBstatement} \label{app:gateerror_bound_matrixelem}}

In this section we give the technical details used to derive Theorem~\ref{cor:shortmatrixBstatement} from the main text.

More concretely, we combine the general lower bounds on the Crawford number from Appendix~\ref{sec:numericalradius} with the exact expressions for the logical gate error of unitary implementations (in particular, Lemma~\ref{lem:gateerrorinnernumericalradius} and~Corollary~\ref{cor:gateerrorunitaryimplement}). 

These bounds exploit the special form of the matrices~$B$ whose Crawford number is relevant for the logical gate error (see Definition~\ref{def:diagonalunitary}). In preparation for the proof of Theorem~\ref{cor:shortmatrixBstatement} from the main text, we show the following:

\begin{lemma} \label{lem: Delta bound matrix elements}
Let~$W:\cHin\rightarrow\cHout$ and~$U:\mathbb{C}^d\rightarrow\mathbb{C}^d$ be unitary. Let~$B=B^{U}_{\cLin,\cLout}(W,U)$ be the operator introduced in Definition~\ref{def:diagonalunitary}. Then 
\begin{align}
    \max_{j}|1 - B_{j,j}| &\le \sqrt{d} \max_{j,k}\left| U_{j,k} - \langle \encoded{j}|_{\cLout} W |\encoded{k} \rangle_{\cLin} \right|  \label{eq: first Bjj claim}
    \end{align}
    and 
    \begin{align}
            \min_{j} \left(\sum_{\ell=0}^{d-1} |B_{j,\ell}|^2,\ \sum_{\ell=0}^{d-1} |B_{\ell,j}|^2 \right)&\leq 1 \label{eq: second Bjj claim}\, .
            \end{align}
\end{lemma}
\begin{proof}
We first prove Eq.~\eqref{eq: first Bjj claim}.
Using the definition~\eqref{eq:bjkdefinition} of the matrix elements~$B_{j,\ell}$ of the operator~$B$ we have for fixed~$j \in \{0,\dots,d-1\}$
\begin{align}
|1 - B_{j,j}| &= \left|1 - \sum_{m=0}^{d-1} \overline{U_{m,j}} \langle \encoded{m}|_{\cLout}W |\encoded{j}\rangle_{\cLin} \right|\\
&= \left|  \sum_{m=0}^{d-1} \left(|{U_{m,j}}|^2 -  \overline{U_{m,j}} \langle \encoded{m}|_{\cLout}W |\encoded{j}\rangle_{\cLin}\right)\right|\\
&= \left|  \sum_{m=0}^{d-1} \overline{U_{m,j}}\left(U_{m,j} -  \langle \encoded{m}|_{\cLout}W |\encoded{j}\rangle_{\cLin}\right)\right|\\
&\le \left(\sum_{m=0}^{d-1} |U_{m,j}|\right) \cdot \max_{\ell,k}\left| U_{\ell,k} - \langle \encodedC{\ell}|_{\cLout} W |\encoded{k} \rangle_{\cLin} \right| \\
&\le \sqrt{d} \max_{\ell,k}\left| U_{\ell,k} - \langle \encodedC{\ell}|_{\cLout} W |\encoded{k} \rangle_{\cLin} \right| \, , \label{eq: bound 1 minus Bjj}
\end{align}
where we used the unitarity of~$U$ in the second identity, and the triangle inequality in the penultimate inequality, as well as the  Cauchy-Schwarz inequality and the unitarity of~$U$ in the last inequality. 

Next, we prove Eq.~\eqref{eq: second Bjj claim}. Using the definition~\eqref{eq:bjkdefinition}  of the matrix elements~$B_{j,k}$ we have for any fixed~$k\in\{0,\ldots, d-1\}$
\begin{alignat}{2}
    \sum_{j=0}^{d-1} |B_{j,k}|^2 &= \sum_{j=0}^{d-1} \left| \sum_{\ell=0}^{d-1} \overline{U_{\ell,j}} \bra{\boldsymbol{\ell}}_\cLout W \ket{\encoded{k}}_\cLin \right|^2 \\
    &= \sum_{m,\ell=0}^{d-1} \left( \sum_{j=0}^{d-1} U_{m,j} \overline{U_{\ell,j}} \right) \bra{\encoded{k}}_\cLin W^\dagger \ket{\encoded{m}}_\cLout \bra{\boldsymbol{\ell}}_\cLout W \ket{\encoded{k}}_\cLin\\ 
    &= \sum_{m,\ell=0}^{d-1} \delta_{m,\ell} \bra{\encoded{k}}_\cLin W^\dagger \ket{\encoded{m}}_\cLout \bra{\boldsymbol{\ell}}_\cLout W \ket{\encoded{k}}_\cLin  &&\hspace{-1.9cm}\text{ by unitarity of~$U$}\\ 
    &= \bra{\encoded{k}}_\cLin W^\dagger \left( \sum_{m=0}^{d-1} \ket{\encoded{m}} _\cLout \bra{\encoded{m}}_\cLout\right) W \ket{\encoded{k}}_\cLin \\ 
    &= \bra{\encoded{k}}_\cLin W^\dagger \pi_\cLout  W \ket{\encoded{k}}_\cLin  &&\hspace{-1.9cm}\text{ because~$ \sum_{m=0}^{d-1} \ket{\encoded{m}} _\cLout \bra{\encoded{m}}_\cLout = \pi_\cLout$} \\ 
    &=  \| \pi_\cLout  W \ket{\encoded{k}}_{\cLin} \|^2 &&\hspace{-1.9cm}\text{ because~$\pi_\cLout^2=\pi_\cLout$} \\
    &\leq 1 \ ,
\end{alignat}
where the inequality follows from the fact that~$\pi_\cLout$ is a projection, the unitarity of~$W$ and the fact that~$\ket{\encoded{k}}_{\cLin}$ is normalized.  
Similarly, for any fixed~$j\in \{ 0, \cdots, d-1 \}$, we have 
\begin{alignat}{2}
    \sum_{k=0}^{d-1} |B_{j,k}|^2 &= \sum_{k=0}^{d-1} \left| \sum_{m=0}^{d-1} \overline{U_{m,j}} \bra{\encoded{m}}_\cLout W \ket{\encoded{k}}_\cLin \right|^2 \\
    &= \sum_{m,\ell=0}^{d-1}  \overline{U_{m,j}}  U_{\ell,j} \bra{\encoded{m}}_\cLout W \left( \sum_{k=0}^{d-1} \ket{\encoded{k}}_\cLin \bra{\encoded{k}}_\cLin\right) W^\dagger \ket{\boldsymbol{\ell}}_\cLout  \\
    &= \sum_{m,\ell=0}^{d-1}  \overline{U_{m,j}}  U_{\ell,j} \bra{\encoded{m}}_\cLout W \pi_\cLin W^\dagger \ket{\boldsymbol{\ell}}_\cLout  
   &&\hspace{-2cm}\text{ because~$\sum_{k=0}^{d-1} \ket{\encoded{k}}_\cLin \bra{\encoded{k}}_\cLin = \pi_\cLin$} \\
   &= \left\| \sum_{\ell=0}^{d-1} U_{\ell,j} \pi_\cLin W^\dagger \ket{\boldsymbol{\ell}}_\cLout  \right\|^2 
   &&\hspace{-2cm}\text{ where we used~$\pi_\cLin^2 = \pi_\cLin$ }\\
   &= \left\| \pi_\cLin W^\dagger \ket{\encoded{j'}}_\cLout  \right\|^2 \\
   &\leq 1
\end{alignat}
where~$\ket{\encoded{j}'} = \sum_{\ell=0}^{d-1} U_{\ell,j} \ket{\boldsymbol{\ell}}_\cLout$ for~$j\in\{0,\ldots, d-1\}$ is normalized, which together with the fact that~$\pi_\cLout$ is a projection and~$W$ a unitary gives the inequality. 
 \end{proof}    

As a direct consequence we can prove Theorem~\ref{cor:shortmatrixBstatement} from the main text. For convenience we restate it here. 

\begin{corollary}[Theorem~\ref{cor:shortmatrixBstatement} from the main text restated]\label{cor:shortmatrixBstatement_app}
Let~$W:\cHin\rightarrow\cHout$ be unitary and~$\cW(\rho)=W\rho W^\dagger$. Let~$U:\mathbb{C}^d\rightarrow\mathbb{C}^d$ be unitary, and let~$B=B^{U}_{\cLin,\cLout}(W,U)$ be the operator introduced in Definition~\ref{def:diagonalunitary}.
Then we have the following.
\begin{enumerate}[(i)]
\item\label{it:ssparsityassumptiononea_app}
Suppose that~$B$ is~$s$-sparse and has positive entries~$\{B_{j,j}\}_{j=0}^{d-1}$ on the diagonal. 
Then
\begin{align}
\max\left\{\gateerror_{\cLin,\cLout} (\cW,U), \gateerror_{\cLout,\cLin} (\cW^\dagger,U^\dagger) \right\} &\leq 3 \left((1-\min_j B_{j,j})+(s-1)\max_{j,\ell:j\neq \ell }|B_{j,\ell}|\right)^{1/2}\ .
\end{align}
\item \label{it:claimsecondbassumptionb_app} Furthermore, we have 
\begin{align} 
    \max\left\{\gateerror_{\cLin,\cLout} (\cW,U), \gateerror_{\cLout,\cLin} (\cW^\dagger,U^\dagger) \right\} &\leq 8 d^{3/8} \left(\max_{j,k}\left| U_{j,k} - \langle \encoded{j}|_{\cLout} W |\encoded{k} \rangle_{\cLin} \right|\right)^{1/4}  \ .
\end{align}
\end{enumerate}
\end{corollary}

\begin{proof}
Recall that 

\begin{align}
    \gateerror_{\cLin,\cLout} (\cW,U)= \gateerror_{\cLout,\cLin} (\cW^\dagger,U^\dagger) &=  2\sqrt{1-\cn(B)^2}\ \label{eq:crawcnb}
\end{align}
according to Lemma~\ref{lem:gateerrorinnernumericalradius}, Corollary~\ref{cor:gateerrorunitaryimplement} and the assumption that~$W$ is unitary. 
Since 
$\sqrt{1-(1-\delta)^2}\leq \sqrt{2\delta}$ for every~$\delta\in [0,1]$, it  follows that
\begin{align}
\gateerror_{\cLin,\cLout} (\cW,U)&\leq  2\sqrt{2\delta} \qquad\textrm{ for any } \delta >0
\textrm{ satisfying}\qquad c(B)\geq 1-\delta\ .
\ \label{eq:upperboundgateerrordeltacb}
\end{align}
(We note that this bound is trivially satisfied for~$\delta>1$ by definition of the logical gate error.)

Let us first prove Claim~\eqref{it:ssparsityassumptiononea}. That is, assume that~$B$ is~$s$-sparse, with~$B_{j,j}>0$ for every~$j\in \{0,\ldots,d-1\}$. Then we have 
\begin{align}
\left|\frac{B_{j,j}}{|B_{j,j}|}-1\right|=0\qquad\textrm{ for all }\qquad j\in \{0,\ldots,d-1\}
\end{align}
and it follows from Corollary~\ref{cor:numericalradiussparse}
that
\begin{align}
\cn(B)&\geq  \left(\min_{j} B_{j,j}\right)-(s-1) \max_{j,\ell:j\neq \ell} |B_{j,\ell}|\\
&=1-\delta\qquad\textrm{ where }\qquad  \delta:=(1-\min_j B_{j,j})+(s-1)\max_{j,\ell:j\neq \ell }|B_{j,\ell}|\ .
\end{align}
Combining this with Eq.~\eqref{eq:upperboundgateerrordeltacb} we obtain Claim~\eqref{it:ssparsityassumptiononea_app}.

Next, we prove Claim~\eqref{it:claimsecondbassumptionb_app}. Lemma~\ref{lem: Delta bound matrix elements} states that the matrix~$B$ of interest satisfies 
\begin{align}
    \max_{j}|1 - B_{j,j}| &\le \sqrt{d} \max_{j,k}\left| U_{j,k} - \langle \encoded{j}|_{\cLout} W |\encoded{k} \rangle_{\cLin} \right|  \label{eq:generalBfirst}
    \end{align}
    and 
    \begin{align}
            \min_{j} \left(\sum_{\ell=0}^{d-1} |B_{j,\ell}|^2,\ \sum_{\ell=0}^{d-1} |B_{\ell,j}|^2 \right)&\leq 1 \label{eq:generalBsec}\, .
            \end{align}
            Because of Eq.~\eqref{eq:generalBsec} we can apply Corollary~\ref{cor: non sparse c(B)} which gives 
\begin{align}
    \cn(B)&\geq 1-7\sqrt{d}\left(\max_{j}|1- B_{j,j}|\right)^{1/2}\\
    &\geq   1-7 d^{3/4}\left(\max_{j,k}\left| U_{j,k} - \langle \encoded{j}|_{\cLout} W |\encoded{k} \rangle_{\cLin} \right| \right)^{1/2}
        \label{eq:generalcB}
\end{align}
when combined with Eq.~\eqref{eq:generalBfirst}. 
Claim~\eqref{it:claimsecondbassumptionb_app} follows by combining Eq.~\eqref{eq:generalcB} with Eq.~\eqref{eq:upperboundgateerrordeltacb} and using that~$2\sqrt{14}\leq  8$. 
\end{proof}

\section{Approximate GKP-states with integer spacing \label{sec:approximateGKPstates}}

Our goal is to analyze the implementations of logical gates introduced in Section~\ref{sec: logical gates in ideal GKP codes} in the context of approximate GKP codes. Here we introduce the relevant definitions and establish some required properties.
Namely, we establish various properties of approximate GKP-states with integer spacing.

In more detail, we introduce four different types of integer-spaced GKP states with parameters~$\kappa,\Delta>0$ and~$\varepsilon\in (0,1/2)$, namely
\begin{enumerate}[(i)]
\item The ``standard'' (untruncated) GKP state~$\ket{\gkp_{\kappa,\Delta}}$ with squeezing parameters~$(\kappa,\Delta)$, see Eq.~\eqref{eq:nontruncatedapproximateGKP}.
This state is commonly used as a basis for defining approximate GKP states. It is a superposition of Gaussians centered on integers with a Gaussian envelope. 
The envelope is introduced as a coefficient for each Gaussian  (local maximum). Correspondingly, we refer to this as a peak-wise envelope.
The state~$\ket{\gkp_{\kappa,\Delta}}$ has support on all of~$\mathbb{R}$.
\item The~$\varepsilon$-truncated GKP state~$\ket{\gkp^\varepsilon_{\kappa,\Delta}}$ with squeezing parameters~$(\kappa,\Delta)$ and truncation parameter~$\varepsilon$, see Eq.~\eqref{eq:GKP_eps}.
In contrast to the untruncated GKP state~$\ket{\gkp_{\kappa,\Delta}}$, this state only has support near integers (namely, on an interval of diameter~$2\varepsilon$ around each integer). It is the result of replacing the Gaussians  by truncated (projected) versions thereof.
\item
The untruncated GKP state with point-wise envelope~$\ket{\tGKP_{\kappa,\Delta}}$, which is defined in terms of two squeezing parameters~$(\kappa,\Delta)$ (see Eq.~\eqref{eq:gkp}). 
Similar to the state~$\ket{\gkp_{\kappa,\Delta}}$, this state has support on all of~$\mathbb{R}$. In contrast to that state, however, the Gaussian envelope is applied point-wise.
\item
The~$\varepsilon$-truncated GKP state with point-wise envelope~$\ket{\tGKP^\varepsilon_{\kappa,\Delta}}$, which is defined in terms of two squeezing parameters~$(\kappa,\Delta)$ and a truncation parameter~$\varepsilon$ (see Eq.~\eqref{eq:gkp_eps}). 
This state again only has support near integers.
\end{enumerate}
 The relationship between these four states is illustrated in  Fig.~\ref{fig:summarysectionapproximate}. We note that via the Fourier transform on~$L^2(\mathbb{R})$, there is an exact relationship between 
 untruncated GKP states with peak-wise and point-wise envelopes, see Lemma~\ref{thm:fouriertransformapproximate} for a detailed statement.
 
 While closed-form analytic expressions for the (inverse) Fourier transform of the untruncated GKP states~$\ket{\gkp_{\kappa,\Delta}}$
 and~$\ket{\tGKP_{\kappa,\Delta}}$
can easily be obtained (see Lemma~\ref{thm:fouriertransformapproximate}),  the truncated GKP states~$\ket{\gkp^\varepsilon_{\kappa,\Delta}}$, $\ket{\tGKP^\varepsilon_{\kappa,\Delta}}$ are computationally convenient because of the specific form of their support. In particular, we use~$\ket{\gkp^\varepsilon_{\kappa,\Delta}}$ as a starting point to define  an orthonormal basis of a~$d$-dimensional approximate GKP code (see Section~\ref{sec:truncatedapproximategkpcode}). 
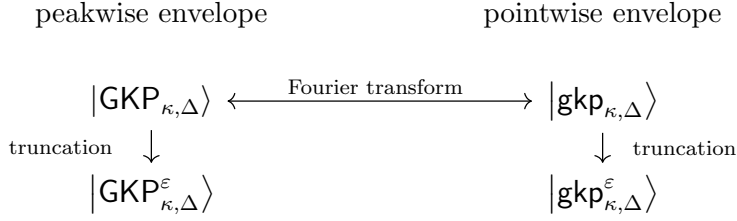
\begin{figure}[h]
\centering
\[
\begin{tikzcd}[row sep=1em, column sep=6em]
\text{\small peakwise envelope} &  \text{\small pointwise envelope}\\
\vspace{-35ex}
 \ket{\gkp_{\kappa,\Delta}} 
   \arrow[r, leftrightarrow, "\text{\quad Fourier transform}\quad "] 
   \arrow[d,  "\text{\quad truncation}\quad ", swap]
 &
 \ket{\tGKP_{\kappa,\Delta}} 
     \arrow[d,  "\text{\quad truncation}\quad "]
 \\
 \ket{\gkp^\varepsilon_{\kappa,\Delta}} 
    &
 \ket{\tGKP^\varepsilon_{\kappa,\Delta}}
\end{tikzcd}
\]
\caption{Illustration of the relationship between the different types of integer-spaced GKP states.
The state~$\ket{\gkp_{\kappa,\Delta}}$ is the ``standard'' approximate GKP state, a superposition of Gaussians centered on integers with a (peak-wise) Gaussian envelope. The state~$\ket{\gkp^\varepsilon_{\kappa,\Delta}}$ is obtained from it by replacing the Gaussians by truncated Gaussians, and has support only near each integer.  Similarly, the state~$\ket{\tGKP_{\kappa,\Delta}}$ is a superposition of Gaussians but with a Gaussian envelope which is applied point-wise. The truncated state~$\ket{\tGKP^\varepsilon_{\kappa,\Delta}}$ is obtained from it 
by replacing the Gaussians by truncated Gaussians, and has support only near each integer. Finally,  peakwise and point-wise GKP states~$\ket{\gkp_{\kappa,\Delta}}$ and~$\ket{\tGKP_{\kappa,\Delta}}$ are related by the Fourier transform, which essentially swaps the role of~$\kappa$ and~$\Delta$, see Lemma~\ref{thm:fouriertransformapproximate} for a detailed statement.\label{fig:summarysectionapproximate}
}
\end{figure}

We will also show that
for an appropriate choice of parameters~$(\kappa,\Delta,\varepsilon)$, the states
$\ket{\gkp_{\kappa,\Delta}}$, $\ket{\gkp^\varepsilon_{\kappa,\Delta}}$, 
$\ket{\tGKP_{\kappa,\Delta}}$
 and~$\ket{\tGKP^\varepsilon_{\kappa,\Delta}}$ are close to each other, see Fig.~\ref{fig:quantitativeversionclose} for  a summary of the results establishing closeness of pairs of such states.

\begin{figure}[h]
\centering
\[
\begin{tikzcd}[row sep=3em, column sep=6em]
 \ket{\gkp_{\kappa,\Delta}} 
   \arrow[r, leftrightarrow, "\text{\quad Lemma~\ref{lem:pointwisepeakwiseGKP}}\quad "] 
   \arrow[d, leftrightarrow, "\text{\quad Lemma~\ref{lem:truncatedapproximateGKPstates}}\quad ", swap]
 &
 \ket{\tGKP_{\kappa,\Delta}} 
   \arrow[d, leftrightarrow, "\text{\quad Lemma~\ref{lem:truncatedapproximategkpstatepointw}}\quad"]
 \\
 \ket{\gkp^\varepsilon_{\kappa,\Delta}}  
   \arrow[ur, 
      "\rotatebox{25}{\normalsize\text{\tiny Lemma~\ref{eq:tgkpgkpeps}}}"{pos=0.6, anchor=center, xshift=-6pt, yshift=3pt}]
 &
 \ket{\tGKP^\varepsilon_{\kappa,\Delta}}
\end{tikzcd}
\]
\caption{The four states~$\ket{\gkp_{\kappa,\Delta}}$, $\ket{\gkp^\varepsilon_{\kappa,\Delta}}$, 
~$\ket{\tGKP_{\kappa,\Delta}}$
 and~$\ket{\tGKP^\varepsilon_{\kappa,\Delta}}$ are close for suitable choices of parameters~$(\kappa,\Delta,\varepsilon)$. Quantitative bounds 
on the overlap of pairs of states are established in the Lemmas indicated.\label{fig:quantitativeversionclose}}
\end{figure}
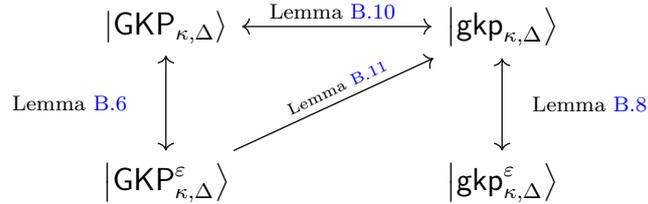

We proceed as follows. In Section~\ref{sec:truncatedgaussians}, we first establish a few properties of (truncated) Gaussians which are used subsequently. 
In Section~\ref{sec:differentnotions} we then relate different notions of approximate GKP-states.
In Section~\ref{sec:propertiesapproximateGKP} we 
compute matrix elements of certain Gaussian unitaries with respect to approximate GKP-states.  Finally, in Section~\ref{sec:fouriertransformapproximateGKP} we discuss the Fourier transform of GKP-states.

We note that most of the statements given here were previously established in Ref.~\cite{brenner2024complexity}. We provide a summary to make this paper self-contained.

\subsection{(Truncated) Gaussians and their properties\label{sec:truncatedgaussians}}

We first show that the centered Gaussian~$\Psi_\Delta$
and its truncated version~$\Psi_\Delta^\varepsilon$ are close to each other for suitably chosen parameters.
  For later use, we give the following (see Eqs.~\eqref{it:firstclosenessgaussiantrunc} and~\eqref{it:fourthclosenessgaussiantrunc}) different versions of this statement.
\begin{lemma}[Gaussians and truncated Gaussians]\label{lem:closenessgaussiantrunc}
Let~$\Delta>0$ and~$\varepsilon\in (0,1/2)$.
Then the following holds:
\begin{enumerate}[(i)]
\item $|\langle \Psi_\Delta,\Psi^\varepsilon_\Delta\rangle|^2 \geq 1-2e^{-(\varepsilon/\Delta)^2}$.\label{it:firstclosenessgaussiantrunc}
\item $\|\Pi_{[-\varepsilon,\varepsilon]}\Psi_\Delta\|^2 \geq 1-2e^{-(\varepsilon/\Delta)^2}$.\label{it:secondclosenessgaussiantrunc}
\item\label{it:thirdclosenessgaussiantrunc}
$\|\Pi_{[-\varepsilon,\varepsilon]}\Psi_\Delta\|\geq 1-2(\Delta/\varepsilon)^4$. 
\item\label{it:fourthclosenessgaussiantrunc}
$\| \Psi_\Delta - \Psi_\Delta^\varepsilon\|\leq 8 (\Delta/\varepsilon)^4$.
\end{enumerate}
\end{lemma}
\begin{proof}
For the proof of Claims~\eqref{it:firstclosenessgaussiantrunc} and~\eqref{it:secondclosenessgaussiantrunc}, we refer to~\cite[Lemma A.2]{brenner2024complexity}, where these inequalities are shown for arbitrary~$\Delta>0$ and~$\varepsilon\in (0,1/2)$. 

Claim~\eqref{it:thirdclosenessgaussiantrunc} follows immediately from Claim~\eqref{it:secondclosenessgaussiantrunc}
since
\begin{align}
    \|\Pi_{[-\varepsilon,\varepsilon]}\Psi_\Delta\| &\geq \|\Pi_{[-\varepsilon,\varepsilon]}\Psi_\Delta\|^2\\
    &\ge 1 - 2e^{-(\varepsilon/ \Delta)^2}\qquad\textrm{ by Claim~\eqref{it:secondclosenessgaussiantrunc}} \\
    &\ge 1 - 2 (\Delta/\varepsilon)^4 \ 
\end{align}
where in the first inequality we use the fact that~$\|\Pi_{[-\varepsilon,\varepsilon]}\Psi_\Delta\| \leq 1$, and in the last inequality we use  that
\begin{align}
e^{-x} \le 1/x^2\qquad\textrm{ for }\qquad x>0\ .\label{eq:emnvx}
\end{align}
The Claim~\eqref{it:fourthclosenessgaussiantrunc}  follows immediately from Claim~\eqref{it:firstclosenessgaussiantrunc}: We have 
\begin{align}
     \langle\Psi_\Delta, \Psi_\Delta^\varepsilon\rangle \ge 
     |\langle\Psi_\Delta, \Psi_\Delta^\varepsilon\rangle|^2 \ge
     1 - 2 e^{-(\varepsilon/\Delta)^2}\, ,
\end{align}
because~$0 \leq \langle\Psi_\Delta, \Psi_\Delta^\varepsilon\rangle \leq 1$.
Therefore
\begin{align} \label{eq: Psi Psi ep L^2}
     \| \Psi_\Delta - \Psi_\Delta^\varepsilon\| &= \sqrt{2\left(1 - \langle\Psi_\Delta, \Psi_\Delta^\varepsilon\rangle\right)}\\ &\le 2e^{-(\varepsilon/\Delta)^2/2}\qquad\textrm{ by Claim~\eqref{it:firstclosenessgaussiantrunc}}\\
     & \le 8 (\Delta/\varepsilon)^4\, ,
\end{align}
where we again made use of Eq.~\eqref{eq:emnvx}.
\end{proof}

\begin{lemma}[Momentum-translated Gaussian states]\label{lem:momentumtranslatedgaussian}
Let~$z\in\mathbb{Z}$ be an integer. Let~$\Delta>0$
and~$\varepsilon\le 1/(4|z|)$. Then the following holds.
\begin{enumerate}[(i)]
\item
$\langle \Psi_\Delta,e^{2\pi i z Q}\Psi_\Delta\rangle\geq 1-10(z\Delta)^2$.
\item\label{it:thirdpropertymomentumtranslated}
$\langle \Psi^\varepsilon_\Delta,e^{2\pi i z Q}\Psi^\varepsilon_\Delta\rangle\geq 1-10(z\Delta)^2-16(\Delta/\varepsilon)^4$.
\end{enumerate}
In particular, these inner products are real.
\end{lemma}
\begin{proof}
         We note that for an arbitrary function~$\Phi \in L^2(\mathbb{R})$ which is even, i.e., which satisfies~$\Phi(-x) = \Phi(x)$ for all~$x \in \mathbb{R}$, 
    the inner product~$\langle \Phi, e^{2\pi i z Q} \Phi\rangle$ is real because
    \begin{align} \label{eq: general overlap Q shift}
        \langle \Phi, e^{2\pi i z Q} \Phi\rangle = \int |\Phi(x)|^2 e^{2\pi i z x} \, dx= \int |\Phi(x)|^2 \cos(2\pi zx) \, dx \, .
    \end{align}
    Eq.~\eqref{eq: general overlap Q shift}
    implies
    that
    \begin{align}
    \langle \Psi_\Delta, e^{2\pi i z Q}\Psi_\Delta\rangle\in\mathbb{R}\qquad\textrm{ and }\qquad \langle \Psi_\Delta^\varepsilon, e^{2\pi i z Q}\Psi_\Delta^\varepsilon\rangle \in\mathbb{R}
    \end{align}
    since both~$\Psi_\Delta$ and~$\Psi_\Delta^\varepsilon$ are even.

    Eq.~\eqref{eq: general overlap Q shift} also implies that
    \begin{align}
    \langle \Psi_\Delta^\varepsilon,   e^{2\pi i z Q} \Psi_\Delta^\varepsilon\rangle \ge 0
    \end{align} 
    if~$\varepsilon\le1/(4|z|)$, since this assumption implies that~$2\pi\varepsilon |z| \le \pi /2$ and thus~$\cos(2\pi z x )\geq 0$ for~$x \in [-\varepsilon,\varepsilon]$. 
    By a simple application of the Cauchy-Schwarz inequality (see Lemma~\ref{lem:continuitymatrixelem})
    we therefore obtain    \begin{align} \label{eq: overlap e 2piQ eps}
        \langle\Psi_\Delta^\varepsilon, e^{2\pi i z  Q} \Psi_\Delta^\varepsilon\rangle 
        \ge |\langle\Psi_\Delta, e^{2\pi i z  Q} \Psi_\Delta\rangle| - 2\| \Psi_\Delta - \Psi_\Delta^\varepsilon\|\, .
    \end{align}
    We can compute the inner product explicitly as
    \begin{align}
        \langle\Psi_\Delta, e^{2\pi i z  Q} \Psi_\Delta\rangle &  = \frac{1}{\sqrt{\pi}\Delta} \int e^{-x^2/\Delta^2} e^{2\pi i z x} dx \\
        &= e^{-\pi^2 \Delta^2 z^2} \\
        &\ge 1 - \pi^2 \Delta^2 z^2\\
        &\ge 1 - 10(z\Delta)^2\, , \label{eq: fourier int bound}
    \end{align}
    where we used~$e^{-x} \ge 1 - x$ for all~$x \in \mathbb{R}$ to obtain the first inequality.
    
    Inserting Eq.~\eqref{eq: fourier int bound} 
    and the inequality~$\|\Psi_\Delta - \Psi_\Delta^\varepsilon\| \le 8 (\Delta/\varepsilon)^4$
    (see  Lemma~\ref{lem:closenessgaussiantrunc}~\eqref{it:fourthclosenessgaussiantrunc})  into Eq.~\eqref{eq: overlap e 2piQ eps} gives 
    \begin{align}
        \langle\Psi_\Delta^\varepsilon, e^{2\pi i z  Q} \Psi_\Delta^\varepsilon\rangle \ge 1 - 10(z\Delta)^2 - 16 (\Delta/\varepsilon)^4\, . \label{eq: bound phase shift trunc psi}
    \end{align}
\end{proof}

Consider the operator
\begin{align}
W_{\Pgate}=e^{\pi i (d Q^2 + c_d Q )}\qquad\textrm{ where }\qquad  
c_d = 
\begin{cases}
    0 & \text{ if } d \text{ is even} \\
    1 & \text{ otherwise}
\end{cases} \ .
\end{align}
 We refer to~$W_{\Pgate}$ as the phase operator because it is an implementation of the phase gate~$\Pgate$ for ideal GKP codes, see Section~\ref{sec: logical gates in ideal GKP codes}.
We need the following statements about the expectation of this operator in a Gaussian state~$\Psi_\Delta$ and a truncated Gaussian~$\Psi^\varepsilon_\Delta$. We consider a slightly more general operator depending on an additional parameter~$y\in\bbR$.
\begin{lemma}[Exptectation of the phase operator in Gaussian states]\label{lem:gaussianphase}
Let~$\Delta>0$, $\varepsilon<1/2$ and~$y\in\bbR$ be arbitrary.
Then the following holds.
\begin{enumerate}[(i)]
\item \label{it:firstgaussianphase}
The Gaussian~$\Psi_\Delta$ satisfies    
\begin{align}
|\langle \Psi_\Delta,  e^{\pi i (d Q^2  + (2dy + c_d)Q )} \Psi_\Delta \rangle| \le e^{- \lambda^2 (y + c_d/(2d))^2}\qquad\textrm{ where }\lambda^2 =\pi^2 d^2 \Delta^2 /(1+ \pi^2 d^2 \Delta^4)\ .
\end{align}
    \item\label{it:secondgaussianphase}
    The truncated Gaussian~$\Psi^\varepsilon_\Delta$ satisfies
    \begin{align}
        |\langle \Psi^\varepsilon_\Delta,  e^{\pi i (d Q^2  + (2dy + c_d)Q )} \Psi^\varepsilon_\Delta \rangle| \le e^{- \lambda^2 (y + c_d/(2d))^2}+16(\Delta/\varepsilon)^4\ .
        \end{align}
\end{enumerate}
\end{lemma}
\begin{proof}
We note that Claim~\eqref{it:secondgaussianphase} follows immediately from Claim~\eqref{it:firstgaussianphase}, the Cauchy-Schwarz inequality (see Lemma~\ref{lem:continuitymatrixelem}), and Lemma~\ref{lem:closenessgaussiantrunc}~\eqref{it:fourthclosenessgaussiantrunc}. 

To prove Claim~\eqref{it:firstgaussianphase}, we 
will show that the quantity of interest is given by a Fresnel integral~$\int  e^{ia x^2 + ibx} \, dx$.
This integral converges if~$\mathrm{Im}(a)>0$ and we have (see e.g., \cite[p. 492-493]{gradshteyn2014table})
\begin{align}
     \int  e^{ia x^2 + ibx} \, dx = \left(\frac{\pi i}{a}\right)^{1/2} e^{- i b^2/(4a)}.
\end{align}
In particular, 
\begin{align}
      \left|\int  e^{ia x^2 + ibx} \, dx\right| = \left(\frac{\pi}{|a|}\right)^{1/2}  e^{- b^2 \mathsf{Im}(a) /(4|a|^2)}\qquad\textrm{ if  }\qquad \mathsf{Im}(a)>0\textrm{ and }b \in \mathbb{R}\ .\label{eq:fresnelintegraltouse}
\end{align}
By definition of the unitary~$e^{\pi i (d Q^2  + (2yd + c_d)Q )}$ we have
\begin{align}
    \langle \Psi_\Delta,  e^{\pi i (d Q^2  + (2yd + c_d)Q )} \Psi_\Delta \rangle &= \int |\Psi_\Delta(x)|^2 e^{i\pi (x^2d + (2yd + c_d)x)} \, dx \\
    &= \frac{1}{\sqrt{\pi}\Delta} \int  e^{- x^2/\Delta^2} e^{i\pi (x^2d + (2yd + c_d)x)} \, dx \\
    &= \frac{1}{\sqrt{\pi}\Delta} \int  e^{ia x^2 + ibx} \, dx   \qquad \textrm{where} \qquad \begin{matrix}
        a &=& \pi d + i/\Delta^2\\
        b &=& \pi (2dy + c_d)
    \end{matrix}\ .  \label{eq: fresnel int}
\end{align}
Using Eq.~\eqref{eq:fresnelintegraltouse}  we thus conclude that
\begin{align}
    |\langle \Psi_\Delta,  e^{\pi i (d Q^2  + (2yd + c_d)Q )} \Psi_\Delta \rangle| &= \left(1 + \pi^2 d^2 \Delta^4\right)^{-1/4} e^{-\pi^2\Delta^2(2yd + c_d)^2/(4(1 + \pi^2 d^2 \Delta^4))} \\
    &\le e^{- \lambda^2 (y + s)^2}\ ,
\end{align}
where we used that~$1+ \pi^2 d^2 \Delta^4 >1$ to obtain the inequality.
\end{proof}

\subsection{Different approximate GKP-states with integer spacing\label{sec:differentnotions}}
Let~$\kappa,\Delta>0$ and~$\varepsilon \in (0,1/2)$. Recall the definition of the approximate GKP-state~$\gkp^\varepsilon_{\kappa,\Delta}$ from Eq.~\eqref{eq:GKP_eps}, that is, 
\begin{align}
    \gkp_{\kappa,\Delta}^\varepsilon(x)&=C_{\kappa} \sum_{z \in \mathbb{Z}} \eta_\kappa(z) \chi_{\Delta}^\varepsilon(z)(x)\qquad\textrm{ for }\qquad x\in\mathbb{R}\,  
\end{align}
and that of the untruncated approximate GKP-state~$\gkp_{\kappa,\Delta} \in L^2(\mathbb{R})$
defined as
\begin{align}
    \gkp_{\kappa,\Delta}(x)&=C_{\kappa,\Delta} \sum_{z \in \mathbb{Z}} \eta_\kappa(z) \chi_{\Delta}(z)(x)\qquad\textrm{ for }\qquad x\in\mathbb{R}\, ,
\end{align}
see Eq.~\eqref{eq:nontruncatedapproximateGKP}.  We note that in the definitions of~$\gkp_{\kappa,\Delta}^\varepsilon$ and~$\gkp_{\kappa,\Delta}$, the envelope~$\eta_\kappa$ appears as a coefficient in the superposition at each peak~$z \in \mathbb{Z}$.  Correspondingly, we sometimes refer to these
states as ``peak-wise'' approximate GKP-states.
Here we introduce additional notions of approximate GKP-states. They all depend on two parameters~$(\kappa,\Delta)$ which determine the width of the envelope and individual peaks, respectively.

Consider what we call a ``point-wise'' approximate GKP-state: Here the envelope~$\eta_{\kappa}$ is applied pointwise, i.e., the function is a pointwise product of the envelope (given by~$\eta_\kappa$) and the peak at~$z\in\mathbb{Z}$ (given by the function~$\chi_\Delta(z)$). That is, we define the point-wise approximate GKP-state~$\tGKP_{\kappa,\Delta}\in L^2(\mathbb{R})$ by
\begin{align}
\label{eq:gkp} \tGKP_{\kappa,\Delta}(x)&=D_{\kappa, \Delta} \sum_{z \in \mathbb{Z}} \eta_\kappa(x) \chi_{\Delta}(z)(x)\qquad\textrm{ for }\qquad x\in\mathbb{R}\ ,
\end{align}
where~$D_{\kappa,\Delta} > 0$ is such that 
$\|\tGKP_{\kappa,\Delta}\|=1$. 
We also introduce the~$\varepsilon$-truncated version
\begin{align}
\label{eq:gkp_eps} \tGKP^\varepsilon_{\kappa,\Delta}(x)=E_{\kappa,\Delta}(\varepsilon)\sum_{z\in\mathbb{Z}}\eta_\kappa(x)\chi_\Delta^\varepsilon(z)(x)\qquad\textrm{ for }\qquad x\in\mathbb{R}\ ,
\end{align}
where the constant~$E_{\kappa,\Delta}(\varepsilon) > 0$ ensures that this function is normalized.

As we argue below,  the constants~$C_{\kappa,\Delta}$ and~$D_{\kappa,\Delta}$ are related by the following.
\begin{lemma}[Normalization of approximate GKP-states]\label{lem:ckd_dkd_eq} 
Let~$\kappa,\Delta>0$ be arbitrary.
 The normalization constant~$C_{\kappa,\Delta}$
of the peak-wise approximate GKP-state~$\gkp_{\kappa,\Delta}$ is equal  to the 
normalization constant~$D_{(2\pi)\Delta,\kappa/(2\pi)}$ of the point-wise approximate GKP-state~$\tGKP_{(2\pi)\Delta,\kappa/(2\pi)}$, i.e.,   
\begin{align}
D_{(2\pi)\Delta,\kappa/(2\pi)} &=C_{\kappa,\Delta} \ .
\end{align}
\end{lemma}

We prove Lemma~\ref{lem:ckd_dkd_eq} using the unitarity of the Fourier transform, see Section~\ref{sec:fouriertransformapproximateGKP}. 

For later reference, we summarize various bounds on the normalization constants~$C_{\kappa,\Delta}$, $C_\kappa$, $D_{\kappa,\Delta}$, $E_{\kappa,\Delta}(\varepsilon)$ of the states~$\gkp_{\kappa,\Delta}$ ~$\gkp_{\kappa,\Delta}^\varepsilon$, $\tGKP_{\kappa,\Delta}$ and~$\tGKP_{\kappa,\Delta}^\varepsilon$ here.

\begin{lemma}[Bounds on normalization constants of GKP-states] 
\label{lem:cte_normaliz_bound}
Let~$\kappa\in(0,1/2), \Delta \in (0,1/8)$ and~$ \varepsilon \in (0,1/2)$. Then
\begin{enumerate}[(i)]
\item
$C_\kappa^2\sum_{z\in\mathbb{Z}}\eta_\kappa(z)^2=1 \ .$ \label{it:normalizationvsfi}
    \item \label{eq:ckappaupperlowerbound} $1-\frac{\kappa}{6} \geq C_\kappa \geq  1-\frac{\kappa}{3} \ .$
    \item \label{eq:ckappadeltalowerbound} $1 -\frac{\kappa}{6} \ge C_{\kappa, \Delta} \ge 1 - \frac{\kappa}{3} - 3\Delta \ .$
    \item  \label{eq:ckd_dkd_eq} $ 1 - \Delta \ge  D_{\kappa,\Delta} \ge 1 - 3\Delta - \frac{\kappa}{2} \ .$ 
    \item \label{eq:mainclaimekappadeltackappa} $ 1 - \frac{\kappa}{6}\geq E_{\kappa,\Delta}(\varepsilon) 
    \geq 1 -  \frac{\kappa}{2}  \ .$
\end{enumerate}
\end{lemma}

\begin{proof}
Assume~$\kappa\in(0,1/2)$ and~$\Delta \in (0,1/8)$.  Claim~\eqref{it:normalizationvsfi} follows immediately from the definition: it amounts to the statement that~$\|\gkp_{\kappa,\Delta}^\varepsilon\|^2=1$. Here we used the fact that the states~$\{\chi^\varepsilon_\Delta(z)\}_{z\in\mathbb{Z}}$ are pairwise orthogonal for~$\varepsilon<1/2$. 

To establish the Claim~\eqref{eq:ckappaupperlowerbound}, we will use the inequalities
\begin{align}
    \label{eq:sqrt1px_LB} (1+x)^{-1/2} &\geq 1 - x/2 \qquad \text{ for } x \geq 0 \ , \\
    \label{eq:sqrt1px_UB} (1+x)^{-1/2} &\leq 1 - x/3 \qquad \text{ for } x \in [0,1/2) \ .
\end{align}
From~\cite[Lemma A.8]{brenner2024complexity} we have 
\begin{align}
    1-\frac{\kappa}{\sqrt{\pi}} \leq C_\kappa^{-2}\leq 1+\frac{\kappa}{\sqrt{\pi}} \ .
\end{align} 
This together with Eqs.~\eqref{eq:sqrt1px_LB} and~\eqref{eq:sqrt1px_UB} gives
\begin{align}
    C_\kappa &\le \left(1-\frac{\kappa}{\sqrt{\pi}}\right)^{-1/2}  \le 1 - \frac{\kappa}{3 \sqrt{\pi}} \le 1-\frac{\kappa}{6} \ ,
    \\
    C_\kappa &\ge \left(1+\frac{\kappa}{\sqrt{\pi}}\right)^{-1/2}  \ge 1 - \frac{\kappa}{2 \sqrt{\pi}} \geq  1-\frac{\kappa}{3} \ ,
\end{align}
which is the Claim~\eqref{eq:ckappaupperlowerbound}.

From~\cite[Eqs.~(158) and~(160)]{brenner2024complexity} we have
\begin{align}
    C_\kappa^{-2}\leq C_{\kappa, \Delta}^{-2} \leq 1+\frac{\kappa}{\sqrt{\pi}}+2(\sqrt{2 \pi}+\kappa) \Delta \ .
\end{align}
This together with Eq.~\eqref{eq:sqrt1px_LB} and Claim~\eqref{eq:ckappaupperlowerbound} gives
\begin{align}
 C_{\kappa, \Delta} &\le C_\kappa \le  1 -  \frac{\kappa}{6} \ , \\
C_{\kappa, \Delta} &\ge \left(1+\frac{\kappa}{\sqrt{\pi}}+2(\sqrt{2 \pi}+\kappa) \Delta\right)^{-1/2} \ge 1-\frac{\kappa}{2\sqrt{\pi}}-(\sqrt{2 \pi}+\kappa) \Delta \ge 1 -  \frac{\kappa}{3} - 3\Delta \ ,
\end{align}
which is the Claim~\eqref{eq:ckappadeltalowerbound}.

From Lemma~\ref{lem:ckd_dkd_eq} we have~$D_{\kappa,\Delta} = C_{2\pi\Delta, \kappa/(2\pi)}$. Using Claim~\eqref{eq:ckappadeltalowerbound} then gives
\begin{align}
    D_{\kappa,\Delta} &= C_{2\pi\Delta, \kappa/(2\pi)} \geq 1 - \frac{2\pi}{3}\Delta - \frac{3}{2\pi} \kappa \geq 1 - 3 \Delta - \frac{\kappa}{2} \ ,\\
    D_{\kappa,\Delta} &= C_{2\pi\Delta, \kappa/(2\pi)} \leq 1 - \frac{2\pi}{6}\Delta \le 1 - \Delta \ ,
\end{align}
which is Claim~\eqref{eq:ckd_dkd_eq}.

We give the proof of Claim~\eqref{eq:mainclaimekappadeltackappa} below.
The normalization constant~$E_{\kappa,\Delta}(\varepsilon)$ satisfies
\begin{align}
E_{\kappa,\Delta}(\varepsilon)^{-2}&=\sum_{z,z'\in\mathbb{Z}}
\int \eta_\kappa(x)^2 \Psi_\Delta^\varepsilon(x-z)\Psi_\Delta^\varepsilon(x-z')dz\\
\label{eq:Einv2_aux} &=\sum_{z\in\mathbb{Z}}
\int \eta_\kappa(x)^2 \Psi_\Delta^\varepsilon(x-z)^2dz\\
&=\sum_{z\in\mathbb{Z}}\int_{-\varepsilon}
^{\varepsilon} \eta_\kappa(z+x)^2 \Psi^\varepsilon_\Delta(x)^2 dx\\
&=\int_{-\varepsilon}^\varepsilon 
    \left(\sum_{z\in\mathbb{Z}} \eta_\kappa(z+x)^2\right)
    \Psi^\varepsilon_\Delta(x)^2 dx\ .\label{eq:toprovezvarepsilonetakappa}
\end{align}
where (in Eq.~\eqref{eq:Einv2_aux}) we used that the states~$\{\Psi^\varepsilon_\Delta(\cdot - z)\}_{z \in \mathbb{Z}}$ have pairwise disjoint support because~$\supp( \Psi_\Delta^\varepsilon(\cdot - z) ) \subset [z-\varepsilon, z+\varepsilon]$.
Lemma~\ref{lem: periodic gaussians} implies that 
\begin{align}
    \label{eq:discgauss_aux} e^{-\kappa^2 \varepsilon^2/2}\sum_{z\in\mathbb{Z}}\eta_\kappa(z)^2
    \leq \sum_{z\in\mathbb{Z}}\eta_\kappa(z+x)^2
    \leq \sum_{z\in\mathbb{Z}}\eta_\kappa(z)^2\qquad\textrm{ for all }x\in [-\varepsilon,\varepsilon] \ ,
\end{align}
which can be rewritten as
\begin{align}
e^{-\kappa^2\varepsilon^2/2}C_\kappa^{-2}\leq  \sum_{z\in\mathbb{Z}}\eta_\kappa(z+x)^2\leq C_\kappa^{-2} \label{eq:discgauss_auxsecond}
\end{align}
because~$C_\kappa^{-2}=\sum_{z\in\mathbb{Z}}\eta_\kappa(z)^2$. Combining Eqs.~\eqref{eq:toprovezvarepsilonetakappa} and~\eqref{eq:discgauss_auxsecond}
gives
\begin{align}
    e^{-\kappa^2\varepsilon^2/2}C_\kappa^{-2} \leq E_{\kappa,\Delta}(\varepsilon)^{-2}\leq
    C_\kappa^{-2}
    \end{align}
    because~$1= \|\Psi^\varepsilon_\Delta\|^2 = \int_{-\varepsilon}^\varepsilon \Psi_\Delta^\varepsilon(x)^2 dx$. This together with Claim~\eqref{eq:ckappaupperlowerbound} as well as the inequality~$e^{x} \geq 1 + x$ for~$x\geq 0$ gives
\begin{align}
    E_{\kappa,\Delta}(\varepsilon) &\le C_\kappa \le 1 - \frac{\kappa}{6} 
    \ , \\
    E_{\kappa,\Delta}(\varepsilon) &\geq e^{\kappa^2\varepsilon^2/4}C_\kappa \geq \left(1 + \frac{\kappa^2\varepsilon^2}{4}\right) \left( 1 - \frac{\kappa}{3} \right)
    \geq 1 -  \frac{\kappa}{3} + \frac{\kappa^2\varepsilon^2}{4} - \frac{\kappa^3\varepsilon^2}{12} 
    \geq 1 -  \frac{\kappa}{2} \ ,
\end{align}
which is the Claim~\eqref{eq:mainclaimekappadeltackappa}.
\end{proof}

The GKP-state~$\gkp_{\kappa,\Delta}$ is close to its truncated version~$\gkp_{\kappa,\Delta}^\varepsilon$ for suitable choices of parameters, as expressed by the following lemma.
\begin{lemma}[Truncated approximate GKP-states, Lemma A.9 in Ref.~\cite{brenner2024complexity}]
\label{lem:truncatedapproximateGKPstates}
For~$\kappa\in(0,1/4), \Delta >0~$ and~$\varepsilon\in(0,1/2)$ we have 
\begin{align}
\left\langle\gkp_{\kappa,\Delta},\gkp^\varepsilon_{\kappa,\Delta}\right\rangle &\ge 1 - 7\Delta - 2(\Delta/\varepsilon)^4 \ . 
\end{align}
\end{lemma}

We note that Lemma~\ref{lem:truncatedapproximateGKPstates} is a consequence of the fact that the (centered) Gaussian~$\Psi_\Delta$ is close to its truncated version~$\Psi_\Delta^\varepsilon$ for the chosen parameters, see Lemma~\ref{lem:closenessgaussiantrunc}.

Let us show the converse, i.e., that the truncated and untruncated GKP-states are not close unless~$\Delta/\varepsilon$ is small.

\begin{lemma} \label{lem:converseoverlapgkp} Let~$\kappa,\Delta \in (0,1/8)$ and~$ \varepsilon \in (0,1/2)$. We have
\begin{align}
\left|\left\langle\gkp_{\kappa,\Delta},\gkp^\varepsilon_{\kappa,\Delta}\right\rangle\right|^2 
\le 15/16 + (\varepsilon/\Delta)^2/(4\pi) + 4 \Delta^4 \, .
\end{align}
\end{lemma}
\begin{proof}
Recall that the support of the truncated GKP-state~$\gkp^\varepsilon_{\kappa,\Delta}$ is contained in the set
\begin{align}
\mathbb{Z}(\varepsilon)&=\{x\in\mathbb{R}\ |\ \mathsf{dist}(x,\mathbb{Z})\leq \varepsilon\}\ .
\end{align}
Denoting by~$\Pi$ the orthogonal projection onto the subspace of~$L^2(\mathbb{R})$ of functions having support contained in this set, we thus have~$\gkp^\varepsilon_{\kappa,\Delta}=\Pi\gkp^\varepsilon_{\kappa,\Delta}$. It follows that
\begin{align}
\left|\left\langle\gkp_{\kappa,\Delta},\gkp^\varepsilon_{\kappa,\Delta}\right\rangle\right|^2 &=
\left|\left\langle\gkp_{\kappa,\Delta},\Pi\gkp^\varepsilon_{\kappa,\Delta}\right\rangle\right|^2\\
&=\left|\left\langle\Pi\gkp_{\kappa,\Delta},\gkp^\varepsilon_{\kappa,\Delta}\right\rangle\right|^2\\
&\leq \|\Pi\gkp_{\kappa,\Delta}\|^2\ , \label{eq:CSproj}
\end{align}
where we used the Cauchy-Schwarz inequality and the fact that~$\gkp^\varepsilon_{\kappa,\Delta}$ is normalized.

Let~$\Pi^\bot=I-\Pi$. Then 
\begin{align}
\|\Pi^\bot\gkp_{\kappa,\Delta}\|^2 &=C_{\kappa,\Delta}^2\sum_{z,z'\in \mathbb{Z}} \eta_\kappa(z)\eta_\kappa(z') \langle \chi_\Delta(z),\Pi^\bot \chi_\Delta(z')\rangle\\
&\geq C_{\kappa,\Delta}^2\sum_{z\in \mathbb{Z}} \eta_\kappa(z)^2 \langle \chi_\Delta(z),\Pi^\bot \chi_\Delta(z)\rangle\\
&\geq R_{\varepsilon,\Delta}\cdot C_{\kappa,\Delta}^2\sum_{z\in \mathbb{Z}} \eta_\kappa(z)^2\\
&=R_{\varepsilon,\Delta}\cdot \left(C_{\kappa,\Delta}/C_{\kappa}\right)^2 \label{eq:boundorthogonal}
\end{align}
where we introduced the quantity
\begin{align}
R_{\varepsilon,\Delta}&:=2\int_\varepsilon^{1-\varepsilon}|\Psi_\Delta(x)|^2 dx
\end{align}
and used that for any~$z\in\mathbb{Z}$, we have 
\begin{align}
\langle \chi_\Delta(z),\Pi^\bot \chi_\Delta(z)\rangle
 &=\int_{\mathbb{R}\backslash \mathbb{Z}(\varepsilon)}
 |\chi_\Delta(z)(x)|^2 dx \\
 &\geq \int_{[z-1+\varepsilon,z-\varepsilon]\cup [z+\varepsilon,z+1-\varepsilon]}
|\chi_\Delta(z)(x)|^2 dx\\
&=\int_{[-1+\varepsilon,-\varepsilon]\cup [\varepsilon,1-\varepsilon]}
|\Psi_\Delta(x)|^2 dx\\
 &=R_{\varepsilon,\Delta}\ .
\end{align}
We note that by definition of~$\Psi_\Delta$ the function~$x \mapsto |\Psi_\Delta(x)|^2$ is the probability density function of the random variable~$X$ 
distributed as~$X \sim \cN(0,\Delta^2/2)$, i.e., a Gaussian distribution with mean~$0$ and variance~$\Delta^2/2$.
Therefore 
\begin{align} 
    R_{\varepsilon,\Delta} &= \mathrm{Pr}[\varepsilon < |X| < 1- \varepsilon]\\
        &= 1 - \left(\mathrm{Pr}[|X| \ge 1 -\varepsilon] + \mathrm{Pr}[|X| \le \varepsilon]\right)\, . \label{eq:probbound1}
\end{align}
Using the Chernoff bound we find 
\begin{align}
    \mathrm{Pr}[|X| \ge 1 -\varepsilon] &\le 2 e^{-(1-\varepsilon)^2/\Delta^2} \\
    &\le 2 (\Delta/(1-\varepsilon))^4\\
    &\le 32 \Delta^4\, , \label{eq:chernoff}
\end{align} where we used the assumption~$\varepsilon <1/2$ and~$e^{-x} \le 1/x^2$ for~$x>0$.
Moreover, a straightforward computation shows that we have 
\begin{align}
    \mathrm{Pr}[|X|\le \varepsilon] &= 2/\sqrt{\pi}\int_0^{\varepsilon/\Delta} e^{-t^2} dt \\
    &\le \left(1 - e^{-4/\pi  (\varepsilon/\Delta)^2}\right)^{1/2}\\
    &\le 1 - e^{-4/\pi  (\varepsilon/\Delta)^2}/2 \label{eq:erfbound}
\end{align} 
where we used~$2/\sqrt{\pi}\int_0^x e^{-t^2} dt \leq \sqrt{1 - e^{4/\pi x^2}}$ (see e.g., \cite{abramowitz1972handbook}) to obtain the second inequality and we used the bound~$\sqrt{1-x} \le 1 - x/2$ for~$x \in [0,1]$ to obtain the third inequality.
Combining Eqs.~\eqref{eq:chernoff} and~\eqref{eq:erfbound} with Eq.~\eqref{eq:probbound1} we find
\begin{align}
    R_{\varepsilon,\Delta} \ge e^{-4/\pi  (\varepsilon/\Delta)^2}/2 - 32 \Delta^4 \ . \label{eq:Repsbound}
\end{align}
Moreover, due to Lemma~\ref{lem:cte_normaliz_bound} we have
\begin{align}
    \frac{C_{\kappa,\Delta}^2}{C_\kappa^2} \ge (1 - \kappa/3 - 3\Delta)^2 \ge 1 - 2\kappa/3 - 6 \Delta \ge 1 - \kappa - 6\Delta \ge 1/8\, , \label{eq:C-kappa-Deltabound}
\end{align} 
where we used~$C_\kappa \le 1$ and~$(1-x)^2\ge 1 - 2x~$ for~$x\ge 0$ and the assumption~$\kappa,\Delta \in (0,1/8)$.
Combining Eqs.~\eqref{eq:C-kappa-Deltabound}, \eqref{eq:Repsbound} and~\eqref{eq:boundorthogonal}
gives 
\begin{align}
    \|\Pi^\bot\gkp_{\kappa,\Delta}\|^2 &\ge  \left(e^{-4/\pi  (\varepsilon/\Delta)^2}- 64 \Delta^4 \right)/16 \\
    &\ge 1/16 - (\varepsilon/\Delta)^2/(4\pi) - 4 \Delta^4\, , \label{eq:orthogonalprojGKPlower}
\end{align}
where we used that~$e^{-x} \ge 1 -x~$ for~$x\ge 0$.
Due to the Pythagorean theorem we have 
\begin{align}
    \|\Pi\gkp_{\kappa,\Delta}\|^2 = 1 -  \|\Pi^\bot\gkp_{\kappa,\Delta}\|^2\, .\label{eq:pythagoras}
\end{align}
The claim follows by combining Eqs.~\eqref{eq:CSproj}, \eqref{eq:pythagoras} and~\eqref{eq:orthogonalprojGKPlower}.
\end{proof}

Let us show that the truncated approximate GKP-state~$\tGKP^\varepsilon_{\kappa,\Delta}$ is close to~$\tGKP_{\kappa,\Delta}$.

\begin{lemma}[Truncated approximate GKP-states: pointwise version]\label{lem:truncatedapproximategkpstatepointw}
Let~$\kappa\in(0,1/2)$, $\Delta\in(0,1/8)$ and~$\varepsilon\geq \Delta$. Then 
\begin{align}
    \langle \tGKP_{\kappa,\Delta},\tGKP^\varepsilon_{\kappa,\Delta}\rangle
    &\geq 1 - \kappa - 3\Delta - 2 (\Delta/\varepsilon)^4 \ .
\end{align}
\end{lemma}
\begin{proof}
By definition of the states~$ \tGKP_{\kappa,\Delta}$ and~$\tGKP^\varepsilon_{\kappa,\Delta}$ we have 
\begin{align}
    \langle \tGKP_{\kappa,\Delta},
    \tGKP^\varepsilon_{\kappa,\Delta}
    \rangle &=D_{\kappa,\Delta} E_{\kappa,\Delta}(\varepsilon)\sum_{z,z'\in\mathbb{Z}}
    \int \eta_\kappa(x)^2 \Psi_\Delta(x-z)\Psi^\varepsilon_\Delta(x-z')dx\\
    &\ge D_{\kappa,\Delta} E_{\kappa,\Delta}(\varepsilon)\sum_{z\in\mathbb{Z}}
    \int \eta_\kappa(x)^2 \Psi_\Delta(x-z)\Psi^\varepsilon_\Delta(x-z)dx\\
    &\ge D_{\kappa,\Delta} E_{\kappa,\Delta}(\varepsilon)\sum_{z\in\mathbb{Z}}
    \int \eta_\kappa(x)^2 \left(\Pi_{[-\varepsilon,\varepsilon]}\Psi_\Delta\right)(x-z)\Psi^\varepsilon_\Delta(x-z)dx\\
    &=\|\Pi_{[-\varepsilon,\varepsilon]}\Psi_\Delta\|\cdot D_{\kappa,\Delta} E_{\kappa,\Delta}(\varepsilon)\sum_{z\in\mathbb{Z}}
    \int \eta_\kappa(x)^2 \Psi^\varepsilon_\Delta(x-z)^2dx\\
    &=\|\Pi_{[-\varepsilon,\varepsilon]}\Psi_\Delta\|\cdot D_{\kappa,\Delta} 
    \cdot E_{\kappa,\Delta}(\varepsilon)^{-1}\ ,\label{eq:generlamvda}
\end{align}
where we used the non-negativity of the integrand to obtain the first inequality, the non-negativity of~$\Psi_\Delta$ for the second inequality, and Eq.~\eqref{eq:Einv2_aux} in the last step.
Thus we have
\begin{alignat}{2}
    &\|\Pi_{[-\varepsilon,\varepsilon]}\Psi_\Delta\|\cdot D_{\kappa,\Delta} 
    \cdot E_{\kappa,\Delta}(\varepsilon)^{-1} \\
    &\qquad\geq  (1 - 2 (\Delta/\varepsilon)^4) \cdot D_{\kappa,\Delta} 
    \cdot E_{\kappa,\Delta}(\varepsilon)^{-1}
    &&\quad \text{ by Lemma~\ref{lem:closenessgaussiantrunc}~\eqref{it:thirdclosenessgaussiantrunc}}  \\
    &\qquad\ge (1 - 2 (\Delta/\varepsilon)^4)  \cdot (1 - 3\Delta-\kappa/2)  \cdot (1 - \kappa/6)^{-1} 
    &&\quad \text{ by Claims~\eqref{eq:ckd_dkd_eq} and~\eqref{eq:mainclaimekappadeltackappa} in Lemma~\ref{lem:cte_normaliz_bound}} 
    \\
    \label{eq:gkpgkpeps_aux1} &\qquad\geq 1 - \kappa - 3\Delta - 2 (\Delta/\varepsilon)^4 
\end{alignat}
where we used that~$(1-x)^{-1} \geq 1-x$ for~$x\in(0,1)$ and~$(1-x)(1-y)\ge  1- x -y~$ for~$x, y \geq 0$. Eq.~\eqref{eq:gkpgkpeps_aux1} gives the claim when inserted into Eq.~\eqref{eq:generlamvda}.
\end{proof}

We show the following lemma which we use below in the proof of Lemma~\ref{lem:lem_squeezed_approx_GKP}.

\begin{lemma}
\label{eq:gkpedgkpinnerproduct}
Let~$d\in\bbN$. Let~$\kappa \in (0, 1/2)$, $\Delta \in (0,1/8)$ and~$\varepsilon \geq d \Delta$.
Then
\begin{align}
    \langle \tGKP_{\kappa,\Delta}^{\varepsilon/d}, \tGKP_{\kappa,\Delta}^{\varepsilon} \rangle \geq 1 - 3 \sqrt{\kappa} - 5\sqrt{\Delta} - 4 (\Delta d/\varepsilon)^2 \ .
\end{align}
\end{lemma}

\begin{proof}
By Lemma~\ref{lem:truncatedapproximategkpstatepointw} we have
\begin{align}
    \langle \tGKP_{\kappa,\Delta}, \tGKP_{\kappa,\Delta}^{\varepsilon} \rangle 
    &\geq 1 - \kappa - 3 \Delta - 2 (\Delta/\varepsilon)^4 \ , \\
    \langle \tGKP_{\kappa,\Delta}, \tGKP_{\kappa,\Delta}^{\varepsilon/d} \rangle 
    &\geq 1 - \kappa - 3 \Delta - 2 (\Delta d /\varepsilon)^4  \ ,
\end{align}
for~$\kappa \in (0, 1/2)$, $\Delta \in (0,1/8)$ and~$\varepsilon \geq d \Delta$. 
We use this in Lemma~\ref{lem:triangle_ineq} (with~$U=I$, $\varphi_1 = \tGKP_{\kappa,\Delta}^{\varepsilon/d}$, $\varphi_2 = \tGKP_{\kappa,\Delta}^{\varepsilon}$ and~$\psi_1=\psi_2=\tGKP_{\kappa,\Delta}$) to obtain
\begin{align}
    \langle \tGKP_{\kappa,\Delta}^{\varepsilon/d}, \tGKP_{\kappa,\Delta}^{\varepsilon} \rangle 
    &\geq \langle \tGKP_{\kappa,\Delta}, \tGKP_{\kappa,\Delta}\rangle 
    - \sqrt{2}\left(  \sqrt{1 - \langle \tGKP_{\kappa,\Delta}, \tGKP_{\kappa,\Delta}^{\varepsilon/d} \rangle} + \sqrt{1 - \langle \tGKP_{\kappa,\Delta}, \tGKP_{\kappa,\Delta}^{\varepsilon} \rangle} \right) \\
    &\geq 1 - \sqrt{2} \left( \sqrt{\kappa + 3 \Delta + 2 (\Delta d/\varepsilon)^4} + \sqrt{\kappa + 3 \Delta + 2 (\Delta/\varepsilon)^4} \right) \\
    &\geq 1 - 3 \sqrt{\kappa} - 5\sqrt{\Delta} - 4 (\Delta d/\varepsilon)^2 \ ,
\end{align}
where we use the fact that the inner products are positive and real because they are the inner products of positive real functions, and~$\sqrt{x+y} \leq \sqrt{x} + \sqrt{y}$ for~$x,y\geq 0$.
\end{proof}

The following shows that for appropriate parameters, the peak-wise  and the point-wise approximate GKP-states are close to each other.

\begin{lemma}\label{lem:pointwisepeakwiseGKP}
    Let~$\kappa\in(0,1/2)$ and~$\Delta\in(0,1/8)$. Then
    \begin{align}
    \langle \gkp_{\kappa,\Delta},\tGKP_{\kappa,\Delta}\rangle &\geq 1-2\kappa - 3\Delta \ .
    \end{align}
    \end{lemma}
    
    \begin{proof}
    Assume~$\kappa\in(0,1/2)$ and~$\Delta\in(0,1/8)$. We have 
    \begin{align}
    \langle \tGKP_{\kappa,\Delta},\gkp_{\kappa,\Delta}\rangle &=
    C_{\kappa,\Delta} D_{\kappa,\Delta} \sum_{z,z'\in\mathbb{Z}}\eta_\kappa(z) \int\Psi_\Delta(x-z) \eta_{\kappa}(x)\Psi_{\Delta}(x-z') dx\\
    &\geq C_{\kappa,\Delta} D_{\kappa,\Delta} \sum_{z\in\mathbb{Z}} \eta_\kappa(z) \int \eta_{\kappa}(x) \Psi_\Delta(x-z)^2 dx \\
    &= C_{\kappa,\Delta} D_{\kappa,\Delta} \cdot 
    \left(1 + \frac{\kappa^2 \Delta^2}{4}\right)^{-\frac{1}{2}} \cdot \frac{\tilde{\kappa}}{\sqrt{\pi}} \sum_{z \in \mathbb{Z}} e^{-\tilde{\kappa}^2 z^2} \\
    &= C_{\kappa,\Delta} D_{\kappa,\Delta} \cdot 
    \left(1 + \frac{\kappa^2 \Delta^2}{4}\right)^{-\frac{1}{2}} \sum_{z \in \mathbb{Z}} \eta_{\tilde{\kappa}}(z)^2 \ , \label{eq: overlap gkp GKP}
    \end{align}
    where
    \begin{align}
    \tilde{\kappa} &= \kappa  \left(1+\frac{(\kappa\Delta)^2}{4}\right)^{\frac{1}{2}} \left(1+\frac{(\kappa\Delta)^2}{2}\right)^{-\frac{1}{2}} \\
    &\leq \kappa  \left(1+\frac{(\kappa\Delta)^2}{8}\right) \left(1-\frac{(\kappa\Delta)^2}{6}\right)\\
    &= \kappa \left( 1 - \frac{(\kappa\Delta)^2}{24} - \frac{(\kappa\Delta)^4}{48}  \right) \\
    &\leq  \kappa \left( 1 - \frac{(\kappa\Delta)^2}{24} \right) 
    \ , \label{eq:asymptotictildekappa}
    \end{align}
    where we used that~$(1+x)^{-1/2} \leq 1-x/3$ for~$x\in[0,1/2)$ and~$(1+x)^{1/2} \leq 1+x/2$ for~$x\geq 0$.
    Furthermore, by Items~\eqref{it:normalizationvsfi} and~\eqref{eq:ckappaupperlowerbound} in Lemma~\eqref{lem:cte_normaliz_bound} we have
      \begin{align}
        \sum_{z\in\mathbb{Z}}\eta_{\tilde {\kappa}}(z)^2
        &=C_{\tilde{\kappa}}^{-2} \\
        &\geq \left(1- \frac{\tilde{\kappa}}{6} \right)^{-2} \\
        &\geq \left(1 + \frac{\tilde{\kappa}}{3} \right) \\
        &\geq 1 - \frac{\kappa}{3} \left( 1 - \frac{(\kappa\Delta)^2}{24} \right) \\
        \label {eq: sum tilde kappa} &\geq 1 - \kappa
    \end{align}
    where we used Eq.~\eqref{eq:asymptotictildekappa} and~$(1-x)^{-2} \geq 1 + 2x$ for~$x \geq 0$. 
    Notice that~$(1+x)^{-1/2} \geq  1 - x/2$ for~$x\ge 0$ gives
    \begin{align}
        \label{eq:GKPgkp_prox_aux1}
        (1+(\kappa\Delta)^2/4)^{-1/2} \geq 1-(\kappa\Delta)^2/8 \ .
    \end{align}
    Inserting Claims~\eqref{eq:ckappadeltalowerbound} and~\eqref{eq:ckd_dkd_eq} from Lemma~\ref{lem:cte_normaliz_bound} and the inequalities~\eqref{eq: sum tilde kappa} and~\eqref{eq:GKPgkp_prox_aux1} into the inequality~\eqref{eq: overlap gkp GKP} gives 
    \begin{align}
        \langle \tGKP_{\kappa,\Delta},\gkp_{\kappa,\Delta}\rangle 
        &\ge \left( 1 - \frac{\kappa}{3} \right) \left( 1 - 3 \Delta - \frac{\kappa}{2}\right) \left( 1 - \frac{(\kappa\Delta)^2}{8} \right) \left( 1 - \kappa \right) \\
        &\ge 1 - \frac{\kappa}{3}- 3 \Delta - \frac{\kappa}{2} - \frac{(\kappa\Delta)^2}{8} - \kappa  \\
        &\ge 1 - 2\kappa - 3\Delta \ ,
    \end{align}
    where we used that~$(1-x)(1-y)\ge  1- x -y~$ for~$x, y \geq 0$. 
    This is the claim.
\end{proof}

Next, we show that the approximate GKP-state~$\tGKP_{\kappa,\Delta}$ is close to the truncated GKP-state~$\gkp_{\kappa,\Delta}^\varepsilon$ for an appropriate choice of parameters.

\begin{lemma}
\label{eq:tgkpgkpeps}
Let~$\kappa\in(0,1/4)$, $\Delta\in(0,1/8)$ and~$\varepsilon\in(0,1/2)$. We have
\begin{align}
    \langle \tGKP_{\kappa,\Delta}, \gkp_{\kappa,\Delta}^\varepsilon \rangle 
    \geq 1 - 2\kappa - 7 \sqrt{\Delta} - 2 (\Delta/\varepsilon)^2 \ .
\end{align}
\end{lemma}
  
\begin{proof}
By Lemma~\ref{lem:triangle_ineq} we have 
\begin{align}
    \label{eq:tgkpgkpepsaux0} |\langle \tGKP_{\kappa,\Delta}, \gkp_{\kappa,\Delta}^\varepsilon \rangle - 1|
    &\leq | \langle \tGKP_{\kappa,\Delta}, \gkp_{\kappa,\Delta} \rangle - 1 |
    + \sqrt{2} \sqrt{ | \langle \gkp_{\kappa,\Delta}^\varepsilon , \gkp_{\kappa,\Delta} \rangle - 1 | }
\end{align}
By Lemmas ~\ref{lem:pointwisepeakwiseGKP} and~\ref{lem:truncatedapproximateGKPstates} we have 
\begin{align}
    \langle \tGKP_{\kappa,\Delta}, \gkp_{\kappa,\Delta}\rangle &\geq 1-2\kappa - 3\Delta \ ,\\
    \langle \gkp^\varepsilon_{\kappa,\Delta}, \gkp_{\kappa,\Delta}\rangle &\ge 1 - 7\Delta - 2(\Delta/\varepsilon)^4 \ ,
\end{align}
which together with Eq.~\eqref{eq:tgkpgkpepsaux0} gives the claim since
\begin{align}
    |\langle \tGKP_{\kappa,\Delta}, \gkp_{\kappa,\Delta}^\varepsilon \rangle - 1|
    &\leq 2\kappa+3\Delta
    + \sqrt{2} \sqrt{ 7 \Delta + 2 (\Delta/\varepsilon)^4 } \\
    &\leq 2\kappa + 7 \sqrt{\Delta} + 2 (\Delta/\varepsilon)^2 \ .
\end{align}
\end{proof}

We say that two (one-parameter families of) states~$\{\Psi_\kappa\}_\kappa,\{\Phi_\kappa\}_\kappa \subset L^2(\mathbb{R})$
are polynomially close to each other  for~$\kappa\rightarrow 0$
if there are constants~$C>0, \alpha>0$ and~$\kappa_0>0$ such that  
\begin{align}\left\|\Psi_\kappa-\Phi_\kappa\right\|\leq C\kappa^\alpha \qquad \textrm{ for all }\qquad \kappa\leq \kappa_0\ .
\end{align}
This definition can be generalized to two-parameter families~$\{\ket{\Psi_{\kappa,\Delta}}\}_{(\kappa,\Delta)}$, 
replacing~$C\kappa^\alpha$ by~$C\kappa^\alpha \Delta^\beta$ for some constants~$C,\alpha,\beta$. 
Combining Lemmas~\ref{lem:truncatedapproximateGKPstates}, \ref{lem:truncatedapproximategkpstatepointw}, \ref{lem:pointwisepeakwiseGKP} and~\ref{eq:tgkpgkpeps} with the inequality~\eqref{eq:upperboundpsiphiinner} relating the distance~$\|\Psi-\Phi\|$ of two states~$\Psi,\Phi$ to their inner product~$\langle \Psi, \Phi \rangle$, we have shown the following, cf. Fig.~\ref{fig:summarysectionapproximate}. 

\begin{lemma}[Equivalence of different notions of approximate GKP-states]
    \label{lem:equivalencefamiliesstates}
    Let~$\varepsilon\in (0,1/2)$.  The three families of states~$\{\ket{\gkp_{\kappa,\Delta}}\}_{\kappa,\Delta}, \{\ket{\gkp^\varepsilon_{\kappa,\Delta}}\}_{\kappa,\Delta}$, $\{\ket{\tGKP_{\kappa,\Delta}}\}_{\kappa,\Delta}$ and $\{\ket{\tGKP^\varepsilon_{\kappa,\Delta}}\}_{\kappa,\Delta}$ are pairwise polynomially close to each other in the limit~$(\kappa,\Delta)\rightarrow (0,0)$. 
    \end{lemma}

\subsection{Properties of approximate GKP-states with integer spacing\label{sec:propertiesapproximateGKP}}

In this section, we establish various statements about the truncated approximate GKP-state~$\gkp^\varepsilon_{\kappa,\Delta}$. In Section~\ref{sec:approximategkpdisplacement}, we  show that this state is approximately stabilized by integer translations~$e^{-izP}$, $z\in\mathbb{Z}$ with small modulus~$|z|$ in position-space, 
and by translations~$e^{2\pi i z Q}$ which are integer multiples of~$2\pi$ in momentum space. In Section~\ref{sec:approximatephaseoperator} we show that the image of~$\gkp^\varepsilon_{\kappa,\Delta}$ under the operator~$e^{\pi i (dQ^2+c_dQ)}$ (where~$c_d=d\pmod 2$) is far from~$\gkp^\varepsilon_{\kappa,\Delta}$.
Finally, in Section~\ref{sec:multiplicationoperationsectionapprox} we consider the action of the squeezing operator~$M_d$ on the state~$\tGKP_{\kappa,\Delta}$.

\subsubsection{Approximate GKP-states and displacements\label{sec:approximategkpdisplacement}}
Here we examine the effect of phase space displacements on approximate GKP-states. We consider displacements in the~$Q$ (position-) or~$P$ (momentum-)directions only.
The following shows that translating the state~$\gkp_{\kappa,\Delta}^\varepsilon$ by a non-integer amount in position-space results in an approximately orthogonal state. 

\begin{lemma}[Orthogonality of displaced GKP-states, Lemma 7.6 in Ref.~\cite{brenner2024factoring}]\label{lem:orthogonalitytruncated gkpstates}
    Let~$\kappa,\Delta>0$, $\varepsilon\in (0,1/4]$ 
    and~$y\in\mathbb{R}$ be such that
    \begin{align}
        \min_{z \in\mathbb{Z}}|y-z| \ge 2\varepsilon\ ,
    \end{align}
    i.e., the distance of~$y$ to~$\mathbb{Z}$ is at least~$2\varepsilon$. Then
    \begin{align}
        \langle \gkp^\varepsilon_{\kappa,\Delta},e^{-iyP}\gkp^\varepsilon_{\kappa,\Delta}\rangle &=0\ .
    \end{align}
\end{lemma}

On the other hand, integer-translations approximately stabilize the state as expressed by the following lemma.

\begin{lemma}[Position-translated  GKP-states, Lemma 7.5 in Ref.~\cite{brenner2024factoring}]\label{lem:translationinvariancegkp}
Let~$\kappa, \Delta \in (0,1/4)$. For any integer~$z\in\mathbb{Z}$ and any~$\varepsilon \in (0,1/2]$, we have
\begin{align}
\langle \gkp_{\kappa,\Delta}^\varepsilon, e^{-izP}\gkp_{\kappa,\Delta}^\varepsilon\rangle &\geq 1 - z^2 \kappa^2 \ .
\end{align}
\end{lemma}

Similarly, integer translations in momentum space also approximately leave the state invariant.

\begin{lemma}[Momentum-translated GKP-states]\label{lem:translationinvariancephasegkpalternative}
Let~$z\in \mathbb{Z}$, $\kappa, \Delta>0$ and~$\varepsilon \le 1/(4|z|)$. Then we have
\begin{align}
    \langle \gkp_{\kappa,\Delta}^\varepsilon, e^{2\pi i z Q}\gkp_{\kappa,\Delta}^\varepsilon\rangle &\geq 1 - 10 \Delta^2 z^2 - 16(\Delta/\varepsilon)^4\ .
\end{align}    
\end{lemma}

\begin{proof}
    We mostly follow the proof of~\cite[Lemma A.14]{brenner2024complexity}. 
    We have
    \begin{align}
        \label{eq:psimpsieps_aux1} \langle \gkp_{\kappa, \Delta}^\varepsilon,e^{2\pi i z Q}\gkp_{\kappa, \Delta}^\varepsilon\rangle &= C_\kappa^2 \sum_{y \in \mathbb{Z}} \eta_\kappa(y)^2  \left\langle \chi_\Delta^\varepsilon(y),  e^{2\pi i z Q} \chi_\Delta^\varepsilon(y)\right\rangle
    \end{align}
    where we used that the unitary~$e^{2\pi i z Q}$ does not change the support of the mutually orthogonal functions~$\{\chi_\Delta^\varepsilon(z)\}_{z\in\bb{Z}}$.
    Moreover, using that~$\ket{\chi_\Delta^\varepsilon(y)} = e^{-i y P} \ket{\Psi_\Delta^\varepsilon}$, we obtain 
    \begin{align}
        \left\langle \chi_\Delta^\varepsilon(y),  e^{2\pi i z Q} \chi_\Delta^\varepsilon(y)\right\rangle &= \left\langle \Psi_\Delta^\varepsilon,  e^{iy P} e^{2\pi i z Q} e^{-i yP}\Psi_\Delta^\varepsilon\right\rangle\\
        &= \left\langle \Psi_\Delta^\varepsilon,   e^{2\pi i z (Q+y)} \Psi_\Delta^\varepsilon\right\rangle\\
        &=  \left\langle \Psi_\Delta^\varepsilon,   e^{2\pi i z Q} \Psi_\Delta^\varepsilon\right\rangle\, ,
    \end{align}
    where we used that~$ e^{iy P} Q e^{-i yP} = Q + y I_{L^2(\bbR)}$ for~$y \in \mathbb{Z}$.
    It thus suffices to bound $\langle\Psi_\Delta^\varepsilon, e^{2\pi i z  Q} \Psi_\Delta^\varepsilon\rangle$.
   Recall that (see Lemma~\ref{lem:momentumtranslatedgaussian}~\eqref{it:thirdpropertymomentumtranslated})
    \begin{align}
        \langle\Psi_\Delta^\varepsilon, e^{2\pi i z  Q} \Psi_\Delta^\varepsilon\rangle \ge 1 - 10z^2 \Delta^2 - 16 (\Delta/\varepsilon)^4\, . \label{eq: bound phase shift trunc psi aux}
    \end{align}
    The claim follows by combining Eq.~\eqref{eq: bound phase shift trunc psi aux} with Eq.~\eqref{eq:psimpsieps_aux1} and using that~$\sum_{y\in\bb{Z}} \eta_\kappa(y)^2 = C_\kappa^{-2}$.
 \end{proof}

 Lemmas~\ref{lem:translationinvariancegkp} and~\ref{lem:translationinvariancephasegkpalternative}
 show that  in the limit $(\kappa,\Delta)\rightarrow (0,0)$, the state~$|\gkp_{\kappa,\Delta}^\varepsilon\rangle$
 is stabilized by the operators~$e^{-iP}$ and $e^{2\pi i Q}$. Thus the state~$|\gkp_{\kappa,\Delta}^\varepsilon\rangle$
 approximately reproduces the stabilization property underlying the formal definition~\cite{gkp} of  the ideal GKP state. 
We emphasize, however, that the latter is not normalizable (i.e., does not belong to the Hilbert space). Furthermore, in the limit $(\kappa,\Delta)\rightarrow (0,0)$, the approximate GKP state~$|\gkp_{\kappa,\Delta}^\varepsilon\rangle$ does not converge to the ideal GKP state in the sense of distributions (i.e., in the weak sense).

\subsubsection{Approximate GKP-states and the phase operator\label{sec:approximatephaseoperator}}
In this section, we consider the action of the phase operator~
\begin{align}
W_{\Pgate}&=e^{\pi i (d Q^2 + c_d Q )}
\end{align}
(where~$c_d=0$ if~$d$ is even and~$c_d=1$ otherwise)  on the approximate GKP-state~$\gkp_{\kappa,\Delta}^\varepsilon$.  We find  that 
(for appropriate choices of parameters), the resulting state
$W_{\Pgate}\gkp_{\kappa,\Delta}^\varepsilon$ is far from the state~$\gkp_{\kappa,\Delta}^\varepsilon$. This result is a key tool to establish our no-go result (Result~\ref{thm:result1} in the introduction). First, we show a lemma used in the proof of Lemma~\ref{lem: bound B_00 P}.

\begin{lemma}\label{lem:gaussiansumsexprva}
Suppose~$\lambda,\kappa \in (0,1/4)$. 
Then
\begin{align}
    C_\kappa^{2}\sum_{y \in \mathbb{Z}} \eta_{\kappa}^2(y)  e^{- \lambda^2 (y + s)^2}&\leq 
    \frac{1}{\sqrt{1+(\lambda/\kappa)^2}}+\kappa \qquad\textrm{ for all }\qquad s\in\mathbb{R}\  .
    \end{align}
        \end{lemma}
        
        \begin{proof}
Using the definition of~$\eta_\kappa$ (see Eq.~\eqref{eq:envelopepeakfunctions}) and by completing the square we have
\begin{align}
\sum_{y \in \mathbb{Z}} \eta_{\kappa}^2(y)  e^{- \lambda^2 (y + s)^2} & = \frac{\kappa}{\widetilde{\kappa}} e^{-\xi^2}\sum_{y \in \mathbb{Z}} \eta_{\widetilde{\kappa}}(y + \widetilde{s})^2\\
&\leq  \frac{\kappa}{\widetilde{\kappa}} e^{-\xi^2}\sum_{y \in \mathbb{Z}} \eta_{\widetilde{\kappa}}(y)^2
 \label{eq: Gaussian sum bound first}
\end{align}
where we defined 
\begin{align}
    \widetilde{\kappa}^2 = \kappa^2 + \lambda^2\, , \qquad \widetilde{s} = \frac{\lambda^2}{\widetilde{\kappa}^2} s \, ,\qquad \xi^2 = \frac{\kappa^2 \lambda^2 s^2}{\widetilde{\kappa}^2}\, 
\end{align}
and used  Lemma~\ref{lem: periodic gaussians}, which states that a periodic Gaussian achieves its maximum at  an integer.
It follows that 
\begin{align}
    C_\kappa^{2}\sum_{y \in \mathbb{Z}} \eta_{\kappa}^2(y)  e^{- \lambda^2 (y + s)^2}
    &\le  C_\kappa^2 \frac{\kappa}{\widetilde{\kappa}}e^{-\xi^2} \sum_{y \in \mathbb{Z}} \eta_{\widetilde{\kappa}}(y)^2\\
    &= \frac{C_\kappa^2}{C_{\widetilde{\kappa}}^2} \cdot \frac{\kappa}{\widetilde{\kappa}}e^{-\xi^2}\, , \label{eq: frac kappa mu P gate}
\end{align}
where we used  Lemma~\ref{lem:cte_normaliz_bound}~\eqref{eq:ckappaupperlowerbound}. 
Since 
\begin{align}
\tilde{\kappa}=\sqrt{\kappa^2+\lambda^2}\leq \kappa+\lambda\ ,\label{eq:upperboundtildekappam}
\end{align}
the assumptions~$\lambda,\kappa \in (0,1/4)$ imply that~$\tilde{\kappa}\in (0,1/2)$. With these assumptions, Lemma~\ref{lem:cte_normaliz_bound} implies that 
\begin{align}
    C_{\kappa} \le 1 - \frac{\kappa}{6}\leq 1  \qquad \textrm{and} \qquad C_{\widetilde{\kappa}} \ge 1 - \frac{\widetilde{\kappa}}{3}\, ,
\end{align}
and it follows that 
\begin{align}
    \frac{C_\kappa^2}{C_{\widetilde{\kappa}}^2} &\le  \left(\frac{1}{1 - \frac{\widetilde{\kappa}}{3}} \right)^2\\
    &\le 1 + \widetilde{\kappa}\\
    &=1+\sqrt{\kappa^2+\lambda^2}\label{eq:uppbthree}
\end{align}
by Eq.~\eqref{eq:upperboundtildekappam}.  Here we used that~$(1-x/3)^{-1} \le 1 + x$ for~$x \in (0,1/2)$.
We also have 
\begin{align}
\frac{\kappa}{\tilde{\kappa}}&=\frac{1}{\sqrt{1+(\lambda/\kappa)^2}} \label{eq:uppbfour}\ .
\end{align}
Combining Eqs.~\eqref{eq:uppbthree} and~\eqref{eq:uppbfour} with Eq.~\eqref{eq: frac kappa mu P gate}, we conclude (with~$\xi^2\geq 0$) that
\begin{align}
    C_\kappa^{2}\sum_{y \in \mathbb{Z}} \eta_{\kappa}^2(y)  e^{- \lambda^2 (y + s)^2}&\leq 
    \frac{1+\sqrt{\kappa^2+\lambda^2}}{\sqrt{1+(\lambda/\kappa)^2}}\\
    &=    \frac{1+\kappa \sqrt{1+(\lambda/\kappa)^2}}{\sqrt{1+(\lambda/\kappa)^2}}\\
    &=\frac{1}{\sqrt{1+(\lambda/\kappa)^2}}+\kappa \qquad\textrm{ for all }\qquad s\in\mathbb{R} 
\end{align}
as claimed.

\end{proof}

\begin{lemma}[GKP-states and the phase operator] \label{lem: bound B_00 P}
    Let~$d\ge 2$ be an integer. Let~$\varepsilon<1/2$ be arbitrary. 
    Let~$\Delta,\kappa>0$ be such that
    \begin{align}
\kappa &<1/4\label{eq:kappabassumptionv} \ , \\
d\Delta&<1/(4\pi)\ .\label{eq:deltadassump}
\end{align}    
    Then
    \begin{align}
    \left|\langle \gkp_{\kappa,\Delta}^\varepsilon, e^{\pi i (d Q^2 + c_d Q )} \gkp_{\kappa,\Delta}^\varepsilon\rangle \right| &\leq       \frac{1}{\sqrt{1+(\pi d \Delta/\kappa)^2/4}}+\kappa    +16(\Delta/\varepsilon)^4\ .
    \label{eq:mainclaimboundb00P}
    \end{align}
\end{lemma}
\begin{proof}
Using the definition of the state~$\gkp_{\kappa,\Delta}^\varepsilon$, we have 
\begin{align}
    \langle \gkp_{\kappa,\Delta}^\varepsilon, e^{\pi i (d Q^2 + c_d Q )} \gkp_{\kappa,\Delta}^\varepsilon\rangle &=  C_\kappa^2 \sum_{y \in \mathbb{Z}} \eta_{\kappa}(y)^2 \langle \chi_\Delta^\varepsilon(y),  e^{\pi i (d Q^2 + c_d Q )} \chi_\Delta^\varepsilon(y)\rangle \\
    &= C_\kappa^2 \sum_{y \in \mathbb{Z}} \eta_{\kappa}(y)^2  \langle \Psi_\Delta^\varepsilon,  e^{\pi i (d (Q+y)^2 + c_d (Q+y) )} \Psi_\Delta^\varepsilon\rangle\\
    &=   C_\kappa^2 \sum_{y \in \mathbb{Z}} \eta_{\kappa}(y)^2 \langle \Psi_\Delta^\varepsilon, e^{i \pi \phi_d(y)} e^{\pi i (d Q^2  + (2dy + c_d)Q )} \Psi_\Delta^\varepsilon\rangle\, .
\end{align}
Here we used that the functions~$\{\chi^\varepsilon_\Delta(y)\}_{y\in\mathbb{Z}}$ have pairwise disjoint support, and the operator~$e^{\pi i (d Q^2 + c_d Q )}$ does not change the support of a function. We also introduced the function
\begin{align}\phi_d(y) = dy^2 +  c_d y\ .
\end{align}
 If~$d$ is even, $c_d=0$, thus~$\phi_d(y)= d y^2$ is even and~$e^{i\pi \phi_d(y)} =1$ for all~$y \in \mathbb{Z}$.
If~$d$ is odd~$c_d =1$. Therefore, $\phi_d(y) = y(dy+1)$. We observe that this expression is even for all~$y\in \mathbb{Z}$. Hence we have~$e^{i\pi \phi_d(y)} =1$ for all~$y \in \mathbb{Z}$ in this case also. We conclude that
\begin{align}
    \langle \gkp_{\kappa,\Delta}^\varepsilon, e^{\pi i (d Q^2 + c_d Q )} \gkp_{\kappa,\Delta}^\varepsilon\rangle = C_\kappa^2 \sum_{y \in \mathbb{Z}} \eta_{\kappa}(y)^2 \langle \Psi_\Delta^\varepsilon,  e^{\pi i (d Q^2  + (2dy + c_d)Q )} \Psi_\Delta^\varepsilon\rangle\, .\label{eq:vdax}
\end{align}
According to Lemma~\ref{lem:gaussianphase}~\eqref{it:secondgaussianphase}, we have the upper bound
    \begin{align}
        |\langle \Psi^\varepsilon_\Delta,  e^{\pi i (d Q^2  + (2dy + c_d)Q )} \Psi^\varepsilon_\Delta \rangle| \le e^{- \lambda^2 (y + c_d/(2d))^2}+16(\Delta/\varepsilon)^4\qquad\textrm{ for every }y\in\mathbb{Z}\ ,\label{eq:yzalldef}
        \end{align}where
                \begin{align}
        \lambda^2 =\pi^2 d^2 \Delta^2 /(1+ \pi^2 d^2 \Delta^4)\ .
        \end{align}
Therefore  Eq.~\eqref{eq:vdax} implies that 
\begin{align}
    \left|\langle \gkp_{\kappa,\Delta}^\varepsilon, e^{\pi i (d Q^2 + c_d Q )} \gkp_{\kappa,\Delta}^\varepsilon\rangle\right|& \le C_\kappa^2 \sum_{y \in \mathbb{Z}} \eta_{\kappa}(y)^2 |\langle \Psi_\Delta^\varepsilon,  e^{\pi i (d Q^2  + (2dy + c_d)Q )} \Psi_\Delta^\varepsilon\rangle |\\
    &\leq \left(C_\kappa^2\sum_{y\in\mathbb{Z}}\eta_\kappa(y)^2 e^{- \lambda^2 (y + c_d/(2d))^2}\right)+16(\Delta/\varepsilon)^4 \ .
    \label{eq: 1- overlap first}
\end{align}
Here we used that~$C_\kappa^2\sum_{z\in\mathbb{Z}}\eta_\kappa(z)^2=1$ (see Lemma~\ref{lem:cte_normaliz_bound}~\eqref{it:normalizationvsfi}).

By definition and the assumption~\eqref{eq:deltadassump}, we have 
\begin{align}
    \lambda = \pi \frac{d\Delta}{\sqrt{1+\pi^2 d^2\Delta^4}}\leq \pi d\Delta \leq 1 / 4\ ,
\end{align}
by the assumption~\eqref{eq:deltadassump}. 
With the assumption~\eqref{eq:kappabassumptionv},  Lemma~\ref{lem:gaussiansumsexprva} implies that
\begin{align}
    C_\kappa^{2}\sum_{y \in \mathbb{Z}} \eta_{\kappa}^2(y)  e^{- \lambda^2 (y + s)^2}&\leq 
    \frac{1}{\sqrt{1+(\lambda/\kappa)^2}}+\kappa\\
   &\leq \frac{1}{\sqrt{1+\frac{\pi^2}{4} \frac{(d\Delta)^2}{\kappa^2}}}+\kappa   \qquad\textrm{ for all }\qquad s\in\mathbb{R}\ , \label{eq:lstidentxv}
    \end{align}
    where we used that~$\pi d \Delta / 2 \leq \lambda$.
    Combining Eqs.~\eqref{eq:lstidentxv} and~\eqref{eq: 1- overlap first} gives the claim.
    \end{proof}

\subsubsection{Approximate GKP-states and squeezing \label{sec:multiplicationoperationsectionapprox}}
Here consider the squeezing operator~$M_{1/d}$ for~$d > 0$, i.e., the Gaussian unitary 
which acts as
\begin{align}
(M_{1/d}\Psi)(x)&=\sqrt{d}\Psi(d x)\qquad x\in\mathbb{R}
\end{align}
on elements~$\Psi\in L^2(\mathbb{R})$. We show the following result, which we will use in the analysis of the Fourier gate for approximate GKP codes (see Section~\ref{sec:matrixlementsfouriertransformsec}).

\begin{lemma}[Truncated GKP-states and squeezing]
    \label{lem:lem_squeezed_approx_GKP_v0}
    Let~$d\geq 2$ be an integer. Let~$\Delta \in (0, 1/(2d^2))$ and~$\kappa, \varepsilon \in(d\Delta,1/(2d))$. 
    Then
    \begin{align}
    \left|\langle\tGKP^\varepsilon_{\kappa,\Delta}, M_{1/d}\tGKP^\varepsilon_{\kappa/d,\Delta d}\rangle-\frac{1}{\sqrt{d}}\right|
    &\leq 4\sqrt{\kappa/d} + 9 \sqrt{\Delta/\varepsilon} \ .
\end{align}
\end{lemma}

\begin{proof}
By definition, we have 
\begin{align}
\langle \tGKP^\varepsilon_{\kappa',\Delta'},M_d\tGKP^\varepsilon_{\kappa,\Delta}\rangle&=\frac{1}{\sqrt{d}}\cdot E_{\kappa',\Delta'}(\varepsilon)E_{\kappa,\Delta}(\varepsilon)\sum_{z,z'\in\mathbb{Z}}
\int \eta_{\kappa'}(x)\eta_{\kappa}(x/d)\Psi^\varepsilon_{\Delta'}(x-z)\Psi^\varepsilon_\Delta(x/d-z')dx\\
&=\frac{1}{\sqrt{d}}\cdot E_{\kappa',\Delta'}(\varepsilon)E_{\kappa,\Delta}(\varepsilon)\sum_{z\in\mathbb{Z}}
\int \eta_{\kappa'}(x)\eta_{\kappa}(x/d)\Psi^\varepsilon_{\Delta'}(x-dz)\Psi^\varepsilon_\Delta(x/d-z)dx\\
&=\frac{1}{\sqrt{d}}\cdot E_{\kappa',\Delta'}(\varepsilon)E_{\kappa,\Delta}(\varepsilon)\sum_{z\in\mathbb{Z}}
\int \eta_{\kappa'}(x)\eta_{\kappa}(x/d)\Psi^\varepsilon_{\Delta'}(d(x/d-z))\Psi^\varepsilon_\Delta(x/d-z)dx \, ,
\end{align}
because~$\Psi^\varepsilon_{\Delta'}(x-z)\Psi^\varepsilon_\Delta(x/d-z')$ is proportional to~$\delta_{z,dz'}$ for~$\varepsilon\le 1/(2d)$. 
Therefore
\begin{align}
    \langle \tGKP^\varepsilon_{\kappa',\Delta'},M_d\tGKP^\varepsilon_{\kappa,\Delta}\rangle& =\sqrt{d}\cdot E_{\kappa',\Delta'}(\varepsilon)E_{\kappa,\Delta}(\varepsilon)\sum_{z\in\mathbb{Z}}
    \int \eta_{\kappa'}(dy)\eta_{\kappa}(y)\Psi^\varepsilon_{\Delta'}(d(y-z))\Psi^\varepsilon_\Delta(y-z)dy\\
    &=\frac{1}{\sqrt{d}}\cdot E_{\kappa',\Delta'}(\varepsilon)E_{\kappa,\Delta}(\varepsilon)\sum_{z\in\mathbb{Z}}
    \int \eta_{d\kappa'}(y)\eta_{\kappa}(y)\Psi^{\varepsilon/d}_{\Delta'/d}(y-z)\Psi^\varepsilon_\Delta(y-z)dy\\
    &=\frac{1}{\sqrt{d}} \cdot \langle \tGKP^{\varepsilon/d}_{d\kappa',\Delta'/d},\tGKP^\varepsilon_{\kappa,\Delta}\rangle \cdot \frac{E_{\kappa',\Delta'}(\varepsilon)}{E_{d\kappa',\Delta'/d}(\varepsilon/d)}\, ,
\end{align}
 where we substituted~$y=x/d$ to obtain the first inequality. The second equality follows from~$\eta_{\kappa'}(dy) = \eta_{d \kappa'}(y)/\sqrt{d}$ (see Eq.~\eqref{eq:envelopepeakfunctions}) and~$\Psi_{\Delta'}^\varepsilon(dy) =  \Psi_{\Delta'/d}^{\varepsilon/d}(y)/\sqrt{d}$ as we have (see Eq.~\eqref{eq:deftruncatedGassian})
 \begin{align}
    \Psi_{\Delta'}^\varepsilon(dy) &= \|\Pi_{[-\varepsilon,\varepsilon]}\Psi_{\Delta'}\|^{-1} \cdot \mathbf{1}_{[-\varepsilon,\varepsilon]}(dy)\cdot \Psi_{\Delta'}(dy)\\
    &= \|\Pi_{[-\varepsilon,\varepsilon]}\Psi_{\Delta'}\|^{-1} \cdot \mathbf{1}_{[-\varepsilon/d,\varepsilon/d]}(y)\cdot \Psi_{\Delta'/d}(y)/\sqrt{d}\\
    &= \|\Pi_{[-\varepsilon/d,\varepsilon/d]}\Psi_{\Delta'/d}\|^{-1} \cdot \mathbf{1}_{[-\varepsilon/d,\varepsilon/d]}(y)\cdot \Psi_{\Delta'/d}(y)/\sqrt{d}\\
    &= \Psi_{\Delta'/d}^{\varepsilon/d}(y)/\sqrt{d}\, ,
 \end{align}
 where we denote the indicator function of the interval~$[-\varepsilon,\varepsilon]$ by~$\mathbf{1}_{[-\varepsilon,\varepsilon]}$. 
 The third identity follows from~$\|\Pi_{[-\varepsilon,\varepsilon]}\Psi_{\Delta'}\|^2 = \int_{-\varepsilon}^{\varepsilon} \pi^{-1/2} (\Delta')^{-1} e^{-x^2/(\Delta')^2} dx = \int_{-\varepsilon/\Delta'}^{\varepsilon/\Delta'} \pi^{-1/2} e^{-x^2}dx$, i.e., the normalization factor only depends on the ratio~$\varepsilon/\Delta'$.
 
 In particular, for the choice~$(\kappa',\Delta')=(\kappa/d,\Delta d)$ we obtain
\begin{align}
\langle \tGKP^\varepsilon_{\kappa/d,\Delta d},M_d\tGKP^\varepsilon_{\kappa,\Delta}\rangle &=\frac{1}{\sqrt{d}} \cdot  \langle \tGKP^{\varepsilon/d}_{\kappa,\Delta},\tGKP^\varepsilon_{\kappa,\Delta}\rangle \cdot  \frac{E_{\kappa/d,\Delta d}(\varepsilon)}{E_{\kappa,\Delta}(\varepsilon/d)}\ .\label{eq:gkpmdgkp}
\end{align}
On the one hand we have 
\begin{align}
    \label{eq:auxLB} \langle \tGKP^{\varepsilon/d}_{\kappa,\Delta},\tGKP^\varepsilon_{\kappa,\Delta}\rangle \leq 1 \ ,
\end{align}
due to the fact that the states are normalized.
On the other hand, from Lemma~\ref{eq:gkpedgkpinnerproduct}, which is valid in the parameter regimes~$\kappa\in(0,1/2)$, $\Delta\in(0,1/8)$ and~$\varepsilon \geq d \Delta~$, we have
\begin{align}
    \label{eq:auxUB} \langle \tGKP_{\kappa,\Delta}^{\varepsilon/d}, \tGKP_{\kappa,\Delta}^{\varepsilon} \rangle \geq 1 - 3 \sqrt{\kappa} - 5\sqrt{\Delta} - 4 (\Delta d/\varepsilon)^2 \ .
\end{align}
Eq.~\eqref{eq:gkpmdgkp} together with Eq.~\eqref{eq:auxLB} and Lemma~\ref{lem:cte_normaliz_bound}~\eqref{eq:mainclaimekappadeltackappa} which we use to bound
$E_{\kappa/d,\Delta d}(\varepsilon)$ and~$E_{\kappa,\Delta}(\varepsilon/d)$ gives
\begin{align}
    \sqrt{d} \langle \tGKP^\varepsilon_{\kappa/d,\Delta d},M_d\tGKP^\varepsilon_{\kappa,\Delta}\rangle 
    &\leq (1-\kappa/(6d))/(1 - \kappa/2) \\ 
    &\leq (1-\kappa/(6d))(1 + \kappa) \\
    &\leq 1 + \kappa \ ,
\end{align}
where we used~$1/(1-x)\leq 1+2x$ for~$x\in(0,1/2)$.
To obtain the lower bound on this inner product we use Eq.~\eqref{eq:auxUB} and Lemma~\ref{lem:cte_normaliz_bound}~\eqref{eq:mainclaimekappadeltackappa}. The result is
\begin{align}
    \sqrt{d} \langle \tGKP^\varepsilon_{\kappa/d,\Delta d},M_d\tGKP^\varepsilon_{\kappa,\Delta}\rangle 
    &\geq (1 - 3 \sqrt{\kappa} - 5\sqrt{\Delta} - 4 (\Delta d/\varepsilon)^2) (1-\kappa/(2d))/(1 - \kappa/6) \\
    &\geq 1 - 3 \sqrt{\kappa} - 5\sqrt{\Delta} - 4 (\Delta d/\varepsilon)^2 -\kappa/(2d) \\
    &\geq 1 - 4 \sqrt{\kappa} - 9\sqrt{\Delta d /\varepsilon} 
\end{align}
where we used~$(1-x)(1-y)\geq 1-x-y$ for~$x,y\geq 0$, $1/(1-x)\geq 1$ for~$x\geq 0$, and the assumptions~$\varepsilon \geq d\Delta$, $\varepsilon,\kappa \in (\Delta d,d/2)$ and~$d \geq 2$ an integer.
Hence we have
\begin{align}
    \label{eq:gkpMgkp_aux1}
    \left|\langle \tGKP^\varepsilon_{\kappa/d,\Delta d},M_d\tGKP^\varepsilon_{\kappa,\Delta}\rangle-\frac{1}{\sqrt{d}}\right| & \leq \max \left\{ \frac{\kappa}{\sqrt{d}}, 4\sqrt{\kappa/d} + 9 \sqrt{\Delta/\varepsilon} \right\} \\
    &\leq 4\sqrt{\kappa/d} + 9 \sqrt{\Delta/\varepsilon} \ , \label{eq:gkpMgkp_aux2}
\end{align}
where we used~$\sqrt{x} \geq x$ for~$0 \leq x \leq 1$.
Because~$M_d^\dagger=M_{1/d}$ and~$\langle \tGKP^\varepsilon_{\kappa/d,\Delta d},M_d\tGKP^\varepsilon_{\kappa,\Delta}\rangle \in \bb{R}$, we have
\begin{align}
    \langle \tGKP^\varepsilon_{\kappa/d,\Delta d},M_d\tGKP^\varepsilon_{\kappa,\Delta}\rangle&=\langle M_{1/d}\tGKP^\varepsilon_{\kappa/d,\Delta d},\tGKP^\varepsilon_{\kappa,\Delta}\rangle
    =\langle\tGKP^\varepsilon_{\kappa,\Delta}, M_{1/d}\tGKP^\varepsilon_{\kappa/d,\Delta d}\rangle\, .
\end{align}
This together with Eq.~\eqref{eq:gkpMgkp_aux2} gives the claim.
\end{proof}

\begin{lemma}[Untruncated GKP-states and squeezing]
\label{lem:lem_squeezed_approx_GKP}
Let~$d\geq 2$ be an integer. Let~$\kappa \in(0,1/2)$ and~$\Delta>0$. Assume that
\begin{align}
    \Delta \leq 1/(4d^3)\ .  \label{eq:deltadassumptionx}
\end{align}
Then 
\begin{align}
    \left|\langle\tGKP_{\kappa,\Delta}, M_{1/d}\tGKP_{\kappa/d,\Delta d}\rangle - \frac{1}{\sqrt{d}}\right|
    &\leq  7\sqrt{\kappa}+16(\Delta d)^{1/4} \ .
\end{align}
\end{lemma}

\begin{proof}
By Lemma~\ref{lem:triangle_ineq} we have
\begin{align}
    &\left|\langle\tGKP_{\kappa,\Delta}, M_{1/d}\tGKP_{\kappa/d,\Delta d}\rangle - \frac{1}{\sqrt{d}}\right| \leq \left|\langle \tGKP_{\kappa,\Delta}^\varepsilon, M_{1/d} \tGKP_{\kappa/d,\Delta d}^\varepsilon  \rangle - \frac{1}{\sqrt{d}}\right| \\
    &\qquad\qquad\qquad\qquad + \sqrt{2}\left(\sqrt{ | \langle \tGKP_{\kappa,\Delta},\tGKP_{\kappa,\Delta}^\varepsilon \rangle - 1 |} + \sqrt{| \langle \tGKP_{\kappa/d,\Delta d} ,\tGKP_{\kappa/d,\Delta d}^\varepsilon  \rangle - 1 | }\right) \ .
    \label{eq:gkpMgkp_aux100}
\end{align}
We make the choice
\begin{align}
    \label{eq:varepschoice}
    \varepsilon=\sqrt{\Delta d} \ , 
\end{align}
which, together with the assumption~\eqref{eq:deltadassumptionx}, satisfies
\begin{align}
    1/2 \geq 1/(2d) \geq \varepsilon \geq d\Delta \ . \label{eq:assump_0}
\end{align}

Because of Eq.~\eqref{eq:assump_0} as well as~$\kappa\in(0,1/2)$ and~$\Delta >0$ we may apply Lemma~\ref{lem:lem_squeezed_approx_GKP_v0} with the choice~\eqref{eq:varepschoice}, which gives 
\begin{align}
    \left|\langle\tGKP^{\sqrt{\Delta d}}_{\kappa,\Delta}, M_{1/d}\tGKP^{\sqrt{\Delta d}}_{\kappa/d,\Delta d}\rangle-\frac{1}{\sqrt{d}}\right|
    &\leq 4\sqrt{\kappa/d} + 9 (\Delta/ d)^{1/4} \ . \label{eq:rhs1}
\end{align}
We have~$\varepsilon \geq \Delta d \geq \Delta$. Therefore, Lemma~\ref{lem:truncatedapproximategkpstatepointw} implies that
\begin{align}
    \langle \tGKP_{\kappa,\Delta},\tGKP_{\kappa,\Delta}^\varepsilon \rangle &\geq 1 - \kappa - 3\Delta - 2 (\Delta/\varepsilon)^4\\
    &\geq 1-\kappa-5\Delta\label{eq:vzmb}
    \end{align}
    where we used that~$\Delta/\varepsilon=\sqrt{\Delta/d}\leq \sqrt{\Delta}$. 
    In addition, we have~$\varepsilon\geq \Delta d$, hence
\begin{align}
    \langle \tGKP_{\kappa/d,\Delta d} ,\tGKP_{\kappa/d,\Delta d}^\varepsilon \rangle &\geq 1 - \kappa/d - 3\Delta d - 2 (\Delta d/\varepsilon)^4\\
    &\geq 1 - \kappa/d - 5\Delta d \label{eq:vdzvd}
\end{align}
by Lemma~\ref{lem:truncatedapproximategkpstatepointw}, where we used that~$\Delta d/\varepsilon = \sqrt{\Delta d}$ and the assumption~\eqref{eq:deltadassumptionx}.

Inserting Eqs.~\eqref{eq:vzmb}, \eqref{eq:vdzvd} and~\eqref{eq:rhs1} into Eq.~\eqref{eq:gkpMgkp_aux100} gives
\begin{align}
    \left|\langle\tGKP_{\kappa,\Delta}, M_{1/d}\tGKP_{\kappa/d,\Delta d}\rangle - \frac{1}{\sqrt{d}}\right| 
    &\leq 4\sqrt{\kappa/d} + 9 (\Delta/d)^{1/4} + \sqrt{2}\left( \sqrt{\kappa + 5\Delta} + \sqrt{\kappa/d+5\Delta d}\right) \\
        &\leq 4\sqrt{\kappa/d} + 9 (\Delta/d)^{1/4} + 2\sqrt{2}\sqrt{\kappa + 5\Delta d}\\
        &\leq 4\sqrt{\kappa} + 9 \Delta^{1/4} +3\sqrt{\kappa}+7 \sqrt{\Delta d}\\
        &\leq 7\sqrt{\kappa}+16(\Delta d)^{1/4}\ ,
\end{align}
where we substituted~$\varepsilon = \sqrt{\Delta d}$ and used~$\Delta d \in (0,1/4)$, the monotonicity of the square root and the inequality~$\sqrt{a+b}\leq \sqrt{a}+\sqrt{b}$ for~$a,b\geq 0$. This is the claim.
\end{proof}

\subsection{The Fourier transform of approximate GKP-states with integer spacing\label{sec:fouriertransformapproximateGKP}}

The Fourier transform~$\cF(f)$ of an element~$f \in L^1(\mathbb{R}) \cap L^2(\mathbb{R})$
is
\begin{align}
\mathcal{F}(f)(p)=\widehat{f}(p)=\frac{1}{\sqrt{2 \pi}} \int f(x) e^{-i p x} d x\ . \label{eq: integral def Fourier}
\end{align}
We note that the unique extension~$\cF: L^2(\bb{R}) \rightarrow L^2(\bb{R})$ is a Gaussian unitary as it maps Gaussian states to Gaussian states. This unitary is generated by the Hamiltonian~$Q^2 + P^2$. We have 
\begin{align}
    \cF = e^{-i \pi/4} e^{i\pi(Q^2 +P^2)/4} \ .
\end{align}

We have the following.
\begin{lemma}\cite[Lemma A.18]{brenner2024complexity}\label{lem:ftinconvenient}
Let~$\kappa,\Delta>0$. Then 
\begin{align}
\widehat{\gkp}_{\kappa, \Delta}(p)=\sqrt{2 \pi} C_{\kappa, \Delta} \sum_{z \in \mathbb{Z}} \eta_{\Delta}(p) \Psi_\kappa(p-2 \pi z)\qquad\textrm{ for all }p\in\mathbb{R}\ .\label{eq:ftgkpkappadelta}
\end{align}
\end{lemma}
\begin{proof}
We include the proof here for completeness. It relies on the Poisson summation formula
\begin{align}
\sum_{z \in \mathbb{Z}} f(z)=\sqrt{2 \pi} \sum_{z \in \mathbb{Z}} \widehat{f}(2 \pi z)
\end{align}
applied to the function~$f(z)=\eta_\kappa(z) e^{-i p z}$, which gives 
\begin{align}
\widehat{\gkp}_{\kappa, \Delta}(p)=C_{\kappa, \Delta} \widehat{\Psi}_{\Delta}(p)\left(\sqrt{2 \pi} \sum_{z \in \mathbb{Z}} \widehat{\eta}_\kappa(p+2 \pi z)\right) \quad \text { for } \quad p \in \mathbb{R}\ .
\end{align}
The proof then follows from the fact that the Fourier transform of the Gaussian envelope- and individual peak-functions defined by Eq.~\eqref{eq:envelopepeakfunctions}
are
\begin{align}
\widehat{\eta}_\kappa(p) &= \frac{e^{-p^2/(2\kappa^2)}}{\sqrt{\kappa} \pi^{1/4}}=\Psi_\kappa(p) \ , \\
\widehat{\Psi}_\Delta(p) &= \sqrt{\Delta}\frac{e^{-p^2\Delta^2/2}}{ \pi^{1/4}}=\eta_\Delta(p)\ .
\end{align}
\end{proof}
Let us rephrase the statement of Lemma~\ref{lem:ftinconvenient} in a more convenient form. 
\begin{lemma}[Fourier transform of approximate GKP-states]\label{thm:fouriertransformapproximate}
Let~$\kappa,\Delta>0$. Then we have 
\begin{align}
\ket{\widehat{\gkp}_{\kappa,\Delta}}&=M_{2\pi}\ket{\tGKP_{\Delta(2\pi),\kappa/(2\pi)}}\ .
\end{align}
\end{lemma}
\begin{proof}
Using that for~$\alpha>0$, the squeezing operation~$M_\alpha$ (see Eq.~\eqref{def:Malpha}) acts unitarily on elements~$\Psi\in L^2(\mathbb{R})$ by 
\begin{align}
\left(M_\alpha \Psi\right)(x)=\frac{1}{\sqrt{\alpha}} \Psi(x / \alpha)\qquad\textrm{ for }x\in \mathbb{R}\ ,
\end{align}
we can rewrite Eq.~\eqref{eq:ftgkpkappadelta} 
with~$\alpha=2\pi$ as
\begin{align}
\widehat{\gkp}_{\kappa, \Delta}(p)&= (M_{2\pi}\Phi)(p)\qquad\textrm{ where }\qquad
\Phi(p)=2\pi C_{\kappa,\Delta} \sum_{z \in \mathbb{Z}} \eta_{\Delta}(2\pi p) \Psi_\kappa(2\pi(p- z))\ .\label{eq:identitm}
\end{align}
Because~$\eta_\Delta(2\pi p)=\frac{1}{\sqrt{2\pi}} \eta_{\Delta(2\pi)}(p)$
and~$\Psi_{\kappa}(2\pi(p-z))=\frac{1}{\sqrt{2\pi}} \Psi_{\kappa/(2\pi)}(p-z)$, we obtain
\begin{align}
\Phi(p)&=C_{\kappa,\Delta}\sum_{z \in \mathbb{Z}} \eta_{\Delta\sqrt{2\pi}}(p) \Psi_{\kappa/(2\pi)}(p- z)\ .
\end{align}
We conclude from this that
\begin{align}
\ket{\Phi}&=\frac{C_{\kappa,\Delta}}{D_{\Delta(2\pi),\kappa/(2\pi)}} \ket{\tGKP_{\Delta(2\pi),\kappa/(2\pi)}}
\end{align}
or
\begin{align}
    \label{eq:FGKP}
\ket{\widehat{\gkp}_{\kappa,\Delta}}&=M_{2\pi} 
\left(\frac{C_{\kappa,\Delta}}{D_{\Delta(2\pi),\kappa/(2\pi)}}\right) \ket{\tGKP_{\Delta(2\pi),\kappa/(2\pi)}}
\end{align}
because of Eq.~\eqref{eq:identitm}. The unitarity of~$M_{2\pi}$ and the Fourier transform implies that~$C_{\kappa,\Delta} = D_{\Delta(2\pi),\kappa/(2\pi)}$.
This concludes the proof. 
\end{proof}

We note that this proof of Lemma~\ref{thm:fouriertransformapproximate} (namely, the unitarity of the Fourier transform together with Eq.~\eqref{eq:FGKP}) also proves Lemma~\ref{lem:ckd_dkd_eq}, i.e., the statement that~$C_{\kappa,\Delta} = D_{(2\pi)\Delta,\kappa/(2\pi)}$.

\section{Matrix elements of linear optics implementations  \label{sec: matrix elements bounds}}
Here we bound the matrix elements of the physical implementations given in Section~\ref{sec:logicalops} of the logical Pauli gates~$X$ and~$Z$, and of the logical Fourier transform~$\Fgate$ for approximate GKP codes of the form~$\gkpcode{\kappa,\Delta}{\varepsilon}{d}$.

Throughout this section, we assume that~$d \ge 2$ is an integer. We are interested in the matrix elements of operators with respect to the basis~$\{
\ket{\gkp_{\kappa,\Delta}^\varepsilon(k)_d}\}_{k\in\mathbb{Z}_d}$ of~$\gkpcode{\kappa,\Delta}{\varepsilon}{d}$, where 
\begin{align}
\ket{\gkp_{\kappa,\Delta}^\varepsilon(k)_d}&= e^{-i\sqrt{2\pi/d}kP} M_{\sqrt{2\pi d}}\ket{\gkp_{\kappa,\Delta}^\varepsilon}
\end{align}
for~$k\in\mathbb{Z}_d$. 

\subsection{Matrix elements of the~$X$-gate}
\begin{lemma}[Matrix elements of the logical~$X$-gate] \label{lem: matrix elements X gate}
    Let~$d\ge2$ be an integer, $\kappa,\Delta >0$ and~$\varepsilon \in (0,1/(2d)]$. 
    For~$j,k \in\mathbb{Z}_d$ define the matrix element 
    \begin{align}
        M_{j,k} = \langle \gkp_{\kappa,\Delta}^\varepsilon(j)_d,e^{-i\sqrt{2\pi/d}P}\gkp_{\kappa,\Delta}^\varepsilon(k)_d\rangle
    \end{align}
    of the operator~$e^{-i\sqrt{2\pi/d}P}$  with respect to the basis~$\{ \gkp_{\kappa,\Delta}^\varepsilon(j)\}_{j \in \mathbb{Z}_d}$ of~$\gkpcode{\kappa,\Delta}{\varepsilon}{d}$.
    We have 
\begin{align}
        (1 - \kappa^2)\cdot\delta_{j,k \oplus 1} \le  M_{j,k} & \le \delta_{j,k \oplus 1} 
\qquad \textrm{for all} \qquad j,k \in \mathbb{Z}_d \, ,
\end{align}
where~$\oplus$ denotes addition modulo~$d$.
\end{lemma}

\begin{proof}
    We have
    \begin{align}
          \left(M_{\sqrt{2\pi d}}\right)^\dagger \left(e^{-i\sqrt{2\pi/d}jP}\right)^\dagger e^{-i\sqrt{2\pi/d}(k+1)P} M_{\sqrt{2\pi d}}
        &=  \left(M_{\sqrt{2\pi d}}\right)^\dagger e^{-i\sqrt{2\pi/d}(k+1-j)P} M_{\sqrt{2\pi d}}\\
        &= e^{-i (k+1-j)P/d}\ , \label{eq: shift overlap}
    \end{align}
    where we used that
    \begin{align}
        M_\alpha^\dagger P M_\alpha = \frac{P}{\alpha}\qquad\textrm{ for any }\qquad \alpha>0\ .  \label{eq: P M_alpha} 
    \end{align}   
    It follows that the matrix elements~$M_{j,k}$ of the operator~$e^{-i\sqrt{2\pi/d}  P}$ can be written as 
    \begin{align}
        M_{j,k}&=\langle \gkp_{\Delta,\kappa}^\varepsilon(j)_d,e^{-i\sqrt{2\pi/d}  P}\gkp_{\kappa,\Delta}^\varepsilon(k)_d\rangle\\
        &=\langle e^{-i\sqrt{2\pi/d}jP}M_{\sqrt{2\pi d}}\gkp_{\Delta,\kappa}^\varepsilon,e^{-i\sqrt{2\pi/d}  P}e^{-i\sqrt{2\pi/d}kP}M_{\sqrt{2\pi d}}\gkp_{\kappa,\Delta}^\varepsilon\rangle \\
        & = \langle \gkp_{\Delta,\kappa}^\varepsilon,e^{-i(k+1-j)P/d}\gkp_{\kappa,\Delta}^\varepsilon\rangle\, .         \label{eq: shift overlap _ v2} 
    \end{align}
    We first consider the case~$k\oplus 1 \neq j$, i.e., when~$k+1-j$  is not an integer multiple of~$d$. 
        In this case, we  clearly have~$\mathsf{dist}\left((k+1-j)/d,\mathbb{Z}\right)\geq 1/d\geq 2\varepsilon$ by the assumption that~$\varepsilon<1/(2d)$.     By Lemma~\ref{lem:orthogonalitytruncated gkpstates} (expressing the orthogonality of the GKP-state~$\ket{\gkp^\varepsilon_{\kappa,\Delta}}$ and a displaced version thereof, where the displacement is far from an integer)          we conclude that the matrix element in Eq.~\eqref{eq: shift overlap _ v2} vanishes.
        
    Next, we analyze the case where~$k\oplus 1 =j$. If~$k+1 = j$, Eq.~\eqref{eq: shift overlap _ v2} evaluates to~$1$. If~$k+1-j =d$, the approximate translation-invariance of approximate GKP-states (see Lemma~\ref{lem:translationinvariancegkp}) gives
    \begin{align}
        \langle \gkp_{\kappa,\Delta}^\varepsilon,e^{-iP}\gkp_{\kappa,\Delta}^\varepsilon\rangle \ge 1 - \kappa^2 \, .
    \end{align}
    The claim follows.
\end{proof}

\subsection{Matrix elements of the~$Z$-gate}
For the logical~$Z$-gate, we also consider integer powers~$Z^m$ with~$|m|\leq d$. This will be useful in the proof of Lemma~\ref{lem: matrix elements fourier}.

\begin{lemma}[Matrix elements of the logical~$Z^m$-gate] 
\label{lem: matrix element Z gate}    
Let~$d\ge2$ be an integer, $m\in\mathbb{Z}$ an integer satisfying~$|m|\leq d$, $\kappa, \Delta>0$ and~$\varepsilon \in (0, 1/(2d)]$.
For~$j,k \in\mathbb{Z}_d$ define the matrix element 
\begin{align}
    M_{j,k} = \langle \gkp_{\kappa,\Delta}^\varepsilon(j)_d,e^{i\sqrt{2\pi/d}mQ}\gkp_{\kappa,\Delta}^\varepsilon(k)_d\rangle
\end{align}
of the operator~$e^{i\sqrt{2\pi/d}mQ}$  with respect to the basis~$
    \{ \gkp_{\kappa,\Delta}^\varepsilon(j)\}_{j \in \mathbb{Z}_d}$ of~$\gkpcode{\kappa,\Delta}{\varepsilon}{d}$. We have 
\begin{align}
    (1- 10 d^2 \Delta^2 - 16(\Delta/\varepsilon)^4 )\cdot  \delta_{j,k} \le  \overline{\omega}_d^{mj} \cdot M_{j,k}&\le \delta_{j,k}  
    \qquad \textrm{for all} \qquad j,k \in \mathbb{Z}_d\, ,
\end{align}
where~$\omega_d = 2^{2\pi i/d}$.
\end{lemma}

\begin{proof}
The identity~$\left(e^{-i\beta P}\right)^\dagger Q e^{-i\beta P}= Q + \beta I$ for any~$\beta\in\mathbb{R}$  implies that
    \begin{align}
\left(e^{-i\sqrt{2\pi/d}jP}\right)^\dagger e^{i\sqrt{2\pi/d}mQ} e^{-i\sqrt{2\pi/d}kP}&=  
\left(\left(e^{-i\sqrt{2\pi/d}jP}\right)^\dagger e^{i\sqrt{2\pi/d}mQ}e^{-i\sqrt{2\pi/d}jP}\right)e^{i\sqrt{2\pi/d}jP} e^{-i\sqrt{2\pi/d}kP}\\
&=e^{i\sqrt{\fr{2\pi}{d}}m (Q + \sqrt{\frac{2\pi}{d}}jI)} e^{i\sqrt{2\pi/d}(j-k)P}\\
&=\omega_d^{mj}  e^{i\sqrt{\fr{2\pi}{d}}m Q} e^{i\sqrt{2\pi/d}(j-k)P}\ .\label{eq:omegadjxv}
\end{align}
Since 
\begin{align}
\begin{aligned}
    \left(M_\alpha\right)^\dagger P M_\alpha &= P/\alpha\\
    \left(M_\alpha\right)^\dagger Q M_\alpha &= \alpha Q
    \end{aligned}\qquad\textrm{ for }\qquad \alpha>0
    \end{align}
    we obtain
    \begin{align}
    \begin{aligned}
        \left(M_{\sqrt{2\pi d}}\right)^\dagger     
        e^{i\sqrt{\fr{2\pi}{d}}m Q}M_{\sqrt{2\pi d}}&= e^{2\pi i m Q}\\
        \left(M_{\sqrt{2\pi d}}\right)^\dagger     
        e^{i\sqrt{\fr{2\pi}{d}}(j-k)P}M_{\sqrt{2\pi d}}&= e^{2\pi i (j-k)/d P}\ .
        \end{aligned}\label{eq:mkjfdv}
    \end{align}
    Combining Eqs.~\eqref{eq:omegadjxv} and~\eqref{eq:mkjfdv} we conclude that
    \begin{align}
            \left(M_{\sqrt{2\pi d}}\right)^\dagger \left(e^{-i\sqrt{2\pi/d}jP}\right)^\dagger e^{i\sqrt{2\pi/d}mQ} e^{-i\sqrt{2\pi/d}kP} M_{\sqrt{2\pi d}} 
           &=\omega_d^{mj} e^{2\pi i m Q} e^{i (j-k)P/d}\ .         \label{eq:Zgateproog_aux1} 
    \end{align}
    It follows from Eq.~\eqref{eq:Zgateproog_aux1} and the definition of the state~$\ket{\gkp_{\kappa,\Delta}^\varepsilon(j)_d}$ that
    \begin{align}
    M_{j,k}&=\langle \gkp_{\kappa,\Delta}^\varepsilon(j)_d, e^{i \sqrt{2\pi/d}Q}\gkp_{\kappa,\Delta}^\varepsilon(k)_d\rangle\\
     &=\omega_d^{mj}\langle \gkp_{\kappa,\Delta}^\varepsilon, e^{2\pi i mQ} e^{i (j-k)P/d} \gkp_{\kappa,\Delta}^\varepsilon\rangle 
    \end{align}
    and thus
    \begin{align}
    \overline{\omega}_d^{mj}M_{j,k}&=\langle \gkp_{\kappa,\Delta}^\varepsilon, e^{2\pi i mQ} e^{i (j-k)P/d} \gkp_{\kappa,\Delta}^\varepsilon\rangle \ . \label{eq: Z gate first}
    \end{align}

    As the unitary~$e^{2\pi i mQ}$ does not change the support of the state~$\gkp_{\kappa,\Delta}^\varepsilon$, by the same argument as in the proof of Lemma~\ref{lem: matrix elements X gate} 
    we conclude that the inner product in Eq.~\eqref{eq: Z gate first} vanishes for~$\varepsilon \le 1/(2d)$ unless~$j =k$. The assumption~$\varepsilon \leq 1/(4|m|) \leq 1/(4d)$ is satisfied, and the claim follows by  Lemma~\ref{lem:translationinvariancephasegkpalternative} which states that
    \begin{align}
        \langle \gkp_{\kappa,\Delta}^\varepsilon, e^{2\pi i m Q} \gkp_{\kappa,\Delta}^\varepsilon\rangle &\geq 1 - 10 \Delta^2 m^2  - 16(\Delta/\varepsilon)^4 \ , \\
         &\geq 1 - 10 d^2 \Delta^2   - 16(\Delta/\varepsilon)^4 \, ,
    \end{align}
     where we used that~$|m|\leq d$.
\end{proof}

\subsection{Matrix elements of the Fourier transform~$\Fgate$\label{sec:matrixlementsfouriertransformsec}}
Here we establish bounds on the matrix elements of the linear optics implementation~$e^{i\frac{\pi}{4}(Q^2+P^2)}$ for the Fourier transformation~$\Fgate$ with respect
to bases associated with the input code space~$\cL_{in}=\gkpcode{\kappa,\Delta}{\varepsilon}{d}$ 
and the output code space~$\cL_{out}=\gkpcode{2\pi d \Delta,\kappa/(2\pi d)}{\varepsilon}{d}$.

We proceed in two steps. In  Lemma~\ref{lem: matrix elements fourier 00} we
first give such a bound for a single matrix element. In Lemma~\ref{lem: matrix elements fourier} we generalize this result and give bounds for all matrix elements of interest.
\begin{lemma}[Matrix element~$M_{0,0}$ of the Fourier transform~$\Fgate$] \label{lem: matrix elements fourier 00}
Let~$d\ge2$ be an integer.
Let~$\kappa\in (0,1/d^2)$ 
and~$\Delta,\varepsilon > 0~$ be such that
\begin{align}
\Delta &\leq \frac{\kappa}{2\pi d}\leq \varepsilon\leq \frac{1}{2d}\ . \label{eq:assumptionpsmvx}
\end{align}
 Define 
\begin{align}
    M_{0,0} = \langle \gkp^\varepsilon_{2\pi d \Delta,\kappa/(2\pi d)}(0)_d,e^{i\frac{\pi}{4}(Q^2+P^2)}\gkp^\varepsilon_{\kappa,\Delta}(0)_d\rangle \ .
\end{align}
Then
\begin{align}
    \left| M_{0,0} - \frac{1}{\sqrt{d}} \right|  
    &\le  10 \left(\frac{\kappa}{2\pi d\varepsilon}\right)^{1/4}+19\kappa^{1/4}\ . \label{eq:upperboundmvx}
\end{align}
\end{lemma}
The proof of this Lemma
involves the states 
\begin{align} 
\ket{\tGKP_{\kappa,\Delta}(0)_d}&= M_{\sqrt{2\pi d}}\ket{\tGKP_{\kappa,\Delta}} \ , \\
\ket{\gkp_{\kappa,\Delta}(0)_d}&= M_{\sqrt{2\pi d}}\ket{\gkp_{\kappa,\Delta}}\ .
\end{align}
Both these states approximate the GKP-state
\begin{align}
\ket{\gkp_{\kappa,\Delta}^\varepsilon(0)_d}= M_{\sqrt{2\pi d}}\ket{\gkp_{\kappa,\Delta}^\varepsilon}
\end{align} for suitably chosen parameters~$(\varepsilon,\kappa,\Delta)$. This follows from the unitarity  of~$M_{\sqrt{2\pi d}}$ as well as  Lemma~\ref{eq:tgkpgkpeps} and Lemma~\ref{lem:truncatedapproximateGKPstates}, respectively.

    \begin{proof}
For brevity, we write 
\begin{align}
    \begin{aligned}
        \kappa'&=2\pi d \Delta \\
        \Delta'&=\kappa/(2\pi d)
    \end{aligned}
    \qquad\textrm{ and }\qquad W_{\Fgate} = e^{i\frac{\pi}{4}(Q^2+P^2)}\ ,\label{eq:fdefshorthand}
\end{align}
and use the shorthand notation 
\begin{align}
\label{eq:notationshortvectors2}
    \ket{\encoded{k}_{in}} = \ket{\gkp^{\varepsilon}_{\kappa ,\Delta}(k)_d}
\quad\text{ and }\quad
    \ket{\encoded{j}_{out}} = \ket{\gkp^{\varepsilon}_{\kappa' ,\Delta'}(j)_d}
    \quad\text{ for }\quad j, k \in \bb{Z}_d
\end{align}
 for basis elements of~$\cLin=\gkpcode{\kappa,\Delta}{\varepsilon}{d}$ and~$\cLout=\gkpcode{\kappa',\Delta'}{\varepsilon}{d}$ respectively.  Expressed in this notation, we have~$M_{0,0} = \langle \encoded{0}_{out}|W_{\Fgate}|\encoded{0}_{in}\rangle$. 

To show the claim, we will use the triangle inequality
\begin{align}
    &\hspace{-9mm}\left| \langle \encoded{0}_{out} |W_{\Fgate}| \encoded{0}_{in}\rangle - \frac{1}{\sqrt{d}}\right| \\
    &\hspace{-9mm}= \left| \langle \encoded{0}_{out} |W_{\Fgate}| \encoded{0}_{in}\rangle - \langle \tGKP_{\kappa',\Delta'}(0)_d,W_{\Fgate} \gkp_{\kappa,\Delta}(0)_d\rangle + \langle \tGKP_{\kappa',\Delta'}(0)_d,W_{\Fgate} \gkp_{\kappa,\Delta}(0)_d\rangle - \frac{1}{\sqrt{d}}\right| \\
    \label{eq:F00aux1} &\hspace{-9mm}\leq \left| \langle \encoded{0}_{out}|W_{\Fgate}|\encoded{0}_{in}\rangle - \langle \tGKP_{\kappa',\Delta'}(0)_d, W_{\Fgate}\gkp_{\kappa,\Delta}(0)_d\rangle \right| \\ &\quad + \left|\langle \tGKP_{\kappa',\Delta'}(0)_d,W_{\Fgate}\gkp_{\kappa,\Delta}(0)_d\rangle - \frac{1}{\sqrt{d}}\right|
\end{align}
and  bound each term in Eq.~\eqref{eq:F00aux1} separately. 

We start by giving an upper bound on the first term in Eq.~\eqref{eq:F00aux1}.
By Lemma~\ref{lem:triangle_ineq} we have
\begin{align}
    \label{eq:boundnew}
&\left|  \langle \encoded{0}_{out}|W_{\Fgate}|\encoded{0}_{in} \rangle
-
\langle \tGKP_{\kappa',\Delta'}(0)_d, W_{\Fgate} \gkp_{\kappa,\Delta}(0)_d\rangle
\right| \\
&\qquad\qquad \leq \sqrt{2}\left(\sqrt{
\left| \langle \encoded{0}_{out},\tGKP_{\kappa',\Delta'}(0)_d\rangle-1
\right|} + \sqrt{\left| {}\langle \encoded{0}_{in},\gkp_{\kappa,\Delta}(0)_d\rangle-1
\right|
}\right)\ .
\end{align}
Because
\begin{align}
\langle \encoded{0}_{out},\tGKP_{\kappa',\Delta'}(0)_d\rangle&=
\langle \gkp^\varepsilon_{\kappa',\Delta'}(0)_d,\tGKP_{\kappa',\Delta'}(0)_d\rangle=\langle \gkp^\varepsilon_{\kappa',\Delta'},\tGKP_{\kappa',\Delta'}\rangle
\end{align} by the unitarity of~$M_{\sqrt{2\pi d}}$, 
we have 
\begin{align}
\left| \langle \encoded{0}_{out},\tGKP_{\kappa',\Delta'}(0)_d\rangle-1
\right|&=|\langle \tGKP_{\kappa',\Delta'},\gkp^\varepsilon_{\kappa',\Delta'}\rangle -1|\\
&\leq 2\kappa' + 7 \sqrt{\Delta'} + 2(\Delta'/\varepsilon)^2\qquad\textrm{ by 
 Lemma~\ref{eq:tgkpgkpeps}}\\
&=4\pi d \Delta+7\sqrt{\kappa/(2\pi d)}+2(\kappa/(2\pi d \varepsilon))^2\\
&\leq 2\kappa+7\sqrt{\kappa/(2\pi d\varepsilon)}+2(\kappa/(2\pi d \varepsilon))^2\\
&\leq 2\kappa+9\sqrt{\kappa/(2\pi d\varepsilon)}\ ,\label{eq:FNEWaux2}
\end{align}
where we  used the assumption~\eqref{eq:assumptionpsmvx}, which implies, in particular, that~$\kappa/(2\pi d \varepsilon)<1$. 
(We note that 
the assumption~\eqref{eq:assumptionpsmvx}
 and~$\kappa\in (0,1/d^2) \subset (0,1/4)$ guarantee that~$\Delta<1/(8\pi d)$,
 which implies that~$\kappa'=2\pi d \Delta<1/4$ and~$\Delta'=\kappa/(2\pi d)< 1/(8\pi d)<1/8$. Hence  Lemma~\ref{eq:tgkpgkpeps} is indeed applicable here.) 
We also have 
\begin{align}
    \left| {}\langle \encoded{0}_{in}|\gkp_{\kappa,\Delta}(0)_d\rangle-1
    \right| &=\left|\langle \gkp^\varepsilon_{\kappa,\Delta}(0)_d,\gkp_{\kappa,\Delta}(0)_d\rangle-1
    \right|\\
&=\left|\langle \gkp^\varepsilon_{\kappa,\Delta},\gkp_{\kappa,\Delta}\rangle-1
    \right| \\
    &\leq 7\Delta + 2(\Delta/\varepsilon)^4\qquad \textrm{by  Lemma~\ref{lem:truncatedapproximateGKPstates} }\\
    &\leq 9 \textfrac{\kappa}{(2\pi d)}\\
    &\leq 9 \textfrac{\kappa}{(2\pi d\varepsilon)}
     \label{eq:FNEWaux1} 
\end{align}
where we used the assumption~\eqref{eq:assumptionpsmvx} in the penultimate  step. (Indeed, using the  assumption 
$\Delta\leq \kappa/(2\pi d)$ we obtain  
$\Delta/\varepsilon\leq \kappa/(2\pi d\varepsilon)\leq 1$, hence 
$(\Delta/\varepsilon)^4\leq \Delta/\varepsilon\leq \kappa/(2\pi d\varepsilon)$.)
The fact that~$\kappa\in  (0,1/d^2) \subset(0,1/4)$ by assumption justifies the use of Lemma~\ref{lem:truncatedapproximateGKPstates} here.

Inserting Eqs.~\eqref{eq:FNEWaux2} and~\eqref{eq:FNEWaux1} into Eq.~\eqref{eq:boundnew} 
and setting~$\tau:=\kappa/(2\pi d\varepsilon)$ we obtain
\begin{align}
    \left|  \langle \encoded{0}_{out}|W_{\Fgate} |\encoded{0}_{in} \rangle
-
\langle \tGKP_{\kappa',\Delta'}(0)_d,W_{\Fgate} \gkp_{\kappa,\Delta}(0)_d\rangle
    \right|&\leq \sqrt{2} \left(\sqrt{2\kappa+9\sqrt{\tau}}+\sqrt{9\tau}\right) \\
    &\leq\sqrt{2}\left(\sqrt{2\pi \tau+9\sqrt{\tau}}+3 \sqrt{\tau}\right)\\
    &\leq 10\tau^{1/4}\\
    &=10 \left(\frac{\kappa}{2\pi d \varepsilon}\right)^{1/4}\ .\label{eq:tildetaubound}
\end{align}
Here we used that 
$2\pi d\varepsilon\leq \pi$ by the assumption~\eqref{eq:assumptionpsmvx} and thus~$\kappa\leq \pi \tau$, as well as~$0 \leq \tau \leq 1$.

We proceed to upper bound the second term in Eq.~\eqref{eq:F00aux1}, i.e., we show that the matrix element~$\langle \tGKP_{\kappa',\Delta'}(0)_d,W_{\Fgate} \gkp_{\kappa,\Delta}(0)_d\rangle$ is close to~$1/\sqrt{d}$.
Observe that because~$W_{\Fgate} M_{\alpha}=M_{1/\alpha} W_{\Fgate}$ for~$\alpha > 0$, we have 
\begin{align}
    \langle \tGKP_{\kappa' ,\Delta'}(0)_d,W_{\Fgate}\gkp_{\kappa,\Delta}(0)_d\rangle 
    &=
    \langle M_{\sqrt{2\pi d}}  \tGKP_{\kappa',\Delta'},
    W_{\Fgate} M_{\sqrt{2\pi d}}   \gkp_{\kappa,\Delta}\rangle \\
    &=
    \langle  \tGKP_{\kappa',\Delta'},
    (M_{\sqrt{2\pi d}})^\dagger W_{\Fgate} M_{\sqrt{2\pi d}}\gkp_{\kappa,\Delta}\rangle \\
    &=
    \langle  \tGKP_{\kappa',\Delta'},
    M_{1/(2\pi d)}W_{\Fgate} \gkp_{\kappa,\Delta}\rangle \\
    &=    \langle  \tGKP_{\kappa',\Delta'},
    M_{1/(2\pi d)}M_{2\pi}\tGKP_{\Delta(2\pi),\kappa/(2\pi)}\rangle\\
    &=     \langle  \tGKP_{\kappa',\Delta'},
    M_{1/d}\tGKP_{\Delta(2\pi),\kappa/(2\pi)}\rangle\\
    &=     \langle  \tGKP_{\kappa',\Delta'},
    M_{1/d}\tGKP_{\kappa'/d,\Delta' d}\rangle \label{eq:fourierMelem_aux1} \ 
\end{align}
by Lemma~\ref{thm:fouriertransformapproximate}, i.e., the identity~$W_{\Fgate} \gkp_{\kappa,\Delta}=M_{2\pi}\tGKP_{\Delta(2\pi),\kappa/(2\pi)}$.

Observe that the assumptions~\eqref{eq:assumptionpsmvx}, 
~$\kappa\leq 1/d^2$ and~$d\geq 2$ imply that 
 \begin{align}
\kappa'&=2\pi d \Delta \leq \kappa \leq 1/d^2 \leq 1/4<1/2\\
\Delta' &=\kappa/(2\pi d) \leq 1/(2\pi d^3) < 1/(4d^3)\ .
\end{align}
In particular, this means that the pair~$(\kappa',\Delta')$ 
satisfies the assumptions of Lemma~\ref{lem:lem_squeezed_approx_GKP}, and we conclude that
\begin{align}
\left|\langle\tGKP_{\kappa',\Delta'}, M_{1/d}\tGKP_{\kappa'/d,\Delta' d}\rangle - \frac{1}{\sqrt{d}}\right|
        \leq  7\sqrt{\kappa'}+16(d \Delta')^{1/4}\ .  \label{eq:upperbonvxvdam}
\end{align}
Combining Eq.~\eqref{eq:upperbonvxvdam} with Eq.~\eqref{eq:fourierMelem_aux1}
we therefore obtain

\begin{align}
    \label{eq:gkpMgkp1d}
    \left|\langle \tGKP_{\kappa' ,\Delta'}(0)_d,W_{\Fgate}\gkp_{\kappa,\Delta}(0)_d\rangle 
-\frac{1}{\sqrt{d}}\right| &=
\left|\langle \tGKP_{\kappa' ,\Delta'}, 
M_{1/d} \tGKP_{\kappa'/d,\Delta' d} \rangle
-\frac{1}{\sqrt{d}}\right| \\
&\leq  7\sqrt{\kappa'}+16 (d \Delta')^{1/4} \\\
&= 7 \sqrt{2\pi }\sqrt{d \Delta} + 16 (2\pi)^{-1/4} \kappa^{1/4}\\
&\leq 18 \sqrt{d\Delta} + 11 \kappa^{1/4} \\
&\leq 19 \kappa^{1/4} \label{eq:secondtermF00aux} 
\end{align}
where in the last step, we used that~$d \Delta \leq \kappa/(2\pi) \leq 1$ by the assumption~\eqref{eq:assumptionpsmvx} and  
~$18 \sqrt{d\Delta} \leq 18 \sqrt{\kappa/(2\pi)} \leq 8 \sqrt{\kappa}$. 

 Finally, inserting Eqs.~\eqref{eq:tildetaubound} and~\eqref{eq:secondtermF00aux} into Eq.~\eqref{eq:F00aux1} gives the claim

\begin{align}
\left|\langle \encoded{0}_{out}|W_{\Fgate}| \encoded{0}_{in}\rangle
-\frac{1}{\sqrt{d}}\right|
&\leq  10 \left(\frac{\kappa}{2\pi d\varepsilon}\right)^{1/4}+ 19 \kappa^{1/4}\ .
\end{align}
\end{proof}
    
    In the following, we consider more general matrix elements
    \begin{align}
      M_{j,k} = \langle \gkp^\varepsilon_{2\pi d \Delta,\kappa/(2\pi d)}(j)_d,e^{i\frac{\pi}{4}(Q^2+P^2)}\gkp^\varepsilon_{\kappa,\Delta}(k)_d\rangle   
      \end{align}
      where~$j,k\in\mathbb{Z}_d$ are arbitrary. To do so, we first relate~$M_{j,k}$ to 
      a matrix element of the operator~$\omega_d^{jk}e^{-ik\sqrt{\frac{2\pi}{d}}Q}e^{i\frac{\pi}{4}(Q^2+P^2)}e^{-ij\sqrt{\frac{2\pi}{d}}Q
}$. Here we again use the notation from Eqs.~\eqref{eq:fdefshorthand} and~\eqref{eq:notationshortvectors2}
      such that~$M_{j,k}=\langle \encoded{j}_{out} ,W_{\Fgate} \encoded{k}_{in}  \rangle$. Specifically, we will use the following identity: We have
      \begin{align}
          \langle \encoded{j}_{out} |W_{\Fgate}| \encoded{k}_{in}  \rangle
&=\omega_{d}^{jk} 
 \langle  \encoded{0}_{out} |e^{-i k \sqrt{\frac{2\pi}{d}}Q}W_{\Fgate} e^{i j \sqrt{\frac{2\pi}{d}}Q}|
\encoded{0}_{in} \rangle\ .\label{eq:omegajkdeijk}
\end{align}
for all~$j,k\in\mathbb{Z}_d$.
\begin{proof}
Let~$j,k\in\mathbb{Z}_d$ be arbitrary. 
 Using that 
\begin{align}
W_{\Fgate} QW_{\Fgate}^\dagger &= -P \ , \\
W_{\Fgate} PW_{\Fgate}^\dagger &= Q 
\end{align}
and the commutation relation
\begin{align}
e^{ik\sqrt{\frac{2\pi}{d}} Q}e^{ij\sqrt{\frac{2\pi}{d}} P}&=\omega_d^{jk}e^{ij\sqrt{\frac{2\pi}{d}} P}e^{ik\sqrt{\frac{2\pi}{d}} Q}
\end{align}
we have 
\begin{align}
e^{i j \sqrt{\frac{2\pi}{d}}P}W_{\Fgate} e^{-i k \sqrt{\frac{2\pi}{d}}P}
&=e^{i j \sqrt{\frac{2\pi}{d}}P} e^{-i k \sqrt{\frac{2\pi}{d}}Q}W_{\Fgate}\\
&=\omega_{d}^{jk} e^{-i k \sqrt{\frac{2\pi}{d}}Q}e^{i j \sqrt{\frac{2\pi}{d}}P}W_{\Fgate}\\
&=\omega_{d}^{jk} e^{-i k \sqrt{\frac{2\pi}{d}}Q}W_{\Fgate} e^{i j \sqrt{\frac{2\pi}{d}}Q}\ .\label{eq:omegadjkcommrel}
\end{align}
We can rewrite the matrix element of interest as 
\begin{align}
\langle \encoded{j}_{out}|W_{\Fgate}| \encoded{k}_{in}\rangle &= \langle \gkp^\varepsilon_{\kappa' ,\Delta'}(j)_d,W_{\Fgate}\gkp^\varepsilon_{\kappa,\Delta}(k)_d\rangle \\
&= \langle \encoded{0}_{out}|
e^{i j \sqrt{\frac{2\pi}{d}}P}W_{\Fgate} e^{-i k \sqrt{\frac{2\pi}{d}}P} |\encoded{0}_{in}\rangle \ .
\label{eq:fourierMelem_aux11}
\end{align}
Inserting Eq.~\eqref{eq:omegadjkcommrel} into Eq.~\eqref{eq:fourierMelem_aux11}
we obtain
\begin{align}
          \langle \encoded{j}_{out} |W_{\Fgate}| \encoded{k}_{in}  \rangle
&=
\omega_{d}^{jk} 
\langle \encoded{0}_{out} |e^{-i k \sqrt{\frac{2\pi}{d}}Q}W_{\Fgate} e^{i j \sqrt{\frac{2\pi}{d}}Q}|
\encoded{0}_{in} \rangle\   
 \end{align}
 which is the claim.
          \end{proof}
We can bound each matrix element~$M_{j,k}$  as follows.
    \begin{lemma}[Matrix elements of the Fourier transform~$\Fgate$] \label{lem: matrix elements fourier}
Let~$d\ge2$ be an integer. 
Let 
~$\kappa\in (0,1/d^2)$ and~$\Delta,\varepsilon>0~$ be such that 
\begin{align}
\Delta &\leq \frac{\kappa}{2\pi d}\leq \varepsilon\leq \frac{1}{2d}\ . \label{eq:assumptionpsmvxfff}
\end{align}
For~$j,k\in\bb{Z}_d$ define the matrix element
\begin{align}
    M_{j,k} = \langle \gkp^\varepsilon_{2\pi d \Delta,\kappa/(2\pi d)}(j)_d,e^{i\frac{\pi}{4}(Q^2+P^2)}\gkp^\varepsilon_{\kappa,\Delta}(k)_d\rangle   
\end{align}
of the operator~$e^{i\frac{\pi}{4}(Q^2+P^2)}$ with respect to the basis~$\{\gkp_{\kappa,\Delta}^\varepsilon(j)\}_{j\in\bb{Z}_d}$ of~$\gkpcode{\kappa,\Delta}{\varepsilon}{d}$  and the basis~$\{\gkp_{2\pi d \Delta,\kappa/(2\pi d)}^\varepsilon(j)\}_{j\in\bb{Z}_d}$ of~$\gkpcode{2\pi d\Delta, \kappa/(2\pi d)}{\varepsilon}{d}$.
Then
 \begin{align}
    | M_{j,k} - \omega^{jk}_d / \sqrt{d} |  &\leq 
     24\left(\frac{\kappa}{2\pi d\varepsilon}\right)^{1/4} + 21\kappa^{1/4}\ .\label{eq:mjkvdadf}
\end{align}
for all~$j,k\in\bb{Z}_d$. In particular, for the optimal truncation parameter~$\varepsilon_d=1/(2d)$ we have 
 \begin{align}
    | M_{j,k} - \omega^{jk}_d / \sqrt{d} |  &\leq  40\kappa^{1/4} \label{eq:optimalsqueezingpmvd}
\end{align}
for all~$j,k\in\mathbb{Z}_d$ if~$(2\pi d)\Delta \leq \kappa< 1/d^2$.
\end{lemma}

\begin{proof}
Recall that~$W_{\Fgate}=e^{i\frac{\pi}{4}(Q^2+P^2)}$ and~$M_{j,k}=\langle \encoded{j}_{out} |W_{\Fgate}| \encoded{k}_{in}  \rangle$.
We show the claim using the triangle inequality 
\begin{align}
    \left| \langle\encoded{j}_{out} |W_{\Fgate}| \encoded{k}_{in}\rangle - \frac{\omega^{jk}}{\sqrt{d}} \right| 
    &= \left| \langle \encoded{j}_{out} |W_{\Fgate} | \encoded{k}_{in}\rangle - \omega^{jk}_d \langle \encoded{0}_{out}|W_{\Fgate} |\encoded{0}_{in}\rangle+ \omega^{jk}_d \langle \encoded{0}_{out}|W_{\Fgate} |\encoded{0}_{in}\rangle - \frac{\omega^{jk}}{\sqrt{d}} \right| \\
    &\leq \left| \langle \encoded{j}_{out} |W_{\Fgate} |\encoded{k}_{in}\rangle - \omega^{jk}_d \langle \encoded{0}_{out}|W_{\Fgate} |\encoded{0}_{in}\rangle\right| + \left| \langle \encoded{0}_{out}|W_{\Fgate}| \encoded{0}_{in}\rangle - \frac{1}{\sqrt{d}} \right| \ .
    \label{eq:tautauprime} 
\end{align}
The second term on the rhs.~is upper bounded by Lemma~\ref{lem: matrix elements fourier 00}, applicable because of the assumptions~\eqref{eq:assumptionpsmvxfff} and 
~$\kappa\in (0,1/d^2)$, obtaining
\begin{align}
\label{eq:boundOWF0}
    \left| \langle \encoded{0}_{out} |W_{\Fgate}| \encoded{0}_{in} \rangle - \frac{1}{\sqrt{d}} \right|
     &  \leq  10 \left(\frac{\kappa}{2\pi d\varepsilon}\right)^{1/4}+ 19\kappa^{1/4} \ .
\end{align}
It remains to upper bound the first term. Because of Eq.~\eqref{eq:omegajkdeijk}   we have
\begin{align}
\left| \langle \encoded{j}_{out} |W_{\Fgate}| \encoded{k}_{in}  \rangle
    - \omega_{d}^{jk} \langle \encoded{0}_{out} |W_{\Fgate}| \encoded{0}_{in}\rangle
    \right| 
    &=\left|\omega_d^{jk}\langle e^{i k \sqrt{\frac{2\pi}{d}}Q} \encoded{0}_{out}|W_{\Fgate}| e^{i j \sqrt{\frac{2\pi}{d}}Q}\encoded{0}_{in}\rangle
    - \omega_{d}^{jk} \langle \encoded{0}_{out} |W_{\Fgate}| \encoded{0}_{in}\rangle
    \right|\\
    &= \left|\langle e^{i k \sqrt{\frac{2\pi}{d}}Q} \encoded{0}_{out}|W_{\Fgate}| e^{i j \sqrt{\frac{2\pi}{d}}Q} \encoded{0}_{in}\rangle - \langle \encoded{0}_{out} |W_{\Fgate}| \encoded{0}_{in}\rangle \right| \\
    &\leq \sqrt{2} \left(\sqrt{|\langle \encoded{0}_{out}|e^{i k \sqrt{\frac{2\pi}{d}}Q}|\encoded{0}_{out}\rangle -1|}+\sqrt{|\langle \encoded{0}_{in}|e^{ij \sqrt{\frac{2\pi}{d}}Q} |\encoded{0}_{in}\rangle -1|} \right)\ . \label{eq:taudefinitionm} 
\end{align}
Here we applied  Lemma~\ref{lem:triangle_ineq} in the last step to 
compare the matrix elements of the unitary~$W_{\Fgate}$ with respect to the  pair of states
$(\psi_1,\psi_2)=(e^{ik \sqrt{\frac{2\pi}{d}}Q}\ket{\encoded{0}_{out}},e^{ij \sqrt{\frac{2\pi}{d}}Q}\ket{\encoded{0}_{in}})$
and the pair
$(\varphi_1,\varphi_2)=(\ket{\encoded{0}_{out}},\ket{\encoded{0}_{in}})$, respectively. 

Lemma~\ref{lem: matrix element Z gate}     
implies that 
\begin{align}
\left|\langle \gkp_{\kappa,\Delta}^\varepsilon(0)_d,e^{i\sqrt{2\pi/d}mQ}\gkp_{\kappa,\Delta}^\varepsilon(0)_d\rangle-1\right|&\leq 10 (d\Delta)^2 + 16(\Delta/\varepsilon)^4 \ \ \textrm{ for every } m\in\mathbb{Z}_d\ \label{eq:lowerboundmatrixlementmve}
\end{align}
for all~$(\kappa,\Delta,\varepsilon)$ with~$\varepsilon\leq 1/(2d)$. 
Applying this with~$m=j$ (which satisfies~$|m|\leq d$) 
and~$(\kappa,\Delta,\varepsilon)$ gives
\begin{align}
   \left|\langle \encoded{0}_{in}|e^{ij \sqrt{\frac{2\pi}{d}}Q}|\encoded{0}_{in}\rangle-1\right|
    &\leq 10 (d \Delta)^2 + 16(\Delta/\varepsilon)^4\\
    &\leq 10\left(\frac{\kappa}{2\pi}\right)^2+16\left(\frac{\kappa}{2\pi d\varepsilon}\right)^4\\
    &\leq \kappa^2 +16\left(\frac{\kappa}{2\pi d\varepsilon}\right)^4\ ,\label{eq:tau12ab}
\end{align}
where we used that~$\Delta/\varepsilon \leq \kappa/(2\pi d\varepsilon)$ 
and~$d\Delta\leq \kappa/(2\pi)$ by the assumption~\eqref{eq:assumptionpsmvxfff}.

Similarly, applying Eq.~\eqref{eq:lowerboundmatrixlementmve} with~$m=k$ and~$(\kappa',\Delta',\varepsilon)$
we obtain
\begin{align}\left|\langle \encoded{0}_{out}|e^{ik \sqrt{\frac{2\pi}{d}}Q}|\encoded{0}_{out}\rangle-1\right|
&\leq  10 (d \Delta')^2 + 16(\Delta'/\varepsilon)^4\\
&\leq  10 d^2 \left(\frac{\kappa}{2\pi d}\right)^2+16\left(\frac{\kappa}{2\pi d\varepsilon}\right)^4\\
&\leq \kappa^2 +16 \left(\frac{\kappa}{2\pi d\varepsilon}\right)^4\\
&\leq (\pi^2+16) \left(\frac{\kappa}{2\pi d\varepsilon}\right)^2\\
&\leq 26 \left(\frac{\kappa}{2\pi d\varepsilon}\right)^2\ .  \label{eq:tau12} 
\end{align}
In the penultimate inequality,  we used that~$1/\pi\leq 1/(2\pi d\varepsilon)$ (because~$\varepsilon\leq 1/(2d)$ by the assumption~\eqref{eq:assumptionpsmvxfff})  and thus~$\kappa^2=\pi^2\left(\kappa/\pi\right)^2 \leq \pi^2 \left(\kappa/(2\pi d\varepsilon)\right)^2$, and the fact that~$\kappa/(2\pi d\varepsilon)\leq 1$ (again by the assumption~\eqref{eq:assumptionpsmvxfff}).

Inserting Eqs.~\eqref{eq:tau12ab}, \eqref{eq:tau12} into~\eqref{eq:taudefinitionm}, we conclude that 
\begin{align}
    \left| \langle \encoded{j}_{out} |W_{\Fgate}| \encoded{k}_{in}  \rangle
        - \omega_{d}^{jk} \langle \encoded{0}_{out} |W_{\Fgate}| \encoded{0}_{in}\rangle
        \right| 
    &\leq  \sqrt{2} \left(\sqrt{ \kappa^2 +16\left(\frac{\kappa}{2\pi d\varepsilon}\right)^4}+\sqrt{ 26 \left(\frac{\kappa}{2\pi d\varepsilon}\right)^2}\right)\  .    \label{eq:taufinalbound} 
\end{align}
Finally, inserting Eq.~\eqref{eq:taufinalbound} and Eq.~\eqref{eq:boundOWF0}
into Eq.~\eqref{eq:tautauprime} gives 
 \begin{align}
    \left|  \langle\encoded{j}_{out}|W_{\Fgate}| \encoded{k}_{in}\rangle - \frac{\omega_d^{jk}}{\sqrt{d}} \right|
    &\leq 10 \tau^{1/4}+ 19\kappa^{1/4} +\sqrt{2}\sqrt{\kappa^2+16\tau^4}+8\tau\\
    &\leq 10 \tau^{1/4}+ 19\kappa^{1/4} +\sqrt{2} (\kappa+4\tau^2)+8\tau\\
    &\leq 24 \tau^{1/4} + 21\kappa^{1/4}\ ,\label{eq:mainclaimxvdf}
\end{align}
 where we introduced the quantity
 \begin{align}
 \tau:=\frac{\kappa}{2\pi d\varepsilon}
 \end{align} and used that~$\tau\leq 1$ by the assumption~\eqref{eq:assumptionpsmvxfff}.  This is the Claim~\eqref{eq:mjkvdadf}.

We note that for the optimal truncation parameter~$\varepsilon_d=1/(2d)$, the
assumption~\eqref{eq:assumptionpsmvx} is satisfied if~$(2\pi d)\Delta\leq \kappa<1/d^2$.
We have
\begin{align}
\tau&= \kappa/\pi 
\end{align}
in this case, hence Eq.~\eqref{eq:mainclaimxvdf} takes the form
 \begin{align}
    \left|  \langle\encoded{j}_{out}|W_{\Fgate}| \encoded{k}_{in}\rangle - \frac{\omega_d^{jk}}{\sqrt{d}} \right|
    &\leq 24 \left(\kappa/\pi\right)^{1/4} + 21 \kappa^{1/4}\\
    &\leq 40 \kappa^{1/4} \ .
\end{align}
\end{proof}

\section{No-go result for asymmetrically squeezed GKP codes\label{sec:nogoresultasymmetric}}
Here we extend our no-go result (Result~\ref{thm:result1}) to consider an
asymmetrically squeezed GKP code~$\gkpcode{\kappa,\Delta}{\star}{d}$, where~$\Delta\neq \kappa/(2\pi d)$. We show that when such an asymmetric code is used in conjunction with the Fourier transform, the linear optics implementation of the phase gate again has a logical gate error lower bounded by a constant.

In more detail, we consider squeezing parameters~$(\kappa,\Delta)$ which are polynomially small in~$d$. Concretely, assume that
\begin{alignat}{2}
\Delta &\leq \frac{1}{80d}&&\qquad\textrm{ and }\label{eq:deltasmallassumptions}\\
\kappa &\leq d^{-(6+\nu)}&&\qquad\textrm{ for some constant }\qquad \nu>0\ .\label{eq:kappasmallassumption}
\end{alignat}
To motivate these assumptions, observe that Eq.~\eqref{eq:deltasmallassumptions}
can be rewritten as~$\Delta<\varepsilon_d/40$ for the optimal squeezing parameter~$\varepsilon_d$.
As argued in Lemma~\ref{lem:converseoverlapgkp}, this kind of assumption is necessary to ensure that our truncated GKP-states are closed to the corresponding untruncated states. 

The exponent in assumption~\eqref{eq:kappasmallassumption} is chosen in accordance with our result about the Fourier transform (see Lemma~\ref{lem: gate error F}): 
Under the additional assumption
\begin{align}
\Delta &\leq\frac{\kappa}{2\pi d}\ ,\label{eq:additionalssumptionDeltakappatwopi}
\end{align}
the assumption~\eqref{eq:kappasmallassumption}
 ensures that 
 application of~$e^{i\frac{\pi}{4}(Q^2+P^2)}$  provides an accurate 
implementation of the (logical) Fourier transform~$\Fgate$ (for large~$d$) where
the input- and output-codes are
\begin{align}
\cLin&=\gkpcode{\kappa,\Delta}{\star}{d} \ , \\
\cLout&=\gkpcode{\kappa',\Delta'}{\star}{d} \ ,
\end{align}
with
\begin{align}
(\kappa',\Delta')=(2\pi d \Delta,\kappa/(2\pi d))\ .\label{eq:newparamsafterft}
\end{align}
Indeed, Lemma~\ref{lem: gate error F} implies  that -- 
under the assumptions~\eqref{eq:kappasmallassumption}  and~\eqref{eq:additionalssumptionDeltakappatwopi} --  the corresponding gate error is bounded as 
\begin{align}
\gateerror_{\gkpcode{\kappa,\Delta}{\star}{d},\gkpcode{\kappa',\Delta'}{\star}{d}}
(e^{i\frac{\pi}{4}(Q^2+P^2)},\Fgate)&\leq 21 d^{3/8}\kappa^{1/16} = 21 d^{-\nu/16}\, .
\end{align}

With this preparation, we can show that  standard linear optics implementation~$W_P$ of the phase gate~$\Pgate$ is no longer accurate when it is used in conjunction with linear optics implementation of the Fourier transform~$\Fgate$. More precisely, the Gaussian~$W_P$ fails to accurately implement the gate~$P$ for at least one of the two codes~$\cL_{in}$ or~$\cL_{out}$.
\begin{lemma}[No-go result for asymmetric squeezing]\label{lem:nogoasymmetric}
Let~$W_P = e^{i(Q^2 +c_d\sqrt{2\pi/d}Q)/2}$. Let~$d\geq 2$ be an integer.
Suppose that
\begin{alignat}{2}
\Delta &\leq \frac{1}{80d}&&\qquad\textrm{ and }\label{eq:deltaassumptionone}\\
\kappa &\leq d^{-(6+\nu)}&&\qquad\textrm{ for }\qquad \nu>0\ .\label{eq:kappasmallasumption}
\end{alignat}
Then 
\begin{align}
\max\left\{\gateerror_{\gkpcode{\kappa,\Delta}{\star}{d}}(W_P, \Pgate),\gateerror_{\gkpcode{\kappa',\Delta'}{\star}{d}}(W_P, \Pgate)\right\} &\geq 
\textfrac{1}{50}\ .
\end{align}
\end{lemma}
\begin{proof}

Consider the following two cases:
\begin{description}
\item[$2\pi d\Delta/\kappa\geq 1$]: In this case the linear optics implementation of the logical gate~$\Pgate$ fails to be accurate for the code~$\cL_{in}$: we have 
\begin{align}
     \gateerror_{\gkpcode{\kappa,\Delta}{\star}{d}}(e^{i(Q^2 +c_d\sqrt{2\pi/d}Q)/2}, \Pgate) &\geq \textfrac{1}{25}-32 (\Delta/\varepsilon_d)^2\  \\
     &\geq 1/50
\end{align}
according to Theorem~\ref{lem:nogoP}  (see Condition~\eqref{it:firstclaimnogoPone})  and the assumption~\eqref{eq:deltaassumptionone}. 
\item[$2\pi d\Delta/\kappa<1$]: In this case,
 application of~$e^{i\frac{\pi}{4}(Q^2+P^2)}$  provides an accurate 
implementation of the (logical) Fourier transform~$\Fgate$ as argued above (for large~$d$). 
Consider the resulting output code~$\cLout=\gkpcode{\kappa',\Delta'}{\star}{d}$.  We argue that the unitary~$e^{i(Q^2 +c_d\sqrt{2\pi/d}Q)/2}$   does not constitute an accurate implementation of the gate~$\Pgate$ for this new code~$\cLout$. Indeed,  Eq.~\eqref{eq:newparamsafterft} and our assumption~$2\pi d\Delta/\kappa<1$ imply that
\begin{align}
2\pi d\Delta' /\kappa'&=\left(2\pi d\Delta/\kappa\right)^{-1}\\
&>1\ .\label{eq:lowerboundasymafterft}
\end{align}
We conclude from Eq.~\eqref{eq:lowerboundasymafterft} and  Theorem~\ref{lem:nogoP}
 (see Condition~\eqref{it:firstclaimnogoPone}) that
 \begin{align}
      \gateerror_{\gkpcode{\kappa',\Delta'}{\varepsilon_d}{d}}(e^{i(Q^2 +c_d\sqrt{2\pi/d}Q)/2}, \Pgate) &\geq \textfrac{1}{25}-32 (\Delta'/\varepsilon_d)^2\\
      &= \textfrac{1}{25}-32 (\kappa/(2\pi d\varepsilon_d))^2\\
      &= \textfrac{1}{25}-8\kappa^2\\
      &\geq \textfrac{1}{25}-8 d^{-2(6+\nu)} \\
      & \geq 1/50 \ ,
 \end{align}
 where we used~$d \geq 2$.
\end{description}

\end{proof}

\section{No-go result for the Gaussian implementation of the $\CZ$ gate\label{sec:nogoresultCZ}}

In this section, we prove a no-go result for the two-mode Gaussian implementation~$e^{iQ_1Q_2}$ of the logical~$\CZ$ gate, analogous to the no-go result for the~$\Pgate$ gate established in Theorem~\ref{lem:nogoP}. We show that the logical gate error of the standard implementation remains bounded away from zero in the limit~$\kappa\rightarrow 0$ of infinite squeezing.

We consider two copies of the single-mode GKP code, i.e., the code space~$\cL=\gkpcode{\kappa,\Delta}{\varepsilon}{d}^{\otimes 2}$  with orthonormal basis
\begin{align}
    \ket{\encoded{j_1,j_2}} := \ket{\gkp_{\kappa,\Delta}^{\varepsilon}(j_1)_d}\otimes\ket{\gkp_{\kappa,\Delta}^{\varepsilon}(j_2)_d} \qquad \text{for} \qquad (j_1,j_2)\in\mathbb{Z}_d^2\ .\label{eq:onbbasisstatestwocopy}
\end{align}

\begin{lemma}[No-go result for Gaussian implementation of $\CZ$]\label{lem:nogoCZ}
Let~$d\geq 2$ be an integer. Let~$\varepsilon\in(0,1/(2d)]$ and~$\kappa<1/250$ be arbitrary. Then the following holds for the logical gate~$\CZ$.
\begin{enumerate}[(i)]
\item\label{it:firstclaimnogoCZone} We have
\begin{align}
    \gateerror_{\gkpcode{\kappa,\Delta}{\varepsilon}{d}^{\otimes 2}}(e^{iQ_1Q_2}, \CZ) \geq \textfrac{1}{25} - 64(\Delta/\varepsilon)^2
\end{align}
for all~$\Delta$ satisfying
\begin{align} 1/2 > 
 2\pi d\Delta\geq \kappa\ .\label{eq:pideltakappaCZ}
\end{align}
\item\label{it:firstclaimnogoCZtwo} For two copies of the symmetrically squeezed GKP code~$\gkpcode{\kappa}{\star}{d}$ we have
\begin{align}
    \gateerror_{\gkpcode{\kappa}{\star}{d}^{\otimes 2}}(e^{iQ_1Q_2}, \CZ) \geq \textfrac{3}{100}\ .
\end{align}
\end{enumerate}
\end{lemma}

\begin{proof}
Claim~\eqref{it:firstclaimnogoCZtwo} follows from Claim~\eqref{it:firstclaimnogoCZone} by substituting~$\Delta=\kappa/(2\pi d)$ (which satisfies the second bound in Eq.~\eqref{eq:pideltakappaCZ} with equality) and~$\varepsilon=\varepsilon_d=1/(2d)$, exactly as in the proof of Theorem~\ref{lem:nogoP}.

It remains to prove Claim~\eqref{it:firstclaimnogoCZone}. Define~$\cL=\cL_{in}=\cL_{out}=\gkpcode{\kappa,\Delta}{\varepsilon}{d}^{\otimes 2}$.
The logical~$\CZ$ gate has matrix elements~$\CZ_{(j_1,j_2),(k_1,k_2)}=\delta_{j_1,k_1}\delta_{j_2,k_2}\cdot\omega_d^{j_1 j_2}$ and its (approximate) implementation is
\begin{align}
    W_{\CZ} = e^{iQ_1Q_2}\ .
\end{align}
The matrix elements of the operator~$B$ in Definition~\ref{def:diagonalunitary} are given by
\begin{align}
B_{(j_1,j_2),(k_1,k_2)}&=\sum_{(m_1,m_2)\in\mathbb{Z}_d^2}\overline{\CZ_{(m_1,m_2),(j_1,j_2)}}\langle\encoded{m_1,m_2}|W_{\CZ}|\encoded{k_1,k_2}\rangle\\
&= \omega_d^{-j_1j_2}\langle\encoded{j_1,j_2}|e^{iQ_1Q_2}|\encoded{k_1,k_2}\rangle\ .
\end{align}
In particular, the diagonal entry at~$(j_1,j_2)=(0,0)$ is
\begin{align}
B_{(0,0),(0,0)}&= \langle\encoded{0,0}|e^{iQ_1Q_2}|\encoded{0,0}\rangle\ .\label{eq:B00CZdef}
\end{align}
We now compute this matrix element explicitly. By definition of the  basis states of the product code (see Eq.~\eqref{eq:onbbasisstatestwocopy})  and the GKP code basis states (see Section~\ref{sec:truncatedapproximategkpcode}), we have
\begin{align}
    \ket{\encoded{0,0}} = M_{\sqrt{2\pi d}}\ket{\gkp_{\kappa,\Delta}^\varepsilon}\otimes M_{\sqrt{2\pi d}}\ket{\gkp_{\kappa,\Delta}^\varepsilon}\ ,
\end{align}
where~$M_{\sqrt{2\pi d}}$ denotes the squeezing operator satisfying~$M_{\sqrt{2\pi d}}^\dagger Q M_{\sqrt{2\pi d}}=\sqrt{2\pi d}\,Q$. Using the identity
\begin{align}
    (M_{\sqrt{2\pi d}}\otimes M_{\sqrt{2\pi d}})^\dagger e^{iQ_1Q_2}(M_{\sqrt{2\pi d}}\otimes M_{\sqrt{2\pi d}}) = e^{2\pi i d\,Q_1Q_2}\ ,
\end{align}
we obtain
\begin{align}
B_{(0,0),(0,0)} &= \bigl\langle\gkp_{\kappa,\Delta}^\varepsilon\otimes\gkp_{\kappa,\Delta}^\varepsilon,\,e^{2\pi i d\,Q_1Q_2}\,\gkp_{\kappa,\Delta}^\varepsilon\otimes\gkp_{\kappa,\Delta}^\varepsilon\bigr\rangle\ .\label{eq:B00CZexplicit}
\end{align}
Since~$\gkp_{\kappa,\Delta}^\varepsilon$ is supported on~$\mathbb{Z}(\varepsilon)=\{x\in\mathbb{R}:\mathsf{dist}(x,\mathbb{Z})\leq\varepsilon\}$, we can expand this matrix element using the definition of~$\gkp_{\kappa,\Delta}^\varepsilon$ (see Eq.~\eqref{eq:GKP_eps}) as
\begin{align}
B_{(0,0),(0,0)} &= C_\kappa^2\sum_{y, y' \in\mathbb{Z}}\eta_\kappa(y) \eta_\kappa(y')\,\langle\chi_\Delta^\varepsilon(y)\otimes\gkp_{\kappa,\Delta}^\varepsilon,\,e^{2\pi id\,Q_1Q_2}\,\chi_\Delta^\varepsilon(y')\otimes\gkp_{\kappa,\Delta}^\varepsilon\rangle \\
&= C_\kappa^2\sum_{y\in\mathbb{Z}}\eta_\kappa(y)^2\,\langle\chi_\Delta^\varepsilon(y)\otimes\gkp_{\kappa,\Delta}^\varepsilon,\,e^{2\pi id\,Q_1Q_2}\,\chi_\Delta^\varepsilon(y)\otimes\gkp_{\kappa,\Delta}^\varepsilon\rangle\\
&=: C_\kappa^2\sum_{y\in\mathbb{Z}}\eta_\kappa(y)^2 I_{\Delta}^\varepsilon(y)\ ,\label{eq:B00CZexpand}
\end{align}
where we used that the functions~$\{\chi_\Delta^\varepsilon(y)\}_{y\in\mathbb{Z}}$ have pairwise disjoint support and the operator~$e^{2\pi id\,Q_1Q_2}$ does not change the support of a function in the first mode. Therefore, cross terms in the double sum over~$y,y'$ vanish. In the last line, we introduced~$I_{\Delta}^\varepsilon(y) = \langle\chi_\Delta^\varepsilon(y)\otimes\gkp_{\kappa,\Delta}^\varepsilon,\,e^{2\pi id\,Q_1Q_2}\,\chi_\Delta^\varepsilon(y)\otimes\gkp_{\kappa,\Delta}^\varepsilon\rangle$ for $y \in \mathbb{Z}$.

We now bound each summand in Eq.~\eqref{eq:B00CZexpand}. Since the operator~$e^{2\pi id\,Q_1Q_2}$ acts as diagonally (i.e., as a phase factor) in position-space, the term in $I_{\Delta}^\varepsilon(y)$ for fixed~$y\in\mathbb{Z}$ can be expressed as
\begin{align}
    I_\Delta^\varepsilon(y) &= \iint |\chi_\Delta^\varepsilon(y)(x_1)|^2\,|\gkp_{\kappa,\Delta}^\varepsilon(x_2)|^2\, e^{2\pi id\,x_1 x_2}\,dx_1\,dx_2\\
    &= \int |\gkp_{\kappa,\Delta}^\varepsilon(x_2)|^2 e^{2\pi idy x_2} \left(\int |\chi_\Delta^\varepsilon(y)(y+t)|^2 e^{2\pi idt x_2}\,dt\right)dx_2\\
& = \int |\gkp_{\kappa,\Delta}^\varepsilon(x_2)|^2 e^{2\pi idy x_2} \left(\int |\Psi_\Delta^\varepsilon(t)|^2 e^{2\pi idt x_2}\,dt\right)dx_2\ .\label{eq:B00CZsummand}
\end{align} where we substituted~$x_1 = y + t$ and used that~$\chi_\Delta^\varepsilon(y)(y+t)=\Psi_\Delta^\varepsilon(t)$.
Taking the absolute value and applying the triangle inequality, we bound
\begin{align}
\left| I_\Delta^\varepsilon(y) \right| &\leq \int |\gkp_{\kappa,\Delta}^\varepsilon(x_2)|^2 \left|\int |\Psi_\Delta^\varepsilon(t)|^2 e^{2\pi idt x_2}\,dt\right|dx_2 \\
&= \int |\gkp_{\kappa,\Delta}^\varepsilon(x_2)|^2 \left| \langle \Psi_\Delta^\varepsilon, e^{2\pi i d x_2 Q}\Psi_\Delta^\varepsilon \rangle \right| dx_2\, .\label{eq:B00CZabs}
\end{align}

We bound the inner product~$|\langle\Psi_\Delta^\varepsilon, e^{2\pi i d x_2 Q}\Psi_\Delta^\varepsilon\rangle|$ for all~$x_2\in\mathbb{R}$ by comparing to the untruncated Gaussian using Lemma~\ref{lem:continuitymatrixelem}, obtaining
\begin{align}
    |\langle\Psi_\Delta^\varepsilon, e^{2\pi idx_2 Q}\Psi_\Delta^\varepsilon\rangle| \leq |\langle\Psi_\Delta, e^{2\pi idx_2 Q}\Psi_\Delta\rangle| + 2\|\Psi_\Delta - \Psi_\Delta^\varepsilon\|\ .
\end{align}
The untruncated inner product is given by the Fourier transform of a Gaussian (cf.\ the proof of Lemma~\ref{lem:momentumtranslatedgaussian}) and gives $\langle\Psi_\Delta, e^{2\pi idx_2 Q}\Psi_\Delta\rangle = e^{-\pi^2\Delta^2 d^2 x_2^2}$. Together with Lemma~\ref{lem:closenessgaussiantrunc}~\eqref{it:fourthclosenessgaussiantrunc}, this gives
\begin{align}
    |\langle\Psi_\Delta^\varepsilon, e^{2\pi idx_2 Q}\Psi_\Delta^\varepsilon\rangle| \leq e^{-\pi^2\Delta^2 d^2 x_2^2} + 16(\Delta/\varepsilon)^4\qquad\textrm{ for all }\qquad x_2\in\mathbb{R}\ .\label{eq:B00CZphasebnd}
\end{align}
Substituting Eq.~\eqref{eq:B00CZphasebnd} into Eq.~\eqref{eq:B00CZabs}, we obtain
\begin{align}
    I_\Delta^\varepsilon(y) \leq \int |\gkp_{\kappa,\Delta}^\varepsilon(x_2)|^2 e^{-\pi^2\Delta^2 d^2 x_2^2}\,dx_2 + 16(\Delta/\varepsilon)^4\ . 
\end{align}
Therefore, using~$C_\kappa^2\sum_{y\in\mathbb{Z}}\eta_\kappa(y)^2=1$, we have
\begin{align} 
    |B_{(0,0),(0,0)}| \le \int |\gkp_{\kappa,\Delta}^\varepsilon(x_2)|^2 e^{-\pi^2\Delta^2 d^2 x_2^2}\,dx_2 + 16(\Delta/\varepsilon)^4\, . 
    \label{eq: B00 CZ expand}
\end{align}

Expanding~$|\gkp_{\kappa,\Delta}^\varepsilon(x_2)|^2 = C_\kappa^2\sum_{z\in\mathbb{Z}}\eta_\kappa(z)^2|\Psi_\Delta^\varepsilon(x_2-z)|^2$ and substituting~$t=x_2-z$ gives
\begin{align}
    \int |\gkp_{\kappa,\Delta}^\varepsilon(x_2)|^2 e^{-\pi^2\Delta^2 d^2 x_2^2}\,dx_2 &= C_\kappa^2\sum_{z\in\mathbb{Z}}\eta_\kappa(z)^2 \int |\Psi_\Delta^\varepsilon(t)|^2 e^{-\pi^2\Delta^2 d^2(z+t)^2}\,dt \\
    &= C_\kappa^2\sum_{z\in\mathbb{Z}}\eta_\kappa(z)^2 \langle \Psi_\Delta^\varepsilon, e^{-\pi^2\Delta^2 d^2(z+Q)^2} \Psi_\Delta^\varepsilon \rangle\ .\label{eq: B00 CZ expand2}
\end{align}
To bound the remaining integral, note that~$A_z := e^{-\pi^2\Delta^2 d^2(z+Q)^2}$ is a bounded positive operator for all $z \in \mathbb{Z}$ with~$\|A_z\|\leq 1$ since $0 \le e^{-\pi^2\Delta^2 d^2(z+x)^2} \leq 1$ for all $x \in \mathbb{R}$. Hence by the same argument as in the proof of Lemma~\ref{lem:continuitymatrixelem} we have
\begin{align}
    \langle \Psi_\Delta^\varepsilon, A_z \Psi_\Delta^\varepsilon \rangle \leq \langle \Psi_\Delta, A_z \Psi_\Delta \rangle + 2\|\Psi_\Delta - \Psi_\Delta^\varepsilon\|\ . \label{eq: A_z continuity}
\end{align}
The inner product over untruncated Gaussians evaluates to
\begin{align}
    \langle \Psi_\Delta, A_z \Psi_\Delta \rangle = \int \frac{e^{-t^2/\Delta^2}}{\sqrt{\pi}\Delta}\, e^{-\pi^2\Delta^2 d^2(n+t)^2}\,dt = \frac{e^{-\lambda^2 n^2}}{\sqrt{1+\pi^2\Delta^4 d^2}} \leq e^{-\lambda^2 n^2}\ ,\label{eq: untruncated Gaussian int}
\end{align}
where~$\lambda^2 = \pi^2 d^2\Delta^2/(1+\pi^2\Delta^4 d^2)$. Combining Eqs.~\eqref{eq: B00 CZ expand}, ~\eqref{eq: A_z continuity} and~\eqref{eq: untruncated Gaussian int} with Lemma~\ref{lem:closenessgaussiantrunc}~\eqref{it:fourthclosenessgaussiantrunc} and~$C_\kappa^2\sum_{z\in\mathbb{Z}}\eta_\kappa(z)^2=1$, we obtain
\begin{align}
    |B_{(0,0),(0,0)}| \leq C_\kappa^2\sum_{z\in\mathbb{Z}}\eta_\kappa(z)^2 e^{-\lambda^2 z^2} + 32(\Delta/\varepsilon)^4\ .\label{eq: B00 CZ intermediate}
\end{align}
Applying Lemma~\ref{lem:gaussiansumsexprva} with~$\lambda = \pi d\Delta/\sqrt{1+\pi^2 d^2\Delta^4}$ (noting that~$\lambda \leq \pi d\Delta < 1/4$ under assumption~\eqref{eq:pideltakappaCZ} and~$\kappa < 1/250 < 1/4$) gives
\begin{align}
    C_\kappa^2 \sum_{y\in\mathbb{Z}} \eta_\kappa(y)^2 e^{-\lambda^2 y^2} \leq \frac{1}{\sqrt{1+(\lambda/\kappa)^2}}+\kappa \leq \frac{1}{\sqrt{1+\pi^2(d\Delta/\kappa)^2/4}}+\kappa\, ,
\end{align}
where we used that under the assumption~\eqref{eq:pideltakappaCZ} \begin{align}
    \lambda = \pi d\Delta/\sqrt{1+ \pi^2 d^2 \Delta^4} \ge \pi d\Delta/\sqrt{1 + \pi^2 d^2 \Delta^2} \ge \pi d\Delta/\sqrt{1 + 1/4} \ge \pi d \Delta /2 \, .
\end{align} Moreover, we have~$\pi d\Delta/\kappa \ge 1/2$, so~$\pi^2(d\Delta/\kappa)^2/4 \ge 1/16$, hence
\begin{align}
    \frac{1}{\sqrt{1+\pi^2(d\Delta/\kappa)^2/4}}+\kappa \leq \frac{1}{\sqrt{17/16}}+\kappa < \frac{49}{50}\qquad\textrm{for all }\kappa<\frac{1}{250}\ ,
\end{align}
since $1/\sqrt{17/16}  < 49/50 - 1/250$. Combining with the Eq.~\eqref{eq: B00 CZ intermediate} gives
\begin{align}
    \left|B_{(0,0),(0,0)}\right| \leq \frac{49}{50}+32(\Delta/\varepsilon)^4\ .\label{eq: B00 CZ final}
\end{align}
By Lemma~\ref{lem:lowerboundgateerrorbzerozero} we have
\begin{align}
    \gateerror_\cL(W_{\CZ},\CZ)&\geq 2\sqrt{1-|B_{(0,0),(0,0)}|^2} \geq 2\left(1-|B_{(0,0),(0,0)}|\right)\ ,\label{eq: gaterror B00 CZ}
\end{align}
where we used~$\sqrt{1-x^2}\geq 1-x$ for~$x\in[0,1]$. Combining Eqs.~\eqref{eq: B00 CZ final} and~\eqref{eq: gaterror B00 CZ} gives the claim.
\end{proof}

\section{Continuity of the (subspace) logical gate error \label{sec:continuity}}
In this section, we show that if a logical unitary~$U$ is approximately implemented by a (physical) unitary~$W$ on a code space~$\mathcal{L} \subset \mathcal{H}$, then the same holds for any subspace~$\widetilde{\mathcal{L}}$ that is close to~$\mathcal{L}$. 
We proceed as follows. 
In Section~\ref{sec:continuitylogicalerrorone}, we establish a continuity bound for the logical gate error. 
In Section~\ref{sec:subspacegateerror}, we introduce a quantity we call the subspace logical gate error. We then show that it satisfies a similar continuity property.

\subsection{Continuity of the logical gate error~\label{sec:continuitylogicalerrorone}}
We show the following.
\newcommand*{\tvarphi}{\tilde{\varphi}}
\begin{lemma}[Continuity of the logical gate error]
    \label{lem:codecontinuity}
    Let~$d\ge 2$ be an integer. Let~$\cH$ be a Hilbert space and let~$\{\varphi_j\}_{j=0}^{d-1}\,,  \{\tvarphi_j\}_{j=0}^{d-1} \subset \cH$ be two families of orthonormal vectors. 
        Define two subspaces~$\cL, \widetilde{\cL} \subset \cH$ by
    \begin{align}
        \cL = \mathrm{span}\{\varphi_j\}_{j=0}^{d-1} \qquad \textrm{and} \qquad \widetilde{\cL}  = \mathrm{span} \{\tvarphi_j\}_{j=0}^{d-1}
    \end{align}
    and corresponding isometric encoding maps~$\encmap_{\cL}: \mathbb{C}^d \rightarrow \cL$ and~$\encmap_{\widetilde{\cL}}: \mathbb{C}^d \rightarrow \widetilde{\cL}$ by 
    \begin{align}
    \begin{matrix}
        \encmap_{\cL}\ket{j} &= &\ket{\varphi_j}\\
        \encmap_{\widetilde{\cL}}\ket{j} &= &\ket{\tvarphi_j}
        \end{matrix}\qquad\textrm{ for }\qquad j\in \{0,\ldots,d-1\}\ .
    \end{align}
    Let~$\decmap_{\cL}:=(\encmap_{\cL})^{-1}$ and~$\decmap_{\widetilde{\cL}}:=(\encmap_{\widetilde{\cL}})^{-1}$ be the associated inverse maps.
    Let~$U: \mathbb{C}^d \rightarrow \mathbb{C}^d$ be a unitary. Define the maps~$\encoded{U}= \encmap_{\cL} U \decmap_{\cL}: \cL \rightarrow \cL$ and~$\widetilde{\encoded{U}}= \encmap_{\widetilde{\cL}} U \decmap_{\widetilde{\cL}}: \widetilde{\cL} \rightarrow \widetilde{\cL}$. 
    Let~$W: \cH \rightarrow \cH$ be a unitary and let~$B=B^{U}_{\cL,\cL}(W,U)$ and~$\widetilde{B}=B^{U}_{\widetilde{\cL},\widetilde{\cL}}(W,U)$ be the corresponding operators introduced in Definition~\ref{def:diagonalunitary}.
    Then 
    \begin{align} 
        \left|\cn(B) - \cn(\widetilde{B})\right| \le 2d \delta\, ,  
    \end{align} 
    where
\begin{align}
\delta:= \max_{j\in \{0,\ldots,d-1\}}\|\varphi_j - \tvarphi_j\|\ .
\end{align}
   In particular, we have
   \begin{align}
\left|\gateerror_{\cL}(W,U)-\gateerror_{\widetilde{\cL}}(W,U)\right| \leq 4\sqrt{d\delta}\  .
\label{eq:mainclaimcomparisonbond}
\end{align}

\end{lemma}
\begin{proof}
       According to Lemma~\ref{lem:continuitycrawford} we have
        \begin{align} \label{eq:cnnormbound}
            |\cn(B) - \cn(\widetilde{B})| \le \|B - \widetilde{B}\|\le \|B - \widetilde{B}\|_2\, .
        \end{align}
        Fix~$j,k \in \{0,\dots,d-1\}$. By definition of~$B$ and~$\widetilde{B}$ we have 
        \begin{align}
            \left|B_{j,k} -  \widetilde{B}_{j,k}\right|^2 &= \left|\sum_{m=0}^{d-1} \overline{U_{m,j}}  \left(\langle \varphi_m, W \varphi_k\rangle -  \langle \tvarphi_m, W \tvarphi_k\rangle\right)\right|^2 \\
            &\le \left(\sum_{m=0}^{d-1} \left|U_{m,j}\right|^2\right)   \cdot \left(\sum_{m=0}^{d-1}  \left|\langle \varphi_m, W \varphi_k\rangle -  \langle \tvarphi_m, W \tvarphi_k\rangle\right|^2\right) \\
            &= \sum_{m=0}^{d-1} \left|\langle \varphi_m, W \varphi_k\rangle -  \langle \tvarphi_m, W \tvarphi_k\rangle\right|^2  \ , \label{eq:boundbound1}
        \end{align} for all~$j,k \in \{0,\dots,d-1\}$. Here we used the triangle inequality and the Cauchy-Schwarz inequality. The last inequality follows from the assumption that~$U$ is unitary, which implies that~$\sum_{m=0}^{d-1}|U_{m,j}|^2= 1$ for all~$j \in \{0,\dots,d-1\}$.
        Note that we have 
        \begin{align} 
            \left|\langle \varphi_m, W \varphi_k\rangle -  \langle \tvarphi_m, W \tvarphi_k\rangle\right|   &\le \| \varphi_k - \tvarphi_k \| + \| \varphi_m - \tvarphi_m \| \\
                &\le 2 \delta \label{eq:boundbound2}
        \end{align} for all~$m \in \{0,\dots,d-1\}$
        by Lemma~\ref{lem:triangle_ineq} and the definition of~$\delta$.   Combining Eqs.~\eqref{eq:boundbound1} and~\eqref{eq:boundbound2} we find that
        \begin{align}
            \|B - \widetilde{B}\|_2 \le 2d\cdot\delta\, ,
        \end{align}
        and thus 
        \begin{align}
            \left|\cn(B) - \cn(\widetilde{B})\right| \le 2d\cdot\delta \, \label{eq:cnupperboundmvacvz}
        \end{align}
        by Eq.~\eqref{eq:cnnormbound}.
        Finally, using \ref{lem:gateerrorinnernumericalradius},
        we have with~$\tau:=\cn(B)-\cn(\widetilde{B})$
        \begin{align}
        \gateerror_{\widetilde{\cL}}(W,U)
        &=2 \sqrt{1-c(\widetilde{B})^2}\\
        &= 2\sqrt{1-(c(B)-\tau)^2}\\
        &= 2\sqrt{1-c(B)^2+2\tau c(B)-\tau^2} \\
        &\leq 2\sqrt{1-c(B)^2+2\tau c(B)}\\
        &\leq 2\sqrt{1-c(B)^2}+2\sqrt{2}\sqrt{\tau}\\
        &\leq 2\sqrt{1-c(B)^2}+4\sqrt{d\delta }\label{eq:mvd5sqr}
        \end{align}
        where we used the inequality~$\sqrt{a+b}\leq \sqrt{a}+\sqrt{b}$ for~$a,b\geq 0$, the fact that~$c(B)\leq 1$ because~$W$ is unitary, and Eq.~\eqref{eq:cnupperboundmvacvz}.  
        According to
 Lemma~\ref{lem:gateerrorinnernumericalradius}, we also have 
 \begin{align}
 \sqrt{1-c(B)^2} &= \frac{1}{2}\gateerror_{\cL}(W,U)\ .\label{eq:gazvdmwu}
 \end{align}
 Combining Eqs.~\eqref{eq:mvd5sqr} and~\eqref{eq:gazvdmwu} implies 
  \begin{align}
          \gateerror_{\widetilde{\cL}}(W,U)
        &\leq \gateerror_{\cL}(W,U)+4\sqrt{d\delta}\  . \label{eq:claimvonecbsvqm}
 \end{align}
 Interchanging the roles of~$\cL$ and~$\widetilde{\cL}$, we obtain 
    \begin{align}
          \gateerror_{\cL}(W,U)
        &\leq \gateerror_{\widetilde{\cL}}(W,U)+4\sqrt{d\delta}\  . \label{eq:claimvonecbsvqmsec}
 \end{align}
The Claim~\eqref{eq:mainclaimcomparisonbond}
follows by combining Eqs.~\eqref{eq:claimvonecbsvqm} and~\eqref{eq:claimvonecbsvqmsec}.
   \end{proof}

The following result establishes that if a family of vectors is close to an orthonormal basis, then we can find an orthogonalization which preserves closeness. 
More precisely, we have the following.
\begin{lemma} \label{lem:symorthogonalization}
    Let~$d\ge 2$ be an integer. Let~$\cH$ be a Hilbert space and let~$\{\phi_j\}_{j=0}^{d-1} \subset \cH$ be a family of orthonormal vectors.
    Let~$\{\psi_j\}_{j=0}^{d-1} \subset \cH$ be a set of linearly independent normalized vectors.     Then there is an orthonormal system~$\{\xi_j\}_{j=0}^{d-1} \subset\cH$ which spans the same subspace as~$\{\psi_j\}_{j=0}^{d-1}$ and satisfies
    \begin{align} \label{eq:orthgeneral}
        \max_{k \in \{0,\dots,d-1\}} \|\xi_k - \phi_k\| \le 2\sqrt{d}\max_{j\in \{0,\dots,d-1\}}\|\psi_j - \phi_j\| \, .
    \end{align}
\end{lemma}

\begin{proof}
    The key ingredient of the proof is the symmetric orthogonalization process, see~\cite{Lowdin1950, Mayer2003}, which can be regarded as a polar decomposition. For completeness, we give a detailed proof.
Define~$\cal{V} := \mathsf{span} \{ \{\psi_j\}_{j=0}^{d-1} \cup \{\phi_j\}_{j=0}^{d-1} \}$. Then~$n := \dim \cal{V} \geq d$.
Using an orthonormal basis of~$\cV$, 
let us represent the families of vectors as columns of~$n \times d$ matrices where 

\begin{align}
    \begin{matrix}
        A_\phi &=& \left(\phi_0, \dots, \phi_{d-1}\right)\\
        A_\psi &=& \left(\psi_0, \dots, \psi_{d-1}\right)\\
    \end{matrix}\ .
\end{align}
Note that by construction~$A_\phi: \bbC^d \rightarrow \cal{V}$ is an isometry, that is, $A_\phi^\dagger A_\phi = I_d$ where~$I_d$ is the identity on~$\mathbb{C}^d$.
Let~$G = A_\psi^\dagger A_\psi$.
Define the matrix 
\begin{align}
    A_\xi = A_\psi G^{-1/2}\, .
\end{align}
Then~$A_{\xi}: \bbC^d \rightarrow \cal{V}$ is an isometry as 
\begin{align}
    A_\xi^\dagger A_\xi &= \left(G^{-1/2}\right)^\dagger A_\psi^\dagger A_\psi G^{-1/2}  \\
                        &= G^{-1/2} G G^{-1/2}\\
                        &= I_d\, .
\end{align}  That is, $A_\psi = A_\xi G^{1/2}$ is the polar decomposition of~$A_\psi$.
Define~$\xi_k$ as the~$k$-th column of~$A_\xi$ for~$k \in \{0,\dots,d-1\}$, i.e., $A_\xi = (\xi_0,\ldots, \xi_{d-1})$.
Our goal is to bound the Frobenius norm~$\|A_\phi - A_\xi\|_2$. This implies a bound on the quantity~$\max_{k \in \{0,\dots, d-1\}}\|\psi_k - \xi_k\|$ as we have
\begin{align}
    \| \phi_k - \xi_k \| \le \left( \sum_{\ell=0}^{d-1} \|\phi_\ell - \xi_\ell\|^2 \right)^{1/2} = \left\| A_\phi - A_\xi\right\|_2 \qquad \textrm{for all} \qquad k \in \{0,\dots,d-1\}\, . \label{eq:boundphipsi}
\end{align}
Using the triangle inequality we have
\begin{align}
    \left\|A_\phi - A_\xi\right\|_2 \le \left\|A_\phi - A_\psi\right\|_2 + \left\|A_\psi - A_\xi\right\|_2\, . \label{eq:Aphipsibound}
\end{align}
We can bound the first term as 
\begin{align}
    \left\|A_\phi - A_\psi\right\|_2 &= \left( \sum_{k=0}^{d-1} \|\phi_k - \psi_k\|^2 \right)^{1/2} \le \sqrt{d}
    \max_{j\in \{0,\dots,d-1\}}\|\psi_j - \phi_j\|\ .\label{eq:boundDelta}
\end{align}
In order to bound the second term, we claim that 
\begin{align}
    \min_{\substack{U: \mathbb{C}^d \rightarrow \cal{V} \\ U^\dagger U = I_d}} \left\|U - A_\psi\right\|_2 = \left\|A_\xi - A_\psi\right\|_2\, , \label{eq:isoboundclaim}
\end{align}
that is, $A_\xi$ is the closest isometry to~$A_\psi$.
\begin{proof}[Proof of Eq.~\eqref{eq:isoboundclaim}.]
Let~$U: \mathbb{C}^d \rightarrow  \cal{V}$ be an isometry. Then 
\begin{align}
    \left\|U - A_\psi\right\|^2_2 &= \left\|U - A_\xi G^{1/2}\right\|^2_2\\
    &= \left\|A_\xi^\dagger U - G^{1/2}\right\|^2_2\\
    &=  \tr\left(\left(A_\xi^\dagger U\right)^\dagger A_\xi^\dagger U\right) - 2 \mathsf{Re}\, \tr\left(\left(A_\xi^\dagger U\right)^\dagger G^{1/2} \right) + \tr\left( G \right)\\
    &= d - 2\mathsf{Re}\, \tr\left(U^\dagger A_\xi G^{1/2}\right) +  \tr\left(G\right)\, , \label{eq:isobound1}
\end{align}
where we used the invariance of the Frobenius norm under isometries in the second step and~\[\left(A_\xi^\dagger U\right)^\dagger A_\xi^\dagger U = I_d\] in the third step. 
The~$U$-dependent term can be bounded as 
\begin{align}
    \mathsf{Re}\, \tr\left(U^\dagger A_\xi G^{1/2}\right) &\le \left| \tr\left(U^\dagger A_\xi G^{1/2}\right) \right|\\
    &\le \left\| U^\dagger A_\xi\right\| \left\| G^{1/2}\right\|_1 \\
    &= \tr\left( G^{1/2}\right)\, , \label{eq:isobound2}
\end{align}
where the second inequality follows from Hölder's inequality. The last identity follows from the unitarity of~$U^\dagger A_\xi$ and the fact that~$G^{1/2}$ is positive semidefinite.
Clearly, the inequalities are equalities for the choice~$U = A_\xi$.
Therefore combining Eqs.~\eqref{eq:isobound1} and~\eqref{eq:isobound2} we find that
\begin{align}
    \left\|U - A_\psi\right\|^2_2 \ge d - \tr\left(G^{1/2}\right) + \tr\left(G\right) = \left\|A_\xi - A_\psi\right\|^2_2\, .
\end{align}
\end{proof}
Using that~$A_\phi$ is an isometry it follows from Eqs.~\eqref{eq:isoboundclaim} and~\eqref{eq:boundDelta} that
\begin{align}
    \left\|A_\xi - A_\psi\right\|_2 \le \left\|A_\phi - A_\psi\right\|_2 \le \sqrt{d}    \max_{j\in \{0,\dots,d-1\}}\|\psi_j - \phi_j\|\  . \label{eq:boundAxipsi}
\end{align}
The Claim~\eqref{eq:orthgeneral} then follows by combining Eq.~\eqref{eq:boundphipsi} with Eqs.~\eqref{eq:boundDelta} and~\eqref{eq:boundAxipsi}. 

\end{proof}

\subsection{Definition and continuity of the subspace logical gate error\label{sec:subspacegateerror}}
In this section, we introduce a quantity called the subspace logical gate error. It is obtained by optimizing the logical gate error over all  encoding- and decoding maps. We then show that this quantity also satisfies a continuity property with respect to the choice of code space~$\cL$.

A logical unitary~$U$ on $\mathbb{C}^d$ is typically specified in terms its action on computational basis states. 
As discussed in Section~\ref{sec:gateerrordef} (see also Section~\ref{sec: approximate implementation of logical unitaries}), this means that for a given  code space~$\cL=\cLin=\cLout$, defining a corresponding logical unitary requires fixing bases of the code space, or, equivalently, a choice of two encoding maps  (see Eq.~\eqref{eq:enccmapsfreedomb}).
In other words, when $U$ is
specified in terms of matrix elements with respect to computational basis states, the logical gate error~$\gateerror_\cL(W,U)$ also depends on the choice of encoding and not only on the underlying code subspace~$\cL$. 

For a fixed physical unitary~$W: \cH \rightarrow \cH$, a $d$-dimensional code space~$\cL\subset\cH$ and a logical unitary~$U$  on~$\mathbb{C}^d$, it is therefore natural to consider the smallest logical gate error achievable when appropriately choosing the encoding maps. We define the corresponding quantity as
\begin{align}
\subspacegateerror{\cL}{W}& = \min_{\encmap_{\cL_{in}},\encmap_{\cL_{out}} \in \cE_{\cL}}\gateerror_{\cL}(\cW,\encmap_{\cL_{out}} U \encmap_{\cL_{in}}^{-1}) \ ,\label{eq:subspacegateerrorlogical}
\end{align}
where 
\begin{align}
\cE_{\cL} := \{\encmap_\cL: \mathbb{C}^d \rightarrow \cL \ |\  \encmap_\cL^\dagger \encmap_\cL = I_{\mathbb{C}^d}\}\, .
\end{align}
denotes the set of isometries from $\mathbb{C}^d$ to $\cL$.
We call the quantity~$\subspacegateerror{\cL}{W}$ the subspace logical gate error. Our notation emphasizes that the quantity only depends on the subspace~$\cL$ considered and the (physical) unitary~$W$: Indeed, the rhs.~of Eq.~\eqref{eq:subspacegateerrorlogical}
does not depend on~$U$  and the subspace logical gate error can be rewritten as 
\begin{align}
\subspacegateerror{\cL}{W}=\min_{\substack{V: \mathbb{C}^d \rightarrow \mathbb{C}^d\\V^\dagger V = I_{\mathbb{C}^d}}}\gateerror_{\cL}(W,\encmap^*_\cL V  (\encmap^*_{\cL})^{-1})\label{eq:subspacegateerrordef}
\end{align}
where $\encmap^*_\cL\in \cE_\cL$ is fixed (but arbitrary), and where the minimum is over all unitaries~$V:\mathbb{C}^d\rightarrow\mathbb{C}^d$. 
Equivalently, 
we have 
\begin{align}
\subspacegateerror{\cL}{W}\ = \min_{\substack{\encoded{V}: \cL \rightarrow \cL \\\encoded{V}^\dagger \encoded{V} = I_\cL}}\gateerror_{\cL}(W,\encoded{V})\ \label{eq:subspacegateerrordeftwo}
\end{align}
where the minimum is over all unitaries~$\encoded{V}:\cL\rightarrow\cL$. 
\begin{proof}
Assume that a physical unitary $W:\cH\rightarrow\cH$ and a logical unitary~$U:\mathbb{C}^d\rightarrow\mathbb{C}^d$  are given. Let $\cL_{in}\cong \cL_{out}\subset \cH$ be $d$-dimensional subspaces of~$\cH$. Furthermore, let   $\encmap^*_{\cL_{in}}\in\cE_{\cL_{in}}$ 
and $\encmap^*_{\cL_{out}}\in\cE_{\cL_{out}}$  be fixed isometric encoding maps.

 Then
any two encoding maps~$\encmap_{\cL_{in}}\in\cE_{\cL_{in}}$ 
and $\encmap_{\cL_{out}}\in\cE_{\cL_{out}}$ 
 can be written as~$\encmap_{\cL_{in}}=\encmap^*_{\cL_{in}} V_{\mathrm{in}}$ and~$\encmap_{\cL_{out}}=\encmap^*_{\cL_{out}} V_{\mathrm{out}}$ for unitaries~$V_{\mathrm{in}},V_{\mathrm{out}}:\mathbb{C}^d\rightarrow\mathbb{C}^d$.
It follows that~$\encmap_{\cL_{out}} U \encmap_{\cL_{in}}^{-1}=\encmap^*_{\cL_{out}} V_{\mathrm{out}} U V_{\mathrm{in}}^\dagger (\encmap^*_{\cL_{in}})^{-1}$. Since the map~$(V_{\mathrm{in}},V_{\mathrm{out}})\mapsto V_{\mathrm{out}} U V_{\mathrm{in}}^\dagger$ is surjective onto the set of unitaries on~$\mathbb{C}^d$, minimizing 
a function of~$\encmap_{\cL_{out}} U \encmap_{\cL_{in}}^{-1}$
over pairs~$(\encmap_{\cL_{in}},\encmap_{\cL_{out}})\in \cE_{\cL_{in}}\times\cE_{\cL_{out}}$ is equivalent to minimizing the same function evaluated on $\encmap^*_{\cL_{out}} V  (\encmap_{\cL_{in}}^*)^{-1}$
over all unitaries~$V$ on~$\mathbb{C}^d$.  
Considering $\cL_{in}=\cL_{out}=\cL$ and the same maps $\encmap^*_{\cL}=\encmap^*_{\cL_{in}}=\encmap^*_{\cL_{out}}$, this implies Eq.~\eqref{eq:subspacegateerrordef}. 

Eq.~\eqref{eq:subspacegateerrordeftwo} follows from Eq.~\eqref{eq:subspacegateerrordef} by observing that any unitary~$\encoded{V}:\cL\rightarrow\cL$ can be written as~$\encoded{V} = \encmap_\cL V \encmap_\cL^{-1}$ for the  unitary~$V = \encmap_{\cL}^{-1} \encoded{V} \encmap_\cL:\mathbb{C}^d\rightarrow\mathbb{C}^d$.
\end{proof}

In what follows, we establish two complementary results about the subspace logical gate error~$\subspacegateerror{\cL}{W}$. First, in Lemma~\ref{lem:subspaceequalityreduction}, we give a lower bound on~$\subspacegateerror{\cL}{W}$ in terms of the action of~$W$ on unit vectors in the code subspace~$\cL$. This bound provides a way of certifying that no encoding into~$\cL$ admits a vanishing logical gate error. Second, in Lemma~\ref{lem:subspacecontinuity}, we give a continuity bound on~$\subspacegateerror{\cL}{W}$ which compares two~$d$-dimensional code subspaces~$\cL$ and~$\widetilde{\cL}$ that are spanned by close families of vectors.

\begin{lemma}[Encoding-independent lower bound on the subspace gate error]\label{lem:subspaceequalityreduction}
Let~$d\geq 2$ be an integer. Let~$\cH$ be a Hilbert space and let~$\cL\subset\cH$ be a~$d$-dimensional subspace with associated orthogonal projection~$\pi_\cL:\cH\rightarrow\cL$. Let~$W:\cH\rightarrow\cH$ be a unitary. Then 
\begin{align}
\subspacegateerror{\cL}{W}\ \geq\ 2\left(1-\min_{\substack{\Psi\in\cL\\ \|\Psi\|=1}}\|\pi_\cL W \Psi\|^{2}\right)^{1/2}\ . \label{eq:subspaceequalityreductionclaim}
\end{align}
\end{lemma}
\begin{proof}
Let~$\encoded{U}:\cL\rightarrow\cL$ be an arbitrary unitary on $\cL$, and let~$B:\cL \rightarrow \cL$ be the operator introduced in Definition~\ref{def:diagonalunitary}. By Lemma~\ref{lem:gateerrorinnernumericalradius} we have 
\begin{align}
\gateerror_\cL(W,\encoded{U})\ =\ 2\sqrt{1-\cn(B)^2}\ , \label{eq:subspeqredstart}
\end{align}
where~$\cn(B)=\min_{\Psi\in \cL, \|\Psi\|=1}|\langle\Psi,B\Psi\rangle|$ denotes the Crawford number of~$B$, see Eq.~\eqref{eq:crawford}.

By definition of~$B$  we have
\begin{align}
\cn(B)\ =\ \min_{\substack{\Psi\in\cL\\ \|\Psi\|=1}}\bigl|\bigl\langle\encoded{U}\Psi,W\Psi\bigr\rangle\bigr|\, .\label{eq:subspeqredstartv}
\end{align}
Using that~$\encoded{U}\Psi\in\cL$ and~$\|\encoded{U}\Psi\|=1$ for every unit vector~$\Psi\in\cL$, the Cauchy-Schwarz inequality gives 
\begin{align}
\bigl|\bigl\langle\encoded{U}\Psi,W\Psi\bigr\rangle\bigr|\ =\ \bigl|\bigl\langle\encoded{U}\Psi,\pi_\cL W\Psi\bigr\rangle\bigr|\ \leq\ \|\pi_\cL W \Psi\|\, .\label{eq:subspeqcauchyschwarz}
\end{align} Combining Eqs.~\eqref{eq:subspeqredstartv} and~\eqref{eq:subspeqcauchyschwarz} yields 
\begin{align}
\cn(B)\ \leq\ \min_{\substack{\Psi\in\cL\\ \|\Psi\|=1}}\|\pi_\cL W \Psi\|\ . \label{eq:cnboundsubsp}
\end{align}
 Inserting Eq.~\eqref{eq:cnboundsubsp} into Eq.~\eqref{eq:subspeqredstart} gives 
\begin{align}
\gateerror_\cL(W,\encoded{U})\ \geq\ 2\left(1-\min_{\substack{\Psi\in\cL\\ \|\Psi\|=1}}\|\pi_\cL W \Psi\|^{2}\right)^{1/2}\ . \label{eq:subspeqgateerrlowerbound}
\end{align}
Since the rhs.~of Eq.~\eqref{eq:subspeqgateerrlowerbound} is independent of the logical unitary~$\encoded{U}$, taking the minimum over all unitaries~$\encoded{U}: \cL \rightarrow \cL$ on the lhs.~establishes Eq.~\eqref{eq:subspaceequalityreductionclaim}
because of Eq.~\eqref{eq:subspacegateerrordeftwo}.
\end{proof}
With this result, we can establish a continuity statement for the subspace logical gate error analogous to Lemma~\ref{lem:codecontinuity}. 
\begin{lemma}[Continuity of the subspace logical gate error]\label{lem:subspacecontinuity}
    Let~$d\ge 2$ be an integer. Let~$\cH$ be a Hilbert space and let~$\{\phi_j\}_{j=0}^{d-1}\subset\cH$ be a family of orthonormal vectors spanning~$\cL=\mathsf{span}\{\phi_j\}_{j=0}^{d-1}$. Let~$\{\widetilde{\phi}_j\}_{j=0}^{d-1}\subset\cH$ be a family of linearly independent normalized vectors spanning~$\widetilde{\cL}=\mathsf{span}\{\widetilde{\phi}_j\}_{j=0}^{d-1}$. Let~$W:\cH\rightarrow\cH$ be a unitary. Then
    \begin{align}
        \left|\subspacegateerror{\cL}{W}-\subspacegateerror{\widetilde{\cL}}{W}\right|\ \leq\ 6d^{3/4}\sqrt{\delta}\ ,\label{eq:subspacecontinuityclaim}
    \end{align}
    where
    \begin{align}
        \delta:=\max_{j\in\{0,\dots,d-1\}}\|\phi_j-\widetilde{\phi}_j\|\ .
    \end{align}
\end{lemma}

\begin{proof}
    
    Due to Lemma~\ref{lem:symorthogonalization} applied to~$\{\phi_j\}_{j=0}^{d-1}$ and~$\{\widetilde{\phi}_j\}_{j=0}^{d-1}$ (cf. Eq.~\eqref{eq:boundphipsi}), there is an orthonormal family~$\{\xi_j\}_{j=0}^{d-1}\subset\cH$ which spans~$\widetilde{\cL}$ and satisfies
    \begin{align}
        \eta:=\max_{j \in \{0,\dots, d-1\}}\|\xi_j-\phi_j\|\ \leq\ 2\sqrt{d}\,\delta\ .\label{eq:etaboundsubsp}
    \end{align}
    Define the isometric encodings~$\encmap_\cL:\mathbb{C}^d\rightarrow\cL$ and~$\encmap_{\widetilde{\cL}}:\mathbb{C}^d\rightarrow\widetilde{\cL}$ by
    \begin{align}
        \encmap_\cL\ket{j}=\ket{\phi_j}\ ,\qquad \encmap_{\widetilde{\cL}}\ket{j}=\ket{\xi_j}\qquad\textrm{for}\qquad j\in\{0,\dots,d-1\}\ .
    \end{align}
    
    Let~$V:\mathbb{C}^d\rightarrow\mathbb{C}^d$ be an arbitrary unitary. Applying Lemma~\ref{lem:codecontinuity} to the two orthonormal families~$\{\phi_j\}_{j=0}^{d-1}$ and~$\{\xi_j\}_{j=0}^{d-1}$ with target~$V$ and closeness parameter~$\eta$ yields
    \begin{align} \label{eq:uniformboundgateerrsubsp}
        \left|\gateerror_{\cL}(W,\encmap_\cL V \encmap_\cL^{-1})-\gateerror_{\widetilde{\cL}}(W,\encmap_{\widetilde{\cL}} V \encmap_{\widetilde{\cL}}^{-1})\right|\ \leq\ 4\sqrt{d\,\eta}\ \leq\ 6d^{3/4}\sqrt{\delta}\ .
    \end{align}
    We emphasize that this bound is independent of~$V$.
    
   Note that for any two non-negative functions~$f,g:\mathcal{X}\rightarrow [0,\infty)$ on a set~$\mathcal{X} \neq \emptyset$  we have 
    \begin{align} \left|\inf_{x\in \mathcal{X}} f(x) - \inf_{x\in \mathcal{X}} g(x)\right| \leq \sup_{x\in \mathcal{X}} |f(x)-g(x)|\ ,\end{align}
    Applying this with~$\mathcal{X}$ being the set of unitaries on~$\mathbb{C}^d$, $f(V)=\gateerror_{\cL}(W,\encmap_\cL V \encmap_\cL^{-1})$, $g(V)=\gateerror_{\widetilde{\cL}}(W,\encmap_{\widetilde{\cL}} V \encmap_{\widetilde{\cL}}^{-1})$ in combination with Eq.~\eqref{eq:subspacegateerrordef}  establishes Eq.~\eqref{eq:subspacecontinuityclaim}.
    
\end{proof}

\section{Encoding-independent no-go result for the $\Pgate$-gate\label{sec:nogoresultphase_noencoding}}

In this section we extend the no-go result of Theorem~\ref{lem:nogoP} from the standard encoding of the symmetrically squeezed GKP code~$\gkpcode{\kappa}{\star}{d}$ to arbitrary isometric encodings into~$\gkpcode{\kappa}{\star}{d}$.
Recall (see  Section~\ref{sec:subspacegateerror})
that the subspace logical gate error~$\subspacegateerror{\gkpcode{\kappa}{\star}{d}}{W_\Pgate}$ 
is the smallest achievable logical gate error  on the code space~$\gkpcode{\kappa}{\star}{d}$ over all isometric encoding maps, for any target logical unitary~$U$.
We show that this quantity remains bounded away from zero even in the limit~$\kappa\rightarrow 0$ of large squeezing.

The argument relies on the encoding-independent reduction (Lemma~\ref{lem:subspaceequalityreduction}) using $\ket{\gkp_{\kappa,\Delta}^\varepsilon(0)_d}$ as a test vector to obtain the desired lower bound on the subspace logical gate error.

\begin{theorem}[Encoding-independent no-go result for the linear optics implementation of~$\Pgate$]\label{thm:noencodingnogoP}
Let~$d\geq 2$ be an integer and let~$c_d=d\bmod 2$. Then the subspace logical gate error 
of the linear optics implementation~$W_\Pgate=e^{i(Q^2+c_d\sqrt{2\pi/d}Q)/2}$ of the logical~$\Pgate$-gate
in the symmetrically squeezed GKP code~$\gkpcode{\kappa}{\star}{d}$  satisfies 
\begin{align}
\subspacegateerror{\gkpcode{\kappa}{\star}{d}}{W_\Pgate} \geq 3/100 \label{eq:noencodingnogoPclaim}
\end{align}
for all~$\kappa<1/250$.

In particular, for any family~$\{\encmap_\kappa:\mathbb{C}^d\rightarrow\gkpcode{\kappa}{\star}{d}\}_{\kappa\in (0,\infty)}$ of isometric encoding maps, we have 
 \begin{align}
 \gateerror_{\gkpcode{\kappa}{\star}{d}}(W_\Pgate,\Pgate)\not\rightarrow 0\qquad \textrm{ as }\qquad \kappa\rightarrow 0\ ,
\end{align}
i.e., even in the limit~$\kappa\rightarrow 0$ of large squeezing, the  logical gate error does not  vanish asymptotically.
The identical conclusion, i.e., the same asymptotic no-go result also  holds  for the symmetrically squeezed untruncated GKP code~$\gkpcode{\kappa,\kappa/(2\pi d)}{}{d}$.
 
\end{theorem}
\begin{proof}
Let~$\cL_\kappa:=\gkpcode{\kappa}{\star}{d}$ and let~$\left\{\ket{\gkp_{\kappa,\kappa/(2\pi d)}^{1/(2d)}(k)_d}\right\}_{k\in\mathbb{Z}}$ denote the basis states of~$\cL_\kappa$. Since these basis states form an orthonormal basis of~$\cL_\kappa$, we have for the test vector~$\Phi:=\gkp_{\kappa}^{\star}(0)_d\in\cL_\kappa$ that
\begin{align}
\|\pi_{\cL_\kappa}W_\Pgate\Phi\|^2\ =\ \sum_{j=0}^{d-1}\bigl|\langle \gkp_{\kappa}^{\star}(j)_d,W_\Pgate \gkp_{\kappa}^{\star}(0)_d\rangle\bigr|^2\ .\label{eq:noencodingnogopnormexpansion}
\end{align}
Using that for $\varepsilon \le 1/(2d)$ we have
\begin{align}
\langle \gkp_{\kappa,\Delta}^\varepsilon(j)_d, W_\Pgate \gkp_{\kappa,\Delta}^\varepsilon(0)_d\rangle\ &=\ \langle \gkp_{\kappa,\Delta}^\varepsilon, M_{\sqrt{2\pi d}}^\dagger e^{ij\sqrt{2\pi/d}P}W_\Pgate M_{\sqrt{2\pi d}}\gkp_{\kappa,\Delta}^\varepsilon\rangle\\
&=\ \langle \gkp_{\kappa,\Delta}^\varepsilon, e^{i(j/d)P}e^{i\pi(dQ^2+c_dQ)}\gkp_{\kappa,\Delta}^\varepsilon\rangle\\
&=\ \int_\mathbb{R}\overline{\gkp_{\kappa,\Delta}^\varepsilon(y-j/d)}\,e^{i\pi(dy^2+c_dy)}\,\gkp_{\kappa,\Delta}^\varepsilon(y)\,dy\\
&=\ \delta_{j,0} \cdot \langle \gkp_{\kappa,\Delta}^\varepsilon, e^{i\pi(dQ^2+c_dQ)}\gkp_{\kappa,\Delta}^\varepsilon\rangle ,\label{eq:columnboundme}
\end{align}
where in the penultimate step we used that~$e^{i(j/d)P}$ acts as the position-shift~$y\mapsto y-j/d$ and the last step follows by considering the support of~$\gkp_{\kappa,\Delta}^\varepsilon$.

The remaining term is exactly the matrix element bounded in Lemma~\ref{lem: bound B_00 P}. Under the assumptions~$\kappa<1/4$ and~$d\Delta<1/(4\pi)$ (which are satisfied for~$\kappa<1/4$ since~$d\Delta=\kappa/(2\pi)<1/(4\pi)$), Lemma~\ref{lem: bound B_00 P} yields
\begin{align}
\bigl|\langle \gkp_{\kappa,\Delta}^\varepsilon(0)_d, W_\Pgate \gkp_{\kappa,\Delta}^\varepsilon(0)_d\rangle\bigr|\ \leq\ \frac{1}{\sqrt{1+(\pi d\Delta/\kappa)^2/4}}+\kappa+16(\Delta/\varepsilon)^4\ .\label{eq:diagboundnogono}
\end{align}

Substituting~$\Delta=\kappa/(2\pi d)$ and~$\varepsilon=1/(2d)$ in Eq.~\eqref{eq:diagboundnogono} and combining with Eq.~\eqref{eq:columnboundme} gives
\begin{align}
\sum_{j=0}^{d-1}\bigl|\langle \gkp_{\kappa,\kappa/(2\pi d)}^{1/(2d)}(j)_d,W_\Pgate \gkp_{\kappa,\kappa/(2\pi d)}^{1/(2d)}(0)_d\rangle\bigr|^2\ \leq\ \left(\frac{1}{\sqrt{1+1/(16)}}+\kappa+16(\kappa/\pi)^4\right)^{2}\ ,\label{eq:noencodingnogoplemmasymsubst}
\end{align}
where we used~$d\Delta/\kappa=1/(2\pi)$ and~$\Delta/\varepsilon=\kappa/\pi$. A direct calculation shows that 
\begin{align}
\frac{1}{\sqrt{1+1/16}}+\kappa + 16(\kappa/\pi)^4\ \leq\ \frac{49}{50}\qquad\textrm{for all}\qquad \kappa<\frac{1}{250}\ . \label{eq:numericalboundsubst}
\end{align}
Combining Eqs.~\eqref{eq:noencodingnogopnormexpansion},~\eqref{eq:noencodingnogoplemmasymsubst}, and~\eqref{eq:numericalboundsubst} yields
\begin{align}
\|\pi_{\cL_\kappa}W_\Pgate\Phi\|\ \leq\ \frac{49}{50}\qquad\textrm{for all}\qquad \kappa<\frac{1}{250}\ .\label{eq:noencodingnogoPnormboundtestvec}
\end{align}

By Lemma~\ref{lem:subspaceequalityreduction} (Eq.~\eqref{eq:subspaceequalityreductionclaim}), we have
\begin{align}
\subspacegateerror{\cL_\kappa}{W_\Pgate}\ \geq\ 2\left(1-\min_{\substack{\Psi\in\cL_\kappa\\ \|\Psi\|=1}}\|\pi_{\cL_\kappa}W_\Pgate\Psi\|^{2}\right)^{1/2}\ge 2\left(1 - \min_{\substack{\Psi\in\cL_\kappa\\ \|\Psi\|=1}}\|\pi_{\cL_\kappa}W_\Pgate\Psi\|\right) \, ,\label{eq:noencodingnogoPapply}
\end{align} where we used the inequality~$\sqrt{1-x^2}\geq 1-x$ for~$x\in [0,1]$ in the last step.
Inserting~$\Psi=\Phi = \gkp_{\kappa, \kappa/(2\pi d)}^{1/(2d)}(0)_d$ (which is a unit vector in~$\cL_\kappa$) in the expression on the rhs, and using Eq.~\eqref{eq:noencodingnogoPnormboundtestvec} establishes the Claim~\eqref{eq:noencodingnogoPclaim}.

The statement for the symmetrically squeezed untruncated GKP code follows by continuity. Specifically, let $\widetilde{\cL}_\kappa = \gkpcode{\kappa,\kappa/(2\pi d)}{}{d}$ denote the untruncated code, and $\cL_\kappa = \gkpcode{\kappa}{\star}{d}$ the truncated code. By Lemma~\ref{lem:truncatedapproximateGKPstates} (see also the proof of Corollary~\ref{lem:gateerruntruncatedcode}), the basis states of $\cL_\kappa$ and $\widetilde{\cL}_\kappa$ can be chosen such that $\|\gkp_{\kappa,\kappa/(2\pi d)}^{1/(2d)}(j)_d - \gkp_{\kappa,\kappa/(2\pi d)}(j)_d\| \leq 2\sqrt{\kappa}$ for all $j$. Applying Lemma~\ref{lem:subspacecontinuity} with $\delta = 2\sqrt{\kappa}$, we see that $\subspacegateerror{\widetilde{\cL}_\kappa}{W_\Pgate}\not\rightarrow 0$  if $\subspacegateerror{\cL_\kappa}{W_\Pgate}\not\rightarrow 0$ as $\kappa \to 0$. This implies the claim. 

\end{proof}

\bibliographystyle{unsrturl}
\bibliography{q}

\end{document}